\documentclass{theoretics}

\title{Extensional Taylor Expansion}

\ThCSauthor[Marseille,LIS,I2M,Lyon]{Lison Blondeau-Patissier}{Lison.Blondeau-Patissier@ens-lyon.org}[]
\ThCSauthor[Marseille,LIS]{Pierre Clairambault}{Pierre.Clairambault@cnrs.fr}[0000-0002-3285-6028]
\ThCSauthor[Marseille,I2M]{Lionel {Vaux Auclair}}{Lionel.Vaux@math.cnrs.fr}[0000-0001-9466-418X]
\ThCSaffil[Marseille]{Aix Marseille Univ, CNRS, Marseille, France}
\ThCSaffil[LIS]{LIS,  Marseille, France}
\ThCSaffil[I2M]{I2M,  Marseille, France}
\ThCSaffil[Lyon]{ ENS de Lyon, CNRS, Université Claude Bernard Lyon 1, LIP, Lyon, France}
\ThCSthanks{This work was partially supported by the French ANR projects
	DyVerSe (ANR-19-CE48-0010-01),
	PPS (ANR-19-CE48-0014), LambdaComb (ANR-21-CE48-0017)
	and Reciprog (ANR-21-CE48-0019).
}
\ThCSshortnames{L.\ Blondeau-Patissier, P.\ Clairambault, L.\ Vaux Auclair}  
\ThCSshorttitle{Extensional Taylor Expansion}
\ThCSyear{2026}
\ThCSarticlenum{7}
\ThCSreceived{Sep 17, 2024}
\ThCSrevised{Oct 2, 2025}
\ThCSaccepted{Dec 13, 2025}
\ThCSpublished{Apr 16, 2026}
\ThCSdoicreatedtrue
\ThCSkeywords{Lambda-calculus, Taylor expansion of
lambda-terms, game semantics}

\usepackage{xparse}

\allowdisplaybreaks[1]
\usepackage{mathtools}
\usepackage[all]{xy}
\usepackage{tikz}
\usepackage{tikz-cd}

\usepackage[T1]{fontenc}
\usepackage{amsfonts}
\usepackage{cmll}
\usepackage{stmaryrd}
\usepackage{alphabeta}

\newcommand{\pdfeta}{\texorpdfstring{$η$}{eta}}
\newcommand{\pdflambda}{\texorpdfstring{$λ$}{lambda}}

\usepackage[capitalise,noabbrev,nameinlink]{cleveref}
\crefname{section}{Section}{Sections}
\crefname{figure}{Figure}{Figures}
\crefname{cfigure}{Figure}{Figures}
\crefname{sidefigure}{Figure}{Figures}
\crefname{csidefigure}{Figure}{Figures}
\crefname{fact}{Fact}{Facts}

\newcommand{\eqdef}{\coloneqq}

\DeclareFontFamily{U}{ourstixex}{}
\DeclareSymbolFont{ourstixdelims}{U}{ourstixex}{m}{n}
\DeclareFontShape{U}{ourstixex}{m}{n}{ <-> s * [1.0] stix-mathex }{}
\DeclareMathDelimiter{\lAngle}{\mathopen}{ourstixdelims}{"EC}{ourstixdelims}{"12}
\DeclareMathDelimiter{\rAngle}{\mathclose}{ourstixdelims}{"ED}{ourstixdelims}{"13}

\DeclareFontFamily{U}{ourstixfrak}{}
\DeclareSymbolFont{ourstixsymbols2}{U}{ourstixfrak}{m}{n}
\DeclareFontShape{U}{ourstixfrak}{m}{n}{ <-> s * [1.0] stix-mathfrak }{}
\DeclareMathSymbol{\typecolon}{\mathbin}{ourstixsymbols2}{"25}

\newcommand{\bij}{\simeq}
\newcommand{\intr}[1]{\|#1\|}
\newcommand{\sintr}[1]{\|#1\|}
\newcommand{\lift}{\square}

\newcommand{\definitive}[1]{\textbf{#1}}

\newcommand{\word}[1]{\vec{#1}}
\newcommand{\bag}[1]{\bar{#1}}
\newcommand{\seq}[1]{\vec{#1}}

\newcommand{\st}{\;|\;}

\NewDocumentCommand{\nodelim}{ m }{}

\NewDocumentCommand{\pars}{ O{} m }{#1(#2#1)}
\NewDocumentCommand{\vpars}{ m m m }
	{\IfBooleanTF{#1}{#2(#3#2)}{#3}}

\NewDocumentCommand{\braces}{ O{} m }{#1\{#2#1\}}

\NewDocumentCommand{\brackets}{ O{} m }{#1[#2#1]}

\NewDocumentCommand{\angles}{ O{} m }{#1\langle#2#1\rangle}

\NewDocumentCommand{\vbars}{ O{} m }{#1|#2#1|}

\newcommand{\K}{\mathbb{K}}

\newcommand{\N}{\mathbb{N}}

\newcommand{\Bool}{\mathbb{B}}
\newcommand{\Realpc}{\overline{\mathbb{R}^+}}
\NewDocumentCommand{\Perms}{m}{\mathbb S_{#1}}
\NewDocumentCommand{\DegreeStreams}{}{\N_{\mathrm s}}

\newcommand{\card}{\#}
\NewDocumentCommand{\CardOf}{sO{}m}{\card\vpars{#1}{#2}{#3}}

\NewDocumentCommand{\WordsOf}{sO{}m}{\vpars{#1}{#2}{#3}^*}
\NewDocumentCommand{\SeqsOf}{sO{}m}{\vpars{#1}{#2}{#3}^{\N}}

\newcommand{\Mf}{\mathfrak{M}_{\mathrm{f}}}
\NewDocumentCommand{\MfOf}{O{}m}{\Mf\pars[#1]{#2}}

\newcommand{\Sf}{\mathcal{S}_{\mathrm{f}}}
\NewDocumentCommand{\SfOf}{O{}m}{\Sf\pars[#1]{#2}}

\newcommand{\Ss}{\mathcal{S}}
\NewDocumentCommand{\SsOf}{O{}m}{\Ss\pars[#1]{#2}}

\NewDocumentCommand{\set}{O{}m}{\braces[#1]{#2}}

\NewDocumentCommand{\MRel}{}{\mathbf{MRel}}
\NewDocumentCommand{\ArIso}{}{\mathbf{ArIso}}

\newcommand{\cons}{\mathop{\dblcolon}}
\NewDocumentCommand{\LengthOf}{O{}m}{\vbars[#1]{#2}}
\newcommand{\seqcat}{\,}

\NewDocumentCommand{\tuple}{mO{}}{#2\langle #1 #2\rangle}
\newcommand{\emptyword}{\varepsilon}

\newcommand{\emptybag}{\mset{\,}}
\newcommand{\mset}[2][]{#1[ #2#1]}
\newcommand{\bagcat}{*}

\newcommand{\emptystream}{\iota}

\newcommand{\restrict}{\upharpoonright}
\newcommand{\splitinto}{\lhd}
\NewDocumentCommand{\preImage}{mO{}m}{#1^{-1}\pars[#2]{#3}}
\NewDocumentCommand{\restrictToPre}{mmO{}m}{#1\restrict\preImage{#2}[#3]{#4}}

\NewDocumentCommand{\bagpack}{mO{}}{#2\llbracket #1 #2\rrbracket}
\NewDocumentCommand{\tuplepack}{mO{}}{#2\lAngle #1#2\rAngle}
\NewDocumentCommand{\strpack}{mO{}}{\tuplepack{#1}[#2]}
\newcommand{\conspacknormal}{
    \mathbin{\ooalign{%
        \raise.52ex\hbox{$\scriptstyle\circ$}%
        \cr%
        \raise-.23ex\hbox{$\scriptstyle\circ$}%
        \cr%
        \raise.52ex\hbox{$\mspace{6mu}\scriptstyle\circ$}%
        \cr%
        \raise-.23ex\hbox{$\mspace{6mu}\scriptstyle\circ$}%
        \cr%
    }}}
\newcommand{\conspacksmall}{
    \mathbin{\ooalign{%
        \raise.41ex\hbox{$\scriptscriptstyle\circ$}%
        \cr%
        \raise-.18ex\hbox{$\scriptscriptstyle\circ$}%
        \cr%
        \raise.41ex\hbox{$\mspace{4.5mu}\scriptscriptstyle\circ$}%
        \cr%
        \raise-.18ex\hbox{$\mspace{4.5mu}\scriptscriptstyle\circ$}%
        \cr%
    }}}
\newcommand{\conspack}{\mathbin{\typecolon\!\typecolon}}

\NewDocumentCommand{\support}{O{} m}{\mathrm{supp}\pars[#1]{#2}}
\NewDocumentCommand{\LinearCombinationsOf}{O{\K}m}{#1^{#2}}
\NewDocumentCommand{\SumsOf}{sO{}m}{\Sigma\vpars{#1}{#2}{#3}}
\NewDocumentCommand{\VectorsOf}{O{\K}O{}m}{#1\angles[#2]{#3}}

\NewDocumentCommand{\Coef}{sO{}msO{}m}{\vpars{#1}{#2}{#3}.\vpars{#4}{#5}{#6}}
\NewDocumentCommand{\BagCoef}{sO{}mm}{\vpars{#1}{#2}{#3}^{#4}}

\newcommand{\LVar}{\mathcal{V}}
\newcommand{\SVar}{\LVar_{\mathrm s}}

\NewDocumentCommand{\LambdaTerms}   {}{\Lambda}
\newcommand{\HNFs} {\LambdaTerms_{\mathrm{hn}}}
\newcommand{\BotTerms} {\LambdaTerms_{\bot}}

\NewDocumentCommand{\NormalValueTerms}
				   {}{\Delta_{\mathrm v}^{\mathrm n}}
\NewDocumentCommand{\NormalBaseTerms}    
				   {}{\Delta_{\mathrm b}^{\mathrm n}}
\NewDocumentCommand{\NormalBagTerms}    
				   {}{\Delta_\oc^{\mathrm n}}
\NewDocumentCommand{\NormalStreamTerms}  
				   {}{\Delta_{\mathrm s}^{\mathrm n}}
\NewDocumentCommand{\NormalTerms}
				   {}{\Delta^{\mathrm n}}
\NewDocumentCommand{\ValueTerms}   {}{\Delta_{\mathrm v}}
\NewDocumentCommand{\HeadExprs}    {}{\Delta_{\mathrm h}}
\NewDocumentCommand{\StreamTerms}  {}{\Delta_{\mathrm s}}
\NewDocumentCommand{\BaseTerms}    {}{\Delta_{\mathrm b}}
\NewDocumentCommand{\BagTerms}     {}{\Delta_\oc}
\NewDocumentCommand{\ResourceTerms}{}{\Delta_{\mathrm t}}
\NewDocumentCommand{\ResourceExprs}{}{\Delta_{\mathrm e}}

\NewDocumentCommand{\ValueRule      }{s}{\IfBooleanF{#1}{\pars}{λ}}
\NewDocumentCommand{\BagRule        }{s}{\IfBooleanF{#1}{\pars}{\oc}}
\NewDocumentCommand{\EmptyStreamRule}{s}{\IfBooleanF{#1}{\pars}{\emptystream}}
\NewDocumentCommand{\ConsStreamRule }{s}{\IfBooleanF{#1}{\pars}{\cons}}
\NewDocumentCommand{\RedexBaseRule  }{s}{\IfBooleanF{#1}{\pars}{app}}
\NewDocumentCommand{\NormalBaseRule }{s}{\IfBooleanF{#1}{\pars}{\LVar}}

\NewDocumentCommand{\ValueSums}   {}{\SumsOf{\ValueTerms}}
\NewDocumentCommand{\HeadSums}    {}{\SumsOf{\HeadExprs}}
\NewDocumentCommand{\StreamSums}  {}{\SumsOf{\StreamTerms}}
\NewDocumentCommand{\BaseSums}    {}{\SumsOf{\BaseTerms}}
\NewDocumentCommand{\BagSums}     {}{\SumsOf{\BagTerms}}
\NewDocumentCommand{\ResourceSums}{}{\SumsOf{\ResourceTerms}}
\NewDocumentCommand{\ExprSums}    {}{\SumsOf{\ResourceExprs}}

\NewDocumentCommand{\ValueVectors}   {O{\K}}{\VectorsOf[#1]{\ValueTerms}}
\NewDocumentCommand{\HeadVectors}    {O{\K}}{\VectorsOf[#1]{\HeadExprs}}
\NewDocumentCommand{\StreamVectors}  {O{\K}}{\VectorsOf[#1]{\StreamTerms}}
\NewDocumentCommand{\BaseVectors}    {O{\K}}{\VectorsOf[#1]{\BaseTerms}}
\NewDocumentCommand{\BagVectors}     {O{\K}}{\VectorsOf[#1]{\BagTerms}}
\NewDocumentCommand{\ResourceVectors}{O{\K}}{\VectorsOf[#1]{\ResourceTerms}}
\NewDocumentCommand{\ExprVectors}    {O{\K}}{\VectorsOf[#1]{\ResourceExprs}}

\NewDocumentCommand{\RelTypes}      {oo}{\mathcal D\IfNoValueF{#1}{_{#1}}\IfNoValueF{#2}{^{#2}}}
\NewDocumentCommand{\RelBagTypes}   {o}{\MfOf{\RelTypes[#1]}}
\NewDocumentCommand{\RelStreamTypes}{o}{\SsOf{\RelTypes[#1]}}

\NewDocumentCommand{\ValueTypes} {o}{\RelTypes[\mathrm v][#1]}
\NewDocumentCommand{\BagTypes}   {o}{\RelTypes[\oc][#1]}
\NewDocumentCommand{\StreamTypes}{o}{\RelTypes[\mathrm s][#1]}
\NewDocumentCommand{\BaseTypes}  {o}{\RelTypes[\mathrm b][#1]}

\NewDocumentCommand{\sem}{O{}mo}{#1\llbracket #2#1\rrbracket\IfNoValueF{#3}{_{#3}}}
\NewDocumentCommand{\gs} {O{} m}{\mathcal{G}\pars[#1]{#2}}
\NewDocumentCommand{\emptyrctx}{}{\star}
\NewDocumentCommand{\rctxcat}{}{\bagcat}
\NewDocumentCommand{\basetype}{}{\mathit{o}}
\NewDocumentCommand{\arrowtype}{sO{}msO{}m}{\vpars{#1}{#2}{#3}\to\vpars{#4}{#5}{#6}}
\NewDocumentCommand{\dualtype}{sO{}m}{\vpars{#1}{#2}{#3}\multimap\basetype}

\NewDocumentCommand{\rjug} {}{\vdash}
\NewDocumentCommand{\rjugb}{}{\vdash_\oc}

\NewDocumentCommand{\rjugBase}{}{\vdash_{\mathrm b}}
\NewDocumentCommand{\rjugStr} {}{\vdash_{\mathrm s}}
\NewDocumentCommand{\rjugBag} {}{\vdash_\oc}
\NewDocumentCommand{\rjugVal} {}{\vdash_{\mathrm v}}

\NewDocumentCommand{\type}{}{\mathsf{type}}
\NewDocumentCommand{\typeOf}{O{}m}{\type\pars[#1]{#2}}
\NewDocumentCommand{\ctx}{}{\mathsf{ctx}}
\NewDocumentCommand{\ctxOf}{O{}m}{\ctx\pars[#1]{#2}}
\NewDocumentCommand{\ctxCtx}{m}{\ctx_{#1}}
\NewDocumentCommand{\ctxCtxOf}{m O{}m}{\ctxCtx{#1}\pars[#2]{#3}}
\NewDocumentCommand{\typeable}{m}{\rjug #1}

\NewDocumentCommand{\witnessOf}   {O{}mm}{\mathsf{w}_{#2}\pars[#1]{#3}}
\NewDocumentCommand{\witnessBagOf}{O{}mm}{\mathsf{w}^\oc_{#2}\pars[#1]{#3}}
\NewDocumentCommand{\witnessStrOf}{O{}mm}{\mathsf{w}_{#2}\pars[#1]{#3}}

\newcommand{\Hs}    {\mathbf H^*}
\newcommand{\eqBeta}{=_{\beta}}
\newcommand{\eqBEta}{=_{\beta\eta}}

\newcommand{\eqTay} {=_{\mathcal T}}
\newcommand{\eqExtTay} {=_{\mathcal T_{\eta}}}
\newcommand{\eqBT}  {=_{\mathcal{B}}}
\newcommand{\eqOh}  {=_{\mathrm{hn}}}
\newcommand{\eqRel} {=_{\RelTypes}}

\newcommand{\ssep}{\bowtie}
\newcommand{\sep} {\mathrel{\ltimes}}
\newcommand{\bsep} {\mathrel{{\ltimes}_{\mathcal B}}}

\NewDocumentCommand{\rotate}{m}{ρ_{#1}}

\NewDocumentCommand{\lsubst}{ s O{} m m s O{} m }
	{\partial_{#4}\vpars{#1}{#2}{#3}\cdot\vpars{#5}{#6}{#7}}
\NewDocumentCommand{\subst}{ s O{} m m O{} m }
    {\vpars{#1}{#2}{#3}\braces[#5]{#6/#4}}
\NewDocumentCommand{\rsubst}{ s O{} m m O{} m }
    {\vpars{#1}{#2}{#3}\brackets[#5]{#6/#4}}
\NewDocumentCommand{\appl}{ s O{} m s O{} m }
	{\vpars{#1}{#2}{#3}\,\vpars{#4}{#5}{#6}}
\NewDocumentCommand{\labs}{ m s O{} m }
	{\lambda{#1}.\vpars{#2}{#3}{#4}}
\NewDocumentCommand{\erase}{ s O{} m s O{} m }
    {\vpars{#1}{#2}{#3}\mathbin\lightning \vpars{#4}{#5}{#6}}
\NewDocumentCommand{\shiftdown}{ s O{} m m o}
    {\vpars{#1}{#2}{#3}{[#4\mathop\downarrow]}\IfNoValueF{#5}{^{#5}}}
\NewDocumentCommand{\shiftup}{ s O{} m m o}
    {\vpars{#1}{#2}{#3}{[#4\mathop\uparrow]}\IfNoValueF{#5}{^{#5}}}

\NewDocumentCommand{\SizeOf}{sO{}m}{\#\vpars{#1}{#2}{#3}}
\NewDocumentCommand{\height}  {}    {\mathsf h}
\NewDocumentCommand{\heightOf}{O{}m}{\height\pars[#1]{#2}}
\NewDocumentCommand{\weight}{O{} m}{w\pars[#1]{#2}}
\NewDocumentCommand{\muldeg}{O{} m}{\mathsf{m}\pars[#1]{#2}}
\NewDocumentCommand{\isodeg}{O{} m}{\mathsf{d}\pars[#1]{#2}}
\NewDocumentCommand{\range}  {}    {\mathsf{r}}
\NewDocumentCommand{\rangeOf}{O{}m}{\range\pars[#1]{#2}}
\NewDocumentCommand{\nocc}{m O{} m}{\vbars[#2]{#3}_{#1}}
\NewDocumentCommand{\fv}{O{} m}{fv\pars[#1]{#2}}
\NewDocumentCommand{\SVarOf}{O{} m}{\SVar\pars[#1]{#2}}
\NewDocumentCommand{\LVarOf}{O{} m}{\LVar\pars[#1]{#2}}

\NewDocumentCommand{\tlId} {}{\mathbf{I}}
\NewDocumentCommand{\tlJ}  {}{\mathbf{J}}
\NewDocumentCommand{\tlOne}{}{\mathbf{1}}
\NewDocumentCommand{\tlJaux}{}{\mathbf{R_J}}

\NewDocumentCommand{\trProj}{mo}{\mathsf{p}_{#1\IfNoValueF{#2}{,{#2}}}}
\NewDocumentCommand{\trCCz}{mo}{\mathsf{c}_{#1\IfNoValueF{#2}{,{#2}}}}
\NewDocumentCommand{\trCCp}{m}{\mathsf{c}'_{#1}}
\NewDocumentCommand{\trJaux}{m}{\mathsf{r}_{#1}}

\NewDocumentCommand{\tvJ}{m}{\mathsf{J}_{#1}}
\NewDocumentCommand{\tvJBag}{m}{\mathsf{J}^{\oc}_{#1}}

\NewDocumentCommand{\etaExp}{sO{}m}{
    \vpars{#1}{#2}{#3}^{η}
}

\NewDocumentCommand{\prom}{sO{}m}{
    \IfBooleanTF{#1}{#2(#3#2)}{#3}^{\oc}
}

\NewDocumentCommand{\IdExp}{sO{}m}{
    \vpars{#1}{#2}{#3}^\eta
}

\NewDocumentCommand{\IdExpCoef}{sO{}mm}{
    \Coef{\vpars{#1}{#2}{#3}^\eta}{#4}
}

\NewDocumentCommand{\IdBagExp}{sO{}m}{
    \vpars{#1}{#2}{#3}^{\oc}
}

\NewDocumentCommand{\IdBagExpCoef}{sO{}mm}{
    \Coef{\vpars{#1}{#2}{#3}^{\,\oc}}{#4}
}

\NewDocumentCommand{\IdSeqExp}{sO{}m}{
    \vpars{#1}{#2}{#3}^{\,\eta}
}

\NewDocumentCommand{\IdSeqExpCoef}{sO{}mm}{
    \Coef{\vpars{#1}{#2}{#3}^{\,\eta}}{#4}
}

\NewDocumentCommand{\IdStrExp}{sO{}m}{
    \vpars{#1}{#2}{#3}^{\,\oc}
}

\NewDocumentCommand{\IdStrExpCoef}{sO{}mm}{
    \Coef{\vpars{#1}{#2}{#3}^{\,\oc}}{#4}
}

\NewDocumentCommand{\ExtTayExp}{O{}m}{
    \mathcal T_{\eta}\pars[#1]{#2}
}

\NewDocumentCommand{\ExtTayExpCoef}{O{}mm}{
    \Coef{\ExtTayExp[#1]{#2}}{#3}
}

\NewDocumentCommand{\ExtArgTayExp}{O{}m}{
    \mathcal T_{\eta}^{\oc}\pars[#1]{#2}
}

\NewDocumentCommand{\HeadTayExp}{O{}m}{
    \mathcal T_{h}\pars[#1]{#2}
}

\NewDocumentCommand{\HeadTayExpCoef}{O{}mm}{
    \Coef{\HeadTayExp[#1]{#2}}{#3}
}

\NewDocumentCommand{\HeadArgTayExp}{O{}m}{
    \mathcal T_{h}^{\oc}\pars[#1]{#2}
}

\NewDocumentCommand{\HeadArgsTayExp}{O{}m}{
    \mathcal T_{h}^{\oc}\pars[#1]{#2}
}

\NewDocumentCommand{\TayExp}{O{}m}{
    \mathcal T\pars[#1]{#2}
}

\NewDocumentCommand{\SeqTayExp}{O{}m}{
    \vec{\mathcal T}\pars[#1]{#2}
}

\NewDocumentCommand{\NFTayExp}{O{}m}{
    \mathcal{NT_{\eta}}\pars[#1]{#2}
}

\NewDocumentCommand{\NFTayExpCoef}{O{}mm}{
    \Coef{\NFTayExp[#1]{#2}}{#3}
}

\NewDocumentCommand{\CCBag}   {mO{}m}{\mathsf{c}^-\angles[#2]{#1,#3}}
\NewDocumentCommand{\CCStr}   {mO{}m}{\mathsf{c}^-\angles[#2]{#1,#3}}
\NewDocumentCommand{\CCValCtx}{mO{}m}{\mathsf{c}^+\angles[#2]{#1,#3}}
\NewDocumentCommand{\CCBagCtx}{mO{}m}{\mathsf{c}^+\angles[#2]{#1,#3}}
\NewDocumentCommand{\CCStrCtx}{mO{}m}{\mathsf{c}^+\angles[#2]{#1,#3}}
\NewDocumentCommand{\CCVar}   {mO{}m}{\mathsf{c}  \angles[#2]{#1,#3}}

\NewDocumentCommand{\bred}{}{\to_{\beta}}

\NewDocumentCommand{\resredsym}{}{\mathrm r}
\NewDocumentCommand{\resredz}{}{\mapsto_{\resredsym_0}}
\NewDocumentCommand{\resred}{}{\mapsto_{\resredsym}}
\NewDocumentCommand{\resredR}{}{\resred^?}
\NewDocumentCommand{\sresred}{}{\to_\resredsym}
\NewDocumentCommand{\sresredRT}{}{\sresred^*}

\NewDocumentCommand{\Resredsym}{}{\mathrm R}
\NewDocumentCommand{\Resred}{}{\mapsto_{\Resredsym}}
\NewDocumentCommand{\ResredR}{}{\Resred^?}
\NewDocumentCommand{\sResred}{}{\to_{\Resredsym}}
\NewDocumentCommand{\sResredRT}{}{\sResred^*}

\NewDocumentCommand{\vembed}{m}{\widetilde{#1}}

\NewDocumentCommand{\vresred}{}{\leadsto}
\NewDocumentCommand{\vresredRT}{}{\mathrel{\leadsto\kern-4pt^*}}

\NewDocumentCommand{\etared}{}{\rightarrow_{\eta}}
\NewDocumentCommand{\etarev}{}{\leftarrow_{\eta}}

\NewDocumentCommand{\betaetared}{}{\rightarrow_{\beta\eta}}

\NewDocumentCommand{\resredrulebeta}{s}{\IfBooleanF{#1}{\pars}{\resredsym_\beta}}
\NewDocumentCommand{\resredruleiota}{s}{\IfBooleanF{#1}{\pars}{\resredsym_\iota}}
\NewDocumentCommand{\resredruleabs} {s}{\IfBooleanF{#1}{\pars}{\resredsym_\lambda}}
\NewDocumentCommand{\resredruleappl}{s}{\IfBooleanF{#1}{\pars}{\resredsym_{appL}}}
\NewDocumentCommand{\resredruleappr}{s}{\IfBooleanF{#1}{\pars}{\resredsym_{appR}}}
\NewDocumentCommand{\resredrulebag} {s}{\IfBooleanF{#1}{\pars}{\resredsym_\oc}}
\NewDocumentCommand{\resredrulestrl}{s}{\IfBooleanF{#1}{\pars}{\resredsym_{\cons L}}}
\NewDocumentCommand{\resredrulestrr}{s}{\IfBooleanF{#1}{\pars}{\resredsym_{\cons R}}}

\NewDocumentCommand{\Resredrulebeta}{s}{\IfBooleanF{#1}{\pars}{\Resredsym_\beta}}
\NewDocumentCommand{\Resredruleabs }{s}{\IfBooleanF{#1}{\pars}{\Resredsym_\lambda}}
\NewDocumentCommand{\Resredruleappl}{s}{\IfBooleanF{#1}{\pars}{\Resredsym_{appL}}}
\NewDocumentCommand{\Resredruleappr}{s}{\IfBooleanF{#1}{\pars}{\Resredsym_{appR}}}
\NewDocumentCommand{\Resredrulebag }{s}{\IfBooleanF{#1}{\pars}{\Resredsym_\oc}}
\NewDocumentCommand{\Resredrulestrl}{s}{\IfBooleanF{#1}{\pars}{\Resredsym_{\cons L}}}
\NewDocumentCommand{\Resredrulestrr}{s}{\IfBooleanF{#1}{\pars}{\Resredsym_{\cons R}}}

\NewDocumentCommand{\testresredsym}{}{\tau}
\NewDocumentCommand{\testresred}{}{\mapsto_{\testresredsym}}
\NewDocumentCommand{\stestresred}{}{\to_{\testresredsym}}
\NewDocumentCommand{\cork}{ s O{} m}
    {\tau\vpars{#1}{#2}{#3}}
\NewDocumentCommand{\uncork}{ s O{} m}
    {\bar\tau\vpars{#1}{#2}{#3}}
\NewDocumentCommand{\testcat}{}{\mid}

\NewDocumentCommand{\extToTest}{sO{}m}{\vpars{#1}{#2}{#3}^{\triangleleft}}
\NewDocumentCommand{\testToExt}{sO{}m}{\widetilde{\vpars{#1}{#2}{#3}}}
\NewDocumentCommand{\extIsTest}{sO{}msO{}m}{\vpars{#1}{#2}{#3}\triangleleft\vpars{#4}{#5}{#6}}

\NewDocumentCommand{\testEq}{}{\sim}

\NewDocumentCommand{\NF}{O{} m}{\mathcal{N}\pars[#1]{#2}}
\NewDocumentCommand{\HOf}{O{} m}{\mathcal{H}\pars[#1]{#2}}
\NewDocumentCommand{\HkOf}{m O{} m}{\mathcal{H}^{#1}\pars[#2]{#3}}
\NewDocumentCommand{\HROf}{O{} m}{\mathcal{H}_{\Resredsym}\pars[#1]{#2}}
\NewDocumentCommand{\HRkOf}{m O{} m}{\mathcal{H}_{\Resredsym}^{#1}\pars[#2]{#3}}

\NewDocumentCommand{\Infle}{}  {\le_\bot}

\NewDocumentCommand{\BT}{O{} m}{\mathcal{B}\pars[#1]{#2}}
\NewDocumentCommand{\BTApprox}{O{} m}{\mathcal{B}_f\pars[#1]{#2}}

\NewDocumentCommand{\Bapp}{m} {{τ}_{#1}}
\NewDocumentCommand{\Bsub}{mm}{{σ}_{#2}^{#1}}
\NewDocumentCommand{\Brot}{mm}{{σ}_{#2}^{#1}}

\NewDocumentCommand{\transf}{sO{}mm}{\vpars{#1}{#2}{#3}#4}

\usepackage{ebproof}
\ebproofset{
  right label template={$\scriptstyle(\inserttext)$},
  center=false,
}

\newcommand{\coefA}{α}

\newcommand{\varA}{x}
\newcommand{\varB}{y}
\newcommand{\varC}{z}

\NewDocumentCommand{\svarA}{o}{\vec{\varA}\IfNoValueF{#1}{(#1)}}
\NewDocumentCommand{\svarB}{o}{\vec{\varB}\IfNoValueF{#1}{(#1)}}
\NewDocumentCommand{\svarC}{o}{\vec{\varC}\IfNoValueF{#1}{(#1)}}

\newcommand{\vtermA}{m}
\newcommand{\vtermB}{n}
\newcommand{\vtermC}{p}

\newcommand{\vtermsA}{M}
\newcommand{\vtermsB}{N}
\newcommand{\vtermsC}{P}

\newcommand{\vtermvA}{\vtermsA}
\newcommand{\vtermvB}{\vtermsB}
\newcommand{\vtermvC}{\vtermsC}

\newcommand{\btermA}{a}
\newcommand{\btermB}{b}
\newcommand{\btermC}{c}

\newcommand{\btermsA}{A}
\newcommand{\btermsB}{B}
\newcommand{\btermsC}{C}

\newcommand{\btermvA}{\btermsA}
\newcommand{\btermvB}{\btermsB}
\newcommand{\btermvC}{\btermsC}

\newcommand{\bagA}{\bag{\vtermA}}
\newcommand{\bagB}{\bag{\vtermB}}
\newcommand{\bagC}{\bag{\vtermC}}

\newcommand{\bagsA}{\bag{\vtermsA}}
\newcommand{\bagsB}{\bag{\vtermsB}}
\newcommand{\bagsC}{\bag{\vtermsC}}

\newcommand{\bagvA}{\bagsA}
\newcommand{\bagvB}{\bagsB}
\newcommand{\bagvC}{\bagsC}

\newcommand{\strA}{\seq{\vtermA}}
\newcommand{\strB}{\seq{\vtermB}}
\newcommand{\strC}{\seq{\vtermC\kern1pt}}

\newcommand{\strsA}{\seq{\vtermsA}}
\newcommand{\strsB}{\seq{\vtermsB}}
\newcommand{\strsC}{\seq{\vtermsC}}

\newcommand{\strvA}{\strsA}
\newcommand{\strvB}{\strsB}
\newcommand{\strvC}{\strsC}

\newcommand{\termA}{u}
\newcommand{\termB}{v}
\newcommand{\termC}{w}

\newcommand{\termsA}{U}
\newcommand{\termsB}{V}
\newcommand{\termsC}{W}

\newcommand{\termvA}{\termsA}
\newcommand{\termvB}{\termsB}
\newcommand{\termvC}{\termsC}

\newcommand{\hexprA}{e}
\newcommand{\hexprB}{f}
\newcommand{\hexprC}{g}

\newcommand{\hexprsA}{E}
\newcommand{\hexprsB}{F}
\newcommand{\hexprsC}{G}

\newcommand{\hexprvA}{\hexprsA}
\newcommand{\hexprvB}{\hexprsB}
\newcommand{\hexprvC}{\hexprsC}

\newcommand{\exprA}{q}
\newcommand{\exprB}{r}
\newcommand{\exprC}{s}

\newcommand{\exprsA}{Q}
\newcommand{\exprsB}{R}
\newcommand{\exprsC}{S}

\newcommand{\exprvA}{\exprsA}
\newcommand{\exprvB}{\exprsB}
\newcommand{\exprvC}{\exprsC}

\newcommand{\ltermA}{M}
\newcommand{\ltermB}{N}
\newcommand{\ltermC}{P}

\newcommand{\ltermsA}{\seq{\ltermA}}
\newcommand{\ltermsB}{\seq{\ltermB}}
\newcommand{\ltermsC}{\seq{\ltermC}}

\newcommand{\rtypeA}{α}
\newcommand{\rtypeB}{β}
\newcommand{\rtypeC}{γ}

\newcommand{\rtermA}{ρ}
\newcommand{\rtermB}{σ}
\newcommand{\rtermC}{τ}

\newcommand{\rbagA}{\bag{\rtypeA}}
\newcommand{\rbagB}{\bag{\rtypeB}}
\newcommand{\rbagC}{\bag{\rtypeC}}

\newcommand{\rstrA}{\seq{\rtypeA}}
\newcommand{\rstrB}{\seq{\rtypeB}}
\newcommand{\rstrC}{\seq{\rtypeC}}

\newcommand{\rctxA}{Γ}
\newcommand{\rctxB}{Δ}
\newcommand{\rctxC}{Φ}

\newcommand{\fbtA}{A}

\newcommand{\stypeA}{α}
\newcommand{\stypeB}{β}
\newcommand{\stypeC}{γ}

\newcommand{\dstrA}{\seq{k}}
\newcommand{\dstrB}{\seq{l}}

\NewDocumentCommand{\tctxA}{o}{C[\IfNoValueTF{#1}{\,}{#1}]}

\newcommand{\btrfA}{\tau}
\newcommand{\btrfB}{\sigma}

\newcommand{\pol}{\mathsf{pol}}
\newcommand{\imc}{\rightarrowtriangle}
\newcommand{\ev}[1]{|#1|}
\newcommand{\tensor}{\otimes}
\newcommand{\bigtensor}{\bigotimes}
\newcommand{\tto}{\Rightarrow}
\newcommand{\iso}{\cong}
\newcommand{\qu}{q}
\newcommand{\fold}{\mathsf{fold}}
\newcommand{\unfold}{\mathsf{unfold}}
\newcommand{\curry}{\mathsf{curry}}
\newcommand{\uncurry}{\mathsf{uncurry}}
\newcommand{\pack}{\mathsf{pack}}
\newcommand{\unpack}{\mathsf{unpack}}
\newcommand{\fun}{\mathsf{fun}}
\newcommand{\unfun}{\mathsf{unfun}}
\newcommand{\conf}{\mathcal{C}}
\newcommand{\pconf}{\mathcal{C}_\bullet}
\newcommand{\display}{\partial}
\newcommand{\init}{\mathsf{init}}
\newcommand{\sym}{\cong}
\newcommand{\pos}{\mathsf{Pos}}
\newcommand{\ppos}{\mathsf{Pos}_\bullet}
\newcommand{\ic}[1]{\overline{#1}}
\newcommand{\rep}[1]{\underline{#1}}
\newcommand{\deseq}[1]{\llparenthesis #1 \rrparenthesis}
\newcommand{\Aug}{\mathsf{Aug}}

\newcommand{\Isog}{\mathsf{Isog}}
\newcommand{\pIsog}{\Isog_\bullet}

\NewDocumentCommand{\confOf}{O{}m}{\conf\pars[#1]{#2}}
\NewDocumentCommand{\pconfOf}{O{}m}{\pconf\pars[#1]{#2}}
\NewDocumentCommand{\posOf}{O{}m}{\pos\pars[#1]{#2}}
\NewDocumentCommand{\pposOf}{O{}m}{\ppos\pars[#1]{#2}}
\NewDocumentCommand{\IsogOf}{O{}m}{\Isog\pars[#1]{#2}}
\NewDocumentCommand{\pIsogOf}{O{}m}{\pIsog\pars[#1]{#2}}

\newcommand{\U}{\mathbf{U}}
\newcommand{\emptyconf}{\varnothing}
\newcommand{\emptypos}{\ic{\emptyconf}}

\newcommand{\emptyaug}{\emptyconf}
\newcommand{\emptyisog}{\ic{\emptyconf}}

\newcommand{\q}{\mathsf{q}}
\newcommand{\p}{\mathsf{p}}
\newcommand{\bq}{\mathbf{q}}
\newcommand{\bp}{\mathbf{p}}
\newcommand{\x}{\mathsf{x}}
\newcommand{\y}{\mathsf{y}}
\newcommand{\z}{\mathsf{z}}

\colorlet{termA} {blue!70!green}
\colorlet{termB} {red!80!white}
\colorlet{termB'}{termB!50!white}
\colorlet{termC} {termB!35!white}
\colorlet{termC'}{termB!20!white}
\colorlet{focus} {violet!80!red}
\usetikzlibrary{positioning,calc,shapes.geometric,fadings,fit}
\pgfdeclareshape{open isosceles triangle}{
    \inheritsavedanchors[from=isosceles triangle]
    \inheritanchorborder[from=isosceles triangle]
    \inheritanchor[from=isosceles triangle]{apex}
    \inheritanchor[from=isosceles triangle]{left corner}
    \inheritanchor[from=isosceles triangle]{right corner}
    \inheritanchor[from=isosceles triangle]{north}
    \inheritanchor[from=isosceles triangle]{west}
    \inheritanchor[from=isosceles triangle]{east}
    \inheritanchor[from=isosceles triangle]{south}
    \inheritanchor[from=isosceles triangle]{center}
    \backgroundpath{%
        \trianglepoints%
        {%
            \pgftransformshift{\centerpoint}%
            \pgftransformrotate{\rotate}%
            \pgfpathmoveto{\lowerleft}%
            \pgfpathlineto{\apex}%
            \pgfpathlineto{\lowerrightanchor}%
        }%
    }%
}
\tikzset{
    arbre/.style = { 
        shape = isosceles triangle,
        isosceles triangle apex angle=80,
        shape border rotate = 90,
        draw,
    },
    foret/.style = { 
        shape = trapezium,
        draw,
    },
    terme/.style={fill opacity=.25,text opacity=1},
    arbre infini/.style={arbre,shape = open isosceles triangle,path fading=south},
    triangle etroit/.style={isosceles triangle apex angle=60},
    triangle large/.style={isosceles triangle apex angle=100},
    termA/.style= {terme,fill=termA, draw=termA},
    termB/.style= {terme,fill=termB, draw=termB},
    termB'/.style={terme,fill=termB',draw=termB'},
    termC/.style= {terme,fill=termC, draw=termC},
    termC'/.style= {terme,fill=termC', draw=termC'},
}

\newdir{|>}{%
!/4.5pt/@{|}*:(1,-.2)@{>}*:(1,+.2)@_{>}}

\begin{document}

\maketitle

\begin{abstract}

We introduce a calculus of extensional resource terms.
These are resource terms \emph{à la} Ehrhard--Regnier,
but in infinitely $η$-long form.
The calculus still retains a finite syntax and dynamics:
in particular, we prove strong confluence and normalization. 

Then we define an extensional version of Taylor expansion, mapping ordinary
$λ$-terms to (possibly infinite) linear combinations of extensional resource terms:
like in the ordinary case,
the dynamics of our resource calculus allows us to simulate the $β$-reduction of $λ$-terms;
the extensional nature of this expansion shows in the fact that we are also
able to simulate $η$-reduction.

In a sense, extensional resource terms contain a language of
finite approximants of Nakajima trees, much like ordinary resource terms can be
seen as a richer version of finite Böhm trees.
We show that the equivalence induced on $λ$-terms by the normalization
of extensional Taylor-expansion is nothing but $\Hs$, the greatest consistent
sensible $λ$-theory -- which is also the theory induced by Nakajima trees.
This characterization provides a new, simple way to exhibit models of $\Hs$:
it becomes sufficient to model the extensional resource calculus and its
dynamics.

The extensional resource calculus moreover allows us to recover, in an untyped
setting, a connection between Taylor expansion and game semantics
that was previously limited to the typed setting.
Indeed, simply typed, $η$-long, $β$-normal resource terms are known to be
in bijective correspondence with plays in the sense of Hyland-Ong game
semantics, up to Melliès' homotopy equivalence.
Extensional resource terms are the appropriate counterpart of $η$-long resource
terms in an untyped setting: we spell out the bijection between normal
extensional resource terms and isomorphism classes of augmentations (a
canonical presentation of plays up to homotopy) in the universal arena.

\end{abstract}

\section{Introduction}

The Taylor expansion of $λ$-terms has profoundly renewed the approximation
theory of the $λ$-calculus by providing a quantitative alternative to order
theoretic approximation techniques, the latter being famously embodied in the
notion of Böhm tree~\cite{DBLP:books/daglib/0067558}:
a key result of Ehrhard and Regnier’s seminal series of
papers~\cite{DBLP:journals/tcs/EhrhardR08,DBLP:conf/cie/EhrhardR06}
is the fact that the normal form of the Taylor expansion of a $λ$-term
is the Taylor expansion of its Böhm tree.
Taylor expansion can thus be seen as mediating
between the potentially infinite dynamics of finite \(λ\)-terms,
and the static but potentially infinite Böhm trees.
In order to expose the context and motivations of our contributions, we find
useful to first review these notions, only assuming knowledge of the ordinary
$λ$-calculus.

\subsection{Böhm trees and ordinary Taylor expansion}

The results we present in this subsection are digested from well-established
literature.\footnote{
For a comprehensive treatment of the theory of Böhm trees, the reader may refer
to classic textbooks such as Barendregt's~\cite{DBLP:books/daglib/0067558} or
Krivine's~\cite{krivine}.
The first chapters of Barendregt and Manzonetto's
\emph{Satellite}~\cite[esp.\ Section 2.3]{DBLP:books/cp/BarendregtM22}
offer a modern, self-sufficient survey of the theory.
Ehrhard and Regnier obtained their seminal commutation theorem, relating Taylor
expansion with Böhm trees via normalization (\cref{eqn:TayExpBT} below),
by combining the results of two
papers~\cite{DBLP:journals/tcs/EhrhardR08,DBLP:conf/cie/EhrhardR06}.
Both papers involve considerable technical developments to unveil deep,
distinctive properties of the Taylor expansion of $λ$-terms
(a uniformity property of the support of Taylor expansion, an explicit formula
for coefficients, a precise connection with execution in an abstract machine)
of which the commutation theorem is but a consequence.
The third author showed that the same theorem could be established
in a more direct fashion, by simulating \(β\)-reduction through Taylor
expansion~\cite{DBLP:journals/lmcs/Vaux19}:
although we do not even sketch a proof of the commutation theorem,
our exposition of ordinary Taylor expansion is inspired by that latter route,
as we will leverage similar techniques in our treatment of the extensional case.
}
For the reader discovering one of the subjects, or both, it can serve as 
a very quick and opinionated survey.
And for the reader versed in both subjects, as well as for the newcomer,
we hope it will convey intuitions that they can advantageously summon up
when we turn to the main matter of the paper.
The expert reader might still prefer to jump directly to 
\cref{section:intro:eta}, where we discuss the literature,
and some less established folklore, about extensionality 
in relation to Taylor expansion;
or even to \cref{section:intro:contrib}, where we outline our contributions.

\paragraph{Böhm trees.}
We can always write a $λ$-term as
\(\ltermA=\labs{\varA_1}{\cdots\labs{\varA_k}{\appl{\ltermA'}{\ltermB_1\cdots\ltermB_l}}}\),%
\footnote{
  We use standard notational conventions to avoid
  repetition of parentheses:
  application has precedence over abstraction,
  and we associate applications on the left.
  So
  $\labs{\varA_1}{\cdots\labs{\varA_k}{\appl{\ltermA'}{\ltermB_1\cdots\ltermB_l}}}$
  should be read as
  $\labs{\varA_1}*[\big]{\cdots\labs{\varA_k}*[\big]{\appl*{\appl{\ltermA'}{\ltermB_1}}{\cdots\ltermB_l}}\cdots}$.
}
where:
\begin{itemize}
  \item either $\ltermA'=\varB$ is the \definitive{head variable} of
    $\ltermA$, and then $\ltermA$ is in \definitive{head normal form};
  \item or $\ltermA'=\appl*{\labs{\varC}{\ltermC}}{\ltermB_0}$
    is the \definitive{head redex} of $\ltermA$,
    in which case the \definitive{head reduction strategy}
    deterministically reduces $\ltermA$ to
    $\HOf{\ltermA}\eqdef\labs{\varA_1}{\cdots\labs{\varA_k}{\appl{\subst{\ltermC}{\varC}{\ltermB_0}}{\ltermB_1\cdots\ltermB_l}}}$,
    where $\subst{\ltermC}{\varC}{\ltermB_0}$ denotes the usual capture avoiding substitution.
\end{itemize}

Head reduction plays a central rôle in the theory, chiefly because 
head normalizable terms are exactly those terms that are \definitive{solvable}:
informally, a term $\ltermA$ is solvable when it can interact with its
evaluation context via normalization;
a possible definition is to require the existence of a context
$\tctxA$ of the shape
$\appl*{\labs{\varC_1}{\cdots\labs{\varC_m}{[\,]}}}{\ltermC_1\cdots\ltermC_n}$
such that $\tctxA[\ltermA]$ $β$-normalizes to the identity term
$\labs{\varA}{\varA}$.
\begin{fact}
  \label{fact:solvability}
  The following three properties are equivalent:
  \begin{enumerate}[label=(\roman*)]
    \item the sequence of head reductions starting from $\ltermA$ is finite
      ($\ltermA$ is \definitive{head normalizable});
    \item $\ltermA$ is $β$-equivalent to a head normal form;
    \item $\ltermA$ is solvable.
  \end{enumerate}
\end{fact}
Conversely, an \definitive{unsolvable} term is one whose structure cannot be probed
by the environment via normalization: applying an unsolvable term to
an argument, or substituting a variable in that term with any other term, will
never yield a head normal form (let alone a normal form).

\begin{sidefigure}
  \includegraphics[page=1,width=.55\textwidth]{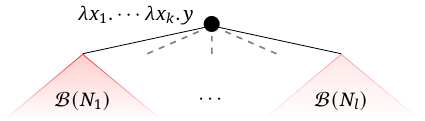}
    \caption{Shape of a (non-\(\bot\)) Böhm tree}
    \label{fig:bohm}
\end{sidefigure}    
\begin{csidefigure}
   \includegraphics[page=2,width=.45\textwidth]{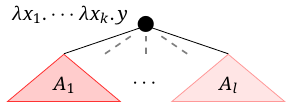}
    \caption{Shape of a (non-\(\bot\)) approximant}
    \label{fig:approx}
\end{csidefigure}

The \definitive{Böhm tree} of a term $\ltermA$ is then a possibly infinite tree
$\BT{\ltermA}$, defined coinductively:
\begin{itemize}
  \item if $\ltermA$ head normalizes to 
    $\labs{\varA_1}{\cdots\labs{\varA_k}{\appl{\varB}{\ltermB_1\cdots\ltermB_l}}}$
    then 
    $\BT{\ltermA}\eqdef
    \labs{\varA_1}{\cdots\labs{\varA_k}{\appl{\varB}{\BT{\ltermB_1}\cdots\BT{\ltermB_l}}}}$
    which we consider as a tree whose root is labelled with the abstractions
    and head variable, and with $l$ immediate subtrees, depicted as in
    \cref{fig:bohm}
    (we leave the bottom of each triangle open to indicate that the
    tree is possibly infinite);
  \item if $\ltermA$ is unsolvable, then 
    $\BT{\ltermA}$ is reduced to a leaf denoted $\bot$.
\end{itemize}
Like $λ$-terms, Böhm trees are considered up to $α$-equivalence.\footnote{
  Although they are possibly infinite trees,
  the set of free variables of any subtree of $\BT{\ltermA}$ is always finite,
  so this poses no particular difficulty.
}
\begin{fact}
Writing $\ltermA\eqBT\ltermA'$ when $\BT{\ltermA}=\BT{\ltermA'}$,
we obtain a $λ$-theory, \emph{i.e.}~a congruence on $λ$-terms
containing $β$-conversion.
This $λ$-theory is moreover \definitive{sensible}: it equates all unsolvable terms.
\end{fact}

It turns out that the only difficult part in establishing 
the previous fact is to show that $\eqBT$ is compatible with application:
if $\ltermA\eqBT\ltermA'$ and $\ltermB\eqBT\ltermB'$ then
$\appl{\ltermA}{\ltermB}\eqBT\appl{\ltermA'}{\ltermB'}$,
which amounts to showing that $\BT{\appl{\ltermA}{\ltermB}}$ is determined
by the sole information of $\BT{\ltermA}$ and $\BT{\ltermB}$.
A standard route to establish the compatibility of $\eqBT$ with syntactic
constructs is to rely on \definitive{finite approximants} of Böhm trees.
The latter are particular \(β\)-normal terms of $\BotTerms$, the $λ$-calculus
augmented with the “undefined” constant $\bot$, whose shape follows that of
Böhm trees:
\begin{itemize}
  \item $\bot$ is an approximant;
  \item if $\fbtA_1,\dotsc,\fbtA_l$ are approximants, then so is
    $\labs{\varA_1}{\cdots\labs{\varA_k}{\appl{\varB}{\fbtA_1\cdots\fbtA_l}}}$,
    which we depict as in \cref{fig:approx}
    (we use closed triangles for finite trees).
\end{itemize}

These approximants are thus both terms of $\BotTerms$
and finite Böhm-like trees (\emph{i.e.}\ trees with internal nodes as 
in \cref{fig:bohm}, and \(\bot\) leaves).
The \definitive{information order} $\Infle$ is defined both on $\BotTerms$ and on Böhm-like trees,
as the partial order such that $\ltermA\Infle\ltermB$ when $\ltermA$ is
obtained from $\ltermB$ by replacing any number of subterms with $\bot$
(possibly infinitely many in case $\ltermA$ is an infinite Böhm-like tree).
The set of finite approximants of a $λ$-term $\ltermA$ is then 
\[\BTApprox{\ltermA}=\set{\fbtA\st \exists \ltermA'\eqBeta\ltermA,\, \fbtA\Infle\ltermA'}\,.\]
The \emph{syntactic approximation}
theorem~\cite[Theorem 2.32]{DBLP:books/cp/BarendregtM22} states that
$\BT{\ltermA}$ is nothing but the supremum (in the directed-complete partial
order of Böhm-like trees equipped with $\Infle$) of $\BTApprox{\ltermA}$.
And the \emph{syntactic continuity}
theorem~\cite[Proposition 2.34]{DBLP:books/cp/BarendregtM22} establishes that 
the notion of approximant is compatible with syntactic constructs:
given a term $\ltermA$ and a context $\tctxA$, $\BTApprox{\tctxA[\ltermA]}$
depends only on $\BTApprox{\ltermA}$ and $\tctxA$.\footnote{
  Another well-known route to show that $\eqBT$ is a congruence,
  also discussed in the \emph{Satellite} \cite{DBLP:books/cp/BarendregtM22},
  is via infinitary $λ$-calculi:
  in some sense (that can be made formal \cite{DBLP:journals/tcs/KennawayKSV97}),
  $\BT{\appl{\ltermA}{\ltermB}}$ is the normal form,
  for an infinitary extension of \(β\)-reduction, of the application
  $\appl{\BT{\ltermA}}{\BT{\ltermB}}$.
}

\paragraph{The resource calculus.}
We have seen that the approximants associated with the Böhm tree interpretation
can be considered as partial $λ$-terms in normal form.
By contrast, the target of the Taylor expansion of $λ$-terms is supported by a
language of approximants called \definitive{resource terms}, that are
not necessarily normal, but whose reduction behaves linearly.
These are just like ordinary $λ$-terms, except for the application constructor:
a resource term $\vtermA$ is applied not just to one argument, but to a
bag (a finite multiset) of arguments $\bagB=\mset{\vtermB_1,\dotsc,\vtermB_l}$,
yielding a new term $\appl{\vtermA}{\bagB}$, which we may depict as 
in \cref{fig:resource:appl}
(we use trapezia for forests, representing bags).

\begin{sidefigure}
   \includegraphics[page=3,width=0.3\textwidth]{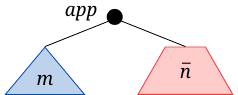}
    \caption{Depiction of \(\appl{\vtermA}{\bagB}\)}
    \label{fig:resource:appl}
    \end{sidefigure}
    \begin{sidefigure}
       \includegraphics[page=4,width=0.4\textwidth]{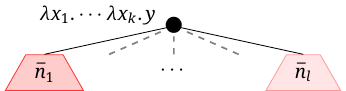}
    \caption{Shape of a normal resource term}
    \label{fig:resource:normal}
    \end{sidefigure}

Resource terms retain a dynamics, induced by a \emph{linear} variant 
of $β$-reduction: if $\bagB={\mset{\vtermB_1,\dotsc,\vtermB_l}}$
and $\varA_1,\dotsc,\varA_k$ enumerate the 
occurrences of $\varA$ in $\vtermA$,
then $\appl*{\labs{\varA}{\vtermA}}{\bagB}$
reduces to $\rsubst{\vtermA}{\varA}{\bagB}$,
the \definitive{resource substitution} of $\bagB$ for $\varA$ in $\vtermA$,
defined as the finite sum of resource terms
\[
\sum_{σ\in\Perms k}
\subst{\vtermA}{\varA_{σ(1)},\dotsc,\varA_{σ(k)}}{\vtermB_1,\dotsc,\vtermB_k}
\]
in case \(k=l\)
(\(\Perms k\) denotes the set of permutations of \(\set{1,\dotsc,k}\)),
or the empty sum \(0\) in case \(k\not=l\):
each summand of \(\rsubst{\vtermA}{\varA}{\bagB}\)
is the result of a one-to-one substitution of the elements of \(\bagB\) (taking
multiplicities into account) for the occurrences of \(\varA\) in \(\vtermA\).
The dynamics thus involves finite formal sums of expressions,
and all syntactic constructs are extended to sums by linearity:
e.g., application is bilinear
\(\appl*{\sum_{i\in I}\vtermA_i}*{\sum_{j\in J}\bagB_j}
\eqdef \sum_{i\in I}\sum_{j\in J}\appl{\vtermA_i}{\bagB_j}\).
The base case of reduction is then extended to a \definitive{resource
reduction} relation on finite sums of resource terms by allowing to fire a
redex in any context and inside sums:
for instance, if \(\vtermA\) reduces to \(\vtermsA'=\sum_{i\in I}\vtermA'_i\),
then \(\labs{\varA}{\vtermA}\) reduces to
\(\labs{\varA}{\vtermsA'}=\sum_{i\in I}\labs{\varA}{\vtermA'_i}\),
and \(\vtermA+\vtermsB\) reduces to \(\vtermsA'+\vtermsB\),
for any finite sum \(\vtermsB\).
Again, reducing any redex \(\appl*{\labs{\varA}{\vtermA}}{\bagB}\)
such that the cardinality of \(\bagB\) does not match the number of occurrences
of \(\varA\) in \(\vtermA\) will yield \(0\):
by linearity, this will annihilate the summand containing this redex.

In any summand $\vtermC'$ of $\rsubst{\vtermA}{\varA}{\bagB}$, the elements of
\(\bagB\) are substituted for variable occurrences in \(\vtermA\) but never
duplicated:
it follows that the size (\emph{i.e.}\ the number of syntactic constructs) of
$\vtermC'$ is strictly smaller than that of the redex
$\appl*{\labs{\varA}{\vtermA}}{\bagB}$.
It is then easy to establish that:
\begin{fact}
  Resource reduction is confluent and strongly normalizing.
\end{fact}
So any finite sum of resource terms \(\vtermsA\) reduces to a unique normal
form \(\NF{\vtermsA}\), in such a way that \(\NF{\vtermsA}
=\sum_{i\in I}\NF{\vtermA_i}\) whenever \(\vtermsA=\sum_{i\in I}\vtermA_i\).
A normal resource term is necessarily of the shape
\(\labs{\varA_1}{\cdots\labs{\varA_k}{\appl{\varB}{\bagB_1\cdots\bagB_l}}}\)
-- see \cref{fig:resource:normal}

Observe that non-$\bot$ Böhm approximants are nothing but normal resource terms
with bags of size at most one:
more precisely, any non-$\bot$ approximant corresponds to a normal resource
term using the empty multiset for any $\bot$ subterm, and singleton multisets
for the non-$\bot$ subterms.
Bags of arbitrary size are nonetheless essential for the resource calculus
to also provide an approximation of $β$-reduction and normalization
via Taylor expansion, as we outline below.

\paragraph{Taylor expansion at work.}
The Taylor expansion $\TayExp{\ltermA}$ of a $λ$-term $\ltermA$
is the vector (\emph{i.e.}~the possibly infinite linear combination) of resource terms
inductively defined by:
\[
  \TayExp{\varA} \eqdef \varA
  \qquad
  \TayExp{\labs{\varA}{\ltermA}} \eqdef\labs{\varA}{\TayExp{\ltermA}}
  \qquad
  \TayExp{\appl{\ltermA}{\ltermB}}\eqdef
  \sum_{k\in\N}\frac{1}{k!}\appl{\TayExp{\ltermA}}{{\TayExp{\ltermB}}^k}
\]
where, again, syntactic constructs are extended to arbitrary
weighted sums of terms by linearity, and
\(\TayExp{\ltermB}^k\eqdef\mset{\TayExp{\ltermB},\dotsc,\TayExp{\ltermB}}\)
is a weighted sum of bags all of size $k$ -- and \(k!\) is the cardinality
of \(\Perms{k}\).%
\footnote{
  Note that the case of application in the definition of Taylor expansion is
  nothing but the usual formula defining the Taylor series of an infinitely
  differentiable map at $0$, provided one interprets the resource application
  $\appl{\vtermvA}{\mset{\vtermvB_1,\dotsc,\vtermvB_k}}$
  (where $\vtermvA,\vtermvB_1,\dotsc,\vtermvB_k$ denote $λ$-terms,
  or vectors of resource terms) as the application of the $k$-th
  derivative at $0$ of $\vtermvA$ to
  the tuple $\tuple{\vtermvB_1,\dotsc,\vtermvB_k}$,
  this application being $k$-linear and symmetric.
  The resource calculus is precisely the fragment of the differential
  $λ$-calculus \cite{DBLP:journals/tcs/EhrhardR03} supporting the target of
  this recursive Taylor expansion.
  We will briefly discuss this analytic interpretation again at 
  the end of this subsection.
}
Here we choose to dispense with the technicalities of dealing with infinite sums by
considering coefficients in the extended half line \([0,+\infty]\),
or more generally in any suitably complete semiring -- a precise
definition of the necessary structure will be recalled in 
\cref{section:resourcevectors}.

Given a vector of resource terms \(\vtermsA\),
it is often useful to consider the \definitive{promotion} of \(\vtermsA\),
which is the vector of bags defined as
\(\prom{\vtermsA}\eqdef \sum_{k\in\N}\frac{1}{k!}{\vtermsA}^k\).
In particular, the Taylor expansion of an application can be written as
\(\TayExp{\appl{\ltermA}{\ltermB}}=\appl{\TayExp{\ltermA}}{\prom{\TayExp{\ltermB}}}\)
(the application on the right hand side being that of the resource calculus).
The crucial feature of Taylor expansion is that it allows us to decompose the
usual substitution operation on \(λ\)-terms into resource substitution and promotion:
\begin{equation}
  \label{eqn:substitution}
  \TayExp{\subst{\ltermA}{\varA}{\ltermB}}=
  \rsubst{\TayExp{\ltermA}}{\varA}{\prom{\TayExp{\ltermB}}}
  \,.
\end{equation}
The proof of \cref{eqn:substitution} is by a simple induction on \(\ltermA\),
relying on basic combinatoric arguments in the case of an application
\cite[Lemmas 4.3 and 4.7]{DBLP:journals/lmcs/Vaux19}.\footnote{
  The original proof by Ehrhard and Regnier \cite[Theorem 32]{DBLP:journals/tcs/EhrhardR08}
  follows a more contorted path, because they insist on establishing an
  explicit formula for the coefficients of resource terms in Taylor expansions.
}
Now, we can consider the normal form of any vector as defined by linearity
\(\NF{\sum_{i\in I} \coefA_i\vtermA_i}\eqdef\sum_{i\in I} \coefA_i\NF{\vtermA_i}\)
(where $I$ is not necessarily finite), and then write $\ltermA\eqTay\ltermB$
when $\NF{\TayExp{\ltermA}}=\NF{\TayExp{\ltermB}}$.

Showing that the equivalence relation \(\eqTay\) is a \(λ\)-theory is easy.
Indeed, its compatibility with syntactic constructs follows from the
confluence property of resource reduction:
in particular, \(\NF{\appl{\vtermA}{\mset{\vtermB_1,\dotsc,\vtermB_k}}}\)
is also the normal form of 
\(\appl{\NF{\vtermA}}{\mset{\NF{\vtermB_1},\dotsc,\NF{\vtermB_k}}}\),
which ensures that 
\(\NF{\TayExp{\appl{\ltermA}{\ltermB}}}\)
is also the normal form of 
\( \appl{\NF{\TayExp{\ltermA}}}{\prom{\NF{\TayExp{\ltermB}}}} \),
thus settling the case of application straightforwardly.
Similarly, thanks to \cref{eqn:substitution},
\(\NF{\TayExp{\subst{\ltermA}{\varA}{\ltermB}}}\)
is the normal form of 
\(\rsubst{\TayExp{\ltermA}}{\varA}{\prom{\TayExp{\ltermB}}}\)
hence of 
\(\appl*{\labs{\varA}{\TayExp{\ltermA}}}{\prom{\TayExp{\ltermB}}}
=\TayExp{\appl*{\labs{\varA}{\ltermA}}{\ltermB}}\),
which ensures that \(\eqTay\) contains \(β\)-reduction.

Showing that \(\eqTay\) is also sensible provides a good example of
Taylor expansion at work.
More precisely, we show that:

\begin{lemma}
  \label{lemma:taylor:solvable}
  A \(λ\)-term \(\ltermA\) is head normalizable iff \(\NF{\TayExp{\ltermA}}\not=0\).
\end{lemma}

\begin{proof}
First observe that if \(\ltermA\) is in head normal form,
then \(\TayExp{\ltermA}\) contains a resource term of the shape
\(\labs{\varA_1}{\cdots\labs{\varA_k}{\appl{\varA}{\emptybag\cdots\emptybag}}}\),
which is normal, so \(\NF{\TayExp{\ltermA}}\not=0\).
Now, if \(\ltermA\) is head normalizable, then \(\ltermA\) reduces to some
head normal form \(\ltermA'\),
and then \(\ltermA\eqTay\ltermA'\) because \(\eqTay\) contains \(β\)-reduction,
so \(\NF{\TayExp{\ltermA}}=\NF{\TayExp{\ltermA'}}\not=0\).

For the reverse implication, we start by observing that Taylor expansion
commutes with head reduction, which we define on resource terms in the same way
as on \(λ\)-terms:
\[
\HOf{\labs{\varA_1}{\cdots\labs{\varA_k}{\appl*{\labs{\varB}{\vtermA}}{\bagB_0\cdots\bagB_l}}}}\eqdef
\labs{\varA_1}{\cdots\labs{\varA_k}{\appl{\rsubst{\vtermA}{\varB}{\bagB_0}}{\bagB_1\cdots\bagB_l}}}
\,.
\]
Then we define \(\HOf{\vtermvA}\) for any weighted sum \(\vtermvA\) 
of head reducible terms by linearity, and
\cref{eqn:substitution} gives
\(\TayExp{\HOf{\ltermA}}=\HOf{\TayExp{\ltermA}}\)
for any \(λ\)-term \(\ltermA\) not in head normal form.
Now, assuming \(\NF{\TayExp{\ltermA}}\not=0\), we can pick an element \(\vtermA\)
in the support of \(\TayExp{\ltermA}\) such that \(\NF{\vtermA}\not=0\),
and then we show by induction on the size of \(\vtermA\)
that \(\ltermA\) is head normalizable:
either \(\ltermA\) is already in head normal form;
or at least one resource term \(\vtermA'\) in \(\HOf{\vtermA}\)
is such that \(\NF{\vtermA}\not=0\), and we observe that \(\vtermA'\) is in
\(\HOf{\TayExp{\ltermA}}=\TayExp{\HOf{\ltermA}}\), which ensures
that \(\HOf{\ltermA}\) is head normalizable by induction hypothesis.
\end{proof}

Observe that the previous argument actually provides a proof of
the implication from (ii) to (i) in \cref{fact:solvability} --
this implication and the one from (iii) to (i) are the only ones that are not
easy consequences of the definitions.
This implication is classically proved by 
standardization~\cites[Corollary 11.4.8]{DBLP:books/daglib/0067558}[Corollary 2.7]{DBLP:journals/iandc/Takahashi95}
or reducibility techniques~\cite[Theorem 4.9]{krivine}.
This demonstrates a nice conceptual contribution of Taylor expansion:
instead of reasoning on reduction
paths, we can pick a well-chosen element of the Taylor expansion and use it as
a decreasing measure for a proof by induction, exploiting the fact that
reduction in the resource calculus reflects $β$-reduction.

\paragraph{Taylor expansion as an alternative to Böhm trees.}
The sensible $λ$-theories $\eqBT$ and $\eqTay$ actually coincide.
Indeed, it is straightforward to extend the definition of Taylor expansion to
Böhm-like trees, in such a way that \(\TayExp{\bot}=0\), and
\(\BTApprox{\ltermA}\) is the set of finite approximants
occurring (as resource terms) in \(\TayExp{\BT{\ltermA}}\):
\(\BT{\ltermA}\) is entirely determined by \(\TayExp{\BT{\ltermA}}\).
Ehrhard and Regnier’s commutation theorem~\cite[Corollary 1]{DBLP:conf/cie/EhrhardR06}
establishes the identity:
\begin{equation}
  \label{eqn:TayExpBT}
  \NF{\TayExp{\ltermA}}=\TayExp{\BT{\ltermA}}
\end{equation}
which ensures that \(\ltermA\eqTay\ltermA'\) iff \(\ltermA\eqBT\ltermA'\).
We do not develop the proof of \cref{eqn:TayExpBT}:
the reader may refer to Ehrhard and Regnier's papers for the original
proof~\cite{DBLP:journals/tcs/EhrhardR08,DBLP:conf/cie/EhrhardR06},
or to the arguably more direct approach by the third author~\cite{DBLP:journals/lmcs/Vaux19},
based on the simulation of \(β\)-reduction.
One can thus view Taylor expansion as a practical alternative to Böhm trees:
the normal form of Taylor expansion subsumes the approximation theory
determined by Böhm trees,
but Taylor expansion also provides non-normal approximants, together with an
analysis of \(β\)-reduction via resource reduction.

Barbarossa and Manzonetto~\cite{DBLP:journals/pacmpl/BarbarossaM20} have
demonstrated at length how to leverage this approach to revisit old
results and establish new ones, in a generic and principled way,
systematically reasoning inductively on the size (or the length of a particular
reduction path) of a well-chosen resource term.

It is worth noting that, as these authors focus on
the order-based approximation theory of the \(λ\)-calculus,
this latter work only relies on the qualitative version of Taylor expansion,
that is obtained by replacing \(\TayExp{\ltermA}\) with its support set --
or, equivalently, by taking scalar coefficients in the boolean semiring.
This is all the more justified because, by Ehrhard and Regnier's uniformity
results~\cite{DBLP:journals/tcs/EhrhardR08}, the Taylor expansion of a pure
\(λ\)-term is entirely characterized by its support.
However, the quantitative information of coefficients underpins the analytic
interpretation of Taylor expansion as a sum of iterated derivatives,
and the full (quantitative) Taylor expansion is strictly more
informative as soon as one departs from that uniform setting.

\paragraph{Beyond the pure \(λ\)-calculus and plain \(β\)-reduction.}

Indeed, one strength of Taylor expansion as a framework for programming language
semantics is its modularity, which is essentially inherited from its origins in
\emph{quantitative semantics},
as initiated by Girard~\cite{DBLP:journals/apal/Girard88} and later revisited
by Ehrhard in a typed setting~\cite{DBLP:journals/mscs/Ehrhard05}:
the basic idea of quantitative semantics is to interpret \(λ\)-terms as
generalized power series,
associated with analytic maps between spaces of some suitable category.
Although it will play no explicit rôle in the remainder of the paper,
this analytic interpretation was crucial for the design of the
differential $λ$-calculus and the Taylor expansion of $λ$-terms by Ehrhard and
Regnier:
these models satisfy the usual Taylor formula as an identity,
of which the Taylor expansion of programs can be understood as a syntactic,
computational counterpart.

It then becomes natural to account for various flavours of superposition of
programs via sums.
For instance in a discrete probabilistic setting, one can turn the
probabilistic choice \(\ltermA\oplus_p\ltermB\) (representing a choice between
\(\ltermA\) with probability \(p\) and \(\ltermB\) with probability \(1-p\))
into a weighted sum:
\(\TayExp{\ltermA \oplus_p \ltermB}\eqdef p\TayExp{\ltermA} +(1-p)\TayExp{\ltermB}\).
Dal Lago and Leventis have shown that, again, this corresponds to a notion of
probabilistic Böhm tree~\cite{DBLP:conf/rta/LagoL19}, via normalization:
it is notable that this extension is straightforward on the side of Taylor
expansion, whereas the development of an adequate notion of probabilistic Böhm
trees by Leventis required considerable technical
work~\cite{DBLP:conf/lics/Leventis18}.

Taylor expansion moreover enjoys a tight connection with linear
logic~\cite{DBLP:journals/tcs/Girard87}, which was also founded on quantitative
semantics:
Ehrhard's version of quantitative semantics~\cite{DBLP:journals/mscs/Ehrhard05}
is actually a denotational model of linear logic, and it is possible to
introduce a differential version of linear
logic~\cite{DBLP:journals/tcs/EhrhardR06}, together with a notion of Taylor
expansion which reflects the structure of the model, and refines the Taylor
expansion of \(λ\)-terms.
The paradigm of Taylor expansion can then be ported to various extensions or
variants of the \(λ\)-calculus, and more generally to systems that “play well”
with linear logic~\cite{DBLP:journals/entcs/Chouquet19,DBLP:conf/csl/ChouquetT20,
DBLP:conf/fscd/DufourM24},
at the price of designing a resource calculus providing a suitable
linearization of the source system.

It is reasonable to expect that such notions of Taylor expansion will
yield interesting and robust approximation theories via normalization.
For instance,
Kerinec, Manzonetto and Pagani~\cite{DBLP:journals/lmcs/KerinecMP20}
have followed this path precisely, in the case of the call-by-value
\(λ\)-calculus, and this guided their definition of an adequate notion of
call-by-value Böhm tree.

\subsection{Towards extensionality}\label{section:intro:eta}

Up to the present paper, one notable case falls outside of the scope that we
have just delineated: extensionality and the \(η\)-rule.
A \(λ\)-theory \(\sim\) is \definitive{extensional} when \(\ltermA\sim\ltermA'\)
as soon as, for each term \(\ltermB\), \(\appl{\ltermA}{\ltermB}\sim\appl{\ltermA'}{\ltermB}\).
Equivalently, \(\sim\) is extensional if it contains the \definitive{\(η\)-rule},
reducing \(\labs{\varA}{\appl{\ltermA}{\varA}}\) to \(\ltermA\)
when \(\varA\) is fresh, \emph{i.e.}\ not a free variable of \(\ltermA\).
The least extensional \(λ\)-theory is thus the reflexive, symmetric, and
transitive closure \(\eqBEta\) of the union of \(β\)- and \(η\)-reductions.

\paragraph{Extensionality via a global transformation.}
From the viewpoint of ordinary Taylor expansion, and in contrast with
\(β\)-reduction, \(η\)-reduction cannot be captured as a superposition of
independent reductions on resource terms.
Indeed, picking variables \(\varA\not=\varB\) and any \(λ\)-term \(\ltermA\),
observe that 
\(
  \TayExp{\labs{\varA}{\appl{\varB}{\ltermA}}}
  =
  \labs{\varA}{\appl{\varB}{\prom{\TayExp{\ltermA}}}} 
\)
always contains the resource term \(\labs{\varA}{\appl{\varB}{\emptybag}}\) as
a summand,
but \(\labs{\varA}{\appl{\varB}{\ltermA}}\) \(η\)-reduces to \(\varB\)
only in case \(\ltermA=\varA\).

Nevertheless,
Manzonetto and Ruoppolo~\cite{DBLP:journals/entcs/ManzonettoR14} introduced a
\emph{global} notion of \(η\)-reduction on sets of normal resource terms
(satisfying a technical condition), which they used to characterize Morris'
equivalence (the observational equivalence induced by \(β\)-normal forms).
More precisely:
they first consider the  support set of the normal form of ordinary Taylor
expansion;
then they apply a further step of \(η\)-normalization guided by 
the global structure of this set,
which yields a new set of ordinary resource terms, still in normal form;
and they prove that this construction induces the same equational theory as
Böhm trees up to countably many finitely nested \(η\)-expansions, 
which are known to capture Morris' equivalence~\cite{DBLP:conf/lambda/Hyland75}.

This process is limited to a qualitative setting, and
the \(η\)-rule is not reflected in resource reduction.
By contrast, in the present paper,
we enforce extensionality during Taylor expansion,
so that \(η\)-reduction is treated just like \(β\)-reduction during normalization,
all this in a quantitative setting.
In the typed case, our approach can be reduced to a well-known trick:
considering \(η\)-long forms.%
\footnote{
  For instance, the restriction of the resource calculus to \(η\)-long forms
  was leveraged by Tsukada \emph{et al.}~\cite{DBLP:conf/lics/TsukadaO16,
  DBLP:conf/lics/TsukadaAO17}, as well as the authors of the
  present paper~\cite{BPCVA25},
  in connection with game semantics.
  More basically, the fact that extensionality can be enforced in a typed
  setting by considering \(β\)-reduction on \(η\)-long forms is very standard
  knowledge.
}
We readily expose this piece of folklore as a stepping stone to the untyped
case.

\paragraph{The typed case.}
Consider simply typed \(λ\)-terms \emph{à la} Church,
where the grammar of types
(denoted by \(\stypeA\), \(\stypeB\), \(\stypeC\), …)
is inductively generated from a single base type \(\basetype\)
by the formation of arrow types \(\arrowtype{\stypeA}{\stypeB}\).
Each type \(\stypeA\) can be written uniquely as 
\(\stypeA = \arrowtype{\stypeB_1}*{\arrowtype{\cdots}*{\arrowtype{\stypeB_k}{\basetype}}\cdots} \):
\(k\) is the arity of \(\stypeA\), and we use the notation
\( \arrowtype{\tuple{\stypeB_1,\dotsc,\stypeB_k}}{\basetype} \)
in this case.
A typed term \(\ltermA\) is \definitive{\(η\)-long} if each occurrence of a
subterm with arrow type \(\arrowtype{\stypeA}{\stypeB}\) is either an
abstraction or applied to a subterm of type \(\stypeA\):
performing a step of typed \(η\)-expansion in \(\ltermA\) will always generate
a \(β\)-redex.
Equivalently, \(\ltermA\) is \(η\)-long if each occurrence \(\ltermA'\) of a variable or redex in
\(\ltermA\) is \definitive{fully applied}, \emph{i.e.}\ it occurs in a subterm
\(\appl{\ltermA'}{\ltermB_1\cdots\ltermB_k}\)
where \(k\) is the arity of the type of \(\ltermA'\).
In particular, each application is part of such a full application sequence.

Given a typed \(λ\)-term \(\ltermA\), an \definitive{\(η\)-long form} of \(\ltermA\)
is any \(η\)-expansion of \(\ltermA\) that is \(η\)-long.
We can always compute such an \(η\)-long form for a term of type \(\stypeA\),
by setting:
\[
  \etaExp{\varA}
  \eqdef
  \labs{\varB_{1}}{\cdots\labs{\varB_{k}}{
      \appl{\varA}{{\etaExp{\varB_{1}}}\cdots{\etaExp{\varB_{k}}}}
  }}
  \qquad
  \etaExp*{\labs{\varC}{\ltermA}}
  \eqdef
  \labs{\varC}{\etaExp{\ltermA}}
  \qquad
  \etaExp*{\appl{\ltermB}{\ltermC}}
  \eqdef 
  \labs{\varB_{1}}{\cdots\labs{\varB_{k}}{
    \appl{\etaExp{\ltermB}}
    {{\etaExp{\ltermC}}\,{\etaExp{\varB_{1}}}\cdots{\etaExp{\varB_{k}}}}
  }}
\]
where \(k\) is the arity of 
\( \stypeA=\arrowtype{\tuple{\stypeB_1,\dotsc,\stypeB_k}}{\basetype} \),
and each \(\varB_{i}\) is a fresh variable of type \(\stypeB_{i}\)
(this is a valid inductive definition because recursive calls are either on
immediate subterms or on variables of strictly smaller type).%
\footnote{
  This choice of \(η\)-long form may introduce redexes even if
  \(\ltermA\) is \(β\)-normal.
  One can avoid this defect by inspecting the head structure of
  \(\ltermA\) instead of its top-level constructor:
  we stick to this naïve version to ease the exposition here,
  but both approaches will be used in our treatment of untyped extensional
  Taylor expansion, in \cref{section:taylor}.
}
It is an easy exercise to check that:
if \(\ltermA\) \(βη\)-reduces to \(\ltermB\)
then \(\etaExp{\ltermA}\) \(β\)-reduces to \(\etaExp{\ltermB}\).
Defining \(\ExtTayExp{\ltermA}\eqdef\TayExp{\etaExp{\ltermA}}\),
we can thus leverage the already established results on Taylor expansion
to simulate both \(β\)- and \(η\)-reduction via resource reduction:
it becomes immediate that
\(\NF{\ExtTayExp{\ltermA}}=\NF{\ExtTayExp{\ltermB}}\) as soon as \(\ltermA\)
and \(\ltermB\) are \(βη\)-equivalent terms of the same type.

Interestingly, the typing and extensionality constraints we consider on 
\(η\)-long terms admit straightforward counterparts in resource terms.
Indeed, it is easy to adapt the simple type system to resource terms,
in such a way that the elements of \(\TayExp{\ltermA}\) are all typed resource
terms of the same type as \(\ltermA\).
Again, we say that a typed resource term is in \(η\)-long form if 
each occurrence of a subterm of arrow type \(\arrowtype{\stypeA}{\stypeB}\)
is either an abstraction, or applied to a bag of terms all of type
\(\stypeA\).

We can thus directly define the extensional Taylor expansion
of a typed \(λ\)-term as a vector of \(η\)-long resource terms of the same
type, by setting inductively:
\begin{align*}
  \ExtTayExp{\varA}
  & \eqdef
  \labs{\varB_{1}}{\cdots\labs{\varB_{k}}{
      \appl{\varA}{\prom{\ExtTayExp{\varB_{1}}}\cdots\prom{\ExtTayExp{\varB_{k}}}}
  }}
  \\
  \ExtTayExp{\labs{\varC}{\ltermA}}
  & \eqdef
  \labs{\varC}{\ExtTayExp{\ltermA}}
  \\
  \ExtTayExp{\appl{\ltermB}{\ltermC}}
  & \eqdef 
  \labs{\varB_{1}}{\cdots\labs{\varB_{k}}{
    \appl{\ExtTayExp{\ltermB}}
    {\prom{\ExtTayExp{\ltermC}}\,\prom{\ExtTayExp{\varB_{1}}}\cdots\prom{\ExtTayExp{\varB_{k}}}}
  }}
\end{align*}
where, again, \(k\) is the arity of the type, and each \(\varB_i\) is a fresh
variable of appropriate type.
Moreover, \(η\)-long resource terms are stable under resource reduction, so the
dynamics on which we rely is purely local, without any reference to a side condition
or global rewriting constraints.

\paragraph{Enforcing \(η\)-longness in the untyped setting.}

To guide the design of an extensional version of Taylor expansion, it is thus
essentially sufficient to rely on \(η\)-long terms…
The only issue is that there is no such thing in the untyped setting:
without typing constraint, it is always possible to 
\(η\)-expand a term without creating any \(β\)-redex.
For instance, given a variable \(\varA\),
we can consider iterated \(η\)-expansions of the head structure:
\[
  \varA
  \etarev
  \labs{\varB_{1}}{\appl{\varA}{\varB_{1}}}
  \etarev
  \labs{\varB_{1}}{\labs{\varB_{2}}{\appl{\varA}{\varB_{1}\,\varB_{2}}}}
  \etarev
  \cdots
  \etarev
  \labs{\varB_{1}}{\cdots\labs{\varB_{k}}{
      \appl{\varA}{\varB_{1}\cdots\varB_{k}}
  }}
\]
or even nested \(η\)-expansions of fresh variables:
\[
  \varA
  \etarev
  \labs{\varC_{1}}{\appl{\varA}{\varC_{1}}}
  \etarev
  \labs{\varC_{1}}{\appl{\varA}*{\labs{\varC_{2}}{\varC_{1}\,\varC_{2}}}}
  \etarev
  \cdots
  \etarev
  \labs{\varC_{1}}{\appl{\varA}*{\labs{\varC_{2}}{\appl{\varC_{1}}*{\cdots
        \labs{\varC_{k}}{\appl{\varC_{k-1}}}{\varC_{k}}
  \cdots}}}}
  \,.
\]

\begin{figure}
    \includegraphics[page=5,width=0.5\textwidth]{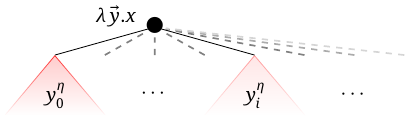}
  \caption{Infinite \(η\)-expansion of a variable}
  \label{fig:idexp}
\end{figure}

\begin{cfigure}
    \includegraphics[page=6,width=0.75\textwidth]{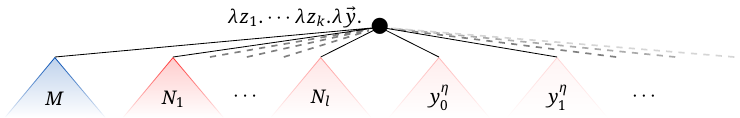}
  \caption{Shape of an infinitely \(η\)-long \(λ\)-term}
  \label{fig:etalong}
\end{cfigure}

We can nonetheless consider the limit of iterating 
the combination of those two processes, as given by an infinite tree
\(\etaExp{\varA}\) that we depict in \cref{fig:idexp},
where \(\svarB\) is a sequence \(\tuple{\varB_0,\varB_1,\dotsc}\) of fresh variables,
and each \(\etaExp{\varB_i}\) is recursively produced in the same way.
If we accept syntactic constructs with countable arity, we may thus write
\(\etaExp{\varA}=\labs{\svarB}{\appl{\varA}{\etaExp{\svarB\,}}}\)
where \(\etaExp{\svarB\,}\) denotes the sequence
\(\tuple{\etaExp{\varB_0},\etaExp{\varB_1},\cdots}\),
and we may understand \(\etaExp{\varA}\) as a kind of infinite term:
\(\labs{\varB_0}{\labs{\varB_1}{\cdots\appl{\varA}{\etaExp{\varB_0}\,\etaExp{\varB_1}\cdots}}}\),
where sequences of abstractions and applications account for countably iterated
head expansions, and recursive calls account for nesting.

More generally, one can consider an intuitive depiction of infinitely
\(η\)-long terms as given by infinite trees, as in \cref{fig:etalong},
where each \(\ltermB_i\) denotes recursively such a tree, and 
\(\ltermA\) is either a variable or a tree itself.
The idea is to ensure that the head of the term is fully applied to
countably many arguments, and this constraint is recursively applied both to
subterms and fresh variables.
In case \(\ltermA=\varA\) and \(k=l=0\), we recover 
the particular case of \(\etaExp{\varA}\).
Now, if we restrict to the case of \(\ltermA\) being a variable,
but extend the construction to allow for \(\bot\)-trees,
the objects we have just described are nothing but Nakajima trees
\cite{DBLP:conf/lambda/Nakajima75}, which are canonical
representatives of Böhm trees up to infinite \(η\)-expansion
\cite[Exercise 19.4.4]{DBLP:books/daglib/0067558}.

One might attempt to equip those infinitely \(η\)-long terms with an
infinitary dynamics, in the style of the infinitary
\(λ\)-calculus~\cite{DBLP:journals/tcs/KennawayKSV97}
-- note that the latter does not account for the application of a term to an
infinite sequence of arguments.
To our knowledge, however, this work has never been carried out, and it would
require to tackle a number of technical issues, 
among which having terms with infinitely many free variables is the least
problematic:
e.g., one also needs to consider countably iterated head reduction,
hence the simultaneous application of countably many substitutions,
which is carefully avoided in the usual infinitary \(λ\)-calculus
approach.

Fortunately, however, we will not need to follow that path:
we only relied on infinitely \(η\)-long \(λ\)-terms as a pedagogical detour,
preparing the reader for the introduction of extensional Taylor expansion,
whose target is supported by a syntax of infinitely \(η\)-long, yet finite,
resource terms.

\subsection{Our contributions}\label{section:intro:contrib}

In the present paper, we introduce a variant of Taylor expansion for 
pure, untyped \(λ\)-terms, in such a way that reduction in the 
associated resource calculus allows us to simulate both \(β\)- and
\(η\)-reduction.
We characterize the equational theory induced via normalization as the
maximal consistent and sensible \(λ\)-theory,
and apply this result to a particular relational model,
demonstrating how this extensional Taylor expansion 
can be leveraged similarly to ordinary Taylor expansion.
We moreover exhibit a precise correspondence between this framework and game
semantics.

\paragraph{The extensional resource calculus.}
With the intuitions and notations of the previous subsection,
one could try to naively apply ordinary Taylor expansion
to an infinitely \(η\)-long term that we may denote
\(
  \labs{\varC_1}{\dotsc{\labs{\varC_k}{\labs{\svarB}{\appl{\ltermA}{
      \ltermB_1\cdots
      \ltermB_l\,
      \etaExp{\varB_0}\,
      \etaExp{\varB_1}\cdots
  }}}}}
\)
or even
\(
  \labs{\svarB}{\appl{\ltermA}{
      \ltermB_1\cdots
      \ltermB_l\,
      \etaExp{\varB_{k}}\,
      \etaExp{\varB_{k+1}}\cdots
  }}
\),
up to \(α\)-equivalence.
In the support of this Taylor expansion, one should find resource terms such as:
\(
  \labs{\svarB}{\appl{\vtermA}{
      \bagB_1\cdots
      \bagB_l\,
      \bagC_0\,
      \bagC_1\cdots
  }}
\)
where each \(\bagB_i\) (resp.\ \(\bagC_j\)) is a bag of terms in the expansion
of \(\ltermB_i\) (resp.\ of \(\etaExp{\varB}\) for some variable \(\varB\)).
Such a term still retains infinite sequences of abstracted variables
and bags of arguments, but there are natural solutions to restrict
this syntax to a finite setting:
\begin{itemize}
  \item we consider \(\svarB\) as a single abstracted variable,
    that we will call a \definitive{sequence variable},
    and refer to the former variables \(\varB_j\) as derived objects;
  \item and we impose bag arguments to be ultimately empty,
    considering only sequences of bags of the shape
    \(\tuple{\bagB_1,\dotsc \bagB_l,\emptybag,\emptybag,\dotsc}\),
    that we will call \definitive{streams}.
\end{itemize}

An extensional resource term will then be \(
\labs{\svarB}{\appl{\vtermA}{\strB}} \) 
where, inductively, \(\vtermA\) is either an ordinary variable or a term
itself, and \(\strB\) is a stream of terms.
Note that we obtain for free that a stream of terms ultimately contains
approximations of successive abstracted variables,
just because empty bags satisfy this condition!
We may again depict such a term as in \cref{fig:resource:ext},
where \(\bagB_i\) is empty for every sufficiently large \(i\).

\begin{sidefigure}
   \includegraphics[page=7,width=.5\textwidth]{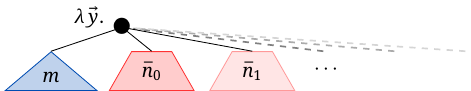}
  \caption{Shape of an extensional resource term}
  \label{fig:resource:ext}
\end{sidefigure}

After preliminary definitions in \cref{section:preliminaries},
we detail the syntax of this extensional resource calculus
in \cref{section:calculus}.
We equip it with a reduction derived from that of the ordinary resource
calculus.
In particular, one can simultaneously fire the countable sequence of redexes at
the head of an expression like
\(\appl*{\labs{\svarA}{\vtermA}}{\strB}\) in a single \emph{full step},
to obtain a finite sum of strictly smaller terms:
this process itself is essentially finite, because the induced sequence of
resource substitutions is ultimately effectless --
replacing non-occurring variables with the elements of empty bags.
The obtained dynamics retains essential properties of resource reduction:
it is confluent in a strong sense,
and the size of terms is non-increasing under reduction, and even strictly
decreasing for full steps.
In particular, each term reduces to a unique normal form, which is a finite sum.

\paragraph{Extensional Taylor expansion.}
Then we turn our attention to vectors of extensional resource terms
in \cref{section:resourcevectors},
and show that ordinary substitution can be obtained as the composition
of resource substitution and promotion, establishing an analogue
of \cref{eqn:substitution} for arbitrary vectors (instead of Taylor
expansions only).
We moreover extend resource reduction to vectors, and show that it is
compatible with promotion.

We leverage these results in \cref{section:taylor}, where we define an
extensional version of Taylor expansion, mapping ordinary $λ$-terms to vectors
of extensional resource terms, subject to the identities:
\begin{align*}
  \ExtTayExp{\varA}
  =
  \labs{\svarB}{\appl{\varA}{\IdStrExp{\svarB}}}
  \qquad
  \ExtTayExp{\labs{\varC}{\ltermA}}
  =
  \labs{\varC}{\ExtTayExp{\ltermA}}
  \qquad
  \ExtTayExp{\appl{\ltermB}{\ltermC}}
  =
  \labs{\svarB}{
    \appl{\ExtTayExp{\ltermB}}
    {\prom{\ExtTayExp{\ltermC}}\cons\IdStrExp{\svarB}}
  }
\end{align*}
where \(\IdStrExp{\svarB}\) (resp.\ \(\prom{\ExtTayExp{\ltermC}}\cons\IdStrExp{\svarB}\))
is the vector of streams induced by the sequence
\(\tuple{\prom{\ExtTayExp{\varB_{0}}},\prom{\ExtTayExp{\varB_{1}}},\dotsc}\)
(resp.\ 
\(\tuple{\prom{
  \ExtTayExp{\ltermC}},\prom{\ExtTayExp{\varB_{0}}},\prom{\ExtTayExp{\varB_{1}}},\dotsc
}\)).
We show that this extensional Taylor expansion also enjoys a version of
\cref{eqn:substitution}, although as a reduction rather than as an identity
--  this is analogous to the fact that, even in a typed setting,
the terms \(\subst{\etaExp{\ltermA}}{\varA}{\etaExp{\ltermB}}\)
and \(\etaExp*{\subst{\ltermA}{\varA}{\ltermB}}\) might differ,
but the former \(β\)-reduces to the latter.
This allows us to simulate both \(β\)- and \(η\)-reduction.

\paragraph{A characterization of \(\Hs\).}
Given the constructions we have outlined, one can reasonably consider the
extensional resource calculus as a language of (non-necessarily normal)
finite approximants of Nakajima trees,
much like ordinary resource terms for Böhm trees.
We are indeed confident that an analogue of \cref{eqn:TayExpBT},
where one replaces ordinary Taylor expansion with extensional Taylor expansion,
and Böhm trees with Nakajima trees, could be established.
But our point is precisely that the technicalities of dealing with infinite
\(η\)-expansion in the already infinite Böhm trees can be avoided, and
that this kind of technology can profitably be replaced with 
Taylor expansion.

In support of this claim, we characterize the \(λ\)-theory \(\eqExtTay\) induced
by the normalization of extensional Taylor expansion, in \cref{section:Hs}.
That \(\eqExtTay\) is indeed a \(λ\)-theory follows from the inductive definition
of Taylor expansion and the simulation of \(β\)-reduction,
like in the ordinary case.
It is moreover extensional, thanks to the simulation of \(η\)-reduction,
and sensible, thanks to a variant of \cref{lemma:taylor:solvable},
that we establish essentially in the same way --
although, like for substitution, extensional Taylor expansion does not commute
with head reduction on the nose.

Finally, we show that \(\eqExtTay\) is nothing but $\Hs$, the greatest consistent
sensible $λ$-theory.
The proof is naturally based on a separability argument, showing that 
\(\eqExtTay\)-distinct terms can be separated by a context,
sending one of them to a head normalizable term,
and the other one to a non-solvable one.
Thanks to the properties of Taylor expansion, we are able to reason on 
the structure of resource terms, which allows us to adapt a well-known proof of
separability for \(η\)-distinct \(β\)-normal forms \cite[Chapter 5]{krivine}:
that we can do so, instead of having to reason on infinite objects, is a
testimony of the applicability of extensional Taylor expansion.

This characterization moreover allows us to revisit previous results about
\(\Hs\) -- or Nakajima trees, which are canonical representatives
for \(\Hs\) \cite[Exercise 19.4.4]{DBLP:books/daglib/0067558}.
E.g., to exhibit a model of $\Hs$, it becomes sufficient to provide a
model of the extensional resource calculus.
Illustrating this strategy, \cref{section:rel} gives a new proof of a result
by Manzonetto \cite{DBLP:conf/mfcs/Manzonetto09}:
$\Hs$ is the $λ$-theory induced by a well-chosen reflexive object \(\RelTypes\)
in the relational model of the simply typed $λ$-calculus
\cite{DBLP:conf/csl/BucciarelliEM07}.

\paragraph{Where this all comes from: Taylor expansion and game semantics.}
The present work was actually motivated by an ongoing effort to expose the close
connections between Taylor expansion and game semantics.
In a typed setting, $η$-long, $β$-normal resource terms were known to be in
bijective correspondence with plays in the sense of Hyland-Ong game semantics
\cite{DBLP:journals/iandc/HylandO00}, up to Opponent's scheduling of the
independent explorations of separate branches of the term, as formalized by
Melliès' homotopy equivalence on plays \cite{DBLP:journals/tcs/Mellies06}:
this correspondence was first unveiled by Tsukada and Ong
\cite{DBLP:conf/lics/TsukadaO16} via two bijections with particular
elements of the relational model of the simply typed \(λ\)-calculus;
and we later exhibited a direct correspondence, underlying a quantitative
denotational interpretation of (non-necessarily normal) resource terms
as strategies \cite{BPCVA25}.

To recast this correspondence in an untyped setting, we needed an untyped
analogue of $η$-long, $β$-normal resource terms: these are the normal forms 
of our extensional resource calculus.
The first sections of the paper make no reference to game semantics, as we
focus on developing the theory of extensional Taylor expansion,
and its applications.
Nonetheless, we dedicate \cref{section:gs} to spelling out the bijection
between normal extensional resource terms and (isomorphism classes of)
augmentations in the universal arena: augmentations were introduced by the
first two authors \cite{DBLP:conf/fscd/Blondeau-Patissier21}, as an alternative
presentation of plays up to homotopy; and the universal arena is the standard
interpretation of the “type” of pure, untyped $λ$-terms in game
semantics~\cite{KerNO02}.\footnote{
  Notably, Ker \emph{et al.}~\cite{KerNO02} have shown that this arena provides
  an extensional reflexive object in the category of innocent strategies, and
  that the induced interpretation of untyped $λ$-terms enjoys a strong
  connection with Nakajima trees.
  In particular, its $λ$-theory is also $\Hs$.
}

We do not go beyond this static correspondence:
as discussed in our concluding \cref{section:conclusion},
we leave for future work the definition of a denotational semantics of
extensional resource terms as strategies,
as well as its relationship with Taylor expansion.
In passing, we also establish a correspondence between \emph{positions}
of the universal arena, which represent states of computation (a position
essentially records the part of the arena that is explored by a play), and
\emph{relational types},
\emph{i.e.}\ the elements of the reflexive object \(\RelTypes\).
We moreover show that, through both correspondences,
the position reached by an augmentation is nothing but the (uniquely defined)
relational type of the associated normal term.

\paragraph{Comparison with the resource calculus with tests.}
It is notable that, in the introduction of their seminal paper
\cite{DBLP:conf/lics/TsukadaO16}, Tsukada and Ong claimed that their results
could be adapted to the untyped setting, relying on the
\emph{resource calculus with tests} of Bucciarelli, Carraro, Ehrhard and
Manzonetto~\cite{DBLP:journals/corr/abs-1209-2890}.
The latter is an extension of the ordinary resource calculus designed to
associate a syntactic counterpart to \emph{every} point of \(\RelTypes\):
from this, the authors derive a full abstraction result for the resource
calculus with tests, that they are able to lift to a version with promotion
(itself an extension of the differential
\(λ\)-calculus~\cite{DBLP:journals/tcs/EhrhardR03}).
Tsukada and Ong's claim was prompted by the fact that this calculus provides
constructions both for applying a term to a denumerable sequence of empty bags
(the \emph{cork} constructor \(\cork{\vtermA}\), yielding a \emph{test}),
and for abstracting over a denumerable sequence of fresh variables
(the dual \emph{uncork} constructor \(\uncork{\btermA}\), where \(\btermA\) is
a test, yielding a term).

It turns out, however, that this calculus is not readily fit for the task:
the original version of its language is too rich (it contains normal forms that
do not correspond to plays) so it must be restricted;
and at the same time its constructions for infinite sequences of abstractions,
and for applications to infinite sequences of bags are not canonical.
As a consequence, even though one can devise an appropriate notion of
\(η\)-longness in that setting,\footnote{
  To our knowledge, such a notion remained to be introduced before our own
  work: Bucciarelli \emph{et al.}\ did not discuss extensionality nor
  \(η\)-longness in the context of their calculus, as their interest was
  elsewhere.
}
the syntax still distinguishes between normal forms that represent the same
play up to homotopy.
For the same reasons, the resource calculus with tests is not an appropriate
target language for extensional Taylor expansion.
For the sake of comparison, we provide a brief account of the resource calculus
with tests in \cref{section:tests}, in light of the key features
of the extensional resource calculus.
In particular, we describe an \(η\)-long fragment of the former
and outline how one could recover the latter as a quotient.

\paragraph{How to read this paper.}
The paper is long, but it is structured in such a way that the reader can
browse through a selection of its content depending on their interests.

A programming language semanticist who is already convinced of the merits of
the Taylor expansion might be content as early as after reaching the end
of \cref{section:taylor}, where we show that extensional Taylor expansion is
compatible with both \(β\)- and \(η\)-reductions.
A \(λ\)-calculist in search of a practical alternative to Nakajima trees 
will want to continue with \cref{section:Hs}, and might also read 
\cref{section:rel} for an example of application.
Both profiles can dispense with \cref{section:gs} altogether.
In any case, we discourage the reader without any background in game semantics
to discover the subject with the present paper:
an introductory account of game semantics in relation to Taylor
expansion can rather be found in the above-cited works
\cite{DBLP:conf/lics/TsukadaO16,BPCVA25}.
On the other hand, a game semanticist might want to jump
to \cref{section:gs} just after getting acquainted with the extensional
resource calculus in \cref{section:calculus} (the correspondence between
positions and relational types additionally refers to the relational model of
the extensional resource calculus given \cref{section:rel:resource}, which
can be read independently from \crefrange{section:resourcevectors}{section:Hs}, as well as
from the rest of \cref{section:rel}).
We summarize those possible routes through the content of the paper in \cref{figure:map}.

All along \crefrange{section:calculus}{section:Hs}, we maintain our bias in
favour of the quantitative version of Taylor expansion.
We have already explained why we consider quantitative Taylor expansion as the
primitive notion, of which the qualitative version is a mere by-product
-- \emph{qui peut le plus, peut le moins}.
The coefficients of Taylor expansion,
as well as those generated by resource reduction,
are moreover relevant in connection with game semantics:
the correspondence between terms and strategies is
quantitative~\cite{BPCVA25}!
Nonetheless, for the results we develop in \cref{section:Hs,section:rel}, the
qualitative version is sufficient;
and \cref{section:gs} involves normal terms only, and does not mention resource
reduction nor Taylor expansion.
A reader interested only in those applications might thus read 
the earlier sections without worrying too much about coefficients,
and even skip some computations.

\begin{figure}
  \scriptsize
  \begin{center}
\scalebox{1.228}{
    \begin{tikzpicture}[
    section/.style={node distance=10pt, font=\em},
    desc/.style={node distance=1pt, align=center, font=\scriptsize,text width=3cm},
    wide desc/.style={desc, text width=3.3cm},
    blob/.style={draw, rounded corners, line width=1pt, inner sep=7pt},
    blob title/.style={node distance=1em, align=center},
    depend/.style={every edge/.style={draw,->,line width=1.5pt}},
  ]
  \node[section](preliminaries)
    {\cref{section:preliminaries}, p.\,\pageref{section:preliminaries}:};
  \node[wide desc,below=of preliminaries, text width=3cm](preliminaries desc)
    {
      Preliminaries:
      notations and
      operations on
      tuples and bags
    };
  \node[section, below=of preliminaries desc](calculus)
    {\cref{section:calculus}, p.\,\pageref{section:calculus}:};
  \node[wide desc,below=of calculus](calculus desc)
    {
      Syntax and dynamics
      of the extensional
      resource calculus
    };
  \node[blob, fit=(preliminaries)(preliminaries desc)(calculus)(calculus desc)](calculus blob){};
  \node[blob title, above=of calculus blob, anchor=base](calculus title)
    {Extensional resource calculus};
  \node[section, right=2.5cm of preliminaries](vectors)
    {\cref{section:resourcevectors}, p.\,\pageref{section:resourcevectors}:};
  \node[wide desc,below=of vectors](vectors desc)
    {
      Properties of\\
      substitution and\\
      reduction on vectors
      of resource terms
    };
  \node[section,below=of vectors desc](taylor)
    {\cref{section:taylor}, p.\,\pageref{section:taylor}:};
  \node[desc,below=of taylor](taylor desc)
    {
      Definition and properties
      of extensional
      Taylor expansion
    };
  \node[blob, fit=(vectors)(vectors desc)(taylor)(taylor desc)](taylor blob){};
  \node[blob title, above=of taylor blob, anchor=base](taylor title)
    {Extensional Taylor expansion};
  \node[blob, fit=(calculus blob)(calculus title)(taylor blob)(taylor title)](core blob){};
  \node[blob title, above=of core blob, anchor=base]
     {Main matter};
  \node[section,right=3cm of vectors,yshift=.4cm](Hs)
    {\cref{section:Hs}, p.\,\pageref{section:Hs}:};
  \node[wide desc,below=of Hs](Hs desc)
    {Characterization of \(\Hs\)};
  \node[section,below=of Hs desc](rel)
    {\cref{section:rel}, p.\,\pageref{section:rel}:};
  \node[desc,below=of rel](rel desc)
    {Relational semantics};
  \node[blob, fit=(Hs)(Hs desc)(rel)(rel desc)](semantics blob){};
  \node[blob title, above=of semantics blob, anchor=base]
    {Denotational semantics\\of the pure \(λ\)-calculus};
  \node[section,below=1.8cm of semantics blob](gs){\cref{section:gs}, p.\,\pageref{section:gs}:};
  \node[wide desc,below=of gs](gs desc)
    {
      Augmentations in the
      universal arena,
      and bijection with
      normal resource terms
    };
  \node[blob, fit=(gs)(gs desc)](gs blob){};
  \node[blob title, above=of gs blob, anchor=base]
    {Connection with\\game semantics};
  \coordinate(out) at (semantics blob.-45);
  \path[depend]
    (calculus blob) edge (taylor blob)
    (core blob) edge (semantics blob)
    (calculus blob) edge[out=-55,in=west, looseness=.8] (gs blob.200)
    (out) edge[dashed] (out|-gs blob.north)
    ;
  \end{tikzpicture}
}\end{center}

  \vskip-.7cm
  \caption{
    A map of the contributions of the paper.
    Arrows represent dependencies;
    the dashed arrow stands for the dependency of some results of 
    \cref{section:gs} on the definition of relational typing in 
    \cref{section:rel:resource}.
  }
  \label{figure:map}
\end{figure}

\section{Preliminaries on sequences and bags}\label{section:preliminaries}

\paragraph{Tuples and bags.} If $X$ is a set, we write
$\WordsOf{X}=\bigcup_{n\in\N}X^n$ for the set of finite lists, or tuples, of
elements of $X$, ranged over by $\word{a},\word{b}$, \emph{etc.} We write
$\tuple{a_1,\dotsc,a_n}=\tuple{a_i}_{1\le i\le n}$ to list the elements of a tuple, 
$\emptyword$ for the empty tuple, $\LengthOf{\word{a}}$ for the length of $\word{a}$,
and denote concatenation simply by juxtaposition, e.g., $\word{a}\seqcat\word{b}$.
If $a\in X$ and $\vec{b}$ is a tuple, we write $a\cons \vec{b}$ 
for the tuple obtained by pushing $a$ at the head of $\vec{b}$:
this \definitive{cons} operation generates $\WordsOf{X}$ inductively from $\emptyword$.

We write $\Mf(X)$ for the
set of finite multisets of elements of $X$, which we call
\definitive{bags}, ranged over by $\bag{a}, \bag{b}$, \emph{etc.}
We write $\mset{a_1, \dotsc, a_n}$ for the bag $\bag{a}$ defined by a list
$\word{a}=\tuple{a_1,\dotsc,a_n}$ of elements:
we say that $\word{a}$ is an \definitive{enumeration} of $\bag{a}$ in this case.
We write $\emptybag$ for the empty bag, and use $\bagcat$ for bag concatenation.
We moreover write $\LengthOf{\bag{a}}$ for the length of $\bag{a}$:
$\LengthOf{\bag{a}}$ is the length of any enumeration of $\bag{a}$.
We may abuse notation and use a tuple $\word{a}$ or a bag $\bag{a}$ for the
set of its elements: e.g., we write $a_i\in\mset{a_1, \dotsc, a_n}$.

We will often need to \emph{partition} bags, which requires some care.
For $k\in\N$, a \definitive{$k$-partitioning} of
$\bag{a}$ is a function
$p:\{1,\dotsc,\LengthOf{\bag{a}}\}\to\{1,\dotsc,k\}$:
we write $p:\bag{a}\splitinto k$.
Given an enumeration $\tuple{a_1,\dotsc,a_n}$ of $\bag{a}$
and $J=\{j_1,\dotsc,j_l\}\subseteq\{1,\dotsc,n\}$ with $\CardOf J=l$,
we write $\bag{a}\restrict J\eqdef \mset{a_{j_i},\dotsc,a_{j_l}}$ 
for the \definitive{restriction} of $\bag{a}$ to $J$.
The \definitive{$k$-partition} of $\bag{a}$ associated with $p:\bag{a}\splitinto k$
is then the tuple
$\tuple{\restrictToPre{\bag{a}}{p}{1},\dotsc,\restrictToPre{\bag{a}}{p}{k}}$,
where $\preImage{p}{i} \eqdef \{ j\mid p(j)=i\}$ for $1\le i \le k$, so that 
\[\bag{a}=\restrictToPre{\bag{a}}{p}{1}\bagcat\cdots\bagcat\restrictToPre{\bag{a}}{p}{k}\;.\]

There is a (temporary) abuse of notation here, as the definitions of
restrictions and $k$-partitions depend on the chosen enumeration of $\bag{a}$.
But having fixed $\bag{a}$ and $k$, neither the set of $k$-partitions
of $\bag{a}$, nor the \emph{number} of partitionings $p$ of $\bag{a}$
yielding a given $\tuple{\bag{a}_1, \dots, \bag{a}_n}$, depend on the enumeration.
So for any function $f : \Mf(X)^k \to \mathcal{M}$ (for $\mathcal{M}$ a
commutative monoid, noted additively), the sum 
\[
\sum_{\bag{a} \splitinto \bag{a}_1\bagcat\cdots\bagcat\bag{a}_k} f(\bag{a}_1,\dotsc,\bag{a}_k)
\quad\eqdef\quad
\sum_{\mathclap{p : \bag{a} \splitinto k}} f(\restrictToPre{\bag{a}}{p}{1}, \dotsc, \restrictToPre{\bag{a}}{p}{k})
\]
is independent of the enumeration.
When indexing a sum with ${\bag{a} \splitinto \bag{a}_1\bagcat\cdots\bagcat\bag{a}_k}$
we thus mean to sum over all partitionings $p : \bag{a} \splitinto k$, $\bag{a}_i$
being shorthand for $\restrictToPre{\bag{a}}{p}{i}$ in the summand, and
the result being independent of the choice of an enumeration.
This construction is easily proved to be associative, in the sense that, e.g.:
\[
  \sum_{\bag{a} \splitinto \bag{a}_1\bagcat\bag{a}'}
  ~
  \sum_{\bag{a}' \splitinto \bag{a}_2\bagcat\bag{a}_3} f(\bag{a}_1,\bag{a}_2,\bag{a}_3)
  \qquad
  =
  \qquad
  \sum_{\mathclap{\bag{a} \splitinto \bag{a}_1\bagcat\bag{a}_2\bagcat\bag{a}_3}} f(\bag{a}_1,\bag{a}_2,\bag{a}_3)
  \;.
\]

The \definitive{isotropy degree}\label{definition:isotropy} $\isodeg{\bag{a}}$
of a bag $\bag{a}$ of length $k$ is the cardinality of the stabilizer of any
enumeration $\tuple{a_1,\dotsc,a_k}$ of $\bag{a}$ under the action of the group
$\Perms k$ of permutations of $\set{1,\dotsc,k}$:
namely,
\(
  \isodeg{\bag{a}}\eqdef\CardOf{\set{
      σ\in\Perms k\st
      \tuple{a_1,\dotsc,a_k}
      =\tuple{a_{σ(1)},\dotsc,a_{σ(k)}}
  }}
\).
The following result is a routine exercise in combinatorics:
\begin{fact}\label{fact:isodeg}
  If $\bag{a}=\bag{a}_1\bagcat\cdots\bagcat\bag{a}_n$ then
  \(
    \isodeg{\bag{a}}
    =\CardOf{\set{
      p:\bag{a}\splitinto n
      \st \restrictToPre{\bag{a}}{p}{i} = \bag{a}_i
      \text{ for }1\le i\le n
    }}\times\prod_{i=1}^{n}\isodeg{\bag{a}_i}
  \).
\end{fact}

\paragraph{Sequences of bags and streams.}
We will also use possibly infinite sequences of bags, with a finiteness
constraint: only finitely many bags may be non-empty.
We write $\SfOf{X}$ for the set $\WordsOf{\MfOf{X}}$ of tuples of bags,
and we write $\SsOf{X}$ for the subset of $\SeqsOf{\MfOf{X}}$ such that
$\tuple{\bag{a}_i}_{i\in\N}\in\SsOf{X}$ iff 
$\{ i\in\N\st \LengthOf{\bag{a}_i}>0 \}$ is finite.
We denote elements of $\SfOf{X}$ or $\SsOf{X}$ as $\seq{a},\seq{b}$,
\emph{etc.}\ just like for plain tuples,
and we reserve the name \definitive{stream} for the elements of $\SsOf{X}$.

We write $\emptystream\eqdef\tuple{\emptybag}_{i\in\N}$ 
for the \definitive{empty stream}.
Note that streams are inductively generated from $\emptystream$,
by the \definitive{cons} operation defined by 
\[
  (\bag{a}\cons\seq{b})_i\eqdef\begin{cases}
    \bag{a}&\text{if $i=0$}\\
    \bag{b}_{j}&\text{if $i=j+1$}
  \end{cases}
  \qquad \text{(writing $\seq{b}=\tuple{\bag{b}_{j}}_{j\in\N}$)}
\]
subject to the identity $\emptybag\cons\emptystream=\emptystream$.
We can thus reason inductively on streams, treating $\emptystream$ 
as the base case, and considering $\seq{b}$ as a ``strict sub-stream''
of $\bag{a}\cons\seq{b}$ when $\bag{a}\cons\seq{b}\not=\emptystream$.

We also define the \definitive{range} of a stream
$\seq{a}=\tuple{\bag{a}_i}_{i\in\N}\in\SfOf{A}$
as the minimal length of a prefix containing all non-empty bags:
$\rangeOf{\seq{a}}\eqdef\max\set{i+1\in\N\st\bag{a}_i\not=0}$.
Equivalently, we can define $\range$ inductively by setting:
$\rangeOf{\emptystream}\eqdef 0$ and
$\rangeOf{\bag{a}\cons\seq{b}}\eqdef\rangeOf{\seq{b}}+1$
if $\bag{a}\cons\seq{b}\not=\emptystream$.

A \definitive{$k$-partitioning} $p:\seq{a}\splitinto k$ of
$\seq{a}=\tuple{\bag{a}_1,\dotsc,\bag{a}_n}\in\SfOf{X}$
is a tuple $p=\tuple{p_1,\dotsc,p_n}$ of $k$-partitionings 
$p_i: \bag{a}_i \splitinto k$.
This defines a \definitive{partition} $\tuple{\restrictToPre{\seq{a}}{p}{1},\dotsc,\restrictToPre{\seq{a}}{p}{k}}$,
component-wise:
each $\restrictToPre{\seq{a}}{p}{i}$ is the sequence
$\tuple{\restrictToPre{\bag{a}_1}{p_1}{i},\dotsc,\restrictToPre{\bag{a}_n}{p_n}{i}}$.
We obtain $\seq{a}=\restrictToPre{\seq{a}}{p}{1}\bagcat\cdots\bagcat\restrictToPre{\seq{a}}{p}{k}$,
where we apply the concatenation of bags component-wise, to sequences all
of the same length.
Just as before, we write
\[
\sum_{\seq{a} \splitinto \seq{a}_1\bagcat\cdots\bagcat\seq{a}_k} f(\seq{a}_1,\dotsc,\seq{a}_k)
\eqdef
\sum_{p : \seq{a} \splitinto k} f(\restrictToPre{\seq{a}}{p}{1}, \dotsc, \restrictToPre{\seq{a}}{p}{k})
\ ,
\]
the result of the sum being independent from the enumerations
of the bags of $\seq{a}$.

Similarly, a \definitive{$k$-partitioning} $p:\seq{a}\splitinto k$ of
a stream $\seq{a}=\tuple{\bag{a}_i}_{i\in\N}$
is a sequence $p=\tuple{p_i}_{i\in\N}$ of $k$-partitionings 
$p_i: \bag{a}_i \splitinto k$:
note that a stream $\seq{a}$ has only finitely many $k$-partitionings,
because $\bag{a}_i$ is empty for sufficiently large values of $i$.
A $k$-partitioning of a stream $\seq{a}$ defines a \definitive{partition}
$\tuple{\restrictToPre{\seq{a}}{p}{1},\dotsc,\restrictToPre{\seq{a}}{p} {k}}$,
component-wise:
each $\restrictToPre{\seq{a}}{p}{j}$ is the sequence
$\tuple{\restrictToPre{\bag{a}_i}{p_i}{j}}_{i\in\N}$.
We obtain $\seq{a}=\restrictToPre{\seq{a}}{p}{1}\bagcat\cdots\bagcat\restrictToPre{\seq{a}}{p}{k}$,
where we apply the concatenation of bags component-wise.
And we write
\[
\sum_{\seq{a} \splitinto \seq{a}_1\bagcat\cdots\bagcat\seq{a}_k} f(\seq{a}_1,\dotsc,\seq{a}_k)
\eqdef
\sum_{p : \seq{a} \splitinto k} f(\restrictToPre{\seq{a}}{p}{1}, \dotsc, \restrictToPre{\seq{a}}{p}{k})
\ ,
\]
which is always a finite sum, whose result is independent from the enumerations
of the bags of $\seq{a}$.

\section{The extensional resource calculus}\label{section:calculus}

In this section, we introduce our extensional version of the resource calculus,
whose terms are the infinitely \(η\)-long resource terms described in the
introduction: \(\labs{\svarA}{\appl{\hexprA}{\strB}}\) where
\(\hexprA\) is a term or variable and \(\strB\) is a stream of terms.
It will be practical to more generally introduce various syntactic categories,
such as \emph{base terms} of the shape \(\appl{\hexprA}{\strB}\),
as in the body of the previous term.
We will collectively refer to these categories as \emph{resource terms},
calling \emph{value terms} those of the first form.

We then discuss suitable notions of resource reduction, taking into account the
presence of infinite sequences of abstractions, applied to streams of
arguments.

\subsection{Syntax of the calculus}

We fix an infinite countable set $\LVar$ of \definitive{value variables}
(or, simply, \definitive{variables}),
which we denote by letters $\varA,\varB,\varC$.
We also fix an infinite countable set $\SVar$ of \definitive{sequence variables},
which we denote by letters $\svarA,\svarB,\svarC$,
and with each sequence variable $\svarA$,
we associate a sequence $\tuple{\svarA[i]}_{i\in\N}$ of value variables,
in such a way that for each $\varA\in\LVar$, there exists a unique 
pair $\tuple{\svarA,i}$ such that $\varA=\svarA[i]$:
sequence variables partition value variables.
We will in general identify $\svarA$ with the corresponding sequence
of value variables.
We may also abuse notation and use $\svarA$ for its image set:
for instance we may write $\varA\in\svarA$ instead of $\varA\in\set{\svarA[i]\st i\in\N}$.
The use of sequence variables will allow us to manage infinite sequences of
$λ$-abstractions, without needing to resort to de Bruijn indices or other
techniques for dealing with $α$-equivalence.

\paragraph{Terms.}
We define \definitive{value terms} ($\vtermA, \vtermB, \vtermC\in\ValueTerms$),
\definitive{base terms} ($\btermA, \btermB, \btermC \in\BaseTerms$),
\definitive{bag terms} ($\bagA,\bagB,\bagC\in\BagTerms$) and
\definitive{stream terms} ($\strA,\strB,\strC\in\StreamTerms$),
inductively by the rules of \cref{fig:terms}.\footnote{
  For now, we overload notations and use \(\mset{-}\), \(\emptystream\) and
  \(\cons\) as term formers:
  they will soon recover their meaning as constructions of bags and streams.
}
\begin{figure}
\begin{gather*}
  \begin{prooftree}
    \hypo{\svarA\in\SVar}
    \hypo{\btermA\in\BaseTerms}
    \infer2[\ValueRule*]{\labs{\svarA}{\btermA}\in\ValueTerms}
  \end{prooftree}
  \qquad
  \begin{prooftree}
    \hypo{\vtermA\in\ValueTerms}
    \hypo{\strB\in\StreamTerms}
    \infer2[\RedexBaseRule*]{\appl{\vtermA}{\strB}\in\BaseTerms}
  \end{prooftree}
  \qquad
  \begin{prooftree}
    \hypo{\varA\in\LVar}
    \hypo{\strB\in\StreamTerms}
    \infer2[\NormalBaseRule*]{\appl{\varA}{\strB}\in\BaseTerms}
  \end{prooftree}
  \\[1em]
  \begin{prooftree}
    \hypo{\vtermA_1\in\ValueTerms}
    \hypo{\cdots}
    \hypo{\vtermA_k\in\ValueTerms}
    \infer3[\BagRule*]{\mset{\vtermA_1,\dotsc,\vtermA_k}\in\BagTerms}
  \end{prooftree}
  \qquad
  \begin{prooftree}
    \infer0[\EmptyStreamRule*]{\emptystream\in\StreamTerms}
  \end{prooftree}
  \qquad
  \begin{prooftree}
    \hypo{\bagA\in\BagTerms}
    \hypo{\strB\in\StreamTerms}
    \infer2[\ConsStreamRule*]{\bagA\cons\strB\in\StreamTerms}
  \end{prooftree}
\end{gather*}
\caption{Rules for constructing extensional resource terms}
\label{fig:terms}
\end{figure}
A \definitive{head expression} ($\hexprA,\hexprB,\hexprC\in\HeadExprs$) is
a value term or variable,
so that a base term is necessarily of the form $\appl{\hexprA}{\strA}$
where $\hexprA$ is a head expression and $\strA$ is a stream term.
A \definitive{resource term} (denoted by $\termA$, $\termB$, $\termC$) is any of 
a value term, base term, stream term or bag term;
and a \definitive{resource expression} (denoted by $\exprA,\exprB,\exprC$) is any of 
a value variable or a resource term.
Note that, despite the seemingly infinitary flavour of abstractions and
applications, the inductive definition of the syntax yields a countable
set of resource expressions.

The actual objects of the calculus are value terms, as these will form the
target of Taylor expansion.
Nonetheless, base terms, bag terms and stream terms also constitute meaningful
computational entities, as will be clear in the next sections.
On the other hand, plain variables \emph{should not} be considered as entities
of the calculus by themselves: in the context of extensional Taylor expansion,
a variable must always come with the stream it is applied to,
and the Taylor expansion of a single variable will always be an infinite
sum of normal value terms (see \cref{section:taylor}).\footnote{
  This distinction also shows in the semantics.
  The various categories of resource terms will have a finite interpretation in
  the relational semantics (\cref{section:rel}).
  And the normal forms among each category will represent isomorphism classes
  of augmentations on corresponding arenas in game semantics (\cref{section:gs}).
  On the other hand, the relational interpretation of a variable
  is an infinite set; and, in game semantics, it corresponds to a whole
  \emph{copycat} strategy.
}
We consider head expressions only because it allows us to treat uniformly both
forms of base terms;
and we consider resource expressions because, when we reason inductively on
terms, it is sometimes simpler to have variables as a base case.
We may write $\ResourceTerms$ (resp.\ $\ResourceExprs$) for any of
the four sets $\ValueTerms$, $\BaseTerms$, $\BagTerms$, or $\StreamTerms$
(resp.\ $\HeadExprs$, $\BaseTerms$, $\BagTerms$, or $\StreamTerms$).\footnote{
  We do mean these as placeholders, not as unions of sets:
  this will be especially relevant when we consider, e.g., sums of resource terms,
  which we always implicitly restrict to a given syntactic category.
}

The set of \definitive{free variables} of a resource expression is defined as follows:
\begin{gather*}
  \LVarOf{\varA} \eqdef \set{\varA}\,,
  \qquad
  \LVarOf{\labs{\svarA}{\btermA}} \eqdef \LVarOf{\btermA}\setminus\svarA\,,
  \qquad
  \LVarOf{\appl{\hexprA}{\strA}} \eqdef \LVarOf{\hexprA}\cup\LVarOf{\strA}\,,
  \\
  \LVarOf{\mset{\vtermA_1,\dotsc,\vtermA_k}} \eqdef \bigcup_{1\le i\le k}\LVarOf{\vtermA_i}\,,
  \qquad
  \LVarOf{\emptystream} \eqdef \emptyset\,,
  \qquad
  \LVarOf{\bagA\cons\strB} \eqdef \LVarOf{\bagA}\cup\LVarOf{\strB}\,.
\end{gather*}

This is always a finite set.
Then we define
$\SVarOf{\termA}\eqdef\bigcup_{\varA\in\LVarOf{\termA}}\SVarOf{\varA}$
with $\SVarOf{\svarA[i]}\eqdef\set{\svarA}$, which is also finite.
We can thus define $α$-equivalence as usual,
despite the fact that binding a single sequence variable
simultaneously binds an infinite family of value variables.
In particular,
given a value term $\vtermA$ and
a finite set $V$ of sequence variables that are not free in \(\vtermA\),
we can always assume, up to $α$-equivalence,
that $\vtermA$ is of the form $\labs{\svarB}{\btermA}$
where $\svarB$ contains no variable $\svarA[i]$ 
with $\svarA\in V$.

In addition to $α$-equivalence, we consider resource expressions up to
permutations of elements in a bag $\mset{\vtermA_1,\dotsc,\vtermA_k}$,
and up to the identity $\emptybag\cons\emptystream =\emptystream$,
so that $\BagTerms$ is identified with $\MfOf{\ValueTerms}$
and $\StreamTerms$ is identified with $\SsOf{\ValueTerms}$.

We write $\subst{\exprA}{\varA}{\hexprB}$ for the ordinarily defined,
capture avoiding \definitive{substitution} of a head expression $\hexprB$
for a value variable $\varA$ in any resource expression $\exprA$:
note that this preserves the syntactic category of $\exprA$ in the sense
that if $\exprA\in\ValueTerms$ (resp.\ $\HeadExprs$, $\BagTerms$, $\BaseTerms$, $\StreamTerms$)
then $\subst{\exprA}{\varA}{\hexprB}\in\ValueTerms$
(resp.\ $\HeadExprs$, $\BagTerms$, $\BaseTerms$, $\StreamTerms$).

We define the \definitive{size} $\SizeOf{\exprA}\in\N$
of a resource expression $\exprA$ inductively as follows:
\begin{gather*}
  \SizeOf{\varA}
  \eqdef 1
  \qquad
  \SizeOf*{\labs{\svarA}{\btermA}}
  \eqdef 1+\SizeOf{\btermA}
  \qquad
  \SizeOf*{\appl{\hexprA}{\strA}}
  \eqdef 1+\SizeOf{\hexprA}+\SizeOf{\strA}
  \\
  \SizeOf{\mset{\vtermA_1,\dotsc,\vtermA_k}}
  \eqdef \sum_{i=1}^k \SizeOf{\vtermA_i}
  \qquad
  \SizeOf*{\bagA\cons\strB}
  \eqdef
  \begin{cases}
    0
    &\text{if \(\bagA\cons\strB=\emptystream\)}
    \\
    \SizeOf{\bagA}+\SizeOf{\strB}
    &\text{otherwise}
  \end{cases}
  \;.
\end{gather*}

In particular, 
$\SizeOf{\strA}=\sum_{i\in\N}\SizeOf{\bagA_i}$
for any stream term $\strA=\tuple{\bagA_i}_{i\in\N}$.
In short, $\SizeOf{\exprA}$ is nothing but the number of abstractions,
applications and variable occurrences in $\exprA$:
this number is always finite, even though each abstraction on a sequence variable
binds an infinite sequence of value variables,
and each application is to a stream of terms, which is an infinite sequence of
bags.
In particular,
$\SizeOf{\hexprA}\ge 1$ for any head expression $\hexprA$,
$\SizeOf{\btermA}\ge 2$ for any base term $\btermA$,
$\SizeOf{\vtermA}\ge 3$ for any value term $\vtermA$,
and $\SizeOf{\bagA}\ge 3\LengthOf{\bagA}$ for any bag term $\bagA$.

We similarly define the \definitive{number of occurrences}
$\nocc{\varA}{\exprA}$ of \(\varA\in\LVar\) in \(\exprA\)
inductively as follows:
\begin{gather*}
  \nocc{\varA}{\varA}
  \eqdef 1
  \qquad
  \nocc{\varA}{\varB}
  \eqdef 0
  ~\text{(for each \(\varB\not=\varA\))}
  \qquad
  \nocc{\varA}{\labs{\svarA}{\btermA}}
  \eqdef \nocc{\varA}{\btermA}
  \qquad
  \nocc{\varA}{\appl{\hexprA}{\strA}}
  \eqdef \nocc{\varA}{\hexprA}+\nocc{\varA}{\strA}
  \\
  \nocc{\varA}{\mset{\vtermA_1,\dotsc,\vtermA_k}}
  \eqdef \sum_{i=1}^k \nocc{\varA}{\vtermA_i}
  \qquad
  \nocc{\varA}{\bagA\cons\strB}
  \eqdef
  \begin{cases}
    0
    &\text{if \(\bagA\cons\strB=\emptystream\)}
    \\
    \nocc{\varA}{\bagA}+\nocc{\varA}{\strB}
    &\text{otherwise}
  \end{cases}
\end{gather*}
(choosing \(\svarA\not\ni\varA\) in the abstraction case)
and obtain that \(\nocc{\varA}{\exprA}\le\SizeOf{\exprA}\),
and that \(\nocc{\varA}{\exprA}=0\) iff \(\varA\not\in\LVarOf{\exprA}\),
by straightforward inductions.

\begin{example}\label{example:terms}
  Let us review the smallest (in terms of the size just defined,
  as well as in the sense of inductive constructions)
  terms we can consider in each syntactic category.
  The smallest bag term (resp.\ stream term) is of course the empty bag \(\emptybag\)
  (resp.\ the empty stream \(\emptystream\)).
  Given a variable \(\varA\), we can then form the base term
  \(\appl{\varA}{\emptystream}\).
  Abstracting over a sequence variable \(\svarB\),
  we obtain the closed value term
  \(\trProj{i}\eqdef\labs{\svarB}{\appl{\svarB[i]}{\emptystream}}\)
  if \(\varA=\svarB[i]\),
  or the value term \(\trProj{\varA}\eqdef\labs{\svarB}{\appl{\varA}{\emptystream}}\)
  where \(\varA\) occurs free if \(\svarB\not\ni\varA\).

  Then one can form more elaborate terms such as
  \[
    \trCCz{\varA}
    = \labs{\svarB}{
      \appl{\varA}{\mset{\trProj{\svarB[0]}}\cons\emptystream}
    }
    \qquad
    \text{and}
    \qquad
    \trCCp{\varA}
    = \labs{\svarB}{
        \appl{\varA}{
          \mset{ \trCCz{\svarB[0]}, \trCCz{\svarB[0]} }
          \cons\emptystream
        }
    }
  \]
  for each \(\varA\in\LVar\),
  that we will use as running examples below.
\end{example}

\paragraph{Sums of Terms.} As in the ordinary resource calculus, 
the reduction of resource terms produces sums:
a value (resp.\ base, bag, stream) term will reduce to a finite sum of value
(resp.\ base, bag, stream) terms.

If $X$ is a set, we write $\SumsOf{X}$ for the set of finite formal
sums of elements $X$ -- those may be equivalently presented as finite
multisets, but we adopt a distinct notation if only to impose a distinction
with bag terms.
Given a sum $A=\sum_{i\in I} a_i$, we write
$\support{A}\eqdef\set{a_i\st i\in I}$ for its \definitive{support} set.
We may abuse notation and write $a\in A$ instead of $a\in\support{A}$.

We call \definitive{value sums} (resp.\ \definitive{base sums}, \definitive{bag
sums}, \definitive{stream sums})
the elements of $\ValueSums$ (resp.\ $\BaseSums$, $\BagSums$, $\StreamSums$),
which we denote with capital letters $\vtermsA,\vtermsB,\vtermsC$ (resp.\ 
$\btermsA,\btermsB,\btermsC$; $\bagsA,\bagsB,\bagsC$; $\strsA,\strsB,\strsC$).
As announced we may write $\ResourceSums$
for any of $\ValueSums$, $\BaseSums$, $\BagSums$, or $\StreamSums$;
and we may call \definitive{term sum} any value sum, base sum, bag sum, or
stream sum, which we then denote by $\termsA,\termsB,\termsC$.

We also call \definitive{head sum} (resp.\ \definitive{resource sum})
any of a value sum (resp.\ term sum) or of a value variable,
which we denote by $\hexprsA,\hexprsB,\hexprsC$ (resp.\ by $\exprsA,\exprsB,\exprsC$).
We abuse notation and write $\HeadSums$ for $\LVar\cup\ValueSums$, and
then write $\ExprSums$ for any of $\HeadSums$, $\BaseSums$, $\BagSums$, or $\StreamSums$.
Again, we introduce head sums and expression sums only as technical devices
allowing us to simplify some definitions or proofs.
Note that we do not need to consider sums of head expressions mixing value
terms and variables.

We then extend term formers and operations to all syntactic categories by
linearity so that, for any finite sets $I$ and $J$, we have:
\begin{gather*}
  \labs{\svarA}*[\big]{\sum_{i\in I}\btermA_i}
  =\sum_{i\in I}\labs{\svarA}{\btermA_i}
  \qquad
  \appl*[\big]{\sum_{i\in I}{\vtermA_i}}{\strsB}
  =\sum_{i\in I} \appl{\vtermA_i}{\strsB} 
  \qquad
  \appl{\hexprA}*[\big]{\sum_{j\in J}{\strB_j}}
  =\sum_{j\in J} \appl{\hexprA}{\strB_j} 
  \\[1ex]
  \mset[\big]{\sum_{i\in I}\vtermA_i}
  =\sum_{i\in I} \mset{\vtermA_i}
  \qquad
  \pars[\big]{\sum_{i\in I}\bagA_i}\bagcat\pars[\big]{\sum_{j\in J}\bagB_j}
  =\sum_{\tuple{i,j}\in I\times J} \bagA_i\bagcat \bagB_j
  \\[1ex]
  \pars[\big]{\sum_{i\in I}\bagA_i}\cons\pars[\big]{\sum_{j\in J}\strB_j}
  =\sum_{\tuple{i,j}\in I\times J} \bagA_i\cons \strB_j\;.
\end{gather*}

We can thus extend ordinary substitution to sums of resource expressions:
\(\subst{\exprsA}{\varA}{\termsA}\) is obtained by
substituting \(\termsA\) for \(\varA\) in each summand of \(\exprsA\),
up to the previous identities.
Note this operation is linear in \(\exprsA\) but not in \(\termsA\).
In particular, \(\subst{\exprA}{\varA}{0}=0\) if \(\varA\in\LVarOf{\exprA}\),
and \(\subst{\exprA}{\varA}{0}=\exprA\) otherwise.

\paragraph{Two flavours of extensional resource reduction.}
A redex in the extensional resource calculus is naturally a base term of the
shape \(\appl*{\labs{\svarA}{\btermA}}{\strB}\).
Intuitively, this involves an infinite sequence of abstractions,
over the family of value variables \(\tuple{\svarA[i]}_{i\in\N}\),
applied to the sequence of bags in the stream \(\strB=\tuple{\bagB_i}_{i\in\N}\).
One way to reduce this redex is to apply what we will call a \emph{full step}:
this yields the base sum \(\rsubst{\btermA}{\svarA}{\strB}\) obtained from
the body \(\btermA\) by the resource substitution of each bag \(\bagB_i\) for
the corresponding variable \(\svarA[i]\).
As we have argued in the introduction, although this seems to involve
infinitely many simultaneous substitutions, only a finite number of them
have an effect:
for a sufficiently large \(i\), \(\svarA[i]\) does not occur in \(\btermA\),
and \(\bagB_i\) is empty.

Moreover, full-step reduction does not allow us to single out the application of
the first abstraction in the sequence to the first bag in the stream,
which would be the analogue of a redex in the \(λ\)-calculus,
via Taylor expansion.

Writing \(\labs{\svarA}{\btermA}=\labs{\varA}{\vtermA}\) and
\(\strB=\bagB\cons\strC\), we will therefore also consider 
the \emph{fine step} reducing \(\appl*{\labs{\svarA}{\btermA}}{\strB}\)
to \(\appl{\rsubst{\vtermA}{\varA}{\bagB}}{\strC}\).
Beware that the abstraction over \(\varA\) is
just a notation rather than a constructor in our syntax:
by this we mean that we consider the value term \(\vtermA=\labs{\svarB}{\btermB}\),
where \(\btermB\) is like \(\btermA\), but
with \(\svarA[0]\) renamed as \(\varA\),
and each \(\svarA[i+1]\) renamed as \(\svarB[i]\)
-- we will make this definition formal below.
Also note that each term in the base sum
\(\appl{\rsubst{\vtermA}{\varA}{\bagB}}{\strC}\) is again a redex:
to actually consume a redex, some form of full step is still needed.
Since each stream only contains finitely many non empty bags,
it is ultimately sufficient to consider a redex of the shape
\(\appl*{\labs{\svarA}{\btermA}}{\emptystream}\),
whose full-step reduction amounts to erase the whole term
(replacing it with \(0\)) in case any variable \(\svarA[i]\) occurs in \(\btermA\).
We take this as an additional case for fine-step reduction.

Both forms of resource reduction will be useful in our study of
extensional Taylor expansion.
We start with fine-step reduction, for which we prove confluence.
Then we consider full-step reduction, for which we prove strong normalization
(which fails in the fine-step case).
Combining those results, we deduce that each resource term 
has a unique normal form, common to both reductions.

\subsection{Fine-step dynamics}

To define fine-step reduction, we need three ingredients:
the \emph{resource substitution} of a bag for a variable;
a way to formalize the notation
\(\labs{\svarA}{\btermA}=\labs{\varA}{\vtermA}\),
which involves \emph{shifting} the indices of
the value variables \(\svarA[i]\);
and a description of the \emph{erasure} of \(\svarA\)
performed in \(\btermA\) when reducing 
\(\appl*{\labs{\svarA}{\btermA}}{\emptystream}\).

\paragraph{Resource substitution.}
We consider the resource substitution $\rsubst{\exprA}{\varA}{\bagB}$ of a bag
$\bagB=\mset{\vtermB_1,\dotsc,\vtermB_k}$ for a value variable $\varA$ in a resource
expression $\exprA$.
As for the ordinary resource calculus, the definition amounts to enumerate
the occurrences $\varA_1,\dotsc,\varA_l$ of $\varA$ in $\exprA$, and then set:
\[ 
  \rsubst{\exprA}{\varA}{\bagB}\eqdef\begin{cases}\sum_{σ\in\Perms k}
    \subst{\exprA}{\varA_{σ(1)},\dotsc,\varA_{σ(k)}}{\vtermB_1,\dotsc,\vtermB_k}
    &\text{if $k=l$}\\
    0&\text{otherwise}
  \end{cases}
\]
where $\Perms k$ is the set of all permutations of $\set{1,\dotsc,k}$,
and $\subst{\exprA}{\varA_{1},\dotsc,\varA_{k}}{\vtermB_1,\dotsc,\vtermB_k}$
denotes the simultaneous capture avoiding substitution of each term $\vtermA_i$
for the corresponding occurrence~$\varA_i$.

Note that, although it is intuitively clear, this definition relies on a notion
of occurrence that is not well defined because the order of elements in a bag
is not fixed.
To be carried out formally, one should introduce a rigid variant of the calculus,
and then show that the result does not depend on the choice of a rigid
representative.
A more annoying issue is that this global definition does not follow the
inductive structure of expressions.
We will prefer the following presentation:
\begin{definition}\label{definition:rsubst}
  We define the \definitive{resource substitution} $\rsubst{\exprA}{\varA}{\bagB}$
  of a bag term $\bagB$ for a value variable $\varA$ in a resource expression $\exprA$
  by induction on $\exprA$ as follows:
\begin{align*}
  \rsubst{\varB}{\varA}{\bagB}&\eqdef
  \begin{cases}
    \vtermB & \text{if $\varA=\varB$ and $\bagB=\mset{\vtermB}$}\\
    \varB  & \text{if $\varA\not =\varB$ and $\bagB=\emptybag$}\\
      0 & \text{otherwise}
  \end{cases}
  \\
  \rsubst*{\labs{\svarB}{\btermA}}{\varA}{\bagB}&\eqdef
  \labs{\svarB}{\rsubst{\btermA}{\varA}{\bagB}}
  \displaybreak[3]
  \\
  \rsubst*{\appl{\hexprA}{\strA}}{\varA}{\bagB}&\eqdef
  \sum_{\bagB\splitinto\bagB_1\bagcat\bagB_2}
  \appl{\rsubst{\hexprA}{\varA}{\bagB_1}}{\rsubst{\strA}{\varA}{\bagB_2}}
  \displaybreak[3]
  \\
  \rsubst{\mset{\vtermA_1,\dotsc,\vtermA_k}}{\varA}{\bagB}&\eqdef
  \sum_{\bagB\splitinto\bagB_1\bagcat\cdots\bagcat\bagB_k}
  \mset{\rsubst{\vtermA_1}{\varA}{\bagB_1},\dotsc,\rsubst{\vtermA_k}{\varA}{\bagB_k}}
  \displaybreak[3]
  \\
  \rsubst{\emptystream}{\varA}{\bagB}&\eqdef
  \begin{cases}
    \emptystream&\text{if $\bagB=\emptybag$}\\
    0&\text{otherwise}\\
  \end{cases}
  \displaybreak[1]
  \\
  \rsubst*{\bagA\cons\strC}{\varA}{\bagB}&\eqdef
  \sum_{\bagB\splitinto\bagB_1\bagcat\bagB_2}
  {\rsubst{\bagA}{\varA}{\bagB_1}}\cons{\rsubst{\strC}{\varA}{\bagB_2}}
  \qquad\text{if $\bagA\cons\strC\not=\emptystream$}
\end{align*}
where, in the abstraction case, $\svarB$ is chosen so that $\varA\not\in\svarB$ 
and $\svarB\cap\LVarOf{\bagB}=\emptyset$.
\end{definition}

Observe that if $\exprA$ is a variable
(resp.\ value term, base term, bag term, stream term)
then $\rsubst{\exprA}{\varA}{\bagB}$ is a head sum
(resp.\ value sum, base sum, bag sum, stream sum).
So we may write $\rsubst{\exprA}{\varA}{\bagB}\in\ExprSums$
(resp.\ $\rsubst{\termA}{\varA}{\bagB}\in\ResourceSums$)
if $\exprA\in\ResourceExprs$ (resp.\ $\termA\in\ResourceTerms$)
keeping implicit the fact that the underlying
syntactic category is the same.

Moreover note that the distinction between empty and non-empty streams is only made so 
that the definition is inductive and non-ambiguous.
It is nonetheless obvious that:
\[
  \sum_{\bagB\splitinto\bagB_1\bagcat\bagB_2}
  {\rsubst{\emptybag}{\varA}{\bagB_1}}\cons{\rsubst{\emptystream}{\varA}{\bagB_2}}
  =
  \begin{cases}
    \emptybag\cons\emptystream=\emptystream&\text{if $\bagB=\emptybag$}\\
    0&\text{otherwise}\\
  \end{cases}
\]
so that the condition $\bagA\cons\strC\not=\emptystream$
can be ignored in the last case of the definition
(a similar remark straightforwardly applies to the above definitions of
the size of a term, and of the number of occurrences of a variable in a term).

More generally, if $\strA=\bagA_1\cons\cdots\cons\bagA_k\cons\emptystream$,
then 
\[
  \rsubst{\strA}{\varA}{\bagB}=
  \sum_{\bagB\splitinto\bagB_1\bagcat\cdots\bagcat\bagB_k}
  {\rsubst{\bagA_1}{\varA}{\bagB_1}}\cons\cdots\cons
  {\rsubst{\bagA_k}{\varA}{\bagB_k}}\cons\emptystream
\]
and, equivalently,
\[
  \rsubst{\tuple{\bagA_i}_{i\in\N}}{\varA}{\bagB}=
  \sum_{p:\bagB\splitinto\N}
  \tuple{\rsubst{\bagA_i}{\varA}{\restrictToPre{\bagB}{p}{i}}}_{i\in\N}
\]
where we generalize $k$-partitionings to $\N$-partitionings in the obvious way.
Also, one can check that:
\[
  \rsubst{\exprA}{\varA}{\emptybag}
  \quad=\quad
  \begin{cases}
    0&\text{if $\varA\in\LVarOf{\exprA}$}\\
    \exprA&\text{otherwise}
  \end{cases}
  \quad=\quad
  \subst{\exprA}{\varA}{0}
  \;.
\]
And, in general,
\( \rsubst{\exprA}{\varA}{\bagB}\not=0 \)
iff \(\nocc{\varA}{\exprA}=\CardOf{\bagB}\).

\begin{example}\label{example:substitution}
  Using the notations of \cref{example:terms},
  we obtain that \(\rsubst{\trProj{i}}{\varA}{\emptybag}=\trProj{i}\),
  and \(\rsubst{\trProj{i}}{\varA}{\bagA}=0\)
  as soon as \(\bagA\not=\emptybag\).
  And we obtain
  \(\rsubst{\trProj{\varA}}{\varA}{\mset{\vtermA}}
  =\labs{\svarB}{\appl{\vtermA}{\emptystream}}\)
  for any value term \(\vtermA\) (choosing \(\svarB\not\in\SVarOf{\vtermA}\)), 
  and \(\rsubst{\trProj{\varA}}{\varA}{\bagA}=0\)
  as soon as \(\LengthOf{\bagA}\not=1\).

  Minimalistic examples yielding sums are \[
  \rsubst{\mset{\trProj{\varA},\trProj{\varA}}}{\varA}{\mset{\vtermC_1,\vtermC_2}}
  =\mset{
    \labs{\svarB}{\appl{\vtermC_1}{\emptystream}},
    \labs{\svarB}{\appl{\vtermC_2}{\emptystream}}
  }+\mset{
    \labs{\svarB}{\appl{\vtermC_2}{\emptystream}},
    \labs{\svarB}{\appl{\vtermC_1}{\emptystream}}
  }
  =2\mset{
    \labs{\svarB}{\appl{\vtermC_1}{\emptystream}},
    \labs{\svarB}{\appl{\vtermC_2}{\emptystream}}
  }
  \] and \[
  \rsubst*[\big]{\mset{\trProj{\varA}}\cons\mset{\trProj{\varA}}\cons\emptystream}
  {\varA}{\mset{\vtermC_1,\vtermC_2}}
  =
  \mset{\labs{\svarB}{\appl{\vtermC_1}{\emptystream}}}
  \cons\mset{\labs{\svarB}{\appl{\vtermC_2}{\emptystream}}}
  \cons\emptystream
  +
  \mset{\labs{\svarB}{\appl{\vtermC_2}{\emptystream}}}
  \cons\mset{\labs{\svarB}{\appl{\vtermC_1}{\emptystream}}}
  \cons\emptystream
  \]
  for any value terms \(\vtermC_1\) and \(\vtermC_2\)
  (choosing \(\svarB\not\in\SVarOf{\vtermC_1}\cup\SVarOf{\vtermC_2}\)).
\end{example}

Contrarily to what happens with ordinary substitution,
the combinatorics of resource substitution is very regular.
Indeed, except for the occurrences of the variable that is substituted,
nothing is erased or discarded from the substituted bag nor from the expression 
in which the substitution takes place.
In particular, the size of the terms produced by a substitution is determined by
the length and size of the substituted bag and the size of the term in which
the substitution is performed; and free variables are preserved.
Formally, a straightforward induction on resource expressions yields:
\begin{lemma}\label{lemma:rsubst:size}
  If $\exprA'\in\rsubst{\exprA}{\varA}{\bagB}$ then 
  $\nocc{\varA}{\exprA}=\LengthOf{\bagB}$ and
  $\SizeOf{\exprA'}=\SizeOf{\exprA}+\SizeOf{\bagB}-\LengthOf{\bagB}$.
  If moreover $\varB\not=\varA$ then $\nocc{\varB}{\exprA'}=
  \nocc{\varB}{\exprA}+\nocc{\varB}{\bagB}$ (in particular,
  $\LVarOf{\exprA'} = \LVarOf{\exprA}\cup\LVarOf{\bagB}$).
\end{lemma}

Resource substitution is extended to sums by linearity:
\[\rsubst*[\bigg]{\sum_{i=1}^k\exprA_i}{\varA}[\bigg]{\sum_{j=1}^l\bagA_j}
  \eqdef
\sum_{i=1}^k\sum_{j=1}^l \rsubst{\exprA_i}{\varA}{\bagA_j}
\;.
\]

One can easily check that, with that extension,
all the identities defining resource substitution
in \cref{definition:rsubst} also hold if we replace
terms with sums.

The usual result on the commutation of substitutions holds:
\begin{lemma}\label{lemma:rsubst:rsubst}
  We have $
  \rsubst{\rsubst{\exprA}{\varA}{\bagA}}{\varB}{\bagB}=
  \sum_{\bagB\splitinto\bagB_1\bagcat\bagB_2}
  \rsubst{\rsubst{\exprA}{\varB}{\bagB_1}}{\varA}{\rsubst{\bagA}{\varB}{\bagB_2}}$
  whenever $\varA\not\in\LVarOf{\bagB}\cup\set{\varB}$.
\end{lemma}
\begin{proof}
  The proof is straightforward by induction on $\exprA$,
  using the associativity of sums over partitionings.
\end{proof}

It is also easy to describe the effect of performing an ordinary
substitution after a resource substitution:
\begin{lemma}\label{lemma:rsubst:subst}
  We have $
  \subst{\rsubst{\exprA}{\varA}{\bagA}}{\varB}{\vtermsB}=
  \rsubst{\subst{\exprA}{\varB}{\vtermsB}}{\varA}{\subst{\bagA}{\varB}{\vtermsB}}$
  whenever $\varA\not\in\LVarOf{\vtermsB}\cup\set{\varB}$.
\end{lemma}
Reversing the order of substitutions, there is no straightforward way to
describe the resource sum \(\rsubst{\subst{\exprA}{\varA}{\vtermsA}}{\varB}{\bagB}\),
because the occurrences of \(\varB\) in \(\vtermsA\) may be duplicated.
Still, it is easy to treat the special case where \(\vtermsA\) is reduced to a
variable, as the ordinary substitution amounts to a renaming of variables:
\begin{lemma}\label{lemma:subst:rsubst:var}
  We have \(\rsubst{\subst{\exprA}{\varA}{\varB}}{\varB}{\bagB}=
  \rsubst{\exprA}{\varA}{\bagB}\) if \(\varA=\varB\)
  or \(\varB\not\in\LVarOf{\exprA}\).
\end{lemma}

\paragraph{Shifting and erasing sequence variables.}
If $\exprA$ is a resource expression and $\svarA\in\SVar$,
we write $\shiftup{\exprA}{\svarA}$ for the term obtained
by replacing each $\svarA[i]$ occurring free in $\exprA$ with $\svarA[i+1]$.
Similarly, if $\svarA[0]$ does not occur in $\exprA$, we write 
$\shiftdown{\exprA}{\svarA}$ for the expression obtained
by replacing each $\svarA[i+1]$ in $\exprA$ with $\svarA[i]$,
so that $\exprA=\shiftup{\shiftdown{\exprA}{\svarA}}{\svarA}$
in this case
--- and $\exprB=\shiftdown{\shiftup{\exprB}{\svarA}}{\svarA}$
for any expression $\exprB$.
Given a variable $\varA$ and a value term $\vtermA=\labs{\svarA}{\btermA}$, 
with $\svarA$ chosen so that $\varA\not\in\svarA$,
we define $
\labs{\varA}{\vtermA}\eqdef
\labs{\svarA}{\subst{\shiftup{\btermA}{\svarA}}{\varA}{\svarA[0]}}
$.

\begin{example}\label{example:abstraction}
  Using the terms of \cref{example:terms} again, we have (assuming \(\varA\not\in\set{\varB,\varC}\cup\svarC\) and \(\varB\not\in\svarC\))
  \begin{gather*}
    \trProj{\varA}
    =\labs{\varC}{\trProj{\varA}}
    \qquad
    \trProj{0}
    =\labs{\varC}{\trProj{\varC}}
    \qquad
    \trProj{i+1}
    =\labs{\varC}{\trProj{i}}
    \qquad
    \trCCz{\varA}
    = \labs{\varB}{\labs{\svarC}{
        \appl{\varA}{\mset{\trProj{\varB}}\cons\emptystream}
    }}
    \qquad
    \trCCp{\varA}
    = \labs{\varB}{\labs{\svarC}{
        \appl{\varA}{
          \mset{ \trCCz{\varB}, \trCCz{\varB} }
          \cons\emptystream
        }
    }}
    \;.
  \end{gather*}
\end{example}

We define the \definitive{erasure} of the sequence variable $\svarA$ 
in an expression $\exprA$ by:
\[ \erase{\exprA}{\svarA}\eqdef\begin{cases}\exprA &\text{if $\svarA\cap\LVarOf{\exprA}=\emptyset$}\\0&\text{otherwise}\end{cases}\;.\]
In other words, $\erase{\exprA}{\svarA}=0$ if some $\svarA[i]$ occurs free in $\exprA$
and $\erase{\exprA}{\svarA}=\exprA$ otherwise.

\begin{lemma}\label{lemma:erase:shift}
  We have $\erase{\shiftup{\exprA}{\svarA}}{\svarA}=\erase{\exprA}{\svarA}$.
  Moreover, if $\svarA[0]\not\in\LVarOf{\exprA}$ 
  then $\erase{\exprA}{\svarA}=\erase{\shiftdown{\exprA}{\svarA}}{\svarA}$.
\end{lemma}
\begin{proof}
  Direct from the definitions, since
  $\svarA[i+1]\in\LVarOf{\shiftup{\exprA}{\svarA}}$
  iff
  $\svarA[i]\in\LVarOf{\exprA}$,
  and
  $\svarA[i]\in\LVarOf{\shiftdown{\exprA}{\svarA}}$
  iff
  $\svarA[i+1]\in\LVarOf{\exprA}$.
\end{proof}

Both shifts and erasure are linearly extended to resource sums:
\[
  \shiftup*{\sum_{i=1}^k \exprA_i}{\svarA}
  \eqdef \sum_{i=1}^k \shiftup{\exprA_i}{\svarA}
  \qquad
  \shiftdown*{\sum_{i=1}^k \exprA_i}{\svarA}
  \eqdef \sum_{i=1}^k \shiftdown{\exprA_i}{\svarA}
  \qquad
  \erase*{\sum_{i=1}^k \exprA_i}{\svarA}
  \eqdef \sum_{i=1}^k \erase{\exprA_i}{\svarA}
\]
and these operations commute with substitution in the following sense:
\begin{lemma}\label{lemma:shift:rsubst}
  If $\svarA\not\in\SVarOf{\bagsA}$, then
  $\shiftup{\rsubst{\exprsA}{\svarA[i]}{\bagsA}}{\svarA}=\rsubst{\shiftup{\exprsA}{\svarA}}{\svarA[i+1]}{\bagsA}$ and
  $\shiftdown{\rsubst{\exprsA}{\svarA[i+1]}{\bagsA}}{\svarA}=\rsubst{\shiftdown{\exprsA}{\svarA}}{\svarA[i]}{\bagsA}$
  (assuming $\svarA[0]\not\in\LVarOf{\exprsA}$ in that case).
  And if $\varA\not\in\svarA$, then we have
  $\shiftup{\rsubst{\exprsA}{\varA}{\bagsA}}{\svarA}=\rsubst{\shiftup{\exprsA}{\svarA}}{\varA}{\shiftup{\bagsA}{\svarA}}$,
  $\shiftdown{\rsubst{\exprsA}{\varA}{\bagsA}}{\svarA}=\rsubst{\shiftdown{\exprsA}{\svarA}}{\varA}{\shiftdown{\bagsA}{\svarA}}$
  (assuming $\svarA[0]\not\in\LVarOf{\exprsA}\cup\LVarOf{\bagsA}$ in that case)
  and
  $\erase{\rsubst{\exprsA}{\varA}{\bagsA}}{\svarA}=\rsubst*{\erase{\exprsA}{\svarA}}{\varA}{\erase{\bagsA}{\svarA}}$.
\end{lemma}
\begin{proof}
  In case $\exprsA=\exprA\in\ResourceExprs$, each result follows by a
  straightforward induction on $\exprA$.
  We deduce the general result by linearity.
\end{proof}

\paragraph{Resource reduction.}
Now, we can define fine-step resource reduction by:
\[
  \appl*{\labs{\varA}{\vtermA}}{\bagB\cons\strC}
  \resredz \appl{\rsubst{\vtermA}{\varA}{\bagB}}{\strC}
  \qquad
  \text{and}
  \qquad
  \appl*{\labs{\svarA}{\btermA}}{\emptystream}
  \resredz \erase{\btermA}{\svarA}\,.
\]
Note that the first reduction step,
which uses the just introduced notation for the abstraction of a single variable,
amounts to:
\[
  \appl*{\labs{\svarA}{\btermA}}{\bagB\cons\strC} 
  \resredz \appl*[\big]{\labs{\svarA}{\shiftdown{\rsubst{\btermA}{\svarA[0]}{\bagB}}{\svarA}}}{\strC}
\]
(choosing $\svarA\not\in\SVarOf{\bagB}$)
thanks to \cref{lemma:shift:rsubst}.

\begin{figure}
  \begin{gather*}
    \begin{prooftree}
      \infer0[\resredrulebeta*]{
        \appl*{\labs{\varA}{\vtermA}}{\bagB\cons\strC}
        \resred \appl{\rsubst{\vtermA}{\varA}{\bagB}}{\strC}
      }
    \end{prooftree}
    \qquad
    \begin{prooftree}
      \infer0[\resredruleiota*]{
        \appl*{\labs{\svarA}{\btermA}}{\emptystream} \resred \erase{\btermA}{\svarA}
      }
    \end{prooftree}
    \\[1em]
    \begin{prooftree}
      \hypo{\btermA\resred\btermsA'}
      \infer1[\resredruleabs*]{\labs{\svarA}{\btermA}\resred\labs{\svarA}{\btermsA'}}
    \end{prooftree}
    \qquad
    \begin{prooftree}
      \hypo{\vtermA\resred\vtermsA'}
      \infer1[\resredruleappl*]{\appl{\vtermA}{\strB}\resred\appl{\vtermsA'}{\strB}}
    \end{prooftree}
    \qquad
    \begin{prooftree}
      \hypo{\strB\resred\strsB'}
      \infer1[\resredruleappr*]{\appl{\hexprA}{\strB}\resred\appl{\hexprA}{\strsB'}}
    \end{prooftree}
    \\[1em]
    \begin{prooftree}
      \hypo{\vtermA\resred\vtermsA'}
      \infer1[\resredrulebag*]{\mset{\vtermA}\bagcat\bagB\resred\mset{\vtermsA'}\bagcat\bagB}
    \end{prooftree}
    \qquad
    \begin{prooftree}
      \hypo{\bagA\resred\bagsA'}
      \infer1[\resredrulestrl*]{\bagA\cons\strB\resred\bagsA'\cons\strB}
    \end{prooftree}
    \qquad
    \begin{prooftree}
      \hypo{\strB\resred\strsB'}
      \infer1[\resredrulestrr*]{\bagA\cons\strB\resred\bagA\cons\strsB'}
    \end{prooftree}
  \end{gather*}
  \caption{
    Rules of fine-step extensional resource reduction
  }
  \label{fig:finestep}
\end{figure}

We extend these base reduction steps by applying them under any context,
and then applying them in parallel within sums of expressions.
Formally:
\begin{definition}
  \definitive{Fine-step resource reduction} is the relation from resource expressions to
  sums of resource expressions defined by the rules of \cref{fig:finestep}.

  Resource reduction is then extended to a relation on term sums
  by setting $\termsA\sresred\termsA'$ iff 
  $\termsA=\sum_{i=1}^k\termA_i$ and 
  $\termsA'=\sum_{i=1}^k\termsA'_i$ 
  with $\termA_i\resredR \termsA'_i$ for $1\le i\le k$,
  where $\resredR$ is the reflexive closure of $\resred$
  ($\termA\resredR\termsA'$ if $\termA\resred\termsA'$ or $\termsA'=\termA$).

  We also write \(\sresredRT\) for the reflexive and transitive closure
  of \(\sresred\), and \(\sresred^k\) for its \(k\)-th iteration.
\end{definition}

Note that, as particular case,
whenever \(\svarB\not\in\SVarOf{\btermA}\),
we have
\[
  \appl*{\labs{\varA}{\labs{\svarB}{\btermA}}}{\bagB\cons\emptystream}
  \sresred
  \appl*{\labs{\svarB}{\rsubst{\btermA}{\varA}{\bagB}}}{\emptystream}
  \sresred
  \erase{\rsubst{\btermA}{\varA}{\bagB}}{\svarB}
  =\rsubst{\btermA}{\varA}{\bagB}\,.
\]

\begin{example}\label{example:reduction:fine}
  The smallest reducible base terms one can form are
  \(\appl{\trProj{\varA}}{\emptystream}\) for \(\varA\in\LVar\),
  and
  \(\appl{\trProj{i}}{\emptystream}\) for \(i\in\N\).
  Recalling the computations of \cref{example:abstraction},
  we obtain
  \(
  \appl{\trProj{\varA}}{\emptystream}
  \resred
  \appl{\rsubst{\trProj{\varA}}{\varC}{\emptybag}}{\emptystream}
  =
  \appl{\trProj{\varA}}{\emptystream}
  \)
  by \(\resredrulebeta\),
  and 
  \(
  \appl{\trProj{\varA}}{\emptystream}
  \resred
  \erase*{\appl{\varA}{\emptystream}}{\svarB}
  =
  \appl{\varA}{\emptystream}
  \)
  by \(\resredruleiota\).
  Similarly, we obtain 
  \(
    \appl{\trProj{0}}{\emptystream}
    \resred
    \appl{
      \rsubst{\trProj{\varA}}{\varA}{\emptybag}
    }
    {\emptystream}
    =
    \appl{0}{\emptystream}
    = 0
  \)
  and
  \(
    \appl{\trProj{i+1}}{\emptystream}
    \resred
    \appl{
      \rsubst{\trProj{i}}{\varA}{\emptybag}
    }
    {\emptystream}
    =
    \appl{\trProj{i}}{\emptystream}
  \)
  by \(\resredrulebeta\),
  and \(
    \appl{\trProj{i}}{\emptystream}
    \resred 
    \erase*{\appl{\svarB[i]}{\emptystream}}{\svarB}
    =
    \appl{0}{\emptystream}
    = 0
  \)
  by \(\resredruleiota\).

  We moreover have:
  \(
    \rsubst{\trProj{\varB}}{\varB}{\mset{\trProj{\varA}}}
    =
    \labs{\svarC}{\appl{\trProj{\varA}}{\emptystream}}
    \resred
    \labs{\svarC}{\appl{\varA}{\emptystream}}
    =
    \trProj{\varA}
  \)
  for any variables \(\varA\) and \(\varB\),
  whence the following sequence
  of reductions:
  \[
    \appl{\trProj{0}}{\mset{\trProj{\varA}}\cons\emptystream}
    \sresred
    \appl{\rsubst{\trProj{\varA}}{\varA}{\mset{\trProj{\varA}}}}{\emptystream}
    \sresred
    \appl{\trProj{\varA}}{\emptystream}
    \sresred
    \appl{\varA}{\emptystream}
  \]
  ending on an irreducible base term.
  Note that the middle step can be derived either by reducing
  \(\rsubst{\trProj{\varB}}{\varB}{\mset{\trProj{\varA}}}\)
  to \(\trProj{\varA}\) as above,
  or by reducing 
  \(
  \appl{
    \rsubst{\trProj{\varB}}{\varB}{\mset{\trProj{\varA}}}
  }{\emptystream}=
  \appl*{
    \labs{\svarC}{\appl{\trProj{\varA}}{\emptystream}}
  }{\emptystream}
  \resred
  \erase*{\appl{\trProj{\varA}}{\emptystream}}{\svarC}
  =
  \appl{\trProj{\varA}}{\emptystream}
  \)
  by \(\resredruleiota\).

  The previous computations entail
  \(\labs{\svarC}{\appl{\trProj{0}}{\mset{\trProj{\varA}}\cons\emptystream}}
  \sresred^3\trProj{\varA}\)
  and 
  \(\labs{\svarC}{\appl{\trProj{0}}{\mset{\trProj{\svarC[i]}}\cons\emptystream}}
  \sresred^3\trProj{i}\).
  This suggests defining a family of value terms \(\trProj{\varA}[k]\) for \(k\in\N\),
  by setting \(\trProj{\varA}[0]\eqdef\trProj{\varA}\)
  and \(\trProj{\varA}[k+1]\eqdef\labs{\svarC}{\appl{\trProj{0}}{\mset{\trProj{\varA}[k]}\cons\emptystream}}\).
  The previous computation shows \(\trProj{\varA}[1]\sresred^{3}\trProj{\varA}[0]\),
  and an easy generalization entails \(\trProj{\varA}[k]\sresred^{3k}\trProj{\varA}\).
  Similarly setting \(\trProj{i}[0]\eqdef\trProj{i}\)
  and \(\trProj{i}[k+1]\eqdef\labs{\svarC}{\appl{\trProj{0}}{\mset{\trProj{\svarC[i]}[k]}\cons\emptystream}}\),
  we obtain \(\trProj{i}[k]\sresred^{3k}\trProj{i}\).
  A more intricate example is:
  \begin{align*}
    \appl{\trCCp{\varA}}{\mset{\trProj{0},\trProj{0}}\cons\emptystream}
    &
    \sresred^2
    \rsubst*{
      \appl{\varA}{
        \mset{ \trCCz{\varB}, \trCCz{\varB} }
        \cons\emptystream
      }
    }{\varB}{\mset{\trProj{0},\trProj{0}}}
    \\
    &=
    2 \appl{\varA}{
      \mset{ 
        \rsubst{\trCCz{\varB}}{\varB}{\mset{\trProj{0}}},
        \rsubst{\trCCz{\varB}}{\varB}{\mset{\trProj{0}}}
      }
      \cons\emptystream
    }
    \\
    &=
    2 \appl{\varA}{
      \mset{ 
        \labs{\svarC}{
          \appl{\trProj{0}}{\mset{\trProj{\svarC[0]}}\cons\emptystream}
        },
        \labs{\svarC}{
          \appl{\trProj{0}}{\mset{\trProj{\svarC[0]}}\cons\emptystream}
        }
      }
      \cons\emptystream
    }
    \\
    &=
    2 \appl{\varA}{
      \mset{ 
        \trProj{0}[1],
        \trProj{0}[1]
      }
      \cons\emptystream
    }
    \\
    &\sresred^6
    2 \appl{\varA}{
      \mset{
        \trProj{0},
        \trProj{0}
      }
      \cons\emptystream
    }
    \,.
  \end{align*}
\end{example}

Observe that $\sresred$ is automatically reflexive.
The reader acquainted with the ordinary resource calculus might 
be surprised by this choice of definition, as it immediately forbids 
strong normalizability.
However, $\resred$ is already reflexive on terms
containing a redex of the form
$\btermA=\appl*{\labs{\svarA}{\btermB}}{\emptystream}$ 
with $\svarA\cap\LVarOf{\btermB}=\emptyset$:
applying rule $\resredrulebeta$ yields a reduction
$\btermA\resred\btermA$.
Nonetheless, resource reduction is weakly normalizing,
in the sense that any resource sum reduces to a sum of
irreducible expressions:
this will follow from strong normalizability of full-step reduction
--
the latter essentially amounts to iterate $\resredrulebeta$ until
$\resredruleiota$ can be applied.
For now, we establish the confluence of fine-step reduction.

First, the extension of $\resred$ to $\sresred$ ensures that $\sresred$ is
linear and compatible (with syntactic constructs), in the following sense:
\begin{lemma}\label{lemma:sresred:context}
  Resource reduction $\sresred$ is reflexive and:
  \begin{enumerate}
    \item if $\termsA\sresred\termsA'$
      and $\termsB\sresred\termsB'$
      then $\termsA+\termsB\sresred\termsA'+\termsB'$;
    \item if $\btermsA\sresred\btermsA'$
      then $\labs{\svarA}{\btermsA}\sresred\labs{\svarA}{\btermsA'}$;
    \item if $\strsB\sresred\strsB'$
      then $\appl{\hexprsA}{\strsB}\sresred\appl{\hexprsA}{\strsB'}$
      and $\bagsA\cons\strsB\sresred\bagsA\cons\strsB'$;
    \item if $\vtermsA\sresred\vtermsA'$
      then $\labs{\varA}{\vtermsA}\sresred\labs{\varA}{\vtermsA'}$,
      $\appl{\vtermsA}{\bagsB}\sresred\appl{\vtermsA'}{\bagsB}$
      and $\mset{\vtermsA}\sresred\mset{\vtermsA'}$;
    \item if $\bagsA\sresred\bagsA'$
      then $\bagsA\bagcat\bagsB\sresred\bagsA'\bagcat\bagsB$
      and $\bagsA\cons\strsB\sresred\bagsA'\cons\strsB$.
  \end{enumerate}
  Moreover, 
  $\appl*{\labs{\svarA}{\btermsA}}*{\bagsB\cons\strsC}\sresred
  \appl*[\big]{\labs{\svarA}{\shiftdown{\rsubst{\btermsA}{\svarA[0]}{\bagsB}}{\svarA}}}{\strsC}$
  and
  $\appl*{\labs{\svarA}{\btermsA}}{\emptystream} \sresred \erase{\btermsA}{\svarA}$.
\end{lemma}
\begin{proof}
  Each statement is a direct consequence of the definitions,
  by linearity.
\end{proof}

We then obtain that $\sresred$ is compatible with substitution in the following sense:
\begin{lemma}\label{lemma:sresred:rsubst}
  If $\termsA\sresred\termsA'$ then 
  for any bag sum $\bagsB$ we have
  $\rsubst{\termsA}{\varA}{\bagsB}\sresred\rsubst{\termsA'}{\varA}{\bagsB}$.
  And if $\bagsA\sresred\bagsA'$ then
  for any resource sum $\exprsA$ we have
  $\rsubst{\exprsA}{\varA}{\bagsA}\sresred\rsubst{\exprsA}{\varA}{\bagsA'}$.
\end{lemma}
\begin{proof}
  We establish the first statement in the case $\termsA=\termA\resred\termsA'$,
  by induction on that reduction.
  For the two base cases, we use \cref{lemma:rsubst:rsubst,lemma:shift:rsubst}.
  The other cases follow directly from the induction hypothesis,
  and the extension to $\sresred$ is straightforward.

  We establish the second statement in the case $\exprsA=\exprA\in\ResourceExprs$
  and $\bagsA=\bagA\resred\bagsA'$, by induction on $\exprA$.
  Note that we must have $\bagA=\mset{\vtermA}\bagcat\bagB$
  and $\bagsA'=\mset{\vtermsA'}\bagcat\bagB$
  with $\vtermA\resred\vtermsA'$:
  then all cases are straightforward, noting that any sum
  $\sum_{\bagA\splitinto \bagA_1\bagcat\bagA_2} f(\bagA_1,\bagA_2)$
  can be written as
  $\sum_{\bagB\splitinto \bagB_1\bagcat\bagB_2} f(\mset{\vtermA}\bagcat\bagB_1,\bagB_2)+f(\bagB_1,\mset{\vtermA}\bagcat\bagB_2)$.
  The general result follows by linearity.
\end{proof}

Moreover, resource reduction preserves free variables and is compatible with
shifts and erasure:
\begin{lemma}\label{lemma:resred:nocc}
  If $\termA\resred\termsA'$ 
  and $\termA'\in\termsA'$
  then $\nocc{\varA}{\termA}=\nocc{\varA}{\termA'}$.
  In particular, $\LVarOf{\termA}=\LVarOf{\termA'}$.
\end{lemma}
\begin{proof}
  Straightforward by induction on the reduction
  $\termA\resred\termsA'$, using \cref{lemma:rsubst:size}
  in the case of rule $\resredrulebeta$.
\end{proof}

\begin{lemma}\label{lemma:sresred:shift}
  If $\termsA\sresred\termsA'$ 
  then $\erase{\termsA}{\svarA}\sresred\erase{\termsA'}{\svarA}$
  and $\shiftup{\termsA}{\svarA}\sresred\shiftup{\termsA'}{\svarA}$.
  If moreover $\svarA[0]\not\in\LVarOf{\termsA}$ then 
  $\shiftdown{\termsA}{\svarA}\sresred\shiftdown{\termsA'}{\svarA}$.
\end{lemma}
\begin{proof}
  Again, the result is proved first for a reduction $\termA\resred\termsA'$,
  then generalized by linearity.
  Each step is straightforward.
\end{proof}

\begin{theorem}[Confluence of $\sresred$]\label{theorem:resred:confluent}
  Resource reduction $\sresred$ has the diamond property:
  if $\termsA\sresred\termsA_1$ and $\termsA\sresred\termsA_2$ then
  there exists $\termsA'$ such that $\termsA_1\sresred\termsA'$
  and  $\termsA_2\sresred\termsA'$.
\end{theorem}
\begin{proof}
  We first establish, by induction on resource terms, that:
  if $\termA\resred\termsA_1$ and $\termA\resred\termsA_2$ then 
  there exists $\termsA'$ such that $\termsA_1\sresred\termsA'$
  and  $\termsA_2\sresred\termsA'$.

  The crucial case is that of head reducible base terms.
  Assume, e.g., $\termsA_1$ is obtained by $\resredrulebeta$:
  then $\termA=\appl*{\labs{\svarA}{\btermA}}*{\bagB\cons\strC}$, and
  $\termsA_1=
  \appl*[\big]{\labs{\svarA}{\shiftdown{\rsubst{\btermA}{\svarA[0]}{\bagB}}{\svarA}}}{\strC}$.

  If $\termsA_2$ is obtained by $\resredrulebeta$ too,
  then $\termsA_1=\termsA_2$ and we conclude by the reflexivity of $\sresred$,
  setting $\termsA'\eqdef\termsA_1$.

  If $\termsA_2$ is obtained by $\resredruleiota$ then $\bagB\cons\strC=\emptystream$ and:
  either $\svarA[0]\in\LVarOf{\btermA}$ and $\termsA_1=\termsA_2=0$,
  and we conclude again by reflexivity;
  or $\svarA[0]\not\in\LVarOf{\btermA}$ and
  $\termsA_1
  =\appl*{\labs{\svarA}{\shiftdown{\btermA}{\svarA}}}{\emptystream}
  \resred\erase{\shiftdown{\btermA}{\svarA}}{\svarA}
  =\termsA_2$
  by \cref{lemma:erase:shift}.

  If $\termsA_2$ is obtained by $\resredruleappl$
  then $\termsA_2=\appl*{\labs{\svarA}{\btermsA'}}*{\bagB\cons\strC}$
  with $\btermA\resred\btermsA'$, and we can set 
  $\termsA'\eqdef
  \appl*[\big]{\labs{\svarA}{\shiftdown{\rsubst{\btermsA'}{\svarA[0]}{\bagB}}{\svarA}}}{\strC}$,
  to obtain 
  $\termsA_1\sresred\termsA'$ by \cref{lemma:sresred:rsubst,lemma:sresred:shift,lemma:sresred:context},
  and 
  $\termsA_2\sresred\termsA'$ by \cref{lemma:sresred:context}.
  
  If $\termsA_2$ is obtained by $\resredruleappr$
  then either 
  $\termsA_2=\appl*{\labs{\svarA}{\btermA}}*{\bagsB'\cons\strC}$
  with $\bagB\resred\bagsB'$, and we can set 
  $\termsA'\eqdef
  \appl*[\big]{\labs{\svarA}{\shiftdown{\rsubst{\btermA}{\svarA[0]}{\bagsB'}}{\svarA}}}{\strC}$,
  to obtain 
  $\termsA_1\sresred\termsA'$ by \cref{lemma:sresred:rsubst,lemma:sresred:shift,lemma:sresred:context},
  and 
  $\termsA_2\sresred\termsA'$ by \cref{lemma:sresred:context};
  or 
  $\termsA_2=\appl*{\labs{\svarA}{\btermA}}*{\bagB\cons\strsC'}$
  with $\strC\resred\strsC'$, and we can set 
  $\termsA'\eqdef
  \appl*[\big]{\labs{\svarA}{\shiftdown{\rsubst{\btermA}{\svarA[0]}{\bagB}}{\svarA}}}{\strsC'}$,
  to obtain both
  $\termsA_1\sresred\termsA'$ 
  and 
  $\termsA_2\sresred\termsA'$ by \cref{lemma:sresred:context}.
  
  By symmetry, we have treated all the cases where at least one of $\termsA_1$
  or $\termsA_2$ is obtained by $\resredrulebeta$.
  Now assume, e.g., $\termsA_1$ is obtained by $\resredruleiota$:
  then $\termA=\appl*{\labs{\svarA}{\btermA}}{\emptystream}$, and
  $\termsA_1=\erase{\btermA}{\svarA}$.
  Note that $\resredruleappr$ cannot be applied in this case.
  If $\termsA_2$ is also obtained by $\resredruleiota$,
  we have $\termsA_1=\termsA_2$ and we conclude by reflexivity.
  This leaves only the case of $\termsA_2$ being obtained by 
  $\resredruleappl$:
  $\termsA_2=\appl*{\labs{\svarA}{\btermsA'}}{\emptystream}$,
  with $\btermA\resred\btermsA'$.
  Then we set $\termsA'\eqdef\erase{\btermsA'}{\svarA}$
  and obtain
  $\termsA_1\sresred\termsA'$ by \cref{lemma:sresred:shift} and
  $\termsA_2\sresred\termsA'$ by \cref{lemma:sresred:context}.

  We are only left with compatibility rules, all falling in two cases:
  if both rules reduce the same subterm (\emph{i.e.}\ the rules are the same
  and, in the case of $\resredrulebag$, the reduced subterm is the same),
  we conclude by the induction hypothesis, together with \cref{lemma:sresred:context};
  otherwise, if the rules are distinct
  (\emph{i.e.}\ $\resredruleappl$ \emph{vs} $\resredruleappr$ for base terms,
  or $\resredrulestrl$ \emph{vs} $\resredrulestrr$ for stream terms)
  or are two instances of $\resredrulebag$ on distinct subterms,
  (\emph{i.e.}\ 
  $\termA=\mset{\vtermA_1}\bagcat\mset{\vtermA_2}\bagcat\bagB$,
  $\termsA_1=\mset{\vtermsA'_1}\bagcat\mset{\vtermA_2}\bagcat\bagB$, and
  $\termsA_2=\mset{\vtermA_1}\bagcat\mset{\vtermsA'_2}\bagcat\bagB$,
  with each $\vtermA_i\resred\vtermsA'_i$)
  then we apply \cref{lemma:sresred:context} directly.

  Now, if \(\termA\resredR\termsA_1\) and \(\termA\resredR\termsA_2\),
  then we obtain \(\termsA'\) such that
  \(\termsA_1\sresred\termsA'\) and \(\termsA_2\sresred\termsA'\):
  if \(\termsA_1=\termA\) or \(\termsA_2=\termA\),
  the result is trivial;
  and otherwise, the previous result applies.

  Now we extend the result to the reduction of term sums:
  assume $\termsA\sresred\termsA_1$ and $\termsA\sresred\termsA_2$.
  We can write
  \[
    \termsA=\sum_{i=1}^{k}\termA_i,
    \quad
    \termsA_1=\sum_{i=1}^{k}\termsA^1_i
    \quad\text{and}\quad
    \termsA_2=\sum_{i=1}^{k}\termsA^2_i
  \]
  so that \(\termA_i\resredR\termsA^j_i\) for \(1\le i\le k\) and \(1\le j\le 2\).
  We obtain $\termsA'_i$ such that 
  $\termsA^j_i\sresred\termsA'_i$ for \(1\le i\le k\) and \(1\le j\le 2\),
  and then set \(\termsA'=\sum_{i=1}^{k}\termsA'_i\),
  so that $\termsA_1\sresred\termsA'$ and $\termsA_2\sresred\termsA'$
  by \cref{lemma:sresred:context}.
\end{proof}

\subsection{Full-step dynamics}

We now present the full-step resource reduction whose main quality is that it
enjoys strong normalization.
We first need to introduce the resource substitution of a stream for a sequence
variable.

\begin{definition}\label{definition:rsubst:str}
  We define the \definitive{resource substitution}
  $\rsubst{\exprA}{\svarA}{\strB}$
  of a stream $\strB=\tuple{\bagB_i}_{i\in\N}$ 
  for a sequence variable $\svarA$ in a resource expression $\exprA$
  by induction on $\exprA$ as follows:
\begin{align*}
  \rsubst{\varB}{\svarA}{\strB}&\eqdef
  \begin{cases}
    \vtermB & \text{if $\varB=\svarA[i]$, $\bagB_i=\mset{\vtermB}$ and $\bagB_j=\emptybag$ for $j\in\N\setminus\set i$}\\
    \varB  & \text{if $\varB\not\in\svarA$ and $\strB=\emptystream$}\\
     0 & \text{otherwise}
  \end{cases}
  \displaybreak[3]
  \\
  \rsubst*{\labs{\svarB}{\btermA}}{\svarA}{\strB}&\eqdef
  \labs{\svarB}{\rsubst{\btermA}{\svarA}{\strB}}
  \displaybreak[3]
  \\
  \rsubst*{\appl{\hexprA}{\strA}}{\svarA}{\strB}&\eqdef
  \sum_{\strB\splitinto\strB_1\bagcat\strB_2}
  \appl{\rsubst{\hexprA}{\svarA}{\strB_1}}{\rsubst{\strA}{\svarA}{\strB_2}}
  \displaybreak[3]
  \\
  \rsubst{\mset{\vtermA_1,\dotsc,\vtermA_k}}{\svarA}{\strB}&\eqdef
  \sum_{\strB\splitinto\strB_1\bagcat\cdots\bagcat\strB_k}
  \mset{\rsubst{\vtermA_1}{\svarA}{\strB_1},\dotsc,\rsubst{\vtermA_k}{\svarA}{\strB_k}}
  \\
  \rsubst{\emptystream}{\svarA}{\strB}&\eqdef
  \begin{cases}
    \emptystream&\text{if $\strB=\emptystream$}\\
    0&\text{otherwise}\\
  \end{cases}
  \\
  \rsubst*{\bagA\cons\strC}{\svarA}{\strB}&\eqdef
  \sum_{\strB\splitinto\strB_1\bagcat\strB_2}
  {\rsubst{\bagA}{\svarA}{\strB_1}}\cons{\rsubst{\strC}{\svarA}{\strB_2}}
  \qquad\text{if $\bagA\cons\strC\not=\emptystream$}
\end{align*}
where, in the abstraction case, $\svarB$ is chosen so that $\svarA\neq\svarB$ 
and $\svarB\cap\LVarOf{\strB}=\emptyset$.
\end{definition}

As for fine-step resource substitution, the condition in the last case of the
definition can be dropped.
Moreover, if $\strA=\bagA_1\cons\cdots\cons\bagA_k\cons\emptystream$,
then 
\[
  \rsubst{\strA}{\svarA}{\strB}=
  \sum_{\strB\splitinto\strB_1\bagcat\cdots\bagcat\strB_k}
  {\rsubst{\bagA_1}{\svarA}{\strB_1}}\cons\cdots\cons
  {\rsubst{\bagA_k}{\svarA}{\strB_k}}\cons\emptystream
\]
and, equivalently,
\[
  \rsubst{\tuple{\bagA_i}_{i\in\N}}{\svarA}{\strB}=
  \sum_{p:\strB\splitinto\N}
  \tuple{\rsubst{\bagA_i}{\svarA}{\restrictToPre{\strB}{p}{i}}}_{i\in\N}
  \;.
\]

Also, one can check that:
\begin{lemma}\label{lemma:rsubst:str:def}
  We have:
\[
  \rsubst{\exprA}{\svarA}{\emptystream}=\erase{\exprA}{\svarA}
\]
and, assuming $\svarA\not\in\SVarOf{\bagA}$, we have:
\[
    \rsubst{\exprA}{\svarA}{\bagA\cons\strB}
    =\rsubst{\shiftdown{\rsubst{\exprA}{\svarA[0]}{\bagA}}{\svarA}}{\svarA}{\strB}\;.
\]
More generally, 
if $\strB=\bagB_0\cons\cdots\cons\bagB_{k-1}\cons\emptystream$,
and $\svarA\not\in\SVarOf{\strB}$, we have:
\begin{align*}
    \rsubst{\exprA}{\svarA}{\strB}
    &=\erase{\shiftdown{\rsubst{\shiftdown{\rsubst{\exprA}{\svarA[0]}{\bagB_0}}{\svarA}\cdots}{\svarA[0]}{\bagB_{k-1}}}{\svarA}}{\svarA}
    \\
    &=\erase{\rsubst{\rsubst{\exprA}{\svarA[0]}{\bagB_0}\cdots}{\svarA[k-1]}{\bagB_{k-1}}}{\svarA}\;.
\end{align*}
\end{lemma}
\begin{proof}
  The first two statements are established directly from the definitions, 
  by induction on $\exprA$.
  If $\strB=\bagB_0\cons\cdots\cons\bagB_{k-1}\cons\emptystream$, and $\svarB\cap\LVarOf{\strB}=\emptyset$, 
  we can iterate $k$ times the second statement then use the first one to obtain
\[
    \rsubst{\exprA}{\svarA}{\strB}
    =\erase{\shiftdown{\rsubst{\shiftdown{\rsubst{\exprA}{\svarA[0]}{\bagB_0}}{\svarA}\cdots}{\svarA[0]}{\bagB_{k-1}}}{\svarA}}{\svarA}
\]
and then we obtain the final identity by iterating \cref{lemma:erase:shift,lemma:shift:rsubst}.
\end{proof}

Substitution of streams enjoys the same regularity as substitution of bags,
w.r.t.\ the size of expressions.
Namely, setting $\LengthOf{\strB}\eqdef\sum_{i\in\N}\LengthOf{\bagB_i}$
when $\strB=\tuple{\bagB_i}_{i\in\N}$, we obtain:

\begin{lemma}\label{lemma:rsubst:str:size}
  If $\exprA'\in\rsubst{\exprA}{\svarA}{\strB}$ with
  $\strB=\tuple{\bagB_i}_{i\in\N}$, then
  $\nocc{\svarA[i]}{\exprA}=\LengthOf{\bagB_i}$ for $i\in\N$, and
  $\SizeOf{\exprA'}=\SizeOf{\exprA}+\SizeOf{\strB}-\LengthOf{\strB}$.
  If moreover $\varB\not\in\svarA$ then $\nocc{\varB}{\exprA'}=
  \nocc{\varB}{\exprA}+\nocc{\varB}{\strB}$ (in particular,
  $\varB\in\LVarOf{\exprA'}$ iff $\varB\in\LVarOf{\exprA}\cup\LVarOf{\strB}$).
\end{lemma}
\begin{proof}
  Using \cref{lemma:rsubst:str:def}, it is sufficient to iterate 
  \cref{lemma:rsubst:size}.
\end{proof}

\begin{figure}
  \begin{gather*}
    \begin{prooftree}
      \infer0[\Resredrulebeta*]{
        \appl*{\labs{\svarA}{\btermA}}{\strB}\Resred\rsubst{\btermA}{\svarA}{\strB}
      }
    \end{prooftree}
    \quad
    \begin{prooftree}
      \hypo{\btermA\Resred\btermsA'}
      \infer1[\Resredruleabs*]{\labs{\svarA}{\btermA}\Resred\labs{\svarA}{\btermsA'}}
    \end{prooftree}
    \quad
    \begin{prooftree}
      \hypo{\vtermA\Resred\vtermsA'}
      \infer1[\Resredruleappl*]{\appl{\vtermA}{\strB}\Resred\appl{\vtermsA'}{\strB}}
    \end{prooftree}
    \quad
    \begin{prooftree}
      \hypo{\strB\Resred\strsB'}
      \infer1[\Resredruleappr*]{\appl{\hexprA}{\strB}\Resred\appl{\hexprA}{\strsB'}}
    \end{prooftree}
    \\[1em]
    \begin{prooftree}
      \hypo{\vtermA\Resred\vtermsA'}
      \infer1[\Resredrulebag*]{\mset{\vtermA}\bagcat\bagB\Resred\mset{\vtermsA'}\bagcat\bagB}
    \end{prooftree}
    \qquad
    \begin{prooftree}
      \hypo{\bagA\Resred\bagsA'}
      \infer1[\Resredrulestrl*]{\bagA\cons\strB\Resred\bagsA'\cons\strB}
    \end{prooftree}
    \qquad
    \begin{prooftree}
      \hypo{\strB\Resred\strsB'}
      \infer1[\Resredrulestrr*]{\bagA\cons\strB\Resred\bagA\cons\strsB'}
    \end{prooftree}
  \end{gather*}
  \caption{Rules of full-step extensional resource reduction}
  \label{fig:fullstep}
\end{figure}

\begin{definition}\label{definition:sResred}
  \definitive{Full-step resource reduction}
  is the relation from resource terms to
  resource sums defined by the rules of \cref{fig:fullstep},
  and then extended to a relation on term sums
  by setting $\termsA\sResred\termsA'$ iff 
  $\termsA=\sum_{i=0}^k\termA_i$ and 
  $\termsA'=\sum_{i=0}^k\termsA'_i$,
  with $\termA_0\Resred \termsA'_0$ and
  $\termA_i\ResredR \termsA'_i$ for $1\le i\le k$.
\end{definition}

Note that, here, we do impose at least one summand to be reduced.
Using the notation for abstraction over single variables,
we can reformulate the base case of reduction
\(\Resredrulebeta\) as follows:
\begin{lemma}\label{lemma:Resred:lambda}
  We have
  \(
    \appl*{\labs{\varA_1}{\cdots\labs{\varA_k}{\labs{\svarB}{\btermA}}}}
    {\bagB_1\cons\cdots\cons\bagB_k\cons\strC}
    \Resred
    \rsubst{
      \rsubst{
        \rsubst{\btermA}{\varA_1}{\bagB_1}
        \cdots
      }{\varA_1}{\bagB_1}
    }{\svarB}{\strC}
  \).
  In particular, if 
  \(\svarB\not\in\SVarOf{\btermA}\)
  then
  \(
    \appl*{\labs{\varA_1}{\cdots\labs{\varA_k}{\labs{\svarB}{\btermA}}}}
    {\bagB_1\cons\cdots\cons\bagB_k\cons\emptystream}
    \Resred
    \rsubst{
      \rsubst{\btermA}{\varA_1}{\bagB_1}
      \cdots
    }{\varA_1}{\bagB_1}
  \).
\end{lemma}
\begin{proof}
  The first statement implies the second one by the 
  first identity of \cref{lemma:rsubst:str:def}.
  We have:
  \begin{align*}
    \appl*{\labs{\varA_1}{\cdots\labs{\varA_k}{\labs{\svarB}{\btermA}}}}
    {\bagB_1\cons\cdots\cons\bagB_k\cons\strC}
    &=
    \appl*[\big]{\labs{\svarB}{
      \subst{\shiftup{
        \subst{\shiftup{\btermA}{\svarB}}{\varA_k}{\svarB[0]}
        \cdots
      }{\svarB}}{\varA_1}{\svarB[0]}
    }}{\bagB_1\cons\cdots\cons\bagB_k\cons\strC}
    \\
    &\Resred
    \rsubst{
      \subst{\shiftup{
        \subst{\shiftup{\btermA}{\svarB}}{\varA_k}{\svarB[0]}
        \cdots
      }{\svarB}}{\varA_1}{\svarB[0]}
    }{\svarB}{\bagB_1\cons\cdots\cons\bagB_k\cons\strC}
  \end{align*}
  where we chose
  the \(\varA_i\)'s and \(\svarB\) not free in 
  the \(\bagB_i\)'s nor in \(\strC\).
  Then we conclude by iterating \(k\) times the following observation:
  for any term \(\termA\), bag term \(\bagB\), sequence term \(\strC\),
  variable \(\varA\not\in\LVarOf{\bagB\cons\strC}\)
  and sequence variable \(\svarB\not\in\SVarOf{\bagB\cons\strC}\cup\SVarOf{\varA}\),
  \cref{lemma:shift:rsubst,lemma:rsubst:str:def,lemma:subst:rsubst:var} entail
  \begin{align*}
    \rsubst{\subst{\shiftup{\termA}{\svarB}}{\varA}{\svarB[0]}}{\svarB}{\bagB\cons\strC}
    &=
    \rsubst{\shiftdown{
      \rsubst{\subst{\shiftup{\termA}{\svarB}}{\varA}{\svarB[0]}}{\svarB[0]}{\bagB}
    }{\svarB}}{\svarB}{\strC}
    \\
    &=
    \rsubst{\shiftdown{
      \rsubst{\shiftup{\termA}{\svarB}}{\varA}{\bagB}
    }{\svarB}}{\svarB}{\strC}
    \\
    &=
    \rsubst{\rsubst{
        \shiftdown{{\shiftup{\termA}{\svarB}}}{\svarB}
    }{\varA}{\shiftdown{\bagB}{\svarB}}}{\svarB}{\strC}
    \\
    &=
    \rsubst{\rsubst{
        \termA
    }{\varA}{\bagB}}{\svarB}{\strC}\,.
    \qedhere
  \end{align*}
\end{proof}

\begin{example}\label{example:reduction:full}
  We consider the same reducible terms as in \cref{example:reduction:fine}.
  We have reductions
  \(
  \appl{\trProj{\varA}}{\emptystream}
  \Resred
  \rsubst*{\appl{\varA}{\emptystream}}{\svarB}{\emptystream}
  =
  \appl{\varA}{\emptystream}
  \)
  and
  \(
    \appl{\trProj{i}}{\emptystream}
    \Resred
    \rsubst*{\appl{\svarB[i]}{\emptystream}}{\svarB}{\emptystream}
    = 0
  \)
  by \(\Resredrulebeta\),
  ending on normal forms:
  notice that these are the only possible full-step reductions
  starting from those base terms.

  We moreover have
  \(
    \appl{\trProj{0}}{\mset{\trProj{\varA}}\cons\emptystream}
    \sResred
    \rsubst*{\appl{\svarB[0]}{\emptystream}}{\svarB}{\mset{\trProj{\varA}}\cons\emptystream}
    =
    \appl{\trProj{\varA}}{\emptystream}
    \sResred
    \rsubst*{\appl{\varA}{\emptystream}}{\svarC}{\emptystream}
    =
    \appl{\varA}{\emptystream}
  \):
  notice how the first two fine steps (\(\resredrulebeta\) followed by
  \(\resredruleiota\)) of the corresponding reduction sequence of
  \cref{example:reduction:fine} are performed in a single full step,
  but the last one reduces a newly created redex.
  As a consequence, we naturally obtain
  \(\trProj{\varA}[k]\sResred^{2k}\trProj{\varA}\) and
  \(\trProj{i}[k]\sResred^{2k}\trProj{i}\).

  This allows us to revisit the final reduction sequence of
  \cref{example:reduction:fine} as follows:
  \[
    \appl{\trCCp{\varA}}{\mset{\trProj{0},\trProj{0}}\cons\emptystream}
    \sResred
    \rsubst*{
      \appl{\varA}{
        \mset{ \trCCz{\varB}, \trCCz{\varB} }
        \cons\emptystream
      }
    }{\varB}{\mset{\trProj{0},\trProj{0}}}
    =
    2 \appl{\varA}{
      \mset{ 
        \trProj{0}[1],
        \trProj{0}[1]
      }
      \cons\emptystream
    }
    \sResred^4
    2 \appl{\varA}{
      \mset{\trProj{0},\trProj{0}}
      \cons\emptystream
    }
    \,.\qedhere
  \]
\end{example}

Recall that, here, we require at least one element in a sum to be reduced.
This, together with the fact that the reduction of a redex 
yields a sum of smaller terms, ensures that full-step resource reduction 
is strongly normalizing:
\begin{lemma}
  If $\termA\Resred\termsA'$ and $\termA'\in\termsA'$
  then $\SizeOf{\termA}>\SizeOf{\termA'}$.
\end{lemma}
\begin{proof}
  The proof is direct by induction on the reduction,
  using \cref{lemma:rsubst:str:size} in the redex case.
\end{proof}

By a standard argument, we obtain:
\begin{corollary}[Strong normalization for $\sResred$]\label{corollary:Resred:SN}
  There is no infinite sequence $\tuple{\termsA_i}_{i\in\N}$ with $\termsA_i\sResred\termsA_{i+1}$ for $i\in\N$.
\end{corollary}
\begin{proof}
  To each term sum, we associate the multiset of the sizes of its elements:
  under full-step reduction, this measure is strictly decreasing for the multiset order.
\end{proof}

Moreover, full-step reduction is a particular case of iterated fine-step reduction:
\begin{lemma}\label{lemma:sResred:sresredRT}
  If $\exprsA\sResred\exprsA'$
  then $\exprsA\sresredRT\exprsA'$.
\end{lemma}
\begin{proof}
  It is sufficient to consider the case of 
  $\exprsA=\exprA\Resred\exprsA'$.
  The proof is then by induction on this reduction:
  the case of $\Resredrulebeta$ 
  follows from \cref{lemma:rsubst:str:def}.
  All the other cases follow straightforwardly from the induction hypothesis,
  using \cref{lemma:sresred:context}.
\end{proof}

\begin{theorem}[Weak normalization for $\sresred$]
  For every resource sum $\exprsA$
  there exists a sum $\exprsA'$ of $\resred$-irreducible expressions
  such that $\exprsA\sresredRT\exprsA'$,
  and this sum is uniquely defined.
\end{theorem}
\begin{proof}
  We obtain $\exprsA'$ by the previous lemma,
  observing that an expression is $\resred$-reducible iff
  it is $\Resred$-reducible.
  Unicity follows from the confluence of $\sresred$,
  together with the fact that if $\exprsA'$ is a sum of $\resred$-irreducible
  expressions and $\exprsA'\sresred\exprsA''$ then $\exprsA'=\exprsA''$.
\end{proof}

Given any resource sum $\exprsA$, we denote by $\NF{\exprsA}$ the unique sum
of irreducible expressions such that $\exprsA\sresredRT\NF{\exprsA}$
and call $\NF{\exprsA}$ the \definitive{normal form} of $\exprsA$.

\begin{theorem}[Confluence of $\sResredRT$]
  The reduction $\sResred$ is confluent
  and $\exprsA\sResredRT\NF{\exprsA}$.
\end{theorem}
\begin{proof}
  By the previous two results, \(\NF{\exprsA}\) is the unique
  \(\Resred\)-irreducible form of \(\exprsA\).
\end{proof}

Recall that a base term is either a redex or of the shape \(\appl{\varA}{\strA}\).
In particular, there is no closed normal base term:
it follows that \(\NF{\btermA}=0\) for any closed base term \(\btermA\).

\subsection{Comparison with the resource calculus with tests}
\label{section:tests}

Although it was designed following the intuitions exposed in our introduction
(and by analogy with game semantics, along the lines of \cref{section:gs}),
the extensional resource calculus shares striking features with the
\emph{resource calculus with tests}.
The latter was designed by Bucciarelli
\emph{et al.}~\cite{DBLP:journals/corr/abs-1209-2890}
to ensure the definability of each element of a particular reflexive object in
the cartesian closed category of sets and multirelations (the relational model
of the simply typed $λ$-calculus), which induces an extensional model of the
pure \(λ\)-calculus \cite{DBLP:conf/csl/BucciarelliEM07}.\footnote{
  Knowledge of the latter model helps, but is not necessary for the discussion
  below. Optionally, the reader may refer to the first paragraphs of
  \cref{section:rel}, where we recall its definition.
}

The syntax and dynamics of both calculi involve a finitary notion of resource
reduction in presence of infinite sequences of abstractions, and applications
to infinite sequences of bags: this concomitance is not fortuitous, and it 
deserves a more detailed comparison.
In the remainder of this section, we thus review those similarities more in
detail, and also outline key differences: in particular, we explain how one can
view value terms (resp.\ base terms) in the extensional resource calculus as
\emph{particular} terms (resp.\ tests) of the resource calculus with tests,
up to some natural quotient on the syntax.%
\footnote{
  The reader not interested in the details of the comparison
  can safely jump to \cref{section:resourcevectors}, on the way to extensional
  Taylor expansion;
  or to \cref{section:gs} for the correspondence with
  plays up to homotopy in game semantics.
}

\paragraph{The resource calculus with tests.}
The syntax of the resource calculus with tests involves three categories of
expressions:
\emph{terms}, which include the terms of the ordinary resource calculus;
\emph{bags} of terms;
and \emph{tests}, which are obtained from bags by the so-called
\emph{cork} construct \(\cork{}\),
and are injected back into terms by the dual \emph{uncork} construct \(\uncork{}\).
Intuitively, a test \(\btermA=\cork{\mset{\vtermA_1,\dotsc,\vtermA_k}}\)
is interpreted as the result of feeding each term \(\vtermA_i\)
with an infinite sequence of empty bags
(the empty stream \(\emptystream\) in our terminology);
and a term \(\uncork{\btermA}\) is viewed as the abstraction of \(\btermA\) over
a denumerable sequence of dummy variables
(\(\labs{\svarC}{\btermA}\) with \(\svarC\) chosen fresh,
in our notation).
The reduction reflects these intuitions.

Formally, one defines a reduction relation \(\testresred\) from expressions
(terms, bags or tests) to sums of expressions, as the compatible closure of
the four base steps:
\begin{align*}
  \appl*{\labs{\varA}{\vtermA}}{\bagB}
  &\testresred\rsubst{\vtermA}{\varA}{\bagB}
  &
  \appl*{\uncork{\btermA}}{\bagB}
  &\testresred\begin{cases}
    \uncork{\btermA}&\text{if \(\bagB=\emptybag\)}\\
    0&\text{otherwise}
  \end{cases}
  \\
  \btermB\testcat\cork{\mset{\labs{\varA}{\vtermA}}}
  &\testresred
  \btermB\testcat\cork{\mset{\subst{\vtermA}{\varA}{0}}}
  &
  \btermB\testcat\cork{\mset{\uncork{\btermA}}}
  &\testresred
  \btermB\testcat\btermA
\end{align*}
where \(\btermB\testcat\btermA\) denotes the parallel composition of tests,
obtained by concatenating the underlying bags:
\(\cork{\bagA}\testcat\cork{\bagB}=\cork*{\bagA\bagcat\bagB}\).
This relation from expressions to sums of expressions is extended to a binary
reduction relation on sums of expressions \(\stestresred\) by linearity,
in the same fashion as in \cref{definition:sResred}:
one requires that at least one term of a sum is reduced.
The obtained reduction \(\stestresred\) is then easily shown to be confluent
and strongly normalizing, as for ordinary resource
reduction~\cite[Theorem 3.22]{DBLP:journals/corr/abs-1209-2890}.

\paragraph{Parallel composition of tests.}
A first discrepancy with the extensional resource calculus is the presence of
an operation of parallel composition on tests: this feature was crucially
used by Bucciarelli \emph{et al.}\ to obtain their full abstraction
results.

Indeed, by inspecting the shape of normal forms, one observes that the only
closed test in normal form is the empty test
\(\cork{\emptybag}\)~\cite[Lemma 3.23]{DBLP:journals/corr/abs-1209-2890}.
If one closes an arbitrary test by substituting all its free variables
with closed bags, and then considers (the support of) the resulting normal
form, the outcome is thus boolean:
it is either \(\emptyset\) (failure) or \(\set{\cork{\emptybag}}\) (success).
Semantically, the effect of parallel composition is then to sum over the
resources (denoted by closed bags) required by successful tests:
this is the key to defining separating test contexts, representing all points of
the model, and constitutes the main ingredient of the full abstraction
proof~\cite[Section 5]{DBLP:journals/corr/abs-1209-2890}.

The absence of this feature in the extensional resource calculus shows in both
syntax and semantics:
a normal base term is necessarily the application of a free variable to a
normal stream, so the normal form of a closed base term can only be \(0\)
(all tests fail!);
and in \cref{section:rel}, we will give examples of points of the relational
model that are not in the interpretation of any normal extensional resource
term (see \cref{proposition:rel:nondefinable} in particular).

In the remaining stages of this comparison, we will thus only consider tests of
the shape \(\cork{\mset{\vtermA}}\) that we simply write \(\cork{\vtermA}\).
The two reduction rules on tests become
\(\cork*{\labs{\varA}{\vtermA}}\testresred\cork*{\subst{\vtermA}{\varA}{0}}\)
and
\(\cork*{\uncork{\btermA}}\testresred\btermA\),
which are direct analogues of 
\(\appl*{\labs{\varA}{\vtermA}}{\emptystream}
=\appl*{\labs{\varA}{\vtermA}}{\emptybag\cons\emptystream}
\resred\appl{\rsubst{\vtermA}{\varA}{\emptybag}}{\emptystream}\)
by \(\resredrulebeta\)
and 
\(\appl*{\labs{\svarA}{\btermA}}{\emptystream}
\resred
\erase{\btermA}{\svarA}
=\btermA\)
by \(\resredruleiota\)
when \(\svarA\not\in\SVarOf{\btermA}\).

\paragraph{Extensional resource terms as \(η\)-long terms with tests.}
We are now ready to make the correspondence between both calculi more precise.
Intuitively, given an expression of the extensional resource calculus, one 
can obtain an expression of the resource calculus with tests by:
\begin{itemize}
  \item replacing each base term \(\appl{\hexprA}{\strB}\) with a finite
    sequence of applications to the bags of some prefix of
    \(\strB=\tuple{\bagB_i}_{i\in \N}\),
    containing all non empty bags, terminated by a cork:
    this yields a test \(\cork*{\appl{\hexprA}{\bagB_0\cdots\bagB_{k-1}}}\)
    -- in particular, \(\appl{\hexprA}{\emptystream}\) can be mapped
    to \(\cork{\hexprA}\), but also to \(\cork*{\appl{\hexprA}{\emptybag}}\);
  \item replacing each value term \(\labs{\svarA}{\btermA}\)
    with a sequence of abstractions on a finite prefix of \(\svarA\),
    binding all value variables \(\svarA[i]\) occurring free 
    in \(\btermA\), on top of an uncorked version of \(\btermA\):
    this yields a term
    \(\labs{\svarA[0]}{\cdots\labs{\svarA[k-1]}{\uncork{\btermA}}}\)
    -- in particular, we may obtain \(\uncork{\btermA}\) when no \(\svarA[i]\)
    is free in \(\btermA\), but also \(\labs{\varC}{\uncork{\btermA}}\)
    for any fresh variable \(\varC\).
\end{itemize}

More formally, we define the \definitive{representation relation}
\(\extIsTest{\termA}{\termA'}\) where
\(\termA\) is a value term \(\vtermA\) (resp.\ base term \(\btermA\), bag term
\(\bagA\), or stream term \(\strB\))
and \(\termA'\) is a term \(\vtermA'\) (resp.\ test \(\btermA'\), bag
\(\bagA'\), or finite tuple of bags \(\strB'\)),
by the following inductive rules:
\begin{gather*}
  \begin{prooftree}
    \hypo{\extIsTest{\btermA}{\btermA'}}
    \hypo{\svarA[i]\not\in \LVarOf{\btermA'}\text{ for }i\ge k}
    \infer{2}{
      \strut
      \extIsTest
      {\labs{\svarA}{\btermA}}
      {\labs{\svarA[0]}{\cdots\labs{\svarA[k-1]}{\uncork{\btermA'}}}}
    }
  \end{prooftree}
  \qquad
  \begin{prooftree}
    \hypo{\extIsTest{\strB}{\strB'}}
    \infer{1}{
      \strut
      \extIsTest
      {\appl{\varA}{\strB}}
      {\cork*{\appl{\varA}{\strB'}}}
    }
  \end{prooftree}
  \qquad
  \begin{prooftree}
    \hypo{\extIsTest{\vtermA}{\vtermA'}}
    \hypo{\extIsTest{\strB}{\strB'}}
    \infer{2}{
      \strut
      \extIsTest
      {\appl{\vtermA}{\strB}}
      {\cork*{\appl{\vtermA'}{\strB'}}}
    }
  \end{prooftree}
  \\
  \begin{prooftree}
    \hypo{
      \extIsTest{\vtermA_1}{\vtermA'_1}
    }
    \hypo{
      \cdots
    }
    \hypo{
      \extIsTest{\vtermA_{k}}{\vtermA'_{k}}
    }
    \infer{3}{
      \strut
      \extIsTest
      {\mset{\vtermA_1,\dotsc,\vtermA_k}}
      {\mset{\vtermA'_1,\dotsc,\vtermA'_k}}
    }
  \end{prooftree}
  \qquad
  \begin{prooftree}
    \infer{0}{
      \strut
      \extIsTest
      {\emptystream}
      {\emptyword}
    }
  \end{prooftree}
  \qquad
  \begin{prooftree}
    \hypo{\extIsTest{\bagA}{\bagA'}}
    \hypo{\extIsTest{\strB}{\strB'}}
    \infer{2}{
      \strut
      \extIsTest
      {\bagA\cons\strB}
      {\bagA'\cons\strB'}
    }
  \end{prooftree}
\end{gather*}
where \(\appl{\vtermA'}{\strB'}\) in the application rule denotes 
the iterated application of \(\vtermA'\) to the bags of \(\strB'\).
This relation yields a notion of \(η\)-longness for resource terms with tests:
a term (resp.\ a test, or a bag) is in \definitive{\(η\)-long form} when it
represents a value term (resp.\ a base term, or a bag term).
Of course, one could devise a direct definition, without reference to 
the extensional resource calculus, that would essentially amount to
forgetting about the left-hand side of each judgement in the above rules.

Note in particular that such an \(η\)-long term starts with a denumerable sequence of 
abstractions (a finite sequence of regular abstractions, followed by \(\uncork{}\)),
and that each variable and each redex is applied to a denumerable sequence of bags
(a finite sequence of bags, followed by \(\cork{}\)).
Clearly, each value term admits infinitely many representations:
it is sufficient to write each stream as a sequence of bags followed by
\(\emptystream\), and to pick a rank \(k\) for each abstraction
\(\labs{\svarA}{\btermA}\) so that no \(\svarA[i]\) occurs in \(\btermA\) for
\(i\ge k\).
Moreover, considered from right to left, the representation relation is
functional, so that it defines a map sending each \(η\)-long term \(\vtermA\)
to the unique value term \(\testToExt{\vtermA}\) such that
\(\extIsTest{\testToExt{\vtermA}}{\vtermA}\).
This map induces an equivalence relation \(\testEq\) on \(η\)-long terms:
one can check that it is the congruence generated by the identities
\(\cork*{\appl{\vtermA}{\emptybag}}\testEq\cork{\vtermA}\) and
\(\labs{\varA}{\uncork{\btermA}}\testEq\uncork{\btermA}\) when
\(\varA\not\in\LVarOf{\btermA}\).
One can thus see value terms as \(\testEq\)-classes of \(η\)-long terms with
tests.

A reader comfortable with relational semantics will easily check that the
representation relation is semantically valid:
\(\sem{\exprA}=\sem{\exprA'}\) as soon as
\(\extIsTest{\exprA}{\exprA'}\), where \(\sem{\exprA}\) is defined in
\cref{section:rel} and \(\sem{\exprA'}\) is the semantics given by
Bucciarelli \emph{et al.}~\cite[Section 4.3]{DBLP:journals/corr/abs-1209-2890}
(the reverse implication cannot hold, as the semantics is not even injective
on extensional resource terms, \emph{cf.}\ \cref{example:rel:noninjective}).

\paragraph{Fine-step reduction as resource reduction on terms with tests.}
A simple inspection of the definitions shows that any representation
of a redex is a reducible test,
and that \(\testToExt{\vtermA}\) is normal iff \(\vtermA\) is.
Note, however, that \(\testresred\) does not commute with \(\testEq\) on the
nose. For instance, consider
\(
\extIsTest{\appl{\trProj{\varA}}{\emptystream}}
{\cork{\uncork{\cork{\varA}}}}
\testEq \cork*{\appl*{\uncork{\cork{\varA}}}{\emptybag}}
\testEq \cork*{\labs{\varB}{\uncork{\cork{\varA}}}}
\testEq \cork*{\appl*{\labs{\varB}{\uncork{\cork{\varA}}}}{\emptybag}}
\) with \(\varB\not=\varA\).
Mimicking the reduction
\(\appl{\trProj{\varA}}{\emptystream}\resred\appl{\trProj{\varA}}{\emptystream}\)
by \(\resredrulebeta\),
the last three representations reduce to the first one:
\[
\cork*{\appl*{\uncork{\cork{\varA}}}{\emptybag}}
\testresred
\cork*{\uncork{\cork{\varA}}}
\qquad
\cork*{\labs{\varB}{\uncork{\cork{\varA}}}}
\testresred
\cork{\subst*{\uncork{\cork{\varA}}}{\varB}{0}}
\qquad
\cork*{\labs{\varB}{\appl*{\uncork{\cork{\varA}}}{\emptybag}}}
\testresred
\cork{\rsubst*{\uncork{\cork{\varA}}}{\varB}{\emptybag}}
\]
each time using a different reduction rule.
And reducing the first representation yields 
\(\cork{\uncork{\cork{\varA}}}\testresred\cork{\varA}\),
reflecting the reduction
\(\appl{\trProj{\varA}}{\emptystream}\resred\appl{\varA}{\emptystream}\)
by \(\resredruleiota\).
One can nonetheless observe that (sums of) \(η\)-long terms are stable under
reduction, and that extensional resource reduction amounts to the reduction of
\(η\)-long terms \emph{up to} \(\testEq\) (extended to sums in the obvious way).

We could thus reconstruct the extensional resource calculus,
\emph{a posteriori}, as the quotient by \(\testEq\) of the \(η\)-long fragment
of the resource calculus with tests.
It is nonetheless more convenient to work with a direct definition
of the language and (especially) of the dynamics, rather than having to reason
up to \(\testEq\) -- and a workable presentation of full-step reduction
\(\sResred\) through that lens seems out of reach.

Finally, note that the resource calculus with tests does not feature a
primitive construction for abstractions over infinitely many variables,
all being potentially free in the immediate subterm.
This is of little consequence when dealing with single terms or finite sums of
terms, which have finitely many free variables, but the extensional
Taylor expansion, described in \cref{section:taylor}, naturally involves
such infinite abstractions
(see, e.g., the expansion of variables in \cref{section:taylor:variables}):
having them reflected in the syntax of resource terms simplifies the exposition
dramatically.

\section{Resource vectors}\label{section:resourcevectors}

Extensional Taylor expansion will map \(λ\)-terms to infinite linear
combinations of value terms.
Since resource reduction generates term sums, we will have to consider infinite
weighted sums of term sums which, \emph{a priori}, involve infinite sums of coefficients
(see \cref{example:vresred} below).
It is reasonable to expect that, like in the ordinary case,
the uniformity of Taylor expansion would allow us to consider finite scalar sums only,
but we leave this for future work, as discussed in the conclusion of the paper
(\cref{section:conclusion}).

Instead, we impose that (countably) infinite sums are always defined by taking coefficients
in a \definitive{complete commutative semiring}~\cite[Section VI.2]{EilenbergALMvolA},
\emph{i.e.}\ a set $\K$ equipped with:
a sum operator $\mathord{\sum}:\K^I\to\K$ on countable families that we denote
by $\sum_{i\in I}α_i\eqdef\sum\tuple{α_i}_{i\in A}$,
satisfying $\sum_{i\in\set{j}} α_i=α_j$,
and $\sum_{i\in I}α_i=\sum_{j\in J}\sum_{i\in I_j}α_{i}$
for any partitioning of $I$ into $\set{I_j\st j\in J}$;
and a commutative monoid structure, denoted multiplicatively, which distributes over $\sum$.
A direct consequence of the axioms is that finite sums are associative and commutative.
We write $0\in\K$ for the empty sum and denote binary and finite sums as usual.
Equipped with finite sums and products, $\K$ is then a commutative semiring in the usual sense.
Moreover, $\K$ is automatically \definitive{positive}:
if $α_1+α_2=0$ then $α_1=α_2=0$.
We also write $1\in\K$ for the multiplicative unit.
Then there is a unique semiring morphism from \(\N\) to \(\K\):
this sends $n\in\N$ to $\sum_{i=1}^n1\in\K$,
and is not necessarily injective.

In order to account for the coefficients of Taylor expansion,
we moreover assume that $\K$ \definitive{has fractions},
meaning that each $n\in\N\setminus\set{0}$, seen as an element of \(\K\),
has a multiplicative inverse in $\K$.
The semiring of booleans $\Bool$ and the extended real half line $\Realpc$,
both equipped with the usual operations,
are such complete semirings with fractions.

We write $\LinearCombinationsOf{X}$ for the semimodule of possibly
infinite linear combinations of elements of a given countable set $X$,
with coefficients in $\K$:
equivalently, these are the $X$-indexed families of elements of $\K$.
We will call any such $A\in\LinearCombinationsOf{X}$ a \definitive{vector},
and we write $\Coef{A}{a}\in\K$ for the value of $A$ at $a$,
\emph{i.e.}\ the \definitive{coefficient} of \(a\) in \(A\).
We write $\support{A}=\set{a\in X\st\Coef{A}{a}\not=0}$ for its \definitive{support}.
We will often abuse notation and write $a\in A$ for $\Coef{A}{a}\not=0$.
Given countable families $\tuple{A_i}_{i\in I}\in\pars{\LinearCombinationsOf{X}}^{I}$ of vectors
and $\tuple{α_i}_{i\in I}\in\K^{I}$ of coefficients,
we write $\sum_{i\in I} α_iA_i\in\LinearCombinationsOf{X}$
for the vector $A$ defined by 
$\Coef{A}{a}\eqdef\sum_{i\in I} α_i\Coef{A_i}{a}\in\K$.

Using the additive monoid structure of $\K$,
each finite sum $A\in\SumsOf{X}$ (and in particular each element of $X$)
induces a vector with finite support $\vembed{A}\in\LinearCombinationsOf{X}$.
Then, for any vector $A$, we have $A=\sum_{a\in A}(\Coef{A}{a})\vembed{a}$.
Note that, again, this embedding of $\SumsOf{X}$ in $\LinearCombinationsOf{X}$
need not be injective: for instance if $\K=\Bool$,
$\vembed{A}$ is nothing but the support of $A$.
We will however abuse notation and generally write $A$ instead of $\vembed{A}$:
the implicit application of the embedding should be clear from the context.
E.g., if we write a vector $\sum_{i\in I}α_i A_i\in\LinearCombinationsOf{X}$
where $α_i\in\K$ and $A_i\in\SumsOf{X}$ for every $i\in I$,
this should be read as  $\sum_{i\in I}α_i\vembed{A_i}$.

\subsection{Vectors of resource terms}
As we have already announced above, the extensional Taylor expansion
of a pure \(λ\)-term will yield a vector of value terms;
moreover, as for ordinary Taylor expansion, the case of an application term
will rely on the \emph{promotion} of a vector of value terms to a vector of bag
terms.
And the latter operation of promotion allows relating ordinary substitution
with resource substitution:
an analogue of \cref{eqn:substitution} (\cpageref{eqn:substitution}) holds for
vectors of terms, without reference to Taylor expansion itself
(\cref{lemma:subst:taylor}).
In order to deal with the possible renaming of bound variables during 
substitutions, we will restrict our attention to vectors of terms whose 
global set of free sequence variables is finite.
In the present subsection, we present this notion formally, and introduce both
kinds of substitution.


We call \definitive{value vector} any vector
$\vtermvA\in\LinearCombinationsOf{\ValueTerms}$ such that
$\SVarOf{\vtermvA}\eqdef\bigcup_{\vtermA\in\vtermvA}\SVarOf{\vtermA}$
is finite ---
note that $\LVarOf{\vtermvA}\eqdef\bigcup_{\vtermA\in\vtermvA}\LVarOf{\vtermA}$
might very well be infinite,
but the hypothesis on $\SVarOf{\vtermvA}$ is sufficient to 
ensure that we can always find variables that are not free in $\vtermvA$.
We use the same typographic conventions for value vectors as for value sums
and write $\ValueVectors$ for the set of value vectors
(thus denoted by $\vtermvA,\vtermvB,\vtermvC$).
We similarly define
\definitive{base vectors} (denoted by $\btermvA,\btermvB,\btermvC\in\BaseVectors$),
\definitive{bag vectors} (denoted by $\bagvA,\bagvB,\bagvC\in\BagVectors$), and
\definitive{stream vectors} (denoted by $\strvA,\strvB,\strvC\in\StreamVectors$).
Note that we do not impose any other bound on the shape of terms in the
definition of vectors: e.g., the length of bags in $\bagvA\in\BagVectors$
is not bounded in general.

As for sums, we may call \definitive{head vector}
(denoted $\hexprvA,\hexprvB,\hexprvC\in\HeadVectors$)
any of a value vector or of a value variable.
And we call \definitive{term vector} (resp.\ \definitive{resource vector})
any value vector (resp.\ head vector), base vector, bag vector, or stream vector,
which we then denote by a letter among $\termvA,\termvB,\termvC$
(resp.\ $\exprvA,\exprvB,\exprvC$).
And we write $\ResourceVectors$ (resp.\ $\ExprVectors$) any of 
$\ValueVectors$ (resp.\ $\HeadVectors$), $\BaseVectors$, $\BagVectors$, or $\StreamVectors$.

We extend term constructors to term vectors by linearity 
as we did for sums, and
we extend resource substitution to vectors by bilinearity, by setting:
\[\rsubst{\exprvA}{\varA}{\bagvB}
  \eqdef
  \sum_{\exprA\in\ResourceExprs}
  \sum_{\bagB\in\BagTerms}
  (\Coef{\exprvA}{\exprA})
  (\Coef{\bagvB}{\bagB})
  \,
  \rsubst{\exprA}{\varA}{\bagB}
\;.
\]
Again, one can easily check that, with that extension,
all the identities defining resource substitution
in \cref{definition:rsubst} also hold if we replace
terms with vectors.
Similarly,  for full-step resource substitution, we set:
\[\rsubst{\exprvA}{\svarA}{\strvB}
  \eqdef
  \sum_{\exprA\in\ResourceExprs}
  \sum_{\strB\in\BagTerms}
  (\Coef{\exprvA}{\exprA})
  (\Coef{\strvB}{\strB})
  \,
  \rsubst{\exprA}{\svarA}{\strB}
\;.
\]

Having extended term constructors to resource vectors, it is also
straightforward to define the ordinary (capture avoiding)
\definitive{substitution} $\subst{\exprA}{\varA}{\hexprvB}\in\ExprVectors$
of a head vector $\hexprvB\in\HeadVectors$
for a value variable $\varA$ in any resource expression
$\exprA\in\ResourceExprs$:
\begin{definition}\label{definition:subst}
  Let $\exprA\in\ResourceExprs$ and $\hexprvB\in\HeadVectors$.
  We define $\subst{\exprA}{\varA}{\hexprvB}\in\ExprVectors$
  by induction on $\exprA$:
\begin{align*}
  \subst{\varB}{\varA}{\hexprvB}&\eqdef
  \begin{cases}
    \hexprvB & \text{if $\varA=\varB$}\\
    \varB  & \text{otherwise}
  \end{cases}
  \\
  \subst*{\labs{\svarB}{\btermA}}{\varA}{\hexprvB}&\eqdef
  \labs{\svarB}{\subst{\btermA}{\varA}{\hexprvB}}
  \\
  \subst*{\appl{\hexprA}{\strA}}{\varA}{\hexprvB}&\eqdef
  \appl{\subst{\hexprA}{\varA}{\hexprvB}}{\subst{\strA}{\varA}{\hexprvB}}
  \\
  \subst{\mset{\vtermA_1,\dotsc,\vtermA_k}}{\varA}{\hexprvB}&\eqdef
  \mset{\subst{\vtermA_1}{\varA}{\hexprvB},\dotsc,\subst{\vtermA_k}{\varA}{\hexprvB}}
  \\
  \subst{\emptystream}{\varA}{\hexprvB}&\eqdef \emptystream
  \\
  \subst*{\bagA\cons\strC}{\varA}{\hexprvB}&\eqdef
  {\subst{\bagA}{\varA}{\hexprvB}}\cons{\subst{\strC}{\varA}{\hexprvB}}
  \qquad\text{if $\bagA\cons\strC\not=\emptystream$}
\end{align*}
where, in the abstraction case, $\svarB$ is chosen so that $\varA\not\in\svarB$
and $\LVarOf{\hexprvB}\cap\svarB=\emptyset$.\footnote{
  The assumption that $\SVarOf{\hexprvB}$ is finite ensures that this
  requirement can be fulfilled.
}
\end{definition}

Note that, as for the substitution of finite sums for variables, \[
  \subst{\exprA}{\varA}{0}=\rsubst{\exprA}{\varA}{\emptybag}
  =\begin{cases}0
    &\text{if $\varA\in\LVarOf{\exprA}$}
    \\
    \exprA&\text{otherwise}
  \end{cases}
  \;.
\]

It is also easy to check that, if $\nocc{\varA}{\exprA}=1$,
then $\subst{\exprA}{\varA}{\hexprvB}=\rsubst{\exprA}{\varA}{\mset{\hexprvB}}$,
which is linear in $\hexprvB$.
In general, however,
$\subst{\exprA}{\varA}{\hexprvB}$ is not linear in $\hexprvB$:
e.g., when $\varA\not\in\LVarOf{\exprA}$,
$\subst{\exprA}{\varA}{0}=\exprA\not=0$.
On the other hand, we can extend this definition to substitution inside a
resource vector, by linearity: we set $\subst{\exprvA}{\varA}{\hexprvB}
\eqdef\sum_{\exprA\in\ResourceExprs}(\Coef{\exprvA}{\exprA})\subst{\exprA}{\varA}{\hexprvB}
\in\ExprVectors$
for any $\exprvA\in\ExprVectors$.
Then one can check that, with that extension, all the identities in the
previous definition also hold if we replace terms with vectors.

Following a similar pattern, one defines the simultaneous substitution 
$\subst{\exprvA}{\seq{\varA}}{\seq{\hexprvB}}$
of the tuple $\seq{\hexprvB}=\tuple{\hexprvB_0,\dotsc,\hexprvB_{k-1}}$
(resp.\ the sequence $\seq{\hexprvB}=\tuple{\hexprvB_i}_{i\in\N}$,
assuming $\SVarOf{\seq{\hexprvB}}$ is finite) of head vectors
(note that, despite the similar notations, these are \emph{not} stream vectors)
for the tuple $\seq{\varA}=\tuple{\varA_0,\dotsc,\varA_{k-1}}$ of variables
(resp.\ the sequence $\seq{\varA}=\tuple{\varA_i}_{i\in\N}$)
of variables in $\exprvA$.
In the finite case, moreover assuming
$\seq{\varA}\cap\LVarOf{\seq{\hexprvB}}=\emptyset$,
we have
\[
  \subst{\exprvA}{\svarA}{\seq{\hexprvB}}
  =\subst{\subst{\exprvA}{\varA_0}{\hexprvB_0}\cdots}{\varA_{k-1}}{\hexprvB_{k-1}}\;.
\]

And in the infinite case,
again assuming $\seq{\varA}\cap\LVarOf{\seq{\hexprvB}}=\emptyset$,
we intuitively have
\[
  \subst{\exprvA}{\svarA}{\seq{\hexprvB}}
  =\subst{\subst{\exprvA}{\varA_0}{\hexprvB_0}}{\varA_1}{\hexprvB_1}\cdots
  \;.
\]

Formally, if we also assume $\LVarOf{\exprvA}\cap\svarA$ is finite
(which is automatic when $\exprvA\in\ExprSums$), then we have
\[
  \subst{\exprvA}{\svarA}{\seq{\hexprvB}}
  =\subst{\subst{\exprvA}{\varA_0}{\hexprvB_0}\cdots}{\varA_{k-1}}{\hexprvB_{k-1}}
\]
for any $k$ such that $\varA_i\in\LVarOf{\exprvA}$ implies $i<k$.
We will most often consider the case where $\seq{\varA}$ is in fact a sequence
variable and $\svarA\not\in\SVarOf{\seq{\hexprvB}}$ --- identifying $\svarA$ with
$\tuple{\svarA[i]}_{i\in\N}$ as we announced.
The latter condition is not restrictive: $\exprvA$ being a resource vector, the
additional condition on $\SVarOf{\seq{\hexprvB}}$ ensures that we can always
find $\svarB\not\in\SVarOf{\exprvA}\cup\SVarOf{\seq{\hexprvB}}$ and write
$\subst{\exprvA}{\svarA}{\seq{\hexprvB}}
=\subst{\subst{\exprvA}{\svarA}{\svarB}}{\svarB}{\seq{\hexprvB}}$.

Note that, writing \(\vec0\) for the sequence of empty value sums, we have
\(\subst{\exprvA}{\svarA}{\vec0}=\erase{\exprvA}=\rsubst{\exprvA}{\svarA}{\emptystream}\):
the erasure of \(\svarA\) amounts to substituting \(0\) for each \(\svarA[i]\).

\subsection{Promotion and the Taylor expansion formula for substitution}

We are now ready to define the promotion of a value vector,
and show that the ordinary substitution of a vector amounts to
the resource substitution of its promotion:
this is the Taylor expansion formula for substitution.
Note that the operation of promotion here is exactly the same as for the
ordinary Taylor expansion, and both substitution mechanisms are the usual
ones: as a consequence, this result (\cref{lemma:subst:taylor}) is stated and
proved similarly, although the underlying term language is different.
We thus only provide references and sketches of proofs.

On the other hand, it will also be useful to define the promotion of a
\emph{sequence} of value vectors, and to establish a similar formula
for the substitution of a sequence of value vectors for a sequence variable
(\cref{lemma:subst:taylor:seq}).
The latter requires a bit of care to deal with infinite sequences of terms,
but the actually novel phenomenon regarding substitution will only appear
when we consider extensional Taylor expansion of \(λ\)-terms, in
\cref{section:taylor}: there, the analogue of \cref{eqn:substitution}
does not hold as an identity.


Given a value vector $\vtermvA$,
and $\bagA=\mset{\vtermA_1,\dotsc,\vtermA_k}\in\BagTerms$,
we write $\BagCoef{\vtermvA}{\bagA}\eqdef\prod_{i=1}^k\Coef{\vtermvA}{\vtermA_i}$.
Then we define the bag vector $\vtermvA^k\eqdef\mset{\vtermvA,\dotsc,\vtermvA}$
(with $k$ copies of $\vtermvA$), and obtain:
\[ \vtermvA^k \quad = \quad 
  \sum_{\tuple{\vtermA_1,\dotsc,\vtermA_k}\in\ValueTerms^k} 
  \BagCoef{\vtermvA}{\mset{\vtermA_1,\dotsc,\vtermA_k}}
  \mset{\vtermA_1,\dotsc,\vtermA_k}\;.
\]

Then we set 
\[\prom{\vtermvA}=\sum_{k\in\N} \frac{1}{k!} \vtermvA^k\in\BagVectors\]
which we call the \definitive{promotion} of $\vtermvA$.

A straightforward computation shows that promotion commutes with substitution:
\begin{lemma}\label{lemma:subst:prom}
  For any $\vtermvA$ and $\vtermvB\in\ValueVectors$, 
  $\subst{\prom{\vtermvB}}{\varA}{\vtermvA}=\prom{\subst{\vtermvB}{\varA}{\vtermvA}}$.
  And for any $\seq{\vtermvA}\in\ValueVectors^\N$ 
  such that $\SVarOf{\seq{\vtermvA}}$ is finite, we have
  $\subst{\prom{\vtermvB}}{\svarA}{\seq{\vtermvA}}
  =\prom{\subst{\vtermvB}{\svarA}{\seq{\vtermvA}}}$.
\end{lemma}

It will also be useful to compute the coefficient of a bag in the promotion of a
value vector (recalling that $\isodeg{\bagA}$ is the isotropy degree of $\bagA$,
as defined in \cref{section:preliminaries}):
\begin{lemma}\label{lemma:prom:coef}
  If $\bagA\in\prom{\vtermvA}$ and $\LengthOf{\bagA}=k$ then
  $\Coef{\prom{\vtermvA}}{\bagA}=\frac{\BagCoef{\vtermvA}{\bagA}}{\isodeg{\bagA}}$.
\end{lemma}
\begin{proof}
  This is an easy result on the combinatorics of multisets
  (see, e.g., \cite[Lemma 4.4]{DBLP:journals/lmcs/OlimpieriA22}).
\end{proof}

\begin{lemma}[Taylor expansion of substitution]\label{lemma:subst:taylor}
  For any $\exprvA\in\ExprVectors$ and $\vtermvA\in\ValueVectors$, 
  $\subst{\exprvA}{\varA}{\vtermvA}=\rsubst{\exprvA}{\varA}{\prom{\vtermvA}}$.
\end{lemma}
\begin{proof}
  The proof is essentially the same as in the ordinary resource calculus
  \cite[Lemma 4.8]{DBLP:journals/lmcs/Vaux19}.
  By linearity, it is sufficient to consider the case of 
  $\exprvA=\exprA\in\ResourceExprs$.
  We first show that the identities defining
  $\exprA\mapsto\subst{\exprA}{\varA}{\vtermvA}$
  (as in \cref{definition:subst}) are also valid for
  $\exprA\mapsto\rsubst{\exprA}{\varA}{\prom{\vtermvA}}$:
  here the definition of sums over partitionings of bags,
  as used in \cref{definition:rsubst}, is crucial,
  in conjunction with \cref{fact:isodeg,lemma:prom:coef}.
  The result follows by induction on $\exprA$.
\end{proof}

Similarly, we associate a \definitive{promotion stream vector}
$\prom{\seq{\vtermvA}}\in\StreamVectors$ with each sequence
$\seq{\vtermvA}=\tuple{\vtermvA_i}_{i\in\N}\in\ValueVectors^\N$ of value vectors
such that $\SVarOf{\seq{\vtermvA}}$ is finite
(again note that, despite the similar notation, the latter sequence is
\emph{not} a stream vector).
First observe that, by construction, $\Coef*{\prom{\vtermvA_i}}{\emptybag}=1$ for each $i\in\N$.
Then we define $\prom{\seq{\vtermvA}}$ by its coefficients:
for every $\strA=\tuple{\bagA_i}_{i\in\N}\in\StreamTerms$,
we can set $\Coef*{\prom{\seq{\vtermvA}}}{\strA}\eqdef
\prod_{i\in\N}\Coef*{\prom{\vtermvA_i}}{\bagA_i}$,
where only finitely many factors $\Coef*{\prom{\vtermvA_i}}{\bagA_i}$ are not $1$.
In particular, we have $\Coef*{\prom{\seq{\vtermvA}}}{\emptystream}=1$ and 
$\Coef*{\prom{\vtermvA\cons\strvB}}*{\bagA\cons\strB}
=\Coef*{\prom{\vtermvA}}{\bagA}\times\Coef*{\prom{\strvB}}{\strB}$,
so that $\prom*{\vtermvA\cons\seq{\vtermvB}}=\prom{\vtermvA}\cons\prom{\seq{\vtermvB}}$.

Then we obtain the analogue of \cref{lemma:subst:prom} for the promotion of
sequences of value vectors:
\begin{lemma}\label{lemma:subst:prom:seq}
  For any $\vtermvA\in\ValueVectors$ and $\seq{\vtermvB}\in\ValueVectors^\N$
  such that $\SVarOf{\seq{\vtermvB}}$ is finite, we have
  $\subst{\prom{\seq{\vtermvB}}}{\varA}{\vtermvA}
  =\prom{\subst{\seq{\vtermvB}}{\varA}{\vtermvA}}$.
  And for any $\seq{\vtermvA}\in\ValueVectors^\N$ with
  $\SVarOf{\seq{\vtermvA}}$ finite, we have
  $\subst{\prom{\seq{\vtermvB}}}{\svarA}{\seq{\vtermvA}}
  =\prom{\subst{\seq{\vtermvB}}{\svarA}{\seq{\vtermvA}}}$.
\end{lemma}
\begin{proof}
  Fix $\strC\in\StreamTerms$:
  we can write $\strC=\bagC_1\cons\cdots\cons\bagC_k\cons\emptystream$.
  Then write $\seq{\vtermvB}=\vtermvB_1\cons\cdots\cons\vtermvB_k\cons\seq{\vtermvC}$.
  We obtain
  \begin{align*}
    \Coef*{\subst{\prom{\seq{\vtermvB}}}{\varA}{\vtermvA}}{\strC}
    &=\Coef*{\subst{\prom{\vtermvB_1}}{\varA}{\vtermvA}\cons\cdots\cons\subst{\prom{\vtermvB_k}}{\varA}{\vtermvA}\cons\subst{\prom{\seq{\vtermvC}}}{\varA}{\vtermvA}}{\strC}
    \\
    &=\Coef*{\subst{\prom{\vtermvB_1}}{\varA}{\vtermvA}}{\bagC_1}
    \times\cdots\times
    \Coef*{\subst{\prom{\vtermvB_k}}{\varA}{\vtermvA}}{\bagC_k}
    \times\Coef*{\subst{\prom{\seq{\vtermvC}}}{\varA}{\vtermvA}}{\emptystream}
    \displaybreak[3]
    \\
    &=\Coef*{\prom{\subst{\vtermvB_1}{\varA}{\vtermvA}}}{\bagC_1}
    \times\cdots\times
    \Coef*{\prom{\subst{\vtermvB_k}{\varA}{\vtermvA}}}{{\bagC_k}}
    \times\Coef*{\prom{\subst{\seq{\vtermvC}}{\varA}{\vtermvA}}}{\emptystream}
    \displaybreak[3]
    \\
    &=\Coef*{\prom{\subst{\vtermvB_1}{\varA}{\vtermvA}}\cons\cdots\cons\prom{\subst{\vtermvB_k}{\varA}{\vtermvA}}\cons\prom{\subst{\seq{\vtermvC}}{\varA}{\vtermvA}}}{\strC}
    \\
    &=\Coef*{\prom{\subst{\seq{\vtermvB}}{\varA}{\vtermvA}}}{\strC}
  \end{align*}
  where each identity 
  $\Coef*{\subst{\prom{\vtermvB_i}}{\varA}{\vtermvA}}{\bagC_i}
  =\Coef*{\prom{\subst{\vtermvB_i}{\varA}{\vtermvA}}}{\bagC_i}$
  follows from \cref{lemma:subst:prom}, and
  $\Coef*{\subst{\prom{\seq{\vtermvC}}}{\varA}{\vtermvA}}{\emptystream}
  =\Coef*{\prom{\seq{\vtermvC}}}{\emptystream}
  =1
  =\Coef*{\prom{\subst{\seq{\vtermvC}}{\varA}{\vtermvA}}}{\emptystream}$
  by definition.
  The proof of the second statement is similar.
\end{proof}

We obtain the analogue of \cref{lemma:subst:taylor} as well:
\begin{lemma}\label{lemma:subst:taylor:seq}
  For any $\exprvA\in\ExprVectors$,
  and any $\seq{\vtermvA}\in\ValueVectors^\N$
  such that $\SVarOf{\seq{\vtermvB}}$ is finite,
  we have
  $\subst{\exprvA}{\svarA}{\seq{\vtermvA}}=\rsubst{\exprvA}{\svarA}{\prom{\seq{\vtermvA}}}$.
\end{lemma}
\begin{proof}
  Write $\seq{\vtermvA}=\tuple{\vtermvA_i}_{i\in\N}$.
  By linearity, it is sufficient to consider the case of 
  $\exprvA=\exprA\in\ResourceExprs$.
  It is possible to follow the same pattern as in the proof of 
  \cref{lemma:subst:taylor}, but we can also deduce the
  present result from \cref{lemma:subst:taylor} itself.
  Indeed, $\LVarOf{\exprA}$ is finite, hence
  we can choose $k$ such that $i\ge k$ implies $\svarA[i]\not\in\LVarOf{\exprA}$, 
  to obtain
  \[
    \subst{\exprA}{\svarA}{\seq{\vtermvA}}
    =\subst{\subst{\exprA}{\svarA[0]}{\vtermvA_0}\cdots}{\svarA[k-1]}{\vtermvA_{k-1}}
  \]
  as discussed above, and also
  \[
    \rsubst{\exprA}{\svarA}{\prom{\seq{\vtermvA}}}
    =\rsubst{\rsubst{\exprA}{\svarA[0]}{\prom{\vtermvA_0}}\cdots}{\svarA[k-1]}{\prom{\vtermvA_{k-1}}}
  \]
  by \cref{lemma:rsubst:str:def}
  ---
  assuming, w.l.o.g., that $\svarA\cap\bigcup_{i<k}\LVarOf{\vtermvA_{i}}=\emptyset$.
  It is then sufficient to iterate \cref{lemma:subst:taylor}.
\end{proof}

Alternatively to the above definition, we could introduce
$\prom{\seq{\vtermvA}}$ similarly to the promotion of value vectors,
as follows.
First, we call \definitive{degree stream} any sequence of natural numbers
$\dstrA=\tuple{k_i}_{i\in\N}\in\N^{\N}$ 
with finite support: $\{i\in\N\st k_i\not=0\}$ is finite.
We will write $\DegreeStreams$ for the set of degree streams.
Degree streams could thus be identified with finite multisets
of natural numbers, but we use different notations
to fit the way we use them.
A more relevant intuition is to identify \(\DegreeStreams\) with the set of
streams over a singleton set \(\SsOf{\set{*}}\):
we denote $\emptystream\eqdef\tuple{0}_{i\in\N}\in\DegreeStreams$,
and if $k\in\N$ and $\dstrB\in\DegreeStreams$, we write
$k\cons\dstrB\in\DegreeStreams$ for the stream obtained
by pushing $k$ at the head of $\dstrB$.
Given $\seq{\vtermvA}\in\ValueVectors^\N$ and $\dstrA\in\DegreeStreams$,
we define $\seq{\vtermvA}^{\dstrA}$ inductively as follows:
$\seq{\vtermvA}^{\emptystream}\eqdef\emptystream$
and 
$(\vtermvA\cons\seq{\vtermvB})^{k\cons\dstrB}\eqdef\vtermvA^k\cons\vtermvB^{\dstrB}$
when $k\cons\dstrB\not=\emptystream$.
We moreover define $\dstrA!\in\N$ by setting $\dstrA!=\prod_{i\in\N}k_i!$,
which satisfies: $\emptystream!=1$ and $\pars{k\cons\dstrB}!=k!\times\dstrB!$.

We obtain:
\begin{lemma}\label{lemma:prom:seq}
  For any sequence $\seq{\vtermvA}\in\ValueVectors^\N$ of value vectors
  such that $\SVarOf{\seq{\vtermvA}}$ is finite,
  we have $\prom{\seq{\vtermvA}}
  =\sum_{\dstrA\in\DegreeStreams}\frac{1}{\dstrA!}\seq{\vtermvA}^{\dstrA}$.
\end{lemma}
\begin{proof}
  It is sufficient to check that:
  \(\Coef*[\bigg]{\sum_{\dstrA\in\DegreeStreams}\frac{1}{\dstrA!}\seq{\vtermvA}^{\dstrA}}{\emptystream}=1\)
  and that
  \[
    \Coef*[\bigg]{
      \sum_{\dstrA\in\DegreeStreams}\frac{1}{\dstrA!}(\vtermvA\cons\seq{\vtermvB})^{\dstrA}
    }*{\bagA\cons\strB}
    =
    \Coef*{\prom{\vtermvA}}{\bagA}\times
    \Coef*[\bigg]{\sum_{\dstrA\in\DegreeStreams}\frac{1}{\dstrA!}\seq{\vtermvB}^{\dstrA}}{\strB}
  \] which follows directly from the definitions.
\end{proof}

This alternative presentation of the promotion of a sequence of value vectors
will be useful to establish its compatibility with reduction
(\cref{lemma:vresred:prom} below).

\subsection{Reduction of resource vectors}\label{section:resourcevectors:reduction}

Now we present the notion of reduction on term vectors that we will use to
simulate both \(β\)- and \(η\)-reductions via extensional Taylor expansion:
we define the \definitive{resource reduction on term vectors}
by setting $\termvA\vresred\termvA'$ if $\termvA=\sum_{i\in I}α_i\termA_i$ 
and $\termvA'=\sum_{i\in I}α_i\termsA'_i$ 
with $\termA_i\in\ResourceTerms$, $\termsA'_i\in\ResourceSums$ and
$\termA_i\sresredRT \termsA'_i$ for $i\in I$.
We also define the \definitive{normal form of a term vector}, point-wise:
\[\NF{\termvA}\eqdef \sum_{\termA\in\termvA} \pars{\Coef{\termvA}{\termA}} \NF{\termA}\;.\]

Note in particular that, in the definition of resource reduction, we do not
impose the terms $\termA_i$ to be pairwise distinct, and that the number of
$\sresred$ reduction steps from each $\termA_i$ is not bounded.
We will see that the latter observation is crucial for our purposes:
contrasting with the ordinary case where parallel reduction is
sufficient to capture \(β\)-reduction~\cite{DBLP:journals/lmcs/Vaux19},
the ability to iterate reductions on each element of a term vector is
essential to obtain our simulation results in~\cref{section:taylor}
-- see the proof of \cref{lemma:cc} and the subsequent discussion.

\begin{example}\label{example:vresred}
  Recalling the reductions of \cref{example:reduction:fine},
  observe that we have \(\sum_{i\in\N}\appl{\trProj{i}}{\emptystream}\vresred 0\)
  by picking a reduction \(\appl{\trProj{i}}{\emptystream}\sresred 0\) for each \(i\in\N\).
  On the other hand, we have, e.g., \(\sum_{k\in\N}\trProj{0}[k]\vresred\sum_{k\in\N}\trProj{0}\),
  which demonstrates that, without any particular constraint on the support of term vectors,
  infinite sums of coefficients are generated by reduction and normalization.
\end{example}

\begin{lemma}\label{lemma:vresred:NF}
  For any term vector $\termvA$, we have $\termvA\vresred\NF{\termvA}$.
  If moreover $\termvA=\sum_{i\in I}α_i\termvA_i$
  with $\termvA_i\in\ResourceVectors$ for $i\in I$,
  then $\NF{\termvA}=\sum_{i\in I}α_i\NF{\termvA_i}$.
  Finally, if $\termvA\vresred\termvA'$ then $\NF{\termvA}=\NF{\termvA'}$.
\end{lemma}
\begin{proof}
  The first statement follows from the definitions,
  observing that $\termA\sresredRT\NF{\termA}$.
  The second one follows from the linearity of $\NF{-}$.
  And if $\termvA\vresred\termvA'$ then we can write
  $\termvA=\sum_{i\in I}α_i\termA_i$ and $\termvA'=\sum_{i\in I}α_i\termsA'_i$
  with $\termA_i\sresredRT \termsA'_i$ for $i\in I$:
  then, by confluence of $\sresredRT$, $\NF{\termA_i}=\NF{\termsA'_i}$
  for each $i\in I$,
  and we conclude by the previous point.
\end{proof}

The confluence of $\vresred$ follows directly. Moreover, \(\vresred\) is linear
and compatible, in the following sense:

\begin{lemma}\label{lemma:vresred:context}
  The relation $\vresred$ is reflexive and:
  \begin{enumerate}
    \item if $\termvA_i\vresred\termvA'_i$ with
      $\termvA_i,\termvA'_i\in\ResourceVectors$ for $i\in I$,
      and $\bigcup_{i\in I}\SVarOf{\termvA_i}$ is finite,
      then $\sum_{i\in I}α_i\termvA_i\vresred\sum_{i\in I}α_i\termvA'_i$;
      \label{lemma:vresred:context:sums}
    \item if $\btermvA\vresred\btermvA'$
      then $\labs{\svarA}{\btermvA}\vresred\labs{\svarA}{\btermvA'}$;
      and if $\vtermvA\vresred\vtermvA'$ then
      $\labs{\varA}{\vtermvA}\vresred\labs{\varA}{\vtermvA'}$;
      \label{lemma:vresred:context:vterms}
    \item if $\strvB\vresred\strvB'$
      then $\appl{\varA}{\strvB}\vresred\appl{\varA}{\strvB'}$;
      if moreover $\vtermvA\vresred\vtermvA'$,
      then $\appl{\vtermvA}{\strvB}\vresred\appl{\vtermvA'}{\strvB'}$;
      \label{lemma:vresred:context:bterms}
    \item if $\vtermvA\vresred\vtermvA'$ then
      $\mset{\vtermvA}\vresred\mset{\vtermvA'}$;
      and if $\bagvA\vresred\bagvA'$ and $\bagvB\vresred\bagvB'$
      then $\bagvA\bagcat\bagvB\vresred\bagvA'\bagcat\bagvB'$;
      \label{lemma:vresred:context:bags}
    \item if $\bagvA\vresred\bagvA'$ and $\strvB\vresred\strvB'$
      then $\bagvA\cons\strvB\vresred\bagvA'\cons\strvB'$;
      \label{lemma:vresred:context:strs}
    \item if $\termvA\vresred\termvA'$ and $\bagvB\vresred\bagvB'$ then
      \(\rsubst{\termvA}{\varA}{\bagvB}\vresred\rsubst{\termvA'}{\varA}{\bagvB'}\).
      \label{lemma:vresred:context:rsubst}
  \end{enumerate}
  Moreover, 
  $\appl*{\labs{\svarA}{\btermvA}}*{\bagvB\cons\strvC}\vresred
  \appl*[\big]{\labs{\svarA}{\shiftdown{\rsubst{\btermvA}{\svarA[0]}{\bagvB}}{\svarA}}}{\strvC}$
  and 
  $\appl*{\labs{\svarA}{\btermvA}}{\strvA} \vresred \rsubst{\btermvA}{\svarA}{\strvA}$.%
\end{lemma}
\begin{proof}
  Each result follows from the definitions, 
  also using \cref{lemma:sresred:context} for
  \cref{lemma:vresred:context:vterms,lemma:vresred:context:bterms,%
    lemma:vresred:context:bags,lemma:vresred:context:strs},
  \cref{lemma:sresred:rsubst} for item 6,
  and \cref{lemma:sResred:sresredRT} for the big-step redex case.
\end{proof}

The reduction of value vectors is moreover compatible with promotion.
We first establish:
\begin{lemma}\label{lemma:vresred:dstr}
  If $\vtermvA_i\vresred\vtermvA'_i$ for $i\in\N$,
  then for all $\dstrA\in\DegreeStreams$,
  $\tuple{\vtermvA_i}_{i\in\N}^{\dstrA}\vresred\tuple{\vtermvA'_i}_{i\in\N}^{\dstrA}$.
\end{lemma}
\begin{proof}
  We reason by induction on $\dstrA$.
  If $\dstrA=\emptystream$, the result holds by reflexivity.
  Otherwise, we write $\dstrA=k\cons\dstrB$,
  and $\vtermvB_i=\vtermvA_{i+1}$ and $\vtermvB'_i=\vtermvA'_{i+1}$ for
  $i\in\N$, so that
  $\tuple{\vtermvA_i}_{i\in\N}^{\dstrA}
  =\vtermvA_0^k\cons\tuple{\vtermvB_i}_{i\in\N}^{\dstrB}$
  and similarly for 
  $\tuple{\vtermvA'_i}_{i\in\N}^{\dstrA}$.
  We have $\vtermvA_0^k\vresred{\vtermvA'_0}^k$ by \cref{lemma:vresred:context:bags}
  of \cref{lemma:vresred:context},
  and $\tuple{\vtermvB_i}_{i\in\N}^{\dstrB}\vresred\tuple{\vtermvB'_i}_{i\in\N}^{\dstrB}$
  by induction hypothesis.
  We conclude by \cref{lemma:vresred:context:strs} of \cref{lemma:vresred:context}.
\end{proof}

\begin{lemma}\label{lemma:vresred:prom}
  If $\vtermvA\vresred\vtermvA'$ then $\prom{\vtermvA}\vresred\prom{{\vtermvA'}}$.
  And if $\vtermvA_i\vresred\vtermvA'_i$ with $\bigcup_{i\in\N}\SVarOf{\vtermvA_i}$ finite,
  then $\prom*{\tuple{\vtermvA_i}_{i\in\N}}\vresred\prom*{\tuple{\vtermvA'_i}_{i\in\N}}$.
\end{lemma}
\begin{proof}
  We have $\prom{\vtermvA}=\sum_{k\in\N}\frac{1}{k!}\vtermvA^k$
  and similarly for $\prom{{\vtermvA'}}$, with
  $\vtermvA^k\vresred{\vtermvA'}^k$ 
  for each $k$ by \cref{lemma:vresred:context:bags} of \cref{lemma:vresred:context},
  and we conclude by \cref{lemma:vresred:context:sums} of \cref{lemma:vresred:context}.
  The second statement is established similarly, thanks to
  \cref{lemma:prom:seq,lemma:vresred:dstr}.
\end{proof}
Together with \cref{lemma:vresred:context}
(\cref{lemma:vresred:context:rsubst})
and \cref{lemma:subst:taylor}, the previous lemma entails:
\begin{corollary}\label{corollary:vresred:subst}
  If $\termvA\vresred\termvA'$ and $\vtermvB\vresred\vtermvB'$ then
  \(\subst{\termvA}{\varA}{\vtermvB}\vresred\subst{\termvA'}{\varA}{\vtermvB'}\).
\end{corollary}

We refer to \cref{lemma:vresred:context,lemma:vresred:prom,corollary:vresred:subst}
collectively as the \definitive{compatibility properties} of \(\vresred\),
which straightforwardly extend to its (reflexive and) transitive closure \(\vresredRT\):
we use these properties extensively in the following sections,
and keep them implicit most of the time.

Note that we do not establish the transitivity of $\vresred$.
Consider two reductions
$\termvA=\sum_{i\in I}α_i\termA_i\vresred \sum_{i\in I}α_i\termsA'_i=\termvA'$
and 
$\termvA'=\sum_{j\in J}β_j\termA'_j\vresred \sum_{j\in J}β_j\termsA''_j=\termvA''$.
Intuitively, to deduce a reduction $\termvA\vresred\termvA''$ using the
transitivity of $\sresredRT$, we would need to ``synchronize'' the two writings
of $\termvA'$ in a way that is compatible with the families of reductions
$\termA_i\sresredRT\termsA'_i$ \emph{and} $\termA'_j\sresredRT\termsA''_j$:
there is no obvious way to perform this synchronization, 
and the transitivity of $\vresred$ is an open question at the time of writing.
Fortunately, we do not need to rely on it:
we rather use $\vresredRT$, or resort to reason component-wise and use the
transitivity of $\sresredRT$ instead.

\section{Extensional Taylor expansion}\label{section:taylor}

A first possible definition of extensional Taylor expansion amounts to
performing
infinite $\eta$-expansion on the fly, in order to produce value vectors,
setting:
\[
  \TayExp{\varA}
  \eqdef
  \labs{\svarB}{\appl{\varA}{\prom{\SeqTayExp{\svarB}}}}
  \qquad
  \TayExp{\labs{\varA}{\ltermA}}
  \eqdef
  \labs{\varA}{\TayExp{\ltermA}}
  \qquad
  \TayExp{\appl{\ltermA}{\ltermB}}
  \eqdef
  \labs{\svarB}{\appl{\TayExp{\ltermA}}{\prom{\TayExp N}\cons{\prom{\SeqTayExp{\svarB}}}}}
\]
where $\SeqTayExp{\svarB}$ denotes the sequence 
$\tuple{\TayExp{\svarB[i]}}_{i\in\N}$ of value vectors.

Note that, although this definition seems to be circular in the variable case,
it can be done by defining the coefficient of a resource term
in the Taylor expansion of a $\lambda$-term,
by induction on the resource term, as we will do below.
Moreover observe that this definition does not preserve normal forms:
the application case always introduces redexes.
This can be fixed by defining Taylor expansion based on the head structure of terms,
where \(η\)-expansion does not create redexes.
We thus define extensional Taylor expansion in several steps:
\begin{itemize}
  \item first the \emph{expansion of variables},
    where all the apparent circularity resides: \(\IdExp{\varA}\);
  \item then the \emph{structural expansion} of $λ$-terms,
    straightforwardly as above: \(\ExtTayExp{\ltermA}\);
  \item and the \emph{head expansion}, avoiding the creation of redexes:
    \(\HeadTayExp{\ltermA}\).
\end{itemize}

Then we show that \(\ExtTayExp{\ltermA}\) reduces to \(\HeadTayExp{\ltermA}\),
and that it allows us to simulate $β$-reduction and $η$-reduction steps 
by resource reduction.
The crux of the construction is to show that the expansion of a variable
is compatible with substitution:
\(\subst{\IdExp{\varA}}{\varA}{\vtermA}\) 
and
\(\subst{\vtermA}{\varA}{\IdExp{\varA}}\) reduce to \(\vtermA\).

We thus obtain not just one, but two notions of extensional Taylor
expansion of a \(λ\)-term:
its structural expansion and its head expansion.
And since one reduces to the other, they have the same normal form,
so both notions induce the same \(λ\)-theory via normalization,
and we may pick one or the other depending on our objectives.
For instance, the simulation of \(β\)- and \(η\)-reductions is easier to
establish for the structural version,
but the head version will be used crucially in the next section,
to show that the induced \(λ\)-theory is sensible (it equates all
unsolvable terms).

\subsection{Infinitely \pdfeta-expanded variables}
\label{section:taylor:variables}

We define simultaneously the \definitive{value expansion} $\IdExp{\varA}\in\ValueVectors$
of a variable $\varA$ and the \definitive{stream expansion}
$\IdStrExp{\svarA}\in\StreamVectors$ of a sequence variable $\svarA$
so that:
\begin{align*}
  \IdExp{\varA}&=\labs{\svarB}{\appl{\varA}{\IdStrExp{\svarB}}}\quad\text{(choosing $\svarB\not\ni\varA$)}
  \\
  \text{and}\quad\IdStrExp{\svarA}&= \prom*{\IdSeqExp{\svarA}}\quad\text{where $\IdSeqExp{\svarA}\eqdef\tuple{\IdExp{\svarA[i]}}_{i\in\N}$}
  \;.
\end{align*}

To be formal, we define the coefficients of these vectors by mutual induction
on value terms and on stream terms:
\begin{align*}
  \IdExpCoef{\varA}{\termA}&\eqdef
  \begin{cases}
    \IdStrExpCoef{\svarB}{\strA}&\text{if $\termA=\labs{\svarB}{\appl{\varA}{\strA}}$ with $\svarB\not\ni\varA$,}
    \\
    0&\text{otherwise}
  \end{cases}
  \\
  \IdStrExpCoef{\svarA}{\strA}&\eqdef \prod_{i\in\N}\Coef{\prom*{\IdExp{\svarA[i]}}}{\bagA_i}
  \qquad\text{if $\strA=\tuple{\bagA_i}_{i\in\N}$}
\end{align*}
which ensures that the previous two identities hold.
We moreover write $\IdBagExp{\varA}\eqdef\prom*{\IdExp{\varA}}$.

The stream expansion of a sequence variable is subject to 
the following recursive characterization,
where $\shiftup{\exprvA}{\svarA}[k]$ denotes
$\shiftup{\shiftup{\exprvA}{\svarA}\cdots}{\svarA}$
(with $k$ applications of $\shiftup{-}{\svarA}$):
\begin{lemma}\label{lemma:cc:shift}
  We have
  \[ \IdStrExp{\svarA}=\IdBagExp{\svarA[0]}\cons\shiftup{\IdStrExp{\svarA}}{\svarA} \]
  and, more generally,
  for every $k\in\N$:
  \[
    \IdStrExp{\svarA}
    =\IdBagExp{\svarA[0]}\cons\cdots\cons\IdBagExp{\svarA[k-1]}
    \cons\shiftup{\IdStrExp{\svarA}}{\svarA}[k]
    \,.
  \]
\end{lemma}
\begin{proof}
  The second statement just iterates the first one,
  which is established component-wise:
  it is easy to check from the definition that 
  $\IdStrExpCoef{\svarA}{\strA}
  =\Coef*{\IdBagExp{\svarA[0]}\cons\shiftup{\IdStrExp{\svarA}}{\svarA}}{\strA}$
  for every $\strA\in\StreamTerms$..
\end{proof}

We can compute the coefficient of a value term in $\IdExp{\varA}$ directly.
We first define the \definitive{multiplicity coefficient}\label{definition:muldeg}
$\muldeg{\exprA}$ of any resource expression $\exprA$ inductively as follows:
\begin{align*}
  \muldeg{\varA}
  &\eqdef 1
  &
  \muldeg{\appl{\exprA}{\strB}}
  &\eqdef\muldeg{\exprA}\times\muldeg{\strB}
  \\
  \muldeg{\labs{\svarA}{\btermA}}
  &\eqdef\muldeg{\btermA}
  &
  \muldeg{\mset{\vtermA_1,\dotsc,\vtermA_k}}
  &\eqdef\isodeg{\mset{\vtermA_1,\dotsc,\vtermA_k}}\times\prod_{i=1}^k\muldeg{\vtermA_i}
  \\
  \muldeg{\emptystream}
  &\eqdef 1
  &
  \muldeg{\bagA\cons\strB}
  &\eqdef\muldeg{\bagA}\times\muldeg{\strB}
\end{align*}
so that \(\muldeg{\termA}\) is the product of isotropy degrees
(see \cref{section:preliminaries}, \cpageref{definition:isotropy}) of
the bags of \(\termA\).
Then:
\begin{lemma}\label{lemma:cc:coef}
  If $\vtermA\in\IdExp{\varA}$
  (resp.\ $\bagA\in\IdBagExp{\varA}$; $\strA\in\IdStrExp{\svarA}$)
  then $\IdExpCoef{\varA}{\vtermA}=\frac{1}{\muldeg{\vtermA}}$
  (resp.\ $\IdBagExpCoef{\varA}{\bagA}=\frac{1}{\muldeg{\bagA}}$;
  $\IdStrExpCoef{\svarA}{\strA}=\frac{1}{\muldeg{\strA}}$).
\end{lemma}
\begin{proof}
  The proof is straightforward by induction on terms,
  using \cref{lemma:prom:coef} in the case of a bag.
\end{proof}

\begin{example}
  \label{example:idexp}
  Let us review some of the smallest terms in \(\IdExp{\varA}\) and
  \(\IdStrExp{\svarB}\) (or, rather, in their respective support sets),
  together with their multiplicity coefficients.
  First observe that, by \cref{lemma:cc:shift},
  any element of \(\IdStrExp{\svarB}\)
  is of the shape \(\bagB_1\cons\cdots\cons\bagB_k\cons\emptystream\)
  with \(\bagB_i\in\IdBagExp{\svarB[i]}\) for \(1\le i\le k\).

  For any variable \(\varA\),
  the smallest value term in \(\IdExp{\varA}\) is the term
  \(\trProj{\varA}=\labs{\svarB}{\appl{\varA}{\emptystream}}\) of
  \cref{example:terms}.
  We thus have \(\mset{\trProj{\svarB[0]}}\cons\emptystream\in\IdStrExp{\svarB}\),
  hence \(\trCCz{\varA}
  =\labs{\svarB}{
    \appl{\varA}{\mset{\trProj{\svarB[0]}}\cons\emptystream}
  }\in\IdExp{\varA}\),
  and thus also \(\labs{\svarB}{
    \appl{\varA}{\mset{\trCCz{\svarB[0]}}\cons\emptystream}
  }\in\IdExp{\varA}\).
  This suggests defining \(\trCCz{\varA}[k]\) for \(k\in\N\) inductively by
  \(\trCCz{\varA}[0]\eqdef\trProj{\varA}\) and 
  \(\trCCz{\varA}[k+1]\eqdef\labs{\svarB}{
    \appl{\varA}{\mset{\trCCz{\svarB[0]}[k]}\cons\emptystream}
  }\) so that
  \(\trCCz{\varA}=\trCCz{\varA}[1]\) and
  \(\trCCz{\varA}[k]\in\IdExp{\varA}\) for each \(k\in\N\).

  The multiplicity coefficient of each of the previous examples is \(1\),
  but we also have, e.g.,
  \(\trCCp{\varA}
  =\labs{\svarB}{
    \appl{\varA}{\mset{\trCCz{\svarB[0]},\trCCz{\svarB[0]}}\cons\emptystream}
  }\in\IdExp{\varA}\),
  with multiplicity coefficient \(2\),
  observing that 
  \(\muldeg{\mset{\trCCz{\svarB[0]},\trCCz{\svarB[0]}}}=
  \isodeg{\mset{\trCCz{\svarB[0]},\trCCz{\svarB[0]}}}=2\).
\end{example}

The value expansion of a variable is meant to behave like an identity morphism
on value terms:
as we have announced, we have
$\subst{\IdExp{\varA}}{\varA}{\vtermA}\vresred\vtermA$
and
$\subst{\vtermA}{\varA}{\IdExp{\varA}}\vresred\vtermA$
for any $\vtermA\in\ValueTerms$.
To establish these properties, we will actually show that, for every
$\vtermA\in\ValueTerms$, there is exactly one element of $\IdExp{\varA}$
contributing to each of those reductions.
More precisely:

\begin{lemma}\label{lemma:cc}
  For any resource term $\termA$,
  any value variable $\varA$,
  and any sequence variable $\svarA$,
  the following holds:
  \begin{enumerate}
    \item\label{lemma:cc:var}
      there exists $\CCBag{\varA}{\termA}\in\IdBagExp{\varA}$ such that 
      $\rsubst{\termA}{\varA}{\CCBag{\varA}{\termA}}\sresredRT
      \muldeg{\CCBag{\varA}{\termA}}\termA$,
      and $\rsubst{\termA}{\varA}{\bagC}\sresredRT 0$
      for any other $\bagC\in\IdBagExp{\varA}$;
    \item\label{lemma:cc:svar}
      there exists $\CCStr{\svarA}{\termA}\in\IdStrExp{\svarA}$ such that 
      $\rsubst{\termA}{\svarA}{\CCStr{\svarA}{\termA}}\sresredRT
      \muldeg{\CCStr{\svarA}{\termA}}\termA$,
      and $\rsubst{\termA}{\svarA}{\strC}\sresredRT 0$
      for any other $\strC\in\IdStrExp{\svarA}$;
    \item\label{lemma:cc:vterm}
      if $\termA=\vtermA\in\ValueTerms$ then 
      there exists $\CCValCtx{\varA}{\vtermA}\in\IdExp{\varA}$ such that 
      $\subst{\CCValCtx{\varA}{\vtermA}}{\varA}{\vtermA}\sresredRT
      \muldeg{\CCValCtx{\varA}{\vtermA}}\vtermA$,
      and $\subst{\vtermC}{\varA}{\vtermA}\sresredRT 0$
      for any other $\vtermC\in\IdExp{\varA}$;
    \item\label{lemma:cc:bag}
      if $\termA=\bagA\in\BagTerms$ then
      there exists $\CCBagCtx{\varA}{\bagA}\in\IdBagExp{\varA}$ such that 
      $\rsubst{\CCBagCtx{\varA}{\bagA}}{\varA}{\bagA}\sresredRT
      \muldeg{\CCBagCtx{\varA}{\bagA}}\bagA$,
      and $\rsubst{\bagC}{\varA}{\bagA}\sresredRT 0$
      for any other $\bagC\in\IdBagExp{\varA}$;
    \item\label{lemma:cc:str}
      if $\termA=\strA\in\StreamTerms$ then:
      \begin{itemize}
        \item there exists $\CCStrCtx{\svarA}{\strA}\in\IdStrExp{\svarA}$ such
          that $\rsubst{\CCStrCtx{\svarA}{\strA}}{\svarA}{\strA}\sresredRT
          \muldeg{\CCStrCtx{\svarA}{\strA}}\strA$,
          and $\rsubst{\strC}{\svarA}{\strA}\sresredRT 0$ for any
          other $\strC\in\IdStrExp{\svarA}$;
        \item and there exists $\CCVar{\varA}{\strA}\in\IdExp{\varA}$ such that
          $\appl{\CCVar{\varA}{\strA}}{\strA}\sresredRT
          \muldeg{\CCVar{\varA}{\strA}}\pars{\appl{\varA}{\strA}}$,
          and $\appl{\vtermC}{\strA}\sresredRT 0$ for any other
          $\vtermC\in\IdExp{\varA}$.
      \end{itemize}
  \end{enumerate}
\end{lemma}

We will deduce the announced identity behaviour from 
\cref{lemma:cc:var,lemma:cc:vterm}.
The other items will be leveraged similarly (see the corresponding items of
\cref{lemma:cc:vresred}), but they are already useful here as intermediate
steps in the proof of \cref{lemma:cc}:
we prove all five items simultaneously by induction on \(\termA\).
Along the proof, we define
$\CCBag{\varA}{\termA}\in\IdBagExp{\varA}$,
$\CCStr{\svarA}{\termA}\in\IdStrExp{\svarA}$,
$\CCValCtx{\varA}{\vtermA}\in\IdExp{\varA}$,
$\CCBagCtx{\varA}{\bagA}\in\IdBagExp{\varA}$,
$\CCStrCtx{\svarA}{\strA}\in\IdStrExp{\svarA}$
and $\CCVar{\varA}{\strA}\in\IdExp{\varA}$,
that we call \definitive{copycat terms}:\footnote{
  Although we do not develop the analogy further,
  a game semanticist reader may infer a precise correspondence between copycat
  terms and the isomorphism classes of augmentations in a copycat strategy
  over the universal arena,
  in light of \cref{section:gs}.
}
we summarize their mutually inductive definitions in \cref{fig:cc},
so that it is safe to skip the proof.

\begin{figure}
  \begin{gather*}
    \CCBag{\varA}{\labs{\svarB}{\btermB}}\eqdef\CCBag{\varA}{\btermB}
    \qquad
    \CCBag{\varA}{\mset{\vtermA_1,\dotsc,\vtermA_k}}
    \eqdef\CCBag{\varA}{\vtermA_1}\bagcat\cdots\bagcat\CCBag{\varA}{\vtermA_k}
    \\
    \CCBag{\varA}{\appl{\varB}{\strB}}\eqdef\CCBag{\varA}{\strB}
    \qquad
    \CCBag{\varA}{\appl{\varA}{\strB}}
    \eqdef\mset{\CCVar{\varA}{\strB}}\bagcat\CCBag{\varA}{\strB}
    \qquad
    \CCBag{\varA}{\appl{\vtermA}{\strB}}
    \eqdef\CCBag{\varA}{\vtermA}\bagcat\CCBag{\varA}{\strB}
    \\
    \CCBag{\varA}{\emptystream}\eqdef\emptybag
    \qquad
    \CCBag{\varA}{\bagA\cons\strB}
    \eqdef\CCBag{\varA}{\bagA}\bagcat\CCBag{\varA}{\strB}
    \\
    \CCStr{\svarA}{\termA}\eqdef
    \CCBag{\svarA[0]}{\termA}\cons\cdots\cons\CCBag{\svarA[l-1]}{\termA}\cons\emptystream
    \\
    \CCValCtx{\varA}{\labs{\svarB}{\btermB}}\eqdef\labs{\svarB}{\appl{\varA}{\CCStr{\svarB}{\btermB}}}
    \qquad
    \CCBagCtx{\varA}{\mset{\vtermA_1,\dotsc,\vtermA_k}}
    \eqdef\mset{\CCValCtx{\varA}{\vtermA_1},\dotsc,\CCValCtx{\varA}{\vtermA_k}}
    \\
    \CCStrCtx{\svarA}{\emptystream}\eqdef\emptystream
    \qquad
    \CCStrCtx{\svarA}{\bagA\cons\strB}\eqdef
    \CCBagCtx{\svarA[0]}{\bagA}\cons\shiftup{\CCStrCtx{\svarA}{\strB}}{\svarA}
    \\
    \CCVar{\varA}{\strA}\eqdef\labs{\svarB}{\appl{\varA}{\CCStrCtx{\svarB}{\strA}}}
\end{gather*}
  \caption{
    Definition of copycat terms
    (assuming \(\varB\not=\varA\), and choosing \(\svarB\not\ni\varA\) whenever relevant;
    and choosing \(l\) such that no \(\svarA[i]\) with \(i\ge l\) is free in \(\termA\) in the
    definition of \(\CCStr{\svarA}{\termA}\))
  }
  \label{fig:cc}
\end{figure}

\begin{proof}[Proof of Lemma~\ref{lemma:cc}]
  Note that \cref{lemma:cc:svar} is obtained by iterating
  \cref{lemma:cc:var}.
  Indeed, if $\strC=\tuple{\bagC_i}_{i\in\N}\in\IdStrExp{\svarA}$
  and $k\in\N$ is such that $\svarA[i]\not\in\LVarOf{\termA}$ when $i\ge k$,
  then $\rsubst{\termA}{\svarA}{\strC}\sresredRT0$
  as soon as $\bagC_i\not=\emptybag$ for some $i\ge k$.
  And, otherwise,
  $\rsubst{\termA}{\svarA}{\strC}
  =\rsubst{\rsubst{\termA}{\svarA[0]}{\bagC_0}\cdots}{\svarA[k-1]}{\bagC_{k-1}}$
  --- observing that no $\svarA[i]$ is free in $\bagC_j$ when $i\not=j$.
  We then apply \cref{lemma:cc:var} to the same term \(\termA\) and each \(\svarA[i]\) for \(i<k\), and write
  \[ 
    \termA'\eqdef
    \rsubst{\rsubst{\termA}{\svarA[0]}{\CCBag{\svarA[0]}{\termA}}\cdots}{\svarA[k-1]}{\CCBag{\svarA[k-1]}{\termA}}
  \]
  and
  \[
    \CCStr{\svarA}{\termA}\eqdef
    \CCBag{\svarA[0]}{\termA}\cons\cdots\cons\CCBag{\svarA[k-1]}{\termA}\cons\emptystream
  \]
  so that 
  $\termA'=\rsubst{\termA}{\svarA}{\CCStr{\svarA}{\termA}}$.
  By the properties of each \(\CCBag{\svarA[i]}{\termA}\),
  together with \cref{lemma:sresred:rsubst}, we obtain
  $\termA'
  \sresredRT\pars{\prod_{i=0}^{k-1}\muldeg{\CCBag{\svarA[i]}{\termA}}}\termA
  =\muldeg{\CCStr{\svarA}{\termA}}\termA$.
  Moreover,
  we have $\rsubst{\termA}{\svarA[i]}{\bagC}\sresredRT 0$
  for any $\bagC\in\IdBagExp{\svarA[i]}$ other than $\CCBag{\svarA[i]}{\termA}$,
  for $0\le i<k$:
  recalling that we have $\rsubst{\termA}{\svarA[i]}{\bagC}\sresredRT 0$
  for any $\bagC\not=\emptybag$ for $i\ge k$,
  $\CCStr{\svarA}{\termA}$ is the only $\strC\in\IdStrExp{\svarA}$
  with $\NF{\rsubst{\termA}{\svarA}{\strC}}\not =0$.

  So it is sufficient to establish 
  \cref{lemma:cc:var}, and possibly one of
  \cref{lemma:cc:vterm,lemma:cc:bag,lemma:cc:str},
  depending on the shape of \(\termA\),
  simultaneously by induction on \(\termA\).

  \subparagraph{Value terms:}
  If $\termA=\vtermA\in\ValueTerms$, we prove \cref{lemma:cc:var,lemma:cc:vterm}.
  We can write $\vtermA=\labs{\svarB}{\btermB}$,
  assuming w.l.o.g. that $\varA\not\in\svarB$.
  Applying the induction hypothesis (\cref{lemma:cc:var})
  to $\btermB$ entails 
  $\rsubst{\btermB}{\varA}{\CCBag{\varA}{\btermB}}\sresredRT
  \muldeg{\CCBag{\varA}{\btermB}}\btermB$
  and 
  $\rsubst{\btermB}{\varA}{\bagC}\sresredRT 0$
  for any other $\bagC\in\IdBagExp{\varA}$:
  we set $\CCBag{\varA}{\vtermA}\eqdef\CCBag{\varA}{\btermB}$,
  and deduce \cref{lemma:cc:var} for $\vtermA$ by \cref{lemma:sresred:context},
  observing that $\rsubst{\vtermA}{\varA}{\bagC}=\labs{\svarB}{\rsubst{\btermB}{\varA}{\bagC}}$
  for any $\bagC\in\BagTerms$.
  Moreover, we can write $\IdExp{\varA}=\labs{\svarB}{\appl{\varA}{\IdStrExp{\svarB}}}$:
  for each $\vtermC\in\IdExp{\varA}$, we have
  $\vtermC=\labs{\svarB}{\appl{\varA}{\strC}}$
  with $\strC\in\IdStrExp{\svarB}$.
  The induction hypothesis (\cref{lemma:cc:svar}) yields
  $\rsubst{\btermB}{\svarB}{\CCStr{\svarB}{\btermB}}\sresredRT
  \muldeg{\CCStr{\svarB}{\btermB}}\btermB$,
  and
  $\rsubst{\btermB}{\svarB}{\strC}\sresredRT 0$ for any other 
  $\strC\in\IdStrExp{\svarB}$.
  We set $\CCValCtx{\varA}{\vtermA}\eqdef\labs{\svarB}{\appl{\varA}{\CCStr{\svarB}{\btermB}}}\in\IdExp{\varA}$,
  and obtain \cref{lemma:cc:vterm} by \cref{lemma:sresred:context},
  observing that 
  $\subst*{\labs{\svarB}{\appl{\varA}{\strC}}}{\varA}{\vtermA}
  =\labs{\svarB}{\appl{\vtermA}{\strC}}
  \sresred
  \labs{\svarB}{\rsubst{\btermB}{\svarB}{\strC}}
  $
  for any $\strC\in\StreamTerms$.

  \subparagraph{Base terms:}
  If $\termA=\btermA\in\BaseTerms$, we only have to prove \cref{lemma:cc:var}.
  There are three possible cases:
  either $\btermA=\appl{\varB}{\strB}$ with $\varB\not=\varA$,
  or $\btermA=\appl{\varA}{\strB}$,
  or $\btermA=\appl{\vtermA}{\strB}$.
  If $\btermA=\appl{\varB}{\strB}$ with $\varB\not=\varA$,
  we apply the induction hypothesis (\cref{lemma:cc:var}) to $\strB$,
  set $\CCBag{\varA}{\btermA}\eqdef\CCBag{\varA}{\strB}$,
  and conclude as in the abstraction case.
 
  If $\btermA=\appl{\varA}{\strB}$,
  then for each $\bagC=\mset{\vtermC_1,\dotsc,\vtermC_k}\in\IdBagExp{\varA}$,
  we have $\rsubst{\btermA}{\varA}{\bagC}=0$ if $k=0$ and, otherwise,
  $\rsubst{\btermA}{\varA}{\bagC}
  =\sum_{i=1}^k\appl{\vtermC_i}{\rsubst{\strB}{\varA}{\bagC'_i}}$
  where each $\bagC'_i$ is such that $\bagC=\mset{\vtermC_i}\bagcat\bagC'_i$.
  The induction hypothesis (\cref{lemma:cc:var,lemma:cc:str}) applied
  to $\strB$ yields $\CCBag{\varA}{\strB}\in\IdBagExp{\varA}$
  and $\CCVar{\varA}{\strB}\in\IdExp{\varA}$,
  and we set $\CCBag{\varA}{\btermA}
  \eqdef\mset{\CCVar{\varA}{\strB}}\bagcat\CCBag{\varA}{\strB}
  \in\IdBagExp{\varA}$.
  If $\bagC\not=\CCBag{\varA}{\btermA}$,
  then each element in the previous sum normalizes to $0$.
  And, if $\bagC=\CCBag{\varA}{\btermA}$,
  we obtain $\rsubst{\btermA}{\varA}{\bagC}
  \sresredRT
  \pars{l\times\muldeg{\CCVar{\varA}{\strB}}\times\muldeg{\CCBag{\varA}{\strB}}}
  \btermA
  $
  where $l=\CardOf{\set{i\st \vtermC_i=\CCVar{\varA}{\strB}}}$.
  \Cref{fact:isodeg} entails $\muldeg{\CCBag{\varA}{\btermA}}
  =l\times\muldeg{\CCVar{\varA}{\strB}}\times\muldeg{\CCBag{\varA}{\strB}}$,
  and we obtain \cref{lemma:cc:var} for $\btermA$.

  If $\btermA=\appl{\vtermA}{\strB}$,
  then for each $\bagC\in\IdBagExp{\varA}$,
  we have 
  $\rsubst{\btermA}{\varA}{\bagC}
  =\sum_{\bagC\splitinto\bagC_1\bagcat\bagC_2}
  \appl{\rsubst{\vtermA}{\varA}{\bagC_1}}{\rsubst{\strB}{\varA}{\bagC_2}}$.
  The induction hypothesis (\cref{lemma:cc:var}) applied
  to $\vtermA$ and $\strB$ yields $\CCBag{\varA}{\vtermA}$ and $\CCBag{\varA}{\strB}\in\IdBagExp{\varA}$,
  and we set $\CCBag{\varA}{\btermA}
  \eqdef\CCBag{\varA}{\vtermA}\bagcat\CCBag{\varA}{\strB}
  \in\IdBagExp{\varA}$.
  If $\bagC\not=\CCBag{\varA}{\btermA}$,
  then each element in the previous sum normalizes to $0$.
  And, if $\bagC=\CCBag{\varA}{\btermA}$,
  we obtain $\rsubst{\btermA}{\varA}{\bagC}
  \sresredRT
  \pars[\big]{l\times\muldeg{\CCBag{\varA}{\vtermA}}\times\muldeg{\CCBag{\varA}{\strB}}}
  \btermA
  $
  where $l=\CardOf{\set{p:\bagC\splitinto 2\st \restrictToPre{\bagC}{p}{1}=\CCBag{\varA}{\vtermA}}}$.
  Again, we apply \cref{fact:isodeg} to check that
  $\muldeg{\CCBag{\varA}{\btermA}}
  =l\times\muldeg{\CCBag{\varA}{\vtermA}}\times\muldeg{\CCBag{\varA}{\strB}}$,
  which yields \cref{lemma:cc:var} for $\btermA$.

  \subparagraph{Bag terms:}
  If $\termA=\bagA\in\BagTerms$, we prove \cref{lemma:cc:var,lemma:cc:bag}.
  We can write
  $\bagA=\mset{\vtermA_1,\dotsc,\vtermA_k}$, and
  $\rsubst{\bagA}{\varA}{\bagC}
  =\sum_{\bagC\splitinto\bagC_1\bagcat\dotsc\bagcat\bagC_k}
  \mset{\rsubst{\vtermA_1}{\varA}{\bagC_1},\dotsc,\rsubst{\vtermA_k}{\varA}{\bagC_k}}$.
  The induction hypothesis (\cref{lemma:cc:var}) applied
  to each $\vtermA_i$ yields $\CCBag{\varA}{\vtermA_i}\in\IdBagExp{\varA}$
  and we set $\CCBag{\varA}{\bagA}
  \eqdef\CCBag{\varA}{\vtermA_1}\bagcat\cdots\bagcat\CCBag{\varA}{\vtermA_k}$.
  If $\bagC\not=\CCBag{\varA}{\bagA}$,
  then each element in the previous sum normalizes to $0$.
  And, if $\bagC=\CCBag{\varA}{\bagA}$,
  we obtain $\rsubst{\btermA}{\varA}{\bagC}
  \sresredRT
  \pars{l\times\prod_{i=1}^k\muldeg{\CCBag{\varA}{\vtermA_i}}}
  \bagA
  $
  where $l=\CardOf{\set{p:\bagC\splitinto k\st \restrictToPre{\bagC}{p}{i}=\CCBag{\varA}{\vtermA_i}\text{ for }1\le i\le k}}$.
  \Cref{lemma:cc:var} for $\bagA$ follows, again applying \cref{fact:isodeg}.
  The induction hypothesis (\cref{lemma:cc:vterm}) applied
  to each $\vtermA_i$ also yields $\CCValCtx{\varA}{\vtermA_i}\in\IdExp{\varA}$
  and we set $\CCBagCtx{\varA}{\bagA}
  \eqdef\mset{\CCValCtx{\varA}{\vtermA_1},\dotsc,\CCValCtx{\varA}{\vtermA_k}}\in\IdBagExp{\varA}$.
  Assume $\bagC=\mset{\vtermC_1,\dotsc,\vtermC_l}\in\IdBagExp{\varA}$.
  Since $\varA$ occurs exactly once in each $\vtermC_i$,
  we have $\rsubst{\bagC}{\varA}{\bagA}\sresredRT 0$ when $k\not=l$.
  And if $k=l$, we have 
  \[
    \rsubst{\bagC}{\varA}{\bagA}=
    \sum_{σ\in\Perms k}\mset{\subst{\vtermC_1}{\varA}{\vtermA_{σ(1)}},\dotsc,\subst{\vtermC_1}{\varA}{\vtermA_{σ(k)}}}
  \;.\]

  Hence, if $\bagC\not=\CCBagCtx{\varA}{\bagA}$,
  then each element in the previous sum normalizes to $0$.
  And, if $\bagC=\CCBagCtx{\varA}{\bagA}$,
  we obtain $\rsubst{\btermA}{\varA}{\bagC}
  \sresredRT
  \pars{l\times\prod_{i=1}^k\muldeg{\CCValCtx{\varA}{\vtermA_i}}}
  \bagA
  $
  where
  \( l
  =\CardOf{\set{σ\in\Perms k\st \vtermC_i =\CCValCtx{\varA}{\vtermA_{σ(i)}}\text{ for }1\le i\le k}}
  =\isodeg{\CCBagCtx{\varA}{\bagA}}
  \).
  \Cref{lemma:cc:bag} for $\bagA$ follows
  since $l\times\prod_{i=1}^k\muldeg{\CCValCtx{\varA}{\vtermA_i}}
  =\muldeg{\CCBagCtx{\varA}{\bagA}}$ by definition.

  \subparagraph{Stream terms:}
  Finally, if $\termA=\strA\in\StreamTerms$ then we prove
  \cref{lemma:cc:var,lemma:cc:str}.
  Note that the first statement of \cref{lemma:cc:str} entails the second one.
  Indeed, if $\vtermC\in\IdExp{\varA}$
  then we can write $\vtermC=\labs{\svarB}{\appl{\varA}{\strC}}$
  with $\strC\in\IdStrExp{\svarB}$,
  and then $\appl{\vtermC}{\strA}\sresred\appl{\varA}{\rsubst{\strC}{\svarB}{\strA}}$:
  we set
  $\CCVar{\varA}{\strA}\eqdef\labs{\svarB}{\appl{\varA}{\CCStrCtx{\svarB}{\strA}}}$.
 
  In case $\termA=\emptystream$, \cref{lemma:cc:var,lemma:cc:str} are straightforward,
  with $\CCBag{\varA}{\emptystream}\eqdef\emptybag$,
  and $\CCStrCtx{\svarA}{\emptystream}\eqdef\emptystream$.

  It remains only to establish \cref{lemma:cc:var} and the first statement of \cref{lemma:cc:str}
  for $\termA=\bagA\cons\strB\not=\emptystream$,
  assuming the induction hypothesis holds for $\bagA$ and $\strB$.
  We obtain \cref{lemma:cc:var}
  from the induction hypothesis (\cref{lemma:cc:var})
  applied to $\bagA$ and $\strB$,
  similarly to the case of $\termA=\appl{\vtermA}{\strB}$.
  By \cref{lemma:cc:shift},
  for any $\strC\in\IdStrExp{\svarA}$,
  we can write $\strC=\bagC_0\cons\shiftup{\strC'}{\svarA}$
  with $\bagC_0\in\IdBagExp{\svarA[0]}$ and $\strC'\in\IdStrExp{\svarA}$,
  to obtain 
  $\rsubst{\strC}{\svarA}{\termA}=\rsubst{\bagC_0}{\svarA[0]}{\bagA}\cons\rsubst{\strC'}{\svarA}{\strB}$.
  We apply the induction hypothesis to $\bagA$ (\cref{lemma:cc:bag})
  and $\strB$ (\cref{lemma:cc:str}) to obtain 
  $\CCBagCtx{\svarA[0]}{\bagA}$ and $\CCStrCtx{\svarA}{\strB}$.
  We set $\CCStrCtx{\svarA}{\termA}\eqdef
  \CCBagCtx{\svarA[0]}{\bagA}\cons\shiftup{\CCStrCtx{\svarA}{\strB}}{\svarA}
  \in\IdStrExp{\svarA}$
  (\cref{lemma:cc:shift} again).
  We have
  \(\muldeg{\CCBagCtx{\svarA}{\termA}}\termA
  = \muldeg{\CCBagCtx{\svarA[0]}{\bagA}}\times \muldeg{\shiftup{\CCStrCtx{\svarA}{\strB}}{\svarA}}
  = \muldeg{\CCBagCtx{\svarA[0]}{\bagA}}\times \muldeg{\CCStrCtx{\svarA}{\strB}}
  \) (shifts clearly preserve the isotropy degree).
  Finally,
  we apply 
  \cref{lemma:rsubst:str:def}
  to obtain
  $\rsubst{\CCStrCtx{\svarA}{\termA}}{\svarA}{\termA}
  ={\rsubst{\CCBagCtx{\svarA[0]}{\bagA}}{\svarA[0]}{\bagA}}
  \cons{\rsubst{\CCStrCtx{\svarA}{\strB}}{\svarA}{\strB}}
  \sresredRT
  \pars{\muldeg{\CCBagCtx{\svarA[0]}{\bagA}}\times \muldeg{\CCStrCtx{\svarA}{\strB}}}
  \pars{\bagA\cons\strB}
  =
  \muldeg{\CCBagCtx{\svarA}{\termA}}\termA
  $, and 
  $\rsubst{\strC}{\svarA}{\termA}\sresredRT 0$ for any other
  $\strC\in\IdStrExp{\svarA}$.
\end{proof}

Observe that the use of iterated reduction \(\sresredRT\) is crucial in the
previous proof:
in particular, to obtain \cref{lemma:cc:vterm} in the value case, 
we define \(\CCValCtx{\varA}{\labs{\svarB}{\btermB}}\eqdef\labs{\svarB}{\appl{\varA}{\CCStr{\svarB}{\btermB}}}\),
and use one step of resource reduction to apply \cref{lemma:cc:var} to the underlying base term;
and to obtain the second part of \cref{lemma:cc:str},
we define \(\CCVar{\varA}{\strA}\eqdef\labs{\svarB}{\appl{\varA}{\CCStrCtx{\svarB}{\strA}}}\),
and use one step of resource reduction to apply the first part;
since each item inductively depends on the others,
we do need to iterate those reductions.

A straightforward induction allows checking that 
\(\LengthOf{\CCBag{\varA}{\termA}}=\nocc{\varA}{\termA}\):
in particular, if \(\varA\not\in\LVarOf{\termA}\), then 
\(\CCBag{\varA}{\termA}=\emptybag\).
Moreover, if \(\svarA\not\in\SVarOf{\termA}\)
then \(\CCStr{\svarA}{\termA}=\emptystream\)
directly by definition.

\begin{example}
  \label{example:cc}
  The reader may check that
  \(\CCVar{\varA}{\emptystream}=\trProj{\varA}\),
  hence
  \(
  \CCBag{\varA}{\trProj{\varA}}=
  \CCBag{\varA}{\appl{\varA}{\emptystream}}=
  \mset{\trProj{\varA}}
  \).
  Indeed,
  \(\appl{\trProj{\varA}{\emptystream}}\sresred\appl{\varA}{\emptystream}\)
  and
  \(
  \subst{\trProj{\varA}}{\varA}{\trProj{\varA}}
  =
  \rsubst{\trProj{\varA}}{\varA}{\mset{\trProj{\varA}}}
  \sresred
  \trProj{\varA}
  \) as in \cref{example:reduction:fine}.
  Moreover, 
  \(\CCBag{\varA}{\appl{\varB}{\emptystream}}=
  \emptybag\) when \(\varA\not=\varB\),
  hence
  \(\CCStr{\svarB}{\appl{\svarB[0]}{\emptystream}}
  =\mset{\trProj{\svarB[0]}}\cons\emptystream\)
  and
  \(\CCStr{\svarB}{\appl{\varA}{\emptystream}}
  =\emptystream\)
  when \(\varA\not\in\svarB\).
  We obtain
  \(\CCValCtx{\varA}{\trProj{0}}=
  \labs{\svarB}{\appl{\varA}{\mset{\trProj{\svarB[0]}}\cons\emptystream}}
  =\trCCz{\varA}\)
  (as defined in \cref{example:terms})
  and 
  \(\CCValCtx{\varB}{\trProj{\varA}}=
  \trProj{\varB}\) for any (non-necessarily distinct) variables
  \(\varA\) and \(\varB\).
  And, indeed,
  \(
  \subst{\trCCz{\varA}}{\varA}{\trProj{0}}
  =
  \labs{\svarB}{\appl{\trProj{0}}{\mset{\trProj{\svarB[0]}}\cons\emptystream}}
  \sresredRT
  \trProj{0}
  \)
  as in \cref{example:reduction:fine},
  and we have already recalled that 
  \(\subst{\trProj{\varB}}{\varB}{\trProj{\varA}}\sresred\trProj{\varA}\).

  We also obtain
  \(\CCVar{\varA}{\mset{\trProj{\varB}}\cons\emptystream}
  =
  \labs{\svarC}{\appl{\varA}{\CCStrCtx{\svarC}{\mset{\trProj{\varB}}\cons\emptystream}}}
  =
  \labs{\svarC}{\appl{\varA}{\mset{\CCValCtx{\svarC[0]}{\trProj{\varB}}}\cons\emptystream}}
  =
  \labs{\svarC}{\appl{\varA}{\mset{\trProj{\svarC[0]}}}\cons\emptystream}
  =
  \labs{\varB}{\labs{\svarC}{\appl{\varA}{\mset{\trProj{\varB}}}\cons\emptystream}}
  =
  \trCCz{\varA}
  \)
  for any variables \(\varA\) and \(\varB\).
  And, indeed,
  \(
  \appl{\trCCz{\varA}}{\mset{\trProj{\varB}}\cons\emptystream}
  \sresredRT
  \rsubst*{\appl{\varA}{\mset{\trProj{\varB}}}\cons\emptystream}{\varB}{\mset{\trProj{\varB}}}
  =
  \appl{\varA}{\mset{\rsubst{\trProj{\varB}}{\varB}{\trProj{\varB}}}\cons\emptystream}
  \sresredRT
  \appl{\varA}{\mset{\trProj{\varB}}\cons\emptystream}
  \).

  Now consider the resource sequence
  \(\strA
  = \mset{ \trProj{0}, \trProj{0} } \cons\emptystream
  = \trProj{0}^2 \cons\emptystream\)
  with \(\muldeg{\strA}=2\).
  We have \(
  \CCStrCtx{\svarB}{\strA}
  =\CCBagCtx{\svarB[0]}{\trProj{0}^2}\cons\emptystream
  =\CCValCtx{\svarB[0]}{\trProj{0}}^2\cons\emptystream
  =\trCCz{\svarB[0]}^2\cons\emptystream
  \).
  So \(
  \CCVar{\varA}{\strA}
  =\labs{\svarB}{\appl{\varA}{\trCCz{\svarB[0]}^2\cons\emptystream}}
  =\trCCp{\varA}
  \).
  And, indeed,
  \(\appl{\CCVar{\varA}{\strA}}{\strA}
  =\appl{\trCCp{\varA}}{\trProj{0}^2 \cons\emptystream}
  \sresredRT 2\appl{\varA}{\trProj{0}^2 \cons\emptystream}
  =\muldeg{\strA}\appl{\varA}{\strA}
  \)
  as in \cref{example:reduction:fine}.
\end{example}
\begin{remark}
    It is no accident that
    \(\CCBag{\varA}{\trProj{\varA}}=\mset{\trProj{\varA}}\)
    and 
    \(\CCValCtx{\varB}{\trProj{\varA}}=\trProj{\varB}=\subst{\trProj{\varA}}{\varA}{\varB}\) in \cref{example:cc}:
    as a general fact, elements of the expansion of variables are their own copycat terms.
    More precisely, one can prove that, for any value variables \(\varA\) and
    \(\varB\) and any sequence variables \(\svarA\) and \(\svarB\), we have:
    \begin{enumerate}
      \item if \(\vtermA\in\IdExp{\varA}\) then \(\CCBag{\varA}{\vtermA}=\mset{\vtermA}\)
        and \(\CCValCtx{\varB}{\vtermA}=\subst{\vtermA}{\varA}{\varB}\);
      \item if \(\bagA\in\IdBagExp{\varA}\) then \(\CCBag{\varA}{\bagA}=\bagA\)
        and \(\CCBagCtx{\varB}{\bagA}=\subst{\bagA}{\varA}{\varB}\);
      \item
        if \(\strA\in\IdStrExp{\svarA}\) then \(\CCStr{\svarA}{\strA}=\strA\),
        \(\CCStrCtx{\svarB}{\strA}=\subst{\strA}{\svarA}{\svarB}\) and
        \(\CCVar{\varB}{\strA}=\labs{\svarA}{\appl{\varB}{\strA}}\).
    \end{enumerate}

    We do not develop this result further as it plays no role in the sequel.
    A game semanticist reader detailing the proof might be reminded of the
    composition of copycat plays.
\end{remark}

\begin{lemma}\label{lemma:cc:vresred}
  For any resource term $\termA$,
  any value variable $\varA$,
  and any sequence variable $\svarA$,
  the following holds:
  \begin{enumerate}
    \item\label{lemma:cc:vresred:var}
      $\subst{\termA}{\varA}{\IdExp{\varA}}\vresred\termA$;
    \item\label{lemma:cc:vresred:svar}
      $\subst{\termA}{\svarA}{\IdSeqExp{\svarA}}\vresred\termA$;
    \item\label{lemma:cc:vresred:vterm}
      if $\termA=\vtermA\in\ValueTerms$ then 
      $
      \subst{\IdExp{\varA}}{\varA}{\vtermA}
      \vresred\vtermA$;
    \item\label{lemma:cc:vresred:bag}
      if $\termA=\bagA\in\BagTerms$ then
      $\rsubst{\IdBagExp{\varA}}{\varA}{\bagA}\vresred\bagA$;
    \item\label{lemma:cc:vresred:str}
      if $\termA=\strA\in\StreamTerms$ then 
      $\rsubst{\IdStrExp{\svarA}}{\svarA}{\strA}\vresred\strA$
      and 
      $\appl{\IdExp{\varA}}{\strA}\vresred\appl{\varA}{\strA}$;
    \item\label{lemma:cc:vresred:base}
      if $\termA=\btermA\in\BaseTerms$ then 
      $\appl*{\labs{\svarA}{\btermA}}{\IdStrExp{\svarA}}\vresred\btermA$.
  \end{enumerate}
\end{lemma}
\begin{proof}
  \Crefrange{lemma:cc:vresred:var}{lemma:cc:vresred:str} follow from the
  corresponding items of \cref{lemma:cc}, combined with \cref{lemma:cc:coef}.
  For \cref{lemma:cc:vresred:base}, 
  write $\appl*{\labs{\svarA}{\btermA}}{\IdStrExp{\svarA}}
  =\sum_{\strC\in\IdStrExp{\svarA}}\frac{1}{\muldeg{\strC}}
  \appl*{\labs{\svarA}{\btermA}}{\strC}$
  thanks to  \cref{lemma:cc:coef},
  observe that 
  $\appl*{\labs{\svarA}{\btermA}}{\strC}\resred
  \rsubst{\btermA}{\varA}{\strC}$
  and conclude by \cref{lemma:cc:svar} of  \cref{lemma:cc}.
\end{proof}

\begin{lemma}\label{lemma:cc:vectors}
  For any term vector $\termvA$,
  any value variable $\varA$,
  and any sequence variable $\svarA$,
  the following holds:
  \begin{enumerate}
    \item\label{lemma:cc:vectors:var}
      $\subst{\termvA}{\varA}{\IdExp{\varA}}\vresred\termvA$
      and, for any variable \(\varB\), 
      $\subst{\termvA}{\varB}{\IdExp{\varA}}\vresred\subst{\termvA}{\varB}{\varA}$;
    \item\label{lemma:cc:vectors:svar}
      $\subst{\termvA}{\svarA}{\IdSeqExp{\svarA}}\vresred\termvA$
      and, for any sequence variable \(\svarB\), 
      $\subst{\termvA}{\svarB}{\IdSeqExp{\svarA}}\vresred\subst{\termvA}{\svarB}{\svarA}$;
    \item\label{lemma:cc:vectors:vterm}
      if $\termvA=\vtermvA\in\ValueVectors$ then 
      $
      \subst{\IdExp{\varA}}{\varA}{\vtermvA}
      \vresred\vtermvA$;
      moreover, if $\svarA\not\in\SVarOf{\vtermvA}$, then
      $\labs{\svarA}{\appl{\vtermvA}{\IdStrExp{\svarA}}}\vresred\vtermvA$;
    \item\label{lemma:cc:vectors:bag}
      if $\termvA=\bagvA\in\BagVectors$ then
      $\rsubst{\IdBagExp{\varA}}{\varA}{\bagvA}\vresred\bagvA$;
    \item\label{lemma:cc:vectors:str}
      if $\termvA=\strvA\in\StreamVectors$ then 
      $\rsubst{\IdStrExp{\svarA}}{\svarA}{\strvA}\vresred\strvA$
      and 
      $\appl{\IdExp{\varA}}{\strvA}\vresred\appl{\varA}{\strvA}$.
    \item\label{lemma:cc:vectors:base}
      if $\termvA=\btermvA\in\BaseVectors$ then 
      $\appl*{\labs{\svarA}{\btermvA}}{\IdStrExp{\svarA}}\vresred\btermvA$
      and, for any sequence variable \(\svarB\),
      $\appl*{\labs{\svarB}{\btermvA}}{\IdStrExp{\svarA}}\vresred
      \subst{\btermvA}{\svarB}{\svarA}$.
  \end{enumerate}

  Finally, for any sequence $\seq{\vtermvA}$ of value vectors
  such that $\SVarOf{\seq{\vtermvA}}$ is finite, we have
  $\subst{\IdStrExp{\svarA}}{\svarA}{\seq{\vtermvA}}\vresred\prom{\seq{\vtermvA}}$.
\end{lemma}
\begin{proof}
  Except for the second part of \cref{lemma:cc:vectors:vterm},
  \crefrange{lemma:cc:vectors:var}{lemma:cc:vectors:base} follow directly from
  the corresponding items of \cref{lemma:cc:vresred} by linearity
  and compatibility.
  If $\vtermvA\in\ValueVectors$ and $\svarA\not\in\SVarOf{\vtermvA}$,
  we can write $\vtermvA=\labs{\svarA}{\btermvA}$, and we obtain 
  $\labs{\svarA}{\appl{\vtermvA}{\IdStrExp{\svarA}}}\vresred\vtermvA$
  by \cref{lemma:cc:vectors:base}.
  Finally, for any sequence $\seq{\vtermvA}$ of value vectors
  such that $\SVarOf{\seq{\vtermvA}}$ is finite,
  \cref{lemma:subst:prom:seq} gives
  $\subst{\IdStrExp{\svarA}}{\svarA}{\seq{\vtermvA}}
  =\rsubst{\IdStrExp{\svarA}}{\svarA}{\prom{\seq{\vtermvA}}}
  $ and we conclude by \cref{lemma:cc:vectors:str}.
\end{proof}

\subsection{Two flavours of extensional Taylor expansion}

Write $\LambdaTerms$ for the set of pure, ordinary $λ$-terms,
that we denote by letters $\ltermA,\ltermB,\ltermC$,
and write \(\bred\) (resp.\ \(\etared\)) for the 
\(β\)-reduction (resp.\ \(\eta\)-reduction) relation,
induced by the application of the base case
\(\appl*{\labs{\varA}{\ltermA}}{\ltermB}\bred\subst{\ltermA}{\varA}{\ltermB}\)
(resp.\ \(\labs{\varA}{\appl{\ltermA}{\varA}}\etared\ltermA\) when
\(\varA\) fresh)
in any context.
We write \(\betaetared\) for their union.
We may also write \(\ltermA\etarev\ltermB\) when \(\ltermB\etared\ltermA\).

\begin{example}
  \label{example:IJ}
  We will use the following \(λ\)-terms as running examples:
  \begin{itemize}
    \item the identity \(\tlId\eqdef\labs{\varA}{\varA}\);
    \item the Church numeral \(\tlOne\eqdef\labs{\varA}{\labs{\varB}{\appl{\varA}{\varB}}}\);
    \item the term \(\tlJ=\appl{\tlJaux}{\tlJaux}\) with
      \(\tlJaux\eqdef\labs{\varC}{\labs{\varA}{\labs{\varB}{\appl{\varA}*{\appl{\varC}{\varC\,\varB}}}}}\)
      (a version of Wadsworth's \(\tlJ\) combinator).
  \end{itemize}

  Observe that \(\tlId\etarev\tlOne\).
  Moreover, \(\tlJ\bred\labs{\varA}{\labs{\varB}{\appl{\varA}*{\appl{\tlJ}{\varB}}}}\),
  hence
  \(
  \appl{\tlJ}{\varB}\bred^2
  \labs{\varB'}{\appl{\varB}*{\appl{\tlJ}{\varB'}}}
  \).
  It follows that
  \[
  \tlJ \bred
  \labs{\varA}{\labs{\varB_0}{\appl{\varA}*{\appl{\tlJ}{\varB_0}}}}
  \bred^2
  \labs{\varA}{\labs{\varB_0}{\appl{\varA}*{
    \labs{\varB_1}{\appl{\varB_0}*{\appl{\tlJ}{\varB_1}}}
  }}}
  \bred^2
  \labs{\varA}{\labs{\varB_0}{\appl{\varA}*{
    \labs{\varB_1}{\appl{\varB_0}*{
      \labs{\varB_2}{\appl{\varB_1}*{
        \appl{\tlJ}{\varB_2}
      }}
    }}
  }}}
  \bred^2
  \cdots
  \]
  to be compared with nested \(η\)-expansions of \(\tlId\):
  \[
  \tlId \etarev
  \labs{\varA}{\labs{\varB_0}{\appl{\varA}{{\varB_0}}}}
  \etarev
  \labs{\varA}{\labs{\varB_0}{\appl{\varA}*{
    \labs{\varB_1}{\appl{\varB_0}{{\varB_1}}}
  }}}
  \etarev
  \labs{\varA}{\labs{\varB_0}{\appl{\varA}*{
    \labs{\varB_1}{\appl{\varB_0}*{
      \labs{\varB_2}{\appl{\varB_1}{
        {\varB_2}
      }}
    }}
  }}}
  \etarev
  \cdots
  \]

  Intuitively (and this can be made formal), both sequences converge to the same
  infinite tree, which is the Böhm tree of \(\tlJ\):
  in that sense, \(\tlJ\) computes the result of infinitely many nested
  \(η\)-expansions of the identity.
\end{example}

We define the \definitive{structural Taylor expansion}
$\ExtTayExp{\ltermA}$ of an ordinary $λ$-term $\ltermA$
by induction on $\ltermA$ as follows:
\[
  \ExtTayExp{\varA}\eqdef\IdExp{\varA}
  \qquad
  \ExtTayExp{\labs{\varA}{\ltermA}}\eqdef\labs{\varA}{\ExtTayExp{\ltermA}}
  \qquad
  \ExtTayExp{\appl{\ltermA}{\ltermB}}\eqdef
  \labs{\svarB}{\appl{\ExtTayExp{\ltermA}}{\ExtArgTayExp{\ltermB}\cons\IdStrExp{\svarB}}}
\]
together with 
\(\ExtArgTayExp{\ltermA}\eqdef\prom{\ExtTayExp{\ltermA}}\),
where $\svarB$ is chosen fresh in the application case.

The \definitive{head Taylor expansion} $\HeadTayExp{\ltermA}$
is defined inductively on the head structure of $\ltermA\in\LambdaTerms$,
for which we first need to introduce some notations.
Given a sequence $\ltermsB=\tuple{\ltermB_1,\dotsc,\ltermB_k}$ of $λ$-terms,
we write the \definitive{iterated application}
$\appl{\ltermA}{\ltermsB}\eqdef\appl{\ltermA}{\ltermB_1\cdots\ltermB_k}$.
Similarly, if $\strvB=\tuple{\bagvB_1,\dotsc,\bagvB_k}\in\BagVectors^k$ is a sequence 
of bag vectors and $\strvA\in\StreamVectors$ is a stream vector,
we write the \definitive{concatenation}
$\strvB\seqcat\strvA\eqdef\bagvB_1\cons\cdots\cons\bagvB_k\cons\strvA$.

Recall that if $\ltermA$ is a $λ$-term, then:
\begin{itemize}
  \item either $\ltermA$ is an abstraction;
  \item or we can write $\ltermA=\appl{\varA}{\ltermsB}$;
  \item or we can write $\ltermA=\appl{\ltermC}{\ltermsB}$
    where $\ltermC$ is an abstraction and $\ltermsB\not=\emptyword$.
\end{itemize}
Then we define:
\[
  \HeadTayExp{\labs{\varA}{\ltermA}}
  \eqdef\labs{\varA}{\HeadTayExp{\ltermA}}
  \qquad
  \HeadTayExp{\appl{\varA}{\ltermsB}}
  \eqdef\labs{\svarB}{\appl{\varA}{\HeadArgsTayExp{\ltermsB}\seqcat\IdStrExp{\svarB}}}
  \qquad
  \HeadTayExp{\appl{\ltermC}{\ltermsB}}
  \eqdef\labs{\svarB}{\appl{\HeadTayExp{\ltermC}}{\HeadArgsTayExp{\ltermsB}\seqcat\IdStrExp{\svarB}}}
\]
together with 
$\HeadArgTayExp{\ltermA}\eqdef \prom{\HeadTayExp{\ltermA}}$
and
$\HeadArgsTayExp{\tuple{\ltermA_1,\dotsc,\ltermA_k}}
\eqdef \tuple{\HeadArgTayExp{\ltermA_1},\dotsc,\HeadArgTayExp{\ltermA_k}}$,
choosing $\svarB$ fresh in the last two cases,
and assuming $\ltermC$ is an abstraction and $\ltermsB\not=\emptyword$ in the application case.
We may also write $\HeadTayExp{\tuple{\ltermA_1,\dotsc,\ltermA_k}}
\eqdef \tuple{\HeadTayExp{\ltermA_1},\dotsc,\HeadTayExp{\ltermA_k}}$,
a tuple of value vectors.
Observe that 
$\HeadTayExp{\varA}=\HeadTayExp{\appl{\varA}{\emptyword}}
=\labs{\svarB}{\appl{\varA}{\emptyword\seqcat\IdStrExp{\svarB}}}
=\IdExp{\varA}
=\ExtTayExp{\varA}$.

Finally, if $\seq{\varA}=\tuple{\varA_1,\dotsc,\varA_k}$ we write
$\labs{\seq{\varA}}{\ltermA}\eqdef\labs{\varA_1}{\cdots\labs{\varA_k}{\ltermA}}$
for $\ltermA$ a $λ$-term or a value vector, so that:
$\ExtTayExp{\labs{\seq{\varA}}{\ltermA}}=\labs{\seq{\varA}}{\ExtTayExp{\ltermA}}$ and
$\HeadTayExp{\labs{\seq{\varA}}{\ltermA}}=\labs{\seq{\varA}}{\HeadTayExp{\ltermA}}$.

Although this plays no particular role in the present section,
as we do not rely on uniformity in our study of extensional Taylor expansion,
we obtain the same characterization of coefficients as for ordinary Taylor expansion
(the multiplicity degree \(\muldeg{\vtermA}\) was defined
on \cpageref{definition:muldeg}):\footnote{
  This result, together with \cref{lemma:NFT:coef}, will ensure that the
  \(λ\)-theory induced by the normal form of extensional Taylor expansion,
  to be considered in \cref{section:Hs},
  does not depend on the choice of the semiring of coefficients:
  it is entirely determined by the support of extensional Taylor expansion.
}
\begin{lemma}\label{lemma:taylor:coef}
  If $\vtermA\in\ExtTayExp{\ltermA}$ (resp.\ $\vtermA\in\HeadTayExp{\ltermA}$)
  then $\ExtTayExpCoef{\ltermA}{\vtermA}=\frac{1}{\muldeg{\vtermA}}$
  (resp.\ $\HeadTayExpCoef{\ltermA}{\vtermA}=\frac{1}{\muldeg{\vtermA}}$).
\end{lemma}
\begin{proof}
  As for \cref{lemma:cc:coef}, the proof is straightforward by induction on
  terms, using \cref{lemma:prom:coef} in the case of a bag.
\end{proof}

\begin{example}
  \label{example:taylor}
  We have \(\ExtTayExp{\tlId}=\HeadTayExp{\tlId}\): indeed 
  \(\HeadTayExp{\tlId}=\labs{\varA}{\labs{\svarB}{\appl{\varA}{\IdStrExp{\svarB}}}}
  =\labs{\varA}{\IdExp{\varA}}=\ExtTayExp{\tlId}
  \).
  The elements of \(\ExtTayExp{\tlId}=\HeadTayExp{\tlId}\)
  are the value terms of the shape
  \(\labs{\varA}{\vtermA}\) with \(\vtermA\in\IdExp{\varA}\):
  the smallest such term is \(\labs{\varA}{\trProj{\varA}}=\trProj{0}\),
  but we can consider arbitrarily complex ones such as
  \(\labs{\varA}{\trCCz{\varA}}=
  \labs{\svarB}{\appl{\svarB[0]}{\mset{\trProj{\svarB[1]}}\cons\emptystream}}\),
  or
  \(\labs{\varA}{\trCCp{\varA}}=
  \labs{\svarB}{\appl{\svarB[0]}{\trCCz{\svarB[1]}^2\cons\emptystream}}\),
  all in normal form.
  This contrasts with ordinary Taylor expansion,
  where \(\TayExp{\varA}=\varA\) and \(\TayExp{\tlId}=\tlId\).

  We have \(\HeadTayExp{\tlOne}
  =\labs{\varA}{\labs{\varB}{\labs{\svarC}{\appl{\varA}{\IdBagExp{\varB}\cons\IdStrExp{\svarC}}}}}
  =\labs{\varA}{\labs{\svarC}{\appl{\varA}{\IdStrExp{\svarC}}}}
  =\HeadTayExp{\tlId}\).
  On the other hand, \(\ExtTayExp{\tlOne}
  =\labs{\varA}{\labs{\varB}{\labs{\svarC}{\appl{\IdExp{\varA}}{\IdBagExp{\varB}\cons\IdStrExp{\svarC}}}}}
  =\labs{\varA}{\labs{\svarC}{\appl{\IdExp{\varA}}{\IdStrExp{\svarC}}}}
  \).
  Its elements are of the shape
  \(\labs{\varA}{\labs{\svarC}{\appl{\vtermA}{\strC}}}\)
  with \(\vtermA\in\IdExp{\varA}\) and \(\strC\in\IdStrExp{\svarC}\):
  the smallest one is thus
  \(\labs{\varA}{\labs{\svarC}{\appl{\trProj{\varA}}{\emptystream}}}\),
  whose normal form is
  \(\trProj{0}\in\HeadTayExp{\tlOne}\).
  More generally, by \cref{lemma:cc}, the normal form of such an element is \(0\),
  unless \(\vtermA=\CCVar{\varA}{\strC}\), in which case its normal form
  is \(\muldeg{\strC}\labs{\varA}{\labs{\svarC}{\appl{\varA}{\strC}}}\).
  This ensures that 
  \(\ExtTayExp{\tlOne}\vresred\HeadTayExp{\tlOne}\),
  a particular instance of \cref{theorem:exttohead}
  that we establish below.

  For \(\tlJ\) and \(\tlJaux\), 
  we will only consider head Taylor expansion
  to keep the notations manageable.
  We obtain
  \(\HeadTayExp{\tlJ}=
  \labs{\svarC}{\appl{\HeadTayExp{\tlJaux}}{\HeadArgTayExp{\tlJaux}\cons\IdStrExp{\svarC}}}\),
  where \(\HeadTayExp{\tlJaux}=
  \labs{\varC}{\labs{\varA}{\labs{\varB}{\labs{\svarC_1}{
    \appl{\varA}{\prom*[\big]{\labs{\svarC_2}{
      \appl{\varC}{\IdBagExp{\varC}\cons\IdBagExp{\varB}\cons\IdStrExp{\svarC_2}}
    }}\cons{\IdStrExp{\svarC_1}}}
  }}}}\).
  Elements of \(\HeadTayExp{\tlJaux}\) are all in normal form, and 
  the smallest one is 
  \[
    \trJaux{0}\eqdef\labs{\varC}{\labs{\varA}{\labs{\varB}{\labs{\svarC_1}{
      \appl{\varA}{\emptystream}
    }}}}
    =\labs{\varC}{\labs{\varA}{\labs{\svarC}{
      \appl{\varA}{\emptystream}
    }}}
    =\labs{\varC}{\labs{\varA}{\trProj{\varA}}}
    =\labs{\varC}{\trProj{0}}
  \]
  so that the smallest element of \(\HeadTayExp{\tlJ}\) 
  is \(\labs{\svarC}{\appl{\trJaux{0}}{\emptystream}}\), for which we observe
  \(
  \labs{\svarC}{\appl{\trJaux{0}}{\emptystream}}
  \sresred
  \labs{\svarC}{\appl{\trProj{0}}{\emptystream}}
  \sresred
  0
  \).
  To get a non-zero normal form, we can observe that
  \(
    \appl{\trJaux{0}}{\emptybag\cons\mset{\trProj{\varA}}\cons\emptystream}
    \sresred
    \appl{\trProj{0}}{\mset{\trProj{\varA}}\cons\emptystream}
    \sresredRT
    \appl{\varA}{\emptystream}
  \)
  so that
  \[
    \HeadTayExp{\tlJ}
    \ni
    \labs{\svarC}{\appl{\trJaux{0}}{\emptybag\cons\mset{\trProj{\svarC[0]}}\cons\emptystream}}
    =
    \labs{\varA}{\labs{\svarC}{\appl{\trJaux{0}}{\emptybag\cons\mset{\trProj{\varA}}\cons\emptystream}}}
    \sresredRT
    \labs{\varA}{\labs{\svarC}{
      \appl{\varA}{\emptystream}
    }}=\trProj{0}
    \in\ExtTayExp{\tlId}
    \;.
  \]

  This suggests defining
  \[
    \trJaux{1}\eqdef
    \labs{\varC}{\labs{\varA}{\labs{\varB}{\labs{\svarC_1}{
            \appl{\varA}{\mset[\big]{\labs{\svarC_2}{
                \appl{\varC}{\emptybag\cons\mset{\trProj{\varB}}\cons\emptystream}
            }}\cons{\emptystream}}
    }}}}
    \in\HeadTayExp{\tlJaux}
  \]
  to obtain
  (recalling from \cref{example:cc} that 
  \(\CCVar{\varA}{\mset{\trProj{\varB}}\cons\emptystream}=\trCCz{\varA}\)):
  \begin{eqnarray*}
  \HeadTayExp{\tlJ}\ni
  \labs{\svarC}{
    \appl{\trJaux{1}}{\mset{\trJaux{0}}\cons\mset{\trCCz{\svarC[0]}}\cons\mset{\trProj{\svarC[1]}}\cons\emptystream}
  }
  &=&
  \labs{\varA}{\labs{\varB}{\labs{\svarC}{
    \appl{\trJaux{1}}{\mset{\trJaux{0}}\cons\mset{\trCCz{\varA}}\cons\mset{\trProj{\varB}}\cons\emptystream}
  }}}
  \\
  &\sresredRT&
  \labs{\varA}{\labs{\varB}{\labs{\svarC_1}{
    \appl{\trCCz{\varA}}{
      \mset[\big]{
        \labs{\svarC_2}{\appl{\trJaux{0}}{
          \emptybag\cons\mset[\big]{\rsubst{\trProj{\varB}}{\varB}{\mset{\trProj{\varB}}}}
          \cons\emptystream
        }}
      }\cons\emptystream
    }
  }}}
  \\
  &\sresredRT&
  \labs{\varA}{\labs{\varB}{\labs{\svarC_1}{
    \appl{\trCCz{\varA}}{
      \mset[\big]{
        \labs{\svarC_2}{\appl{\trJaux{0}}{
            \emptybag\cons\mset{\trProj{\varB}}\cons\emptystream
        }}
      }\cons\emptystream
    }
  }}}
  \\
  &\sresredRT&
  \labs{\varA}{\labs{\varB}{\labs{\svarC_1}{
    \appl{\trCCz{\varA}}{
      \mset[\big]{
        \labs{\svarC_2}{\appl{\varB}{\emptystream}}
      }\cons\emptystream
    }
  }}}
  \\
  &\sresredRT&
  \labs{\varA}{\labs{\varB}{\labs{\svarC_1}{
    \appl{\varA}{
      \mset[\big]{
        \trProj{\varB}
      }\cons\emptystream
    }
  }}}
  \in\ExtTayExp{\tlId}
  \;.
  \end{eqnarray*}

  These examples are hints to the fact that
  \(\HeadTayExp{\tlJ}\vresred\ExtTayExp{\tlId}\),
  as we will establish below (\cref{example:NFJ}).
\end{example}

We show that the reduction of resource vectors
simulates both \(β\)-reduction and \(η\)-reduction, through
structural Taylor expansion.
The case of \(η\)-reduction is easy:
\begin{theorem}\label{theorem:simulation:eta}
  If $\ltermA\etared\ltermA'$ then
  $\ExtTayExp{\ltermA}\vresred\ExtTayExp{\ltermA'}$.
\end{theorem}
\begin{proof}
  By compatibility,
  it is sufficient to consider the base case:
  $\ltermA=\labs{\varA}{\appl{\ltermA'}{\varA}}$ with $\varA$ fresh.
  We have
  \(\ExtTayExp{\ltermA}
  =\labs{\varA}{\labs{\svarB}{\appl{\ExtTayExp{\ltermA'}}{\prom*{\IdExp{\varA}}\cons\IdStrExp{\svarB}}}}
  =\labs{\svarB}{\appl{\ExtTayExp{\ltermA'}}{\prom*{\IdExp{\svarB[0]}}\cons\shiftup*{\IdStrExp{\svarB}}{\svarB}}}
  =\labs{\svarB}{\appl{\ExtTayExp{\ltermA'}}{\IdStrExp{\svarB}}}
  \)
  by \cref{lemma:cc:shift}, and we conclude by \cref{lemma:cc:vectors}.
\end{proof}

To deal with \(β\)-reduction, we first need to consider the interplay between
structural Taylor expansion and substitution.
As we have announced, they do not commute on the nose, and the
analogue of \cref{eqn:substitution} (\cpageref{eqn:substitution})
is a reduction rather than an identity:
\begin{lemma}\label{lemma:exttayexp:subst}
  For all $λ$-terms $\ltermA$ and $\ltermB$,
  $\subst{\ExtTayExp{\ltermA}}{\varA}{\ExtTayExp{\ltermB}}
  \vresred\ExtTayExp{\subst{\ltermA}{\varA}{\ltermB}}$.
\end{lemma}
\begin{proof}
  The proof is by induction on $\ltermA$.

  If $\ltermA=\varA$ then 
  \(
    \subst{\ExtTayExp{\ltermA}}{\varA}{\ExtTayExp{\ltermB}}
    =\subst{\IdExp{\varA}}{\varA}{\ExtTayExp{\ltermB}}
    \vresred\ExtTayExp{\ltermB}
  \)
  by \cref{lemma:cc:vectors} and we conclude since
  $\subst{\ltermA}{\varA}{\ltermB}=\ltermB$.

  If $\ltermA=\varB\not=\varA$ then 
  \(
    \subst{\ExtTayExp{\ltermA}}{\varA}{\ExtTayExp{\ltermB}}
    =\subst{\IdExp{\varB}}{\varA}{\ExtTayExp{\ltermB}}
    =\IdExp{\varB}
    =\ExtTayExp{\varB}
  \)
  and we conclude since
  $\subst{\ltermA}{\varA}{\ltermB}=\varB$.

  If $\ltermA=\labs{\varC}{\ltermA'}$ (choosing \(\varC\not\in\LVarOf{\ltermB}\cup\set{\varA}\)) then 
  \begin{align*}
  \subst{\ExtTayExp{\ltermA}}{\varA}{\ExtTayExp{\ltermB}}
  &=\subst*{\labs{\varC}{\ExtTayExp{\ltermA'}}}{\varA}{\ExtTayExp{\ltermB}}
  \\
  &=\labs{\varC}{\subst{\ExtTayExp{\ltermA'}}{\varA}{\ExtTayExp{\ltermB}}}
  \\
  &\vresred\labs{\varC}{\ExtTayExp{\subst{\ltermA'}{\varA}{\ltermB}}}
  &&\text{by induction hypothesis}
  \\
  &=\ExtTayExp{\subst{\ltermA}{\varA}{\ltermB}}\;.
  \end{align*}

  If $\ltermA=\appl{\ltermA'}{\ltermA''}$ then
  \begin{align*}
  &\subst{\ExtTayExp{\ltermA}}{\varA}{\ExtTayExp{\ltermB}}
  \\
  &=\subst*{\labs{\svarB}{\appl{\ExtTayExp{\ltermA'}}{\prom*{\ExtTayExp{\ltermA''}}\cons\IdStrExp{\svarB}}}}{\varA}{\ExtTayExp{\ltermB}}
  \\
  &=\labs{\svarB}{\appl*{\subst{\ExtTayExp{\ltermA'}}{\varA}{\ExtTayExp{\ltermB}}}{\prom*{\subst{\ExtTayExp{\ltermA''}}{\varA}{\ExtTayExp{\ltermB}}}\cons\IdStrExp{\svarB}}}
  &&\text{by \cref{lemma:subst:prom}}
  \\
  &\vresred\labs{\svarB}{\appl*{\ExtTayExp{\subst{\ltermA'}{\varA}{\ltermB}}}{\prom{\ExtTayExp{\subst{\ltermA''}{\varA}{\ltermB}}}\cons\IdStrExp{\svarB}}}
  &&\text{by induction hypothesis}
  \\
  &=\ExtTayExp{\appl{\subst{\ltermA'}{\varA}{\ltermB}}{\subst{\ltermA''}{\varA}{\ltermB}}}\;.
  \\
  &=\ExtTayExp{\subst{\ltermA}{\varA}{\ltermB}}\;.
  && \qedhere
  \end{align*}
\end{proof}

\begin{theorem}\label{theorem:simulation:beta}
  If $\ltermA\bred\ltermA'$ then
  $\ExtTayExp{\ltermA}\vresredRT\ExtTayExp{\ltermA'}$.
\end{theorem}
\begin{proof}
  By compatibility,
  it is sufficient to consider the case of a redex,
  $\ltermA=\appl*{\labs{\varA}{\ltermB}}{\ltermC}$:
  \begin{align*}
    \ExtTayExp{\ltermA}
    =
    \labs{\svarB}{
      \appl*{\labs{\varA}{\ExtTayExp{\ltermB}}}
      {\prom{\ExtTayExp{\ltermC}}\cons\IdStrExp{\svarB}}
    }
    &\vresred
    \labs{\svarB}{
      \appl{\rsubst{\ExtTayExp{\ltermB}}{\varA}{\prom{\ExtTayExp{\ltermC}}}}
      {\IdStrExp{\svarB}}
    }
    \\
    &=
    \labs{\svarB}{
      \appl{\subst{\ExtTayExp{\ltermB}}{\varA}{\ExtTayExp{\ltermC}}}
      {\IdStrExp{\svarB}}
    }
    &&\text{by \cref{lemma:subst:taylor}}
    \\
    &\vresred
    \labs{\svarB}{
      \appl{\ExtTayExp{\subst{\ltermB}{\varA}{\ltermC}}}
      {\IdStrExp{\svarB}}
    }
    &&\text{by \cref{lemma:exttayexp:subst}}
    \\
    &\vresred\ExtTayExp{\subst{\ltermB}{\varA}{\ltermC}}
    &&\text{by \cref{lemma:cc:vectors}}\;.
    \qedhere
  \end{align*}
\end{proof}

Versions of \cref{theorem:simulation:eta,theorem:simulation:beta}
also hold for \(\HeadTayExp[\nodelim]{}\), although the proof is a bit more
contorted, due to the necessary inspection of the head structure of terms:
for that reason, we postpone this result to the end of the present
section (\cref{theorem:simulation:head}).
Instead, we proceed to showing that both flavours
of extensional Taylor expansion have the same normal forms.
We will rely on the following technical lemma, which 
will also be used repeatedly later on:
\begin{lemma}\label{lemma:taylor:head}
  For every $\ltermA\in\LambdaTerms$ and $\ltermsB\in\LambdaTerms^k$
  such that $\svarB\not\in\SVarOf{\ltermA}\cup\SVarOf{\ltermsB}$,
  $\labs{\svarB}{\appl{\HeadTayExp{\ltermA}}{\HeadArgsTayExp{\ltermsB}\seqcat\IdStrExp{\svarB}}}
  \vresredRT\HeadTayExp{\appl{\ltermA}{\ltermsB}}$.
\end{lemma}
\begin{proof}
  If $k=0$, then we conclude directly by \cref{lemma:cc:vectors}.
  If $\ltermA$ is an abstraction and $k>0$, then we apply the reflexivity of $\vresredRT$.

  If $\ltermA=\appl{\varC}{\ltermsC}$ then
  \(
    \HeadTayExp{\ltermA}
    =\labs{\svarB}{\appl{\varC}{\HeadArgsTayExp{\ltermsC}\seqcat\IdStrExp{\svarB}}}
  \)
  hence
  \begin{align*}
    \labs{\svarB}{\appl{\HeadTayExp{\ltermA}}{\HeadArgsTayExp{\ltermsB}\seqcat\IdStrExp{\svarB}}}
    &\vresred
    \labs{\svarB}{\rsubst*{\appl{\varC}{\HeadArgsTayExp{\ltermsC}\seqcat\IdStrExp{\svarB}}}{\svarB}{\HeadArgsTayExp{\ltermsB}\seqcat\IdStrExp{\svarB}}}
    \\
    &=
    \labs{\svarB}{\subst*{\appl{\varC}{\HeadArgsTayExp{\ltermsC}\seqcat\IdStrExp{\svarB}}}{\svarB}{\HeadArgsTayExp{\ltermsB}\seqcat\IdSeqExp{\svarB}}}
    &&\text{by \cref{lemma:subst:taylor:seq}}
    \\
    &=
    \labs{\svarB}{\appl{\varC}{\HeadArgsTayExp{\ltermsC}\seqcat\subst{\IdStrExp{\svarB}}{\svarB}{\HeadArgsTayExp{\ltermsB}\seqcat\IdSeqExp{\svarB}}}}
    \\
    &\vresred
    \labs{\svarB}{\appl{\varC}{\HeadArgsTayExp{\ltermsC}\seqcat\HeadArgsTayExp{\ltermsB}\seqcat\IdSeqExp{\svarB}}}
    &&\text{by \cref{lemma:cc:vectors}}
    \\
    &=
    \HeadTayExp{\appl{\varC}{\ltermsC\seqcat\ltermsB}}
    \;.
  \end{align*}

  If $\ltermA=\appl{\ltermA'}{\ltermsC}$ where $\ltermA'$ is an abstraction and
  $\LengthOf{\ltermsC}>0$, then 
  \(
    \HeadTayExp{\ltermA}
    =\labs{\svarB}{\appl{\HeadTayExp{\ltermsA'}}{\HeadArgsTayExp{\ltermsC}\seqcat\IdStrExp{\svarB}}}
  \)
  and we reason as in the previous case.
\end{proof}

\newpage
\begin{theorem}\label{theorem:exttohead}
  For every $λ$-term $\ltermA$,
  $\ExtTayExp{\ltermA}\vresredRT\HeadTayExp{\ltermA}$.
\end{theorem}
\begin{proof}
  We reason by induction on $\ltermA$.
  If $\ltermA=\varA\in\LVar$, we have already observed that 
  $\HeadTayExp{\ltermA}=\IdExp{\varA}=\ExtTayExp{\ltermA}$.
  If $\ltermA$ is an abstraction,
  we apply the induction hypothesis.
  Otherwise, $\ltermA=\appl{\ltermB}{\ltermC}$ so that:
  \begin{align*}
    \ExtTayExp{\ltermA}
    =
    \labs{\svarB}{\appl{\ExtTayExp{\ltermB}}{\prom{\ExtTayExp{\ltermC}}\cons\IdStrExp{\svarB}}}
    &\vresredRT
    \labs{\svarB}{\appl{\HeadTayExp{\ltermB}}{\prom{\HeadTayExp{\ltermC}}\cons\IdStrExp{\svarB}}}
    &&\qquad\text{by induction hypothesis}
    \\
    &\vresredRT
    \HeadTayExp{\appl{\ltermB}{\ltermC}}
    &&\qquad\text{by \cref{lemma:taylor:head}.}
    \qedhere
  \end{align*}
\end{proof}

\begin{corollary}\label{corollary:NFT}
  For any $λ$-term $\ltermA$, $\NF{\ExtTayExp{\ltermA}}=\NF{\HeadTayExp{\ltermA}}$.
  If moreover $\ltermA\eqBEta\ltermA'$
  then $\NF{\HeadTayExp{\ltermA}}=\NF{\HeadTayExp{\ltermA'}}$.
  In particular if $\ltermA$ is normalizable then
  $\NF{\HeadTayExp{\ltermA}}=\HeadTayExp{\NF{\ltermA}}$.
\end{corollary}

For each $\ltermA\in\LambdaTerms$, we write $\NFTayExp{\ltermA}$ for
$\NF{\ExtTayExp{\ltermA}}$:
by \cref{corollary:NFT}, we also have
$\NFTayExp{\ltermA}=\NF{\HeadTayExp{\ltermA}}$.
In our introduction, we announced that extensional Taylor expansion would
provide a practical alternative to Nakajima trees
\cites{DBLP:conf/lambda/Nakajima75},
\emph{i.e.}\ Böhm trees of \(λ\)-terms quotiented by infinitely
nested infinite \(η\)-expansions.
We will show in the next section that, indeed, the normalization of Taylor
expansion induces the same \(λ\)-theory \(\Hs\) as Nakajima trees, but here we
first demonstrate extensional Taylor expansion at work, by showing that
\(\NFTayExp{\tlJ}=\ExtTayExp{\tlId}\).
It is intuitively obvious that \(\tlJ\) converges to the result of applying
infinitely nested \(η\)-expansions to \(\tlId\), but the latter notion is not so
easy to describe formally nor to reason about;
on the other hand, we show that reasoning inductively on extensional
resource terms is sufficient to obtain the announced identity.

\begin{example}
  \label{example:NFJ}
  We prove \(\NFTayExp{\tlJ}=\ExtTayExp{\tlId}\).
  Recall that
  \(\ExtTayExp{\tlId}=\labs{\varA}{\IdExp{\varA}}\)
  and,
  by \cref{corollary:NFT},
  \(
  \NFTayExp{\tlJ}
  =\NFTayExp{\labs{\varA}{\appl{\tlJ}{\varA}}}
  =\labs{\varA}{\NFTayExp{\appl{\tlJ}{\varA}}}
  \):
  it will thus be sufficient to prove that
  \(\NFTayExp{\appl{\tlJ}{\varA}}=\IdExp{\varA}\)
  for any variable \(\varA\).
  Moreover recall from \cref{example:IJ} that 
  \(\appl{\tlJ}{\varA}\eqBeta\labs{\varB}{\appl{\varA}*{\appl{\tlJ}{\varB}}}\):
  writing
  \(\tvJ{\varA}\eqdef\NFTayExp{\appl{\tlJ}{\varA}}\)
  and
  \(\tvJBag{\varA}\eqdef\prom*{\tvJ{\varA}}\),
  \cref{corollary:NFT} again entails
  \(
  \tvJ{\varA}
  =\NFTayExp{\labs{\varB}{\appl{\varA}*{\appl{\tlJ}{\varB}}}}
  =\labs{\varB}{\labs{\svarC}{\appl{\varA}{\tvJBag{\varB}\cons\IdStrExp{\svarC}}}}
  =\labs{\svarC}{\appl{\varA}{\tvJBag{\svarC[0]}\cons\shiftup{\IdStrExp{\svarC}}}{\svarC}}
  \).
  We establish by induction on resource terms that:
  \begin{itemize}
    \item for any value term \(\vtermA\),
      \(\IdExpCoef{\varA}{\vtermA}=\Coef{\tvJ{\varA}}{\vtermA}\);
    \item for any bag term \(\bagA\),
      \(\IdBagExpCoef{\varA}{\bagA}=\Coef{\tvJBag{\varA}}{\bagA}\);
    \item for any stream term \(\strA\),
      \(\IdStrExpCoef{\svarA}{\strA}=
      \Coef*{\tvJBag{\svarA[0]}\cons\shiftup{\IdStrExp{\svarA}}{\svarA}}{\strA}\).
  \end{itemize}

  Consider a value term \(\vtermA=\labs{\svarC}{\appl{\varB}{\strB}}\).
  If \(\varB\not=\varA\), then 
  \(
  \Coef{\tvJ{\varA}}{\vtermA}
  =0
  =\IdExpCoef{\varA}{\vtermA}
  \).
  Otherwise, by induction hypothesis on \(\strB\),
  \(
  \Coef{\tvJ{\varA}}{\vtermA}
  =\Coef*{\tvJBag{\svarC[0]}\cons\shiftup{\IdStrExp{\svarC}}{\svarC}}{\strB}
  =\IdStrExpCoef{\svarC}{\strB}
  =\IdExpCoef{\varA}{\vtermA}
  \).
  Now consider a bag term \(\bagA=\mset{\vtermA_1,\dotsc,\vtermA_k}\):
  applying the induction to each \(\vtermA_i\), we obtain
  \(
  \Coef{\tvJBag{\varA}}{\bagA}
  =\frac{1}{k!}\prod_{i=1}^k \Coef{\tvJ{\varA}}{\vtermA_i}
  =\frac{1}{k!}\prod_{i=1}^k \IdExpCoef{\varA}{\vtermA_i}
  =\IdBagExpCoef{\varA}{\bagA}
  \).
  Finally consider a stream term \(\strA\).
  If \(\strA=\emptystream\) then 
  \(
  \Coef*{\tvJBag{\svarA[0]}\cons\shiftup{\IdStrExp{\svarA}}{\svarA}}{\strA}
  =
  \Coef{\tvJBag{\svarA[0]}}{\emptybag}
  \times
  \Coef{\shiftup{\IdStrExp{\svarA}}{\svarA}}{\emptystream}
  =1
  =\IdStrExpCoef{\svarA}{\strA}
  \).
  Otherwise, we can write \(\strA=\bagA\cons\strB\)
  and apply the induction hypothesis to \(\bagA\),
  to obtain:
  \(
  \Coef*{\tvJBag{\svarA[0]}\cons\shiftup{\IdStrExp{\svarA}}{\svarA}}{\strA}
  =
  \Coef{\tvJBag{\svarA[0]}}{\bagA}
  \times
  \Coef{\shiftup{\IdStrExp{\svarA}}{\svarA}}{\strB}
  =
  \IdBagExpCoef{\svarA[0]}{\bagA}
  \times
  \Coef{\shiftup{\IdStrExp{\svarA}}{\svarA}}{\strB}
  = 
  \IdStrExpCoef{\svarA}{\strA}
  \)
  by \cref{lemma:cc:shift}.
\end{example}

We conclude this section by showing that head Taylor expansion,
like structural Taylor expansion,
sends \(βη\)-reduction steps to resource reductions on value
vectors.
This result can be taken as an indication of the robustness of the notion,
but is not actually needed: \cref{corollary:NFT}
is sufficient to ensure that \(\NFTayExp{\ltermA}\)
depends only on the \(βη\)-class of \(\ltermA\);
and the important property of \(\HeadTayExp[\nodelim]{}\)
is rather its commutation with head reduction, 
that we establish and exploit in the next section.
The reader may thus safely skip \cref{theorem:simulation:head}
below.
On the other hand, the following variant of \cref{lemma:exttayexp:subst} for
\(\HeadTayExp[\nodelim]{}\) will be useful:
\begin{lemma}\label{lemma:headtayexp:subst}
  For all $λ$-terms $\ltermA$ and $\ltermB$,
  $\subst{\HeadTayExp{\ltermA}}{\varA}{\HeadTayExp{\ltermB}}
  \vresredRT\HeadTayExp{\subst{\ltermA}{\varA}{\ltermB}}$.
\end{lemma}
\begin{proof}
  The proof is by induction on $\ltermA$.
  The case of $\labs{\varA}{\ltermB}$ is settled by applying the induction
  hypothesis to $\ltermB$ as in the proof of \cref{lemma:exttayexp:subst}.

  If $\ltermA=\appl{\ltermA'}{\ltermsC}$ where $\ltermA'$ is an abstraction
  and $\LengthOf{\ltermsC}>0$ then
  \[
    \subst{\HeadTayExp{\ltermA}}{\varA}{\HeadTayExp{\ltermB}}
    =\labs{\svarB}{
      \appl*[\big]{\subst{\HeadTayExp{\ltermA'}}{\varA}{\HeadTayExp{\ltermB}}}
      *[\big]{\subst{\HeadArgsTayExp{\ltermsC}}{\varA}{\HeadTayExp{\ltermB}}\seqcat\IdStrExp{\svarB}}
    }
  \]
  where, by \cref{lemma:subst:prom:seq},
  \[
    \subst{\HeadArgsTayExp{\ltermsC}}{\varA}{\HeadTayExp{\ltermB}}\seqcat\IdStrExp{\svarB}
    =\prom*{\subst{\HeadTayExp{\ltermC_0}}{\varA}{\HeadTayExp{\ltermB}}}
    \cons\cdots\cons
    \prom*{\subst{\HeadTayExp{\ltermC_k}}{\varA}{\HeadTayExp{\ltermB}}}
    \cons\IdStrExp{\svarB}
  \]
  if $\ltermsC=\tuple{\ltermC_0,\dotsc,\ltermC_k}$.
  Applying the induction hypothesis to $\ltermA'$ and to each $\ltermC_i$,
  we obtain
  \[
    \subst{\HeadTayExp{\ltermA}}{\varA}{\HeadTayExp{\ltermB}}
    \vresredRT\labs{\svarB}{
      \appl*[\big]{\HeadTayExp{\subst{\ltermA'}{\varA}{\ltermB}}}
      *[\big]{\HeadArgsTayExp{\subst{\ltermsC}{\varA}{\ltermB}}\seqcat\IdStrExp{\svarB}}
    }
  \]
  and conclude, observing that $\subst{\ltermA'}{\varA}{\ltermB}$ is an abstraction
  and $\LengthOf{\subst{\ltermsC}{\varA}{\ltermB}}>0$.

  If $\ltermA=\appl{\varC}{\ltermsC}$ with $\varC\not=\varA$ then
  \[
    \subst{\HeadTayExp{\ltermA}}{\varA}{\HeadTayExp{\ltermB}}
    =\labs{\svarB}{
      \appl{\varC}
      *[\big]{\subst{\HeadArgsTayExp{\ltermsC}}{\varA}{\HeadTayExp{\ltermB}}\seqcat\IdStrExp{\svarB}}
    }
  \]
  and we obtain
  \(
    \subst{\HeadTayExp{\ltermA}}{\varA}{\HeadTayExp{\ltermB}}
    \vresredRT\labs{\svarB}{
      \appl{\varC}
      *[\big]{\HeadArgsTayExp{\subst{\ltermsC}{\varA}{\ltermB}}\seqcat\IdStrExp{\svarB}}
    }
  \)
  as in the previous case.

  If $\ltermA=\appl{\varA}{\ltermsC}$ then
  \[
    \subst{\HeadTayExp{\ltermA}}{\varA}{\HeadTayExp{\ltermB}}
    =\labs{\svarB}{
      \appl{\HeadTayExp{\ltermB}}
      *[\big]{\subst{\HeadArgsTayExp{\ltermsC}}{\varA}{\HeadTayExp{\ltermB}}\seqcat\IdStrExp{\svarB}}
    }
  \]
  and again
  \(
    \subst{\HeadTayExp{\ltermA}}{\varA}{\HeadTayExp{\ltermB}}
    \vresredRT\labs{\svarB}{
      \appl{\HeadTayExp{\ltermB}}
      *[\big]{\HeadArgsTayExp{\subst{\ltermsC}{\varA}{\ltermB}}\seqcat\IdStrExp{\svarB}}
    }
  \).
  We conclude by \cref{lemma:taylor:head}.
\end{proof}

\begin{theorem}\label{theorem:simulation:head}
  If $\ltermA\betaetared\ltermA'$ then
  $\HeadTayExp{\ltermA}\vresred\HeadTayExp{\ltermA'}$.
\end{theorem}
\begin{proof}
  The proof is by induction on \(\ltermA\).

  We first treat the case of an abstraction \(\ltermA=\labs{\varA}{\ltermB}\),
  necessarily with \(\varA\not\in\LVarOf{\ltermA'}\).
  If moreover \(\ltermA'=\labs{\varA}{\ltermB'}\)
  with \(\ltermB\betaetared\ltermB'\),
  then we can apply the induction hypothesis directly,
  using the compatibility of \(\vresred\).
  Otherwise, we must have \(\ltermB=\appl{\ltermA'}{\varA}\)
  so that \(\ltermA\etared\ltermA'\), 
  and \(\ltermA'\) is not an abstraction
  (if it is, we can write
  \(\ltermA'=\labs{\varA}{\ltermB'}\) and have
  \(\ltermB\bred\ltermB'\)).
  In that case, we write
  \(\ltermA'=\appl{\ltermA_0}{\ltermsC}\)
  where \(\ltermA_0\) is a variable or an abstraction
  (with \(\LengthOf{\ltermsC}>0\) in the latter case).
  Writing \(\hexprvA=\ltermA_0\) if 
  \(\ltermA_0\) is a variable,
  and \(\hexprvA=\HeadTayExp{\ltermA_0}\) otherwise,
  we obtain
  \(
  \HeadTayExp{\ltermA}
  =\labs{\varA}{
    \labs{\svarB}{
      \appl{\hexprvA}{
        \HeadArgsTayExp{\ltermsC}
        \seqcat\IdBagExp{\varA}
        \cons\IdStrExp{\svarB}
      }
    }
  }
  = \labs{\svarB}{
    \appl{\hexprvA}{
      \HeadArgsTayExp{\ltermsC}
      \cons\IdStrExp{\svarB}
    }
  }
  = \HeadTayExp{\ltermA'}
  \).

  If \(\ltermA=\appl{\varA}{\ltermsC}\),
  then the reduction
  \(\ltermA\betaetared\ltermA'\)
  occurs necessarily in an element of \(\ltermsC\)
  and we apply the induction hypothesis to that element
  and the compatibility of \(\vresred\).

  We finally treat the case of
  \(\ltermA=\appl*{\labs{\varA}{\ltermA_0}}*{\ltermB\cons\ltermsC}\).
  If the reduction occurs in \(\ltermA_0\), in \(\ltermB\),
  or in an element of \(\ltermsC\), then again
  we apply the induction hypothesis
  and the compatibility of \(\vresred\),
  the head structure of \(\ltermA'\) being the same.
  If \(\ltermA'=\appl*{\subst{\ltermA_0}{\varA}{\ltermB}}{\ltermsC}\)
  so that \(\ltermA\bred\ltermA'\),
  then
  \begin{align*}
    \HeadTayExp{\ltermA}
    &=\labs{\svarB}{
      \appl*{\labs{\varA}{\HeadTayExp{\ltermA_0}}}
      {\HeadArgTayExp{\ltermB}\cons\HeadArgsTayExp{\ltermsC}\seqcat\IdStrExp{\svarB}}
    }
    \\
    &\vresred
    \labs{\svarB}{
      \appl*{\subst{\HeadTayExp{\ltermA_0}}{\varA}{\HeadTayExp{\ltermB}}}
      {\HeadArgsTayExp{\ltermsC}\seqcat\IdStrExp{\svarB}}
    }
    &&\text{by \cref{lemma:subst:taylor}}
    \\
    &\vresred
    \labs{\svarB}{
      \appl{\HeadTayExp{\subst{\ltermA_0}{\varA}{\ltermB}}}
      {\HeadArgsTayExp{\ltermsC}\seqcat\IdStrExp{\svarB}}
    }
    &&\text{by \cref{lemma:headtayexp:subst}}
    \\
    &\vresred
    \HeadTayExp{\ltermA'}
    &&\text{by \cref{lemma:taylor:head}.}
  \end{align*}

  Finally, if \(\ltermA_0=\appl{\ltermA'_0}{\varA}\)
  with \(\varA\not\in\LVarOf{\ltermA'_0}\)
  and \(\ltermA'=\appl{\ltermA'_0}*{\ltermB\cons\ltermsC}\)
  so that \(\ltermA\etared\ltermA'\),
  then observe that we also have
  \(\ltermA'=\appl*{\subst{\ltermA_0}{\varA}{\ltermB}}{\ltermsC}\),
  and we are back to the previous case.
\end{proof}

\section{Characterization of \texorpdfstring{$\Hs$}{H*}}\label{section:Hs}

In this section, we study the equational theory induced by the normalization
of extensional Taylor expansion.
Writing $\ltermA\eqExtTay\ltermA'$ if $\NFTayExp{\ltermA}=\NFTayExp{\ltermA'}$,
we will show that \(\eqExtTay\) is the maximum consistent and sensible
\(λ\)-theory $\Hs$ (\cref{theorem:eqTay:is:eqHs}).
The latter is known to be characterized by Nakajima trees
\cites{DBLP:conf/lambda/Nakajima75}[Exercise 19.4.4]{DBLP:books/daglib/0067558},
but our results demonstrate how extensional Taylor expansion
allows us to dispense with such intrinsically infinite objects,
in the same fashion as ordinary Taylor expansion provides an alternative
to Böhm trees, based on finite syntactic approximants.\footnote{
  A game semanticist reader could rightfully argue that classical game
  semantics already provide an alternative to Nakajima trees, based
  on finite approximants.
  A key difference is that (ordinary as well as extensional)
  resource terms are not necessarily normal:
  they provide a language of finite, syntactic approximants of
  \(λ\)-terms \emph{and} of their (potentially infinite) normal forms.
}

\begin{lemma}\label{lemma:eqTay:extensional}
  The equivalence relation $\eqExtTay$ on $\LambdaTerms$ is an extensional $λ$-theory.
\end{lemma}
\begin{proof}
  That \(\eqExtTay\) is a congruence follows directly from the compatibility of
  \(\vresred\).
  Moreover, by \cref{corollary:NFT}, $\eqExtTay$ contains $\bred$ and $\etared$.
\end{proof}

\subsection{Sensibility}

We now show that, like ordinary Taylor expansion, extensional Taylor
expansion allows us to characterize the head normalizability of $λ$-terms.

\begin{theorem}\label{theorem:unsolvable}
  A \(λ\)-term is head normalizable iff $\NFTayExp{\ltermA}\not=0$.
\end{theorem}
As in the proof of \cref{lemma:taylor:solvable}, the forward implication is easy:
\begin{lemma}\label{lemma:hn:NFTn0}
  If $\ltermA$ is $βη$-equivalent to a head normal form
  then $\NFTayExp{\ltermA}\not=0$.
\end{lemma}
\begin{proof}
  If $\ltermA\eqBEta\ltermA'=\labs{\svarA}{\appl{\varC}{\ltermsC}}$,
  then $\vtermA\eqdef\labs{\svarA}{\labs{\svarB}{\appl{\varC}{\emptystream}}}\in\HeadTayExp{\ltermA'}$,
  hence $\vtermA\in\NFTayExp{\ltermA}$.
\end{proof}

For the other implication, as for \cref{lemma:taylor:solvable},
we apply the head reduction strategy for resource reduction until we
reach a head normal form of some term in the image of Taylor expansion;
and we show that this reduction path can be reflected as a head normalization
sequence for the source \(λ\)-term.
To ensure that head reduction terminates in the resource calculus,
we have to consider full-step reduction:
this will make the proof significantly more complex than that of 
\cref{lemma:taylor:solvable}, because such a full-step is not necessarily the
translation of a single head \(β\)-reduction step.

Moreover, it will be more practical to rely on \(\HeadTayExp[\nodelim]{}\)
instead of \(\ExtTayExp[\nodelim]{}\),
because the latter does not preserve the head structure of terms:
when \(\ltermA\) is in head normal form,
the elements of \(\ExtTayExp{\ltermA}\) need not be
(consider \(\ltermA=\appl{\varB}{\ltermB}\));
and if \(\ltermA\) has a head redex,
the head redex in an element of \(\ExtTayExp{\ltermA}\)
is not necessarily the image of that of \(\ltermA\)
(consider \(\ltermA=\appl*{\labs{\varA}{\ltermB}}{\ltermC_1\,\ltermC_2}\)).
This will require additional work to establish the properties 
of \(\HeadTayExp[\nodelim]{}\) w.r.t.\ reduction,
instead of relying on those of \(\ExtTayExp[\nodelim]{}\) via
\cref{corollary:NFT}, as was sufficient before.

Given a redex $\btermA\in\BaseTerms$, we write \(\HROf{\btermA}\) for the base
sum obtained by fully reducing it:
namely, if $\btermA=\appl*{\labs{\svarB}{\btermB}}{\strB}$ then
$\HROf{\btermA}\eqdef\rsubst{\btermB}{\svarB}{\strB}$.
We say that a value term $\vtermA=\labs{\svarA}{\btermA}$ is
\definitive{head reducible} if \(\btermA\) is a redex, and then set 
$\HROf{\vtermA}\eqdef\labs{\svarA}{\HROf{\btermA}}$
so that \(\vtermA\Resred\HROf{\vtermA}\);
otherwise we say that $\vtermA$ is in \definitive{head normal form}.
\definitive{Head reduction} is defined on sums and vectors by linearity:
\(\HROf{\vtermvA}\eqdef\sum_{\vtermA\in\vtermvA}\pars{\Coef{\vtermvA}{\vtermA}}\HROf{\vtermA}\)
for any $\vtermvA\in\ValueVectors$ whose support contains head reducible
terms only, so that \(\vtermvA\vresred\HROf{\vtermvA}\) in this case.
By \cref{lemma:rsubst:str:size}, if $\vtermA$
is head reducible and $\vtermA'\in\HROf{\vtermA}$
then $\SizeOf{\vtermA'}<\SizeOf{\vtermA}$:
the head reduction sequence of a value sum always terminates.

With these definitions in place, a \(λ\)-term \(\ltermA\) is head reducible
iff one value term in \(\HeadTayExp{\ltermA}\) is
(and then every term in \(\HeadTayExp{\ltermA}\) is head reducible).
However, as we have already stated, head reduction does not commute with head
Taylor expansion on the nose:
in general \(\HROf{\HeadTayExp{\ltermA}}\not=\HeadTayExp{\HOf{\ltermA}}\)
(recall that \(\HOf{\ltermA}\) denotes the head reduct of \(\ltermA\)).
Consider, for instance, \(\ltermA=\appl*{\labs{\varA_1}{\labs{\varA_2}{\varA_1}}}{\varB_1\,\varB_2}\):
it is easy to check that
\(\HROf{\HeadTayExp{\ltermA}}
= \labs{\svarC}{\appl{\IdExp{\varB_1}}{\IdStrExp{\svarC}}}\),
while \(\varB_2\) still occurs in \({\HOf{\ltermA}}\),
hence in \(\HeadTayExp{\HOf{\ltermA}}\);
on the other hand, 
\(\HROf{\HeadTayExp{\ltermA}}\vresred\IdExp{\varB_1}\) and
\(\HOf{\HOf{\ltermA}}=\varB_1\).
And we will indeed show that a head reduction step from the
expansion of a head redex can always be extended to the image of a head
reduction sequence (\cref{lemma:taylor:H}).

\begin{lemma}\label{lemma:HR}
  For any base vector of the form
  \(\btermvA=\appl*{\labs{\varA_1}{\cdots\labs{\varA_{k}}{\labs{\svarB}{\btermvB}}}}
  {\bagvB_1\cons\cdots\cons\bagvB_k\cons\strvC}\in\BaseVectors\)
  we have \(\HROf{\btermvA}=
  \rsubst{\rsubst{\rsubst{\btermvB}{\varA_1}{\bagvB_1}\cdots}{\varA_k}{\bagvB_k}}{\svarB}{\strvC}\).
\end{lemma}
\begin{proof}
  By linearity, it is sufficient to consider the case of 
  a base term:
  \[
    \btermvA=\btermA=
    \appl*{\labs{\varA_1}{\cdots\labs{\varA_{k}}{\labs{\svarB}{\btermB}}}}
    {\bagB_1\cons\cdots\cons\bagB_k\cons\strC}\in\BaseTerms
    \;.
  \]

  By definition, 
  \[
    \btermA=
    \appl*[\big]{\labs{\svarB}{
      \subst{\shiftup{
        \subst{\shiftup{\btermB}{\svarB}}{\varA_k}{\svarB[0]}\cdots
      }{\svarB}}{\varA_1}{\svarB[0]}
    }}{\bagB_1\cons\cdots\cons\bagB_k\cons\strC}
  \]
  hence
  \[
    \HROf{\btermA}=
    \rsubst[\big]{
      \subst{\shiftup{
        \subst{\shiftup{\btermB}{\svarB}}{\varA_k}{\svarB[0]}\cdots
      }{\svarB}}{\varA_1}{\svarB[0]}
    }{\svarB}{\bagB_1\cons\cdots\cons\bagB_k\cons\strC}
  \]
  and we conclude by iterating \cref{lemma:rsubst:str:def,lemma:shift:rsubst}.
\end{proof}

\begin{lemma}\label{lemma:taylor:H}
  For every head reducible $\ltermA\in\LambdaTerms$, there exists $k\in\N$ such that
  \(\HkOf{k}{\ltermA}\) is defined and
  $\HROf{\HeadTayExp{\ltermA}}\vresredRT\HeadTayExp{\HkOf{k}{\ltermA}}$.
\end{lemma}
\begin{proof}
  We reason by induction on $\ltermA$.
  If $\ltermA=\labs{\varB}{\ltermB}$ then we apply the induction hypothesis to
  $\ltermB$, and conclude by compatibility.
  Otherwise, $\ltermA=\appl{\ltermC}{\ltermsB}$ 
  where $\ltermC$ is an abstraction
  and $\ltermsB=\tuple{\ltermB_0,\dotsc,\ltermB_l}$
  is a non empty sequence of $λ$-terms.
  We write $\ltermC=\labs{\seq{\varA}}{\ltermC'}$
  where $\seq{\varA}=\tuple{\varA_0,\dotsc\varA_k}$ is a non empty tuple of
  pairwise distinct variables and $\ltermC'$ is not an abstraction:
  $\ltermC'=\appl{\ltermC''}{\ltermsB'}$ and, either 
  $\ltermC''=\varC$ is a variable, or
  $\ltermC''$ is an abstraction and $\LengthOf{\ltermsB'}>0$.
  Then we can write
  \(
    \HeadTayExp{\ltermC'}=\labs{\svarC}{\appl{\hexprvA}{\HeadArgsTayExp{\ltermsB'}\IdStrExp{\svarC}}}
  \)
  where $\hexprvA=\varC$ or $\hexprvA=\HeadTayExp{\ltermC''}$,
  and $\svarC$ is fresh.
  Choosing $\svarB$ fresh, we obtain:
  \begin{align*}
    \HeadTayExp{\ltermA}
    =\labs{\svarB}{
      \appl{\HeadTayExp{\ltermC}}{\HeadArgsTayExp{\ltermsB}\seqcat\IdStrExp{\svarB}}
    }
    =\labs{\svarB}{
      \appl*{
        \labs{\seq{\varA}}{\labs{\svarC}{\appl{\hexprvA}{\HeadArgsTayExp{\ltermsB'}\IdStrExp{\svarC}}}}
      }{\HeadArgsTayExp{\ltermsB}\seqcat\IdStrExp{\svarB}}
    }\;.
  \end{align*}

  If $k\le l$ then we write
  $\ltermsB''\eqdef\tuple{\ltermB_{0},\dotsc,\ltermB_{k}}$
  and
  $\ltermsB'''\eqdef\tuple{\ltermB_{k+1},\dotsc,\ltermB_{l}}$,
  and we obtain
  $\HkOf{k}{\ltermA}=\appl*{\subst{\ltermC'}{\seq{\varA}}{\ltermsB''}}{\ltermsB'''}$.
  In this case,
  we obtain
  \begin{align*}
    &\HROf{\HeadTayExp{\ltermA}}
    \\
    &=\labs{\svarB}{
      \rsubst{
        \rsubst*{
          \appl{\hexprvA}{\HeadArgsTayExp{\ltermsB'}
          \seqcat\IdStrExp{\svarC}}
        }{\seq{\varA}}{\HeadArgsTayExp{\ltermsB''}}
      }{\svarC}{\HeadArgsTayExp{\ltermsB'''}\seqcat\IdStrExp{\svarB}}
    }
    &&\text{by \cref{lemma:HR}}
    \\
    &=\labs{\svarB}{
      \rsubst{
        \subst*{
          \appl{\hexprvA}{
            \HeadArgsTayExp{\ltermsB'}\seqcat\IdStrExp{\svarC}
          }
        }{\seq{\varA}}{\HeadTayExp{\ltermsB''}}
      }{\svarC}{\HeadArgsTayExp{\ltermsB'''}\seqcat\IdStrExp{\svarB}}
    }
    &&\text{by \cref{lemma:subst:taylor}}
    \\
    &=\labs{\svarB}{
      \appl{
        \subst{\hexprvA}{\seq{\varA}}{\HeadTayExp{\ltermsB''}}
      }{
        \subst{\HeadArgsTayExp{\ltermsB'}}{\seq{\varA}}{\HeadTayExp{\ltermsB''}}
        \seqcat
        \rsubst{\IdStrExp{\svarC}}{\svarC}
        {\HeadArgsTayExp{\ltermsB'''}\seqcat\IdStrExp{\svarB}}
      }
    }
    &&\text{since $\svarC$ is fresh}
    \\
    &\vresred\labs{\svarB}{
      \appl{
        \subst{\hexprvA}{\seq{\varA}}{\HeadTayExp{\ltermsB''}}
      }{
        \subst{\HeadArgsTayExp{\ltermsB'}}{\seq{\varA}}{\HeadTayExp{\ltermsB''}}
        \seqcat
        \HeadArgsTayExp{\ltermsB'''}\seqcat\IdStrExp{\svarB}
      }
    }
    &&\text{by \cref{lemma:cc:vectors}}
    \\
    &\vresredRT\labs{\svarB}{
      \appl{
        \subst{\hexprvA}{\seq{\varA}}{\HeadTayExp{\ltermsB''}}
      }{
        \HeadArgsTayExp{\subst{\ltermsB'}{\seq{\varA}}{\ltermsB''}}
        \seqcat
        \HeadArgsTayExp{\ltermsB'''}\seqcat\IdStrExp{\svarB}
      }
    }
    &&\text{by 
      \cref{lemma:headtayexp:subst,lemma:subst:prom}
    }
    \\
    &=\labs{\svarB}{
      \appl{
        \subst{\hexprvA}{\seq{\varA}}{\HeadTayExp{\ltermsB''}}
      }{
        \HeadArgsTayExp{
          \subst{\ltermsB'}{\seq{\varA}}{\ltermsB''}\seqcat\ltermsB'''
        }\seqcat\IdStrExp{\svarB}
      }
    }
    \;.
  \end{align*}

  If $\ltermC'=\appl{\varC}{\ltermsB'}$ with $\varC\not\in\seq{\varA}$
  then $\subst{\hexprvA}{\seq{\varA}}{\HeadTayExp{\ltermsB''}}=\varC$
  and we conclude since we have
  $\HkOf{k}{\ltermA}
  =\appl{\varC}{\pars{\subst{\ltermsB'}{\seq{\varA}}{\ltermsB''}}\seqcat\ltermsB'''}$.
  If $\ltermC'=\appl{\varA_i}{\ltermsB'}$
  then $\subst{\hexprvA}{\seq{\varA}}{\HeadTayExp{\ltermsB''}}=\HeadTayExp{\ltermB_i}$
  and we conclude by \cref{lemma:taylor:head} since
  $\HkOf{k}{\ltermA}
  =\appl{\ltermB_i}{\pars{\subst{\ltermsB'}{\seq{\varA}}{\ltermsB''}}\seqcat\ltermsB'''}$.
  And if $\ltermC'=\appl{\ltermC''}{\ltermsB'}$ with $\ltermC''$ an abstraction
  and $\LengthOf{\ltermsB'}>0$,
  then $\hexprvA=\HeadTayExp{\ltermC''}$ hence
  $\subst{\hexprvA}{\seq{\varA}}{\HeadTayExp{\ltermsB''}}
  \vresredRT\HeadTayExp{\subst{\ltermC''}{\seq{\varA}}{\ltermsB''}}$
  by \cref{lemma:headtayexp:subst},
  hence
  \[
    \HROf{\HeadTayExp{\ltermA}}
    \vresredRT
    \labs{\svarB}{
      \appl{
        \HeadTayExp{\subst{\ltermC''}{\seq{\varA}}{\ltermsB''}}
      }{
        \HeadArgsTayExp{
          \subst{\ltermsB'}{\seq{\varA}}{\ltermsB''}\seqcat\ltermsB'''
        }\seqcat\IdStrExp{\svarB}
      }
    }
  \]
  and we conclude again by \cref{lemma:taylor:head} since
  \(
  \HkOf{k}{\ltermA} =\appl*[\big]{\subst{\ltermsC''}{\seq{\varA}}{\ltermsB''}}{\pars[\big]{\subst{\ltermsB'}{\seq{\varA}}{\ltermsB''}}\seqcat\ltermsB'''}
  \).

  Now if $k> l$, then we write $\seq{\varA}'\eqdef\tuple{\varA_0,\dotsc,\varA_l}$
  and $\seq{\varA}''\eqdef\tuple{\varA_{l+1},\dotsc,\varA_k}$
  so that 
  $\HkOf{l}{\ltermA}=\labs{\seq{\varA}''}{\subst{\ltermC'}{\seq{\varA}'}{\ltermsB}}$.
  We also write \(k'\eqdef k-l-1\) and 
  $\tuple{\varA'_{0},\dotsc,\varA'_{k'}}\eqdef\seq{\varA}''$.
  Then we compute:
  \begin{align*}
    &\HROf{\HeadTayExp{\ltermA}}
    \\
    &=\labs{\svarB}{
      \rsubst{
        \rsubst{
          \rsubst{
            \rsubst*{
              \appl{\hexprvA}{
                \HeadArgsTayExp{\ltermsB'}
                \seqcat\IdStrExp{\svarC}
              }
            }{\seq{\varA}'}{\HeadArgsTayExp{\ltermsB}}
          }{\varA'_{0}}{\IdBagExp{\svarB[0]}}
          \cdots
        }{\varA'_{k'}}{\IdBagExp{\svarB[k']}}
      }{\svarC}{\shiftup{\IdStrExp{\svarB}}{\svarB}[k'+1]}
    }
    &&\text{\small by \cref{lemma:HR,lemma:cc:shift}}
    \\
    &=\labs{\svarB}{
      \rsubst{
        \rsubst{
          \rsubst{
            \rsubst*{
              \appl{\hexprvA}{
                \HeadArgsTayExp{\ltermsB'}
                \seqcat\IdStrExp{\svarC}
              }
            }{\seq{\varA}'}{\HeadArgsTayExp{\ltermsB}}
          }{\svarC}{\shiftup{\IdStrExp{\svarB}}{\svarB}[k'+1]}
        }{\varA'_{k'}}{\IdBagExp{\svarB[k']}}
        \cdots
      }{\varA'_{0}}{\IdBagExp{\svarB[0]}}
    }
    &&\text{\small since \(\svarB\) and \(\svarC\) were fresh}
    \\
    &=\labs{\svarB}{
      \rsubst{
        \rsubst{
          \shiftup{
            \rsubst{
              \rsubst*{
                \appl{\hexprvA}{
                  \HeadArgsTayExp{\ltermsB'}
                  \seqcat\IdStrExp{\svarC}
                }
              }{\seq{\varA}'}{\HeadArgsTayExp{\ltermsB}}
            }{\svarC}{\IdStrExp{\svarB}}
          }{\svarB}[k'+1]
        }{\varA'_{k'}}{\IdBagExp{\svarB[k']}}
        \cdots
      }{\varA'_{0}}{\IdBagExp{\svarB[0]}}
    }
    &&\text{\small by \cref{lemma:shift:rsubst}}
    \\
    &=\labs{\svarB}{
      \rsubst{
        \rsubst{
          \shiftup{
            \rsubst*{
              \appl{\hexprvA}{
                \HeadArgsTayExp{\ltermsB'}
                \seqcat
                \rsubst{\IdStrExp{\svarC}}{\svarC}{\IdStrExp{\svarB}}
              }
            }{\seq{\varA}'}{\HeadArgsTayExp{\ltermsB}}
          }{\svarB}[k'+1]
        }{\varA'_{k'}}{\IdBagExp{\svarB[k']}}
        \cdots
      }{\varA'_{0}}{\IdBagExp{\svarB[0]}}
    }
    &&\text{\small since \(\svarB\) and \(\svarC\) were fresh}
    \\
    &\vresred\labs{\svarB}{
      \rsubst{
        \rsubst{
          \shiftup{
            \rsubst*{
              \appl{\hexprvA}{
                \HeadArgsTayExp{\ltermsB'}
                \seqcat
                \IdStrExp{\svarB}
              }
            }{\seq{\varA}'}{\HeadArgsTayExp{\ltermsB}}
          }{\svarB}[k'+1]
        }{\varA'_{k'}}{\IdBagExp{\svarB[k']}}
        \cdots
      }{\varA'_{0}}{\IdBagExp{\svarB[0]}}
    }
    &&\text{\small by \cref{lemma:cc:vectors}}
    \\
    &=\labs{\svarB}{
      \subst{
        \subst{
          \shiftup{
            \subst*{
              \appl{\hexprvA}{
                \HeadArgsTayExp{\ltermsB'}
                \seqcat
                \IdStrExp{\svarB}
              }
            }{\seq{\varA}'}{\HeadTayExp{\ltermsB}}
          }{\svarB}[k'+1]
        }{\varA'_{k'}}{\IdExp{\svarB[k']}}
        \cdots
      }{\varA'_{0}}{\IdExp{\svarB[0]}}
    }
    &&\text{\small by \cref{lemma:subst:taylor}}
    \\
    &\vresredRT\labs{\svarB}{
      \subst{
        \subst{
          \shiftup{
            \subst*{
              \appl{\hexprvA}{
                \HeadArgsTayExp{\ltermsB'}
                \seqcat
                \IdStrExp{\svarB}
              }
            }{\seq{\varA}'}{\HeadTayExp{\ltermsB}}
          }{\svarB}[k'+1]
        }{\varA'_{k'}}{\svarB[k']}
        \cdots
      }{\varA'_{0}}{\svarB[0]}
    }
    &&\text{\small by \cref{lemma:cc:vectors}}
    \displaybreak[3]
    \\
    &=\lefteqn{\labs{\svarB}{
        \subst{
          \shiftup{
            \subst{
              \subst{
                \shiftup{
                  \subst*{
                    \appl{\hexprvA}{
                      \HeadArgsTayExp{\ltermsB'}
                      \seqcat
                      \IdStrExp{\svarB}
                    }
                  }{\seq{\varA}'}{\HeadTayExp{\ltermsB}}
                }{\svarB}[k']
              }{\varA'_{k'}}{\svarB[k'-1]}
              \cdots
            }{\varA'_{1}}{\svarB[0]}
          }{\svarB}
        }{\varA'_{0}}{\svarB[0]}
      }
    }
    \displaybreak[0]
    \\
    &
    &&\text{\small by \cref{lemma:shift:rsubst}}
    \displaybreak[3]
    \\
    &=\labs{\varA'_{0}}{
      \labs{\svarB}{
        \subst{
          \subst{
            \shiftup{
              \subst*{
                \appl{\hexprvA}{
                  \HeadArgsTayExp{\ltermsB'}
                  \seqcat
                  \IdStrExp{\svarB}
                }
              }{\seq{\varA}'}{\HeadTayExp{\ltermsB}}
            }{\svarB}[k']
          }{\varA'_{k'}}{\svarB[k'-1]}
          \cdots
        }{\varA'_{1}}{\svarB[0]}
      }
    }
    &&\text{\small by definition}
    \displaybreak[2]
    \\
    &\vresredRT\labs{\svarA''}{
      \labs{\svarB}{
        \subst*{
          \appl{\hexprvA}{
            \HeadArgsTayExp{\ltermsB'}
            \seqcat
            \IdStrExp{\svarB}
          }
        }{\seq{\varA}'}{\HeadTayExp{\ltermsB}}
      }
    }
    &&\text{\small iterating previous steps}
  \end{align*}

  It will thus be sufficient to establish 
  \(
    \labs{\svarA''}{
      \labs{\svarB}{
        \subst*{
          \appl{\hexprvA}{
            \HeadArgsTayExp{\ltermsB'}
            \seqcat
            \IdStrExp{\svarB}
          }
        }{\seq{\varA}'}{\HeadTayExp{\ltermsB}}
      }
    }
    \vresredRT
    \HeadTayExp{\subst{\ltermC'}{\seq{\varA}'}{\ltermsB}}
  \)
  which is done like in the case $k\le l$,
  by applying \cref{lemma:headtayexp:subst,lemma:subst:prom},
  and inspecting the shape of $\ltermC'$.
\end{proof}

We are now ready to complete the proof of \cref{theorem:unsolvable}:

\begin{lemma}\label{lemma:NFTn0:hn}
  If $\NFTayExp{\ltermA}\not=0$ then $\ltermA$ is head normalizable.
\end{lemma}
\begin{proof}
  If  $\NFTayExp{\ltermA}\not=0$ then there exists $\vtermA\in\HeadTayExp{\ltermA}$
  such that $\NF{\vtermA}\not=0$.
  The proof is by induction on $\SizeOf{\vtermA}$.
  If $\ltermA$ is in head normal form then we conclude directly.
  Otherwise, $\NF{\HROf{\vtermA}}\not=0$, and
  we can pick $\vtermA'\in\HROf{\vtermA}$ such that $\NF{\vtermA'}\not=0$.
  In particular, $\SizeOf{\vtermA'}<\SizeOf{\vtermA}$.
  By \cref{lemma:taylor:H},
  we obtain $\vtermA'\sresredRT\vtermsA''$
  with \(\support{\vtermsA''}\subseteq\support{\ExtTayExp{\HkOf{k}{\ltermA}}}\).
  Again, $\NF{\vtermsA''}=\NF{\vtermA'}\not=0$ and we can pick 
  $\vtermA''\in\vtermsA''$ such that $\NF{\vtermA''}\not=0$.
  We also have $\vtermA''\in\ExtTayExp{\HkOf{k}{\ltermA}}$ and, 
  moreover, $\SizeOf{\vtermA''}\le\SizeOf{\vtermA'}<\SizeOf{\vtermA}$
  so the induction hypothesis applies:
  $\HkOf{k}{\ltermA}$ is head normalizable,
  hence $\ltermA$ is head normalizable too.
\end{proof}

\begin{corollary}\label{corollary:sensible}
  The $λ$-theory $\eqExtTay$ is sensible.
\end{corollary}
\begin{proof}
  By \cref{lemma:NFTn0:hn}, $\NFTayExp{\ltermA}=0$
  as soon as $\ltermA$ is not head normalizable.
\end{proof}

We conclude the present subsection by showing that the characterization of
coefficients in the extensional Taylor expansion of a term
(\cref{lemma:taylor:coef}) remains valid after normalization.
Our proof relies on \cref{lemma:NFTn0:hn}.

\begin{lemma}\label{lemma:NFT:coef}
  If $\vtermA\in\NFTayExp{\ltermA}$
  then $\NFTayExpCoef{\ltermA}{\vtermA}=\frac{1}{\muldeg{\vtermA}}$.
\end{lemma}
\begin{proof}
  The proof is by induction on \(\SizeOf{\vtermA}\).
  Since $\vtermA\in\NFTayExp{\ltermA}$, \cref{lemma:NFTn0:hn}
  ensures that \(\ltermA\) is head normalizable:
  \(\ltermA\eqBeta\ltermA'\) with
  \(\ltermA'=\labs{\varA_1}{\cdots{\labs{\varA_k}{\appl{\varB}{\ltermsB}}}}\),
  and \(\ltermsB=\tuple{\ltermB_0,\dotsc,\ltermB_{l-1}}\).
  By \cref{corollary:NFT}, \(\NFTayExp{\ltermA}=\NFTayExp{\ltermA'}\).
  Moreover,
  \(\HeadTayExp{\ltermA'}
  = \labs{\varA_1}{\cdots{\labs{\varA_k}{\labs{\svarC}{\appl{\varB}{\HeadArgTayExp{\ltermsB}\seqcat\IdStrExp{\svarC}}}}}}\),
  where \(\svarC\) is chosen fresh.
  Then we can write \(\vtermA
  = \labs{\varA_1}{\cdots{\labs{\varA_k}{\labs{\svarC}{\appl{\varB}{\bagB_0\cons\cdots\cons\bagB_{l+l'-1}\cons\emptystream}}}}}\),
  where \(\bagB_i\in\prom{\NFTayExp{\ltermB_i}}\) for \(0\le i< l\)
  and \(\bagB_{l+j}\in\IdBagExp{\svarC[j]}\) for \(0\le j< l'\).
  For \(0\le i< l\),
  applying the induction hypothesis to the elements of \(\bagB_i\),
  and then \cref{lemma:prom:coef}, we obtain 
  \(\Coef{\prom{\NFTayExp{\ltermB_i}}}{\bagB_i}=\frac{1}{\muldeg{\bagB_i}}\);
  and for \(0\le j<l'\), we have
  \(\Coef{\IdBagExp{\svarC[j]}}{\bagB_{l+j}}=\frac{1}{\muldeg{\bagB_{l+j}}}\)
  by \cref{lemma:cc:coef}.
  The result follows by the definition of \(\muldeg{\vtermA}\).
\end{proof}

This ensures that the definition of \(\eqExtTay\) does not depend
on the choice of the semiring of coefficients.
\begin{corollary}\label{corollary:eqExtTay:support}
  We have \(\ltermA\eqExtTay\ltermA'\) iff
  \(\support{\NFTayExp{\ltermA}}=\support{\NFTayExp{\ltermA'}}\).
\end{corollary}

\subsection{Böhm-out via Taylor expansion}

We have established that $\eqExtTay$ is a sensible extensional $λ$-theory.
It is obviously consistent since, e.g., $\varA\not\eqExtTay\varB$
when $\varA\not=\varB$.

To establish that $\eqExtTay$ is indeed maximum among sensible consistent 
$λ$-theories, we show that it contains the observational 
equivalence induced by head normal forms:
writing $\HNFs$ for the set of head normalizable $λ$-terms,
$\ltermA$ and $\ltermB$ are \definitive{observationally equivalent},
and we write $\ltermA\eqOh\ltermB$ if, for each $λ$-term context
$\tctxA$, $\tctxA[\ltermA]\in\HNFs$ iff $\tctxA[\ltermB]\in\HNFs$.
It is easy to check that a sensible $λ$-theory identifying two
$\eqOh$-distinct terms is inconsistent~\cite[Lemma 16.2.4]{DBLP:books/daglib/0067558}.
It follows that all sensible consistent $λ$-theories are included in 
$\eqOh$, so it remains only to prove that $\eqExtTay$ contains~$\eqOh$.

Our objective is thus to show that if $\ltermA\not\eqExtTay\ltermB$
then there is a context $\tctxA$ such that 
one of $\tctxA[\ltermA]$ and $\tctxA[\ltermB]$ is head normalizable,
and the other one is not --- in this case we say that $\tctxA$
\definitive{separates} $\ltermA$ from $\ltermB$.
By \cref{corollary:eqExtTay:support}, the assumption
$\ltermA\not\eqExtTay\ltermB$ amounts to the existence of a normal
value term $\vtermA$ such that
$\vtermA\in\NFTayExp{\ltermA}\setminus\NFTayExp{\ltermB}$
--- or \emph{vice versa}.
We show that the standard Böhm-out technique to separate $βη$-distinct normal
$λ$-terms can be adapted to this setting, by reasoning on normal value terms instead:
most of what follows is standard material about the $λ$-calculus,
and only the final result relies on the properties of extensional Taylor expansion.

Following the Böhm-out technique, we will only use separating contexts corresponding to
\definitive{Böhm transformations}, which are generated by composing
the following basic transformations:
$\Bapp{\ltermB}:\ltermA\mapsto\appl{\ltermA}{\ltermB}$
for $\ltermB\in\LambdaTerms$
(corresponding to the context $\appl{[\,]}{\ltermB}$); and
$\Bsub{\varA}{\ltermB}:\ltermA\mapsto\subst{\ltermA}{\varA}{\ltermB}$
for $\ltermB\in\LambdaTerms$ and $\varA\in\LVar$
(corresponding to the context $\appl*{\labs{\varA}{[\,]}}{\ltermB}$).
We use a postfix notation for the application of Böhm transformations
and use the sequential order for their composition so that
$\transf{\ltermA}{\btrfA\btrfB}=\transf*{\transf{\ltermA}{\btrfA}}{\btrfB}$
for any Böhm transformations $\btrfA$ and $\btrfB$.
We may apply a Böhm transformation $\btrfB$ to a tuple of terms:
$\transf{\tuple{\ltermA_1,\dotsc,\ltermA_k}}{\btrfB}
\eqdef\tuple{\transf{\ltermA_1}{\btrfB},\dotsc,\transf{\ltermA_k}{\btrfB}}$.

We say that $\ltermA$ and $\ltermB$ are \definitive{strongly separable},
and write $\ltermA\ssep\ltermB$,
if there exists a Böhm transformation $\btrfA$ such that
$\transf{\ltermA}{\btrfA}\eqBeta\labs{\varA}{\labs{\varB}{\varA}}$ and
$\transf{\ltermB}{\btrfA}\eqBeta\labs{\varA}{\labs{\varB}{\varB}}$.
And we say that $\ltermA$ is \definitive{separable} from $\ltermB$,
and write $\ltermA\sep\ltermB$,
if there exists a Böhm transformation $\btrfA$ such that
$\transf{\ltermA}{\btrfA}\in\HNFs$ and $\transf{\ltermB}{\btrfA}\not\in\HNFs$.
Note that strong separability implies separability and is
symmetric.\footnote{
  Our terminology deviates from the one used by, e.g., Coppo
  \emph{et al.}~\cite{CoppoDCRR1978}:
  in their terminology, one would say that the set
  \(\set{\ltermA,\ltermB}\) is separable (resp.\ semi-separable)
  if \(\ltermA\ssep\ltermB\) (resp.\ \(\ltermA\sep\ltermB\)
  or \(\ltermB\sep\ltermA\)).
}

The following are direct consequences of the definitions
or basic exercises in $λ$-calculus.
\begin{fact}\label{fact:sep}
  We have $\ltermA\ssep\ltermB$ as soon as one of the following holds:
  \begin{itemize}
    \item $\ltermA=\appl{\varA}{\ltermsA}$
      and $\ltermB=\appl{\varB}{\ltermsB}$
      with $\varA\not=\varB$
      or $\LengthOf{\ltermsA}\not=\LengthOf{\ltermsB}$;
    \item $\ltermA=\appl{\varA}{\ltermsA}$
      and $\ltermB=\labs{\seq{\varB}}{\appl{\varB}{\ltermsB}}$
      with $\varB\in\seq{\varB}$;
    \item $\ltermA=\labs{\seq{\varB}}{\labs{\varA}{\appl{\varA}{\ltermsA}}}$
      and $\ltermB=\labs{\seq{\varC}}{\labs{\varA}{\appl{\varA}{\ltermsB}}}$
      with  $\LengthOf{\seq{\varB}}\not=\LengthOf{\seq{\varC}}$
      or $\LengthOf{\ltermsA}\not=\LengthOf{\ltermsB}$;
  \end{itemize}
  and we have $\ltermA\sep\ltermB$ as soon as one of the following holds:
  \begin{itemize}
    \item $\ltermA\in\HNFs$ and $\ltermB\not\in\HNFs$;
    \item $\appl{\ltermA}{\ltermC}\sep\appl{\ltermB}{\ltermC}$
      for some $\ltermC\in\LambdaTerms$;
    \item $\subst{\ltermA}{\varA}{\ltermC}\sep\subst{\ltermB}{\varA}{\ltermC}$
      for some $\varA\in\LVar$ and $\ltermC\in\LambdaTerms$;
    \item $\ltermA\eqBEta\ltermA'\sep\ltermB'\eqBEta\ltermB$;
    \item $\ltermA=\labs{\varA}{\ltermA'}$
      and $\ltermB=\labs{\varA}{\ltermB'}$
      with $\ltermA'\sep\ltermB'$.
  \end{itemize}
\end{fact}

For each $k\in\N$, we write
$\rotate{k}\eqdef
\labs{\varA_1}{\dotsc\labs{\varA_k}{\labs\varB}{\appl{\varB}{\varA_1\cdots\varA_k}}}
\in\LambdaTerms$.
For $l\in\N$, $\seq{\varA}=\tuple{\varA_1,\dotsc,\varA_l}\in\LVar^l$ and 
$\seq{k}=\tuple{k_1,\dotsc,k_l}\in\N^l$, we define the Böhm transformation 
$\Brot{\seq{\varA}}{\seq{k}}\eqdef
\Bsub{\varA_1}{\rotate{k_1}}\cdots\Bsub{\varA_l}{\rotate{k_l}}$,
which is a sequence of substitutions.
We say that $\ltermA$ is \definitive{Böhm-separable} from $\ltermB$,
and write $\ltermA\bsep\ltermB$, if,
for each $l\in\N$, each tuple $\seq{\varA}\in\LVar^l$ of pairwise distinct
variables and each tuple $\seq{k}\in\N^l$ of pairwise distinct and sufficiently
large integers, we have
$\transf{\ltermA}{\Brot{\seq{\varA}}{\seq{k}}}
\sep\transf{\ltermB}{\Brot{\seq{\varA}}{\seq{k}}}$.
Taking $l=0$, Böhm-separability implies separability.

The results of the following lemma are, again, standard material.
They adapt the setup used by Krivine for the strong separation of 
$η$-distinct $β$-normal forms \cite[Chapter 5]{krivine}.

\begin{lemma}\label{lemma:bsep}
  We have $\ltermA\bsep\ltermB$ as soon as one of the following holds:
  \begin{enumerate}
    \item
      \label{lemma:bsep:unsolvable}
      $\ltermA\in\HNFs$ and $\ltermB\not\in\HNFs$;
    \item 
      \label{lemma:bsep:betaeta}
      $\ltermA\eqBEta\ltermA'\bsep\ltermB'\eqBEta\ltermB$;
    \item 
      \label{lemma:bsep:lambda}
      $\ltermA=\labs{\varA}{\ltermA'}$ and $\ltermB=\labs{\varA}{\ltermB'}$
      with $\ltermA'\bsep\ltermB'$;
    \item
      \label{lemma:bsep:var}
      $\ltermA=\appl{\varA}{\ltermsA}$
      and $\ltermB=\appl{\varB}{\ltermsB}$
      with $\varA\not=\varB$
      or $\LengthOf{\ltermsA}\not=\LengthOf{\ltermsB}$;
    \item
      \label{lemma:bsep:arg}
      $\ltermA=\appl{\varA}{\ltermA_0\cdots\ltermA_k}$
      and $\ltermB=\appl{\varA}{\ltermB_0\cdots\ltermB_k}$
      with $\ltermA_i\bsep\ltermB_i$ for some $i\le k$.
  \end{enumerate}
\end{lemma}

We are now ready to establish the separation theorem.
The canonical structure of normal extensional resource terms
allows us to proceed by a simple induction, very similar
to the proof of strong separation on $β$-normal $λ$-terms
(compare, e.g., with \cite[Lemma 5.10]{krivine}).

\begin{theorem}[Separation]\label{theorem:separation}
  If $\ltermA\not\eqExtTay\ltermB$ then $\ltermA\sep\ltermB$ or $\ltermB\sep\ltermA$.
\end{theorem}
\begin{proof}
  We prove by induction on $\vtermA$ that,
  if $\vtermA\in\NFTayExp{\ltermA}\setminus\NFTayExp{\ltermB}$,
  then $\ltermA\bsep\ltermB$.
  In this case $\NFTayExp{\ltermA}\not=0$,
  hence $\ltermA\in\HNFs$ by \cref{lemma:NFTn0:hn}.
  If $\ltermB\not\in\HNFs$, we conclude directly by 
  \cref{lemma:bsep:unsolvable} of \cref{lemma:bsep}.

  Otherwise, since $\eqExtTay$ is a $λ$-theory, \cref{lemma:bsep:betaeta}
  allows us to $\beta$-reduce both $\ltermA$ and $\ltermB$
  to bring them into head normal form:
  $\ltermA=\labs{\seq{\varA}}{\appl{\varA}{\ltermsA}}$ and
  $\ltermB=\labs{\seq{\varA}'}{\appl{\varA'}{\ltermsB}}$,
  with $\ltermsA=\tuple{\ltermA_1,\dotsc,\ltermA_k}$ and 
  $\ltermsB=\tuple{\ltermB_1,\dotsc,\ltermB_{k'}}$.
  Then we can write
  $\vtermA=\labs{\seq{\varA}}{\labs{\svarB}{\appl{\varA}{\strA}}}$.
  And since $\eqExtTay$ is extensional, \cref{lemma:bsep:betaeta} again allows
  to $η$-expand inside $\ltermA$ and $\ltermB$.
  We can thus ensure that
  $\LengthOf{\seq{\varA}}=\LengthOf{\seq{\varA}'}$,
  and that $k$ is large enough to write
  $\strA=\bagA_1\cons\cdots\cons\bagA_k\cons\emptystream$.
  By $α$-equivalence, we can further assume $\seq{\varA}=\seq{\varA}'$.
  By iterating \cref{lemma:bsep:lambda}, 
  and observing that $\labs{\varC}{\vtermC}\in\NFTayExp{\labs{\varC}{\ltermC}}$
  iff $\vtermC\in\NFTayExp{\ltermC}$,
  we can assume that 
  $\ltermA=\appl{\varA}{\ltermsA}$,
  $\ltermB=\appl{\varA'}{\ltermsB}$
  and
  $\vtermA=\labs{\svarB}{\appl{\varA}{\strA}}$.

  If $\varA\not=\varA'$ or $\LengthOf{\ltermsA}\not=\LengthOf{\ltermsB}$,
  we conclude by \cref{lemma:bsep:var}.
  Otherwise, observe that
  \begin{align*}
    \NFTayExp{\ltermA}
    &=
    \labs{\svarB}{\appl{\varA}{
      \prom{\NFTayExp{\ltermA_1}}
      \cons\cdots\cons
      \prom{\NFTayExp{\ltermA_k}}
      \cons\IdStrExp{\svarB}
    }}
    \quad\text{and}
    \\
    \NFTayExp{\ltermB}
    &=
    \labs{\svarB}{\appl{\varA}{
      \prom{\NFTayExp{\ltermB_1}}
      \cons\cdots\cons
      \prom{\NFTayExp{\ltermB_k}}
      \cons\IdStrExp{\svarB}
    }}
    \;.
  \end{align*}
  Since $\emptystream\in\IdStrExp{\svarB}$ and $\vtermA\not\in\NFTayExp{\ltermB}$,
  there must be $i$ such that $\bagA_i\not\in\prom{\NFTayExp{\ltermB_i}}$.
  Hence there must be $\vtermA_i\in\bagA_i$ such that 
  $\vtermA_i\not\in\NFTayExp{\ltermB_i}$.
  Since moreover $\bagA_i\in\prom{\NFTayExp{\ltermA_i}}$,
  we have $\vtermA_i\in\NFTayExp{\ltermA_i}$.
  By applying the induction hypothesis to $\vtermA_i$, 
  we obtain $\ltermA_i\bsep\ltermB_i$,
  and conclude using \cref{lemma:bsep:arg}.
\end{proof}

We have thus established that, if $\ltermA\not\eqExtTay\ltermB$
then $\ltermA\not\eqOh\ltermB$, which is sufficient to ensure that 
$\eqExtTay$ contains every other consistent and sensible $λ$-theory.

\begin{theorem}\label{theorem:eqTay:is:eqHs}
  Given two $λ$-terms $\ltermA$ and $\ltermB$,
  we have $\ltermA\eqExtTay\ltermB$ iff $\ltermA\eqOh\ltermB$,
  and \(\eqExtTay\) is the maximum consistent and sensible $λ$-theory \(\Hs\).
\end{theorem}

In particular, \cref{example:NFJ} provides a new proof that \(\tlJ\eqOh\tlId\).

\section{A relational model}\label{section:rel}

By \cref{theorem:eqTay:is:eqHs},
to construct a model of $\Hs$, it is sufficient to give a model of the
reduction of resource vectors.
We exploit this approach to give a new proof of the fact
that $\Hs$ is the $λ$-theory induced by a particular 
extensional reflexive object in the relational model of the simply typed $λ$-calculus.

We define the set $\RelTypes$ of \definitive{relational types}
as $\RelTypes\eqdef\bigcup_{k\in\N}\RelTypes[k]$
with $\RelTypes[0]\eqdef\emptyset$ and $\RelTypes[k+1]\eqdef\SsOf{\RelTypes[k]}$.
In other words, the elements of $\RelTypes$ are generated from \(\emptystream\)
by iterating the construction
$\tuple{\rbagA,\rtypeB}\in\MfOf{\RelTypes}\times\RelTypes
\mapsto\rbagA\cons\rtypeB\in\RelTypes$,
subject to the identity $\mset{}\cons\emptystream=\emptystream$.
Note that $\RelTypes$ is nothing but the extensional reflexive object 
of the cartesian closed category $\MRel$ put forward by 
Bucciarelli, Ehrhard and Manzonetto \cite{DBLP:conf/csl/BucciarelliEM07}
--- as an example of the construction of an extensional $λ$-theory based on
reflexive object in a cartesian closed category having ``not enough points''. 
The fact that this $λ$-theory characterizes $\Hs$ was later proved
by Manzonetto \cite{DBLP:conf/mfcs/Manzonetto09}.

\subsection{Relational semantics of \pdflambda-terms as a type system}

Let us recall that the interpretation of $λ$-terms in this model 
can be described by a kind of non-idempotent intersection type system:
the “$\cons$” constructor acts as an arrow connective, while 
the monoid structure of multisets induces the non-idempotent intersection
operator.

More explicitly, we first define \definitive{relational typing contexts}
(denoted by greek capitals $\rctxA$, $\rctxB$, $\rctxC$)
as functions $\LVar\to\MfOf{\RelTypes}$ with finite support, \emph{i.e.}\ 
$\rctxA:\LVar\to\MfOf{\RelTypes}$ is a relational typing context if
$\set{\varA\in\LVar\st\rctxA(\varA)\not=\emptybag}$ is finite.
We write $\emptyrctx$ for the empty context: $\emptyrctx(\varA)=\emptybag$
for each $\varA\in\LVar$.
We write $\varA:\bagA$ for the context $\rctxA$ such that
$\rctxA(\varA)=\bagA$ and $\rctxA(\varB)=\emptybag$ when $\varA\not=\varB$.
And we define the concatenation of contexts point-wise:
$(\rctxA\rctxcat\rctxB)(\varA)\eqdef\rctxA(\varA)\bagcat\rctxB(\varA)$.
We also write $\rctxA-\varA$ for the context such that 
$(\rctxA-\varA)(\varA)=\emptybag$ and 
$(\rctxA-\varA)(\varB)=\rctxA(\varB)$ when $\varA\not=\varB$.

Similarly, if $\rstrA=\tuple{\rbagA_i}_{i\in\N}\in\SsOf{\RelTypes}$, we write 
$\svarA:\rstrA$ for the context $\rctxA$ such that 
$\rctxA(\svarA[i])=\rbagA_i$ for $i\in\N$,
and $\rctxA(\varB)=\emptybag$ if $\varB\not\in\svarA$.
We also write $\rctxA-\svarA$ for the context such that 
$(\rctxA-\svarA)(\svarA[i])=\emptybag$ and 
$(\rctxA-\svarA)(\varB)=\rctxA(\varB)$ when $\varB\not\in\svarA$.
Finally, $\rctxA(\svarA)$ denotes the sequence $\tuple{\rctxA(\svarA[i])}_{i\in\N}$,
which is a stream because $\rctxA$ has finite support.

The relational semantics $\sem{\ltermA}$ can then be computed as
the set \[\set{\tuple{\rctxA,\rtypeA}\st \rctxA\vdash\ltermA:\rtypeA}\]
where the type system is described by the following rules:
\begin{gather*}
  \begin{prooftree}
    \infer0{\varA:\mset{\rtypeA}\rjug \varA:\rtypeA}
  \end{prooftree}
  \qquad
  \begin{prooftree}
    \hypo{\rctxA\rjug \ltermA:\rtypeB}
    \infer1{\rctxA-\varA\rjug\labs{\varA}{\ltermA}:\rctxA(\varA)\cons\rtypeB}
  \end{prooftree}
  \\[1em]
  \begin{prooftree}
    \hypo{\rctxA\rjug \ltermA:\rbagA\cons \rtypeB}
    \hypo{\rctxB\rjugb \ltermB:\rbagA}
    \infer2{\rctxA\rctxcat\rctxB\rjug\appl{\ltermA}{\ltermB}:\rtypeB}
  \end{prooftree}
  \qquad
  \begin{prooftree}
    \hypo{\rctxA_1\rjug \ltermA:\rtypeA_1}
    \hypo{\cdots}
    \hypo{\rctxA_k\rjug \ltermA:\rtypeA_k}
    \infer3{\rctxA_1\rctxcat\cdots\rctxcat\rctxA_k\rjugb\ltermA:\mset{\rtypeA_1,\dotsc,\rtypeA_k}}
  \end{prooftree}
\end{gather*}
where the additional kind of judgement ($\rctxA\rjugb\ltermA:\rbagA$)
used in the application rule reflects the promotion rule of linear logic.

It is well known that this interpretation
is stable under $β$-reduction and $η$-expansion, 
and that it factors through (ordinary) Taylor expansion
(see, e.g., \cite[Proposition 5.9]{DBLP:journals/mscs/Manzonetto12}),
the semantics of resource terms being given by a natural variant
of the above rules:
\[
  \sem{\ltermA}=\sem{\TayExp{\ltermA}}=\sem{\NF{\TayExp{\ltermA}}}
  =\bigcup_{\vtermA\in\TayExp{\ltermA}} \sem{\vtermA}\;.
\]

We will now show a similar result for our extensional version of Taylor
expansion.
Since this semantics is set based, we only consider the qualitative 
version of extensional Taylor expansion, \emph{i.e.}\ we fix the set of 
coefficients to be the boolean semiring \(\Bool\).

A novel feature of our factorization is that, for extensional normal resource
terms, the relational semantics is always a singleton:
each normal resource term is mapped to a single point in the relational semantics.
Even more strikingly, the interpretation of a non-normal resource term
is not only a finite set, but has at most one element.
By contrast, in the ordinary resource calculus, normal terms
always have an infinite semantics:
for instance, the semantics of a single variable \(\varA\) is 
\(\set{\tuple{\varA:\mset{\rtypeA},\rtypeA}\st \rtypeA\in\RelTypes}\) (the same as if it was considered as a \(λ\)-term).
Also, the compatibility of the semantics with $η$-expansion is reflected
by the properties of Taylor expansion, instead of being an additional 
property of the model, to be checked by hand.

\subsection{Relational semantics of extensional resource terms}\label{section:rel:resource}

Each syntactic category of expressions will correspond to a specific category of types:
\begin{itemize}
  \item value expressions will be typed with elements of $\RelTypes$;
  \item bag terms with elements of $\MfOf{\RelTypes}$;
  \item stream terms with elements of $\SsOf{\RelTypes}$;
  \item base terms with a single, newly introduced base type $\basetype$.
\end{itemize}

The underlying idea is that a value expression of type $\rtypeA=\tuple{\rbagA_i}_{i\in\N}$ 
expects a sequence of bags $\strA=\tuple{\bagA_i}_{i\in\N}$ with
$\bagA_i$ of type $\rbagA_i$ to produce a successful interaction
(of type $\basetype$).
Although the types are incidentally taken from the same set $\RelTypes=\SsOf{\RelTypes}$,
the intuitive meaning of typing for value expressions and for stream terms
is thus quite different.
To fit this intuition, we redefine $\RelTypes$ simply to be able to distinguish
between $\RelTypes$ and $\SsOf{\RelTypes}$:
we believe this will clarify the presentation,
although the remainder of this section might be carried out identically
without this notational trick.

We define the sets $\ValueTypes$ of \definitive{value types}
and $\StreamTypes$ of \definitive{stream types},
simultaneously by mutual induction as follows:
\begin{itemize}
  \item $\rtypeA\in\ValueTypes$ if $\rtypeA=\dualtype{\rstrA}$ with $\rstrA\in\StreamTypes$;
  \item $\rstrA\in\StreamTypes$ if $\rstrA\in\SsOf{\ValueTypes}$;
\end{itemize}
so that $\rstrA\mapsto\dualtype{\rstrA}$ defines a bijection 
from $\StreamTypes=\SsOf{\ValueTypes}$ to $\ValueTypes$.
We will also write $\BagTypes\eqdef\MfOf{\ValueTypes}$ for the set of \definitive{bag types}
and $\BaseTypes\eqdef\set{\basetype}$ for the singleton containing the \definitive{base type}.
We call \definitive{type term} (denoted by $\rtermA,\rtermB,\rtermC$) any
of a value type, bag type, stream type or the base type.

A type context is now a function $\rctxA:\LVar\to\BagTypes$,
whose value is almost always the empty bag.
The type system involves four kinds of jugdements:
\[
  \rctxA\rjugVal\vtermA:\rtypeA
  \qquad
  \rctxA\rjugBag\bagA:\rbagA
  \qquad
  \rctxA\rjugStr\strA:\rstrA
  \qquad
  \rctxA\rjugBase\btermA:\basetype
\]
and we denote by $\rctxA\rjug\termA:\rtermA$ any judgement as above.
The rules are in \cref{figure:reltyping}.
\begin{figure}
\begin{gather*}
  \begin{prooftree}
    \hypo{\rctxA\rjugVal \vtermA:\dualtype{\rstrA}}
    \hypo{\rctxB\rjugStr \strB:\rstrA}
    \infer2{\rctxA\rctxcat\rctxB\rjugBase \appl{\vtermA}{\strB}:\basetype}
  \end{prooftree}
  \qquad
  \begin{prooftree}
    \hypo{\rctxA\rjugStr \strA:\rstrA}
    \infer1{\rctxA\rctxcat\varA:\mset{\dualtype{\rstrA}}\rjugBase \appl{\varA}{\strA}:\basetype}
  \end{prooftree}
  \\[1em]
  \begin{prooftree}
    \infer0{\emptyrctx \rjugStr \emptystream:\emptystream}
  \end{prooftree}
  \qquad
  \begin{prooftree}
    \hypo{\rctxA\rjugBag \bagA:\rbagA}
    \hypo{\rctxB\rjugStr \strB:\rstrB}
    \infer2{\rctxA\rctxcat\rctxB \rjugStr \pars{\bagA\cons\strB}:\pars{\rbagA\cons\rstrB}}
  \end{prooftree}
  \\[1em]
  \begin{prooftree}
    \hypo{\rctxA\rjugBase \btermA:\basetype}
    \infer1{\rctxA-\svarA\rjugVal \labs{\svarA}{\btermA}:\dualtype{\rctxA(\svarA)}}
  \end{prooftree}
  \qquad
  \begin{prooftree}
    \hypo{\rctxA_1\rjugVal \vtermA_1:\rtypeA_1}
    \hypo{\cdots}
    \hypo{\rctxA_k\rjugVal \vtermA_k:\rtypeA_k}
    \infer3{\rctxA_1\rctxcat\cdots\rctxcat\rctxA_k
    \rjugBag \mset{\vtermA_1,\dotsc,\vtermA_k}:\mset{\rtypeA_1,\dotsc,\rtypeA_k}}
  \end{prooftree}
\end{gather*}
\caption{Relational typing system for the extensional resource calculus}
\label{figure:reltyping}
\end{figure}
Note that, In particular, the last rule allows us to derive $\emptyrctx\rjugBag\emptybag:\emptybag$.

As we have announced, contrasting with the ordinary resource
calculus, the relational interpretation of an extensional resource
term will have at most one element:
\begin{lemma}\label{lemma:typing:unicity}
  Each resource term $\termA$ admits at most one derivable typing judgement.
  If moreover $\termA$ is normal, then it is typable.
\end{lemma}
\begin{proof}
  The proof is straightforward, by induction on $\termA$.
  Note that $\termA$ being normal forbids the first rule for base terms,
  which is the only one in which a constraint on the premises is imposed.
\end{proof}
We write $\typeable{\termA}$ when $\termA$ is typable,
and in this case we write $\ctxOf{\termA}$ and $\typeOf{\termA}$
respectively for the unique context and type term such that 
$\ctxOf{\termA}\rjug\termA:\typeOf{\termA}$ is derivable.
\begin{lemma}\label{lemma:typing:occurrences}
  If $\typeable{\termA}$ then $\LengthOf{\ctxOf{\termA}(\varA)}$
  is the number of occurrences of $\varA$ in $\termA$.
\end{lemma}
\begin{proof}
  By a straightforward induction on $\termA$.
\end{proof}

Note that the functions $\typeOf{-}$ and $\ctxOf{-}$ are not injective, even jointly:
\begin{example}
  \label{example:rel:noninjective}
  Consider
  \[
    \strA_1\eqdef
    \mset{\labs{\svarB}{\appl{\varA}{\emptystream}}}\cons
    \mset{\labs{\svarB}{\appl{\varA}{\mset{\vtermA}\cons\emptystream}}}\cons
    \emptystream
  \]
  and
  \[
    \strA_2\eqdef
    \mset{\labs{\svarB}{\appl{\varA}{\mset{\vtermA}\cons\emptystream}}}\cons
    \mset{\labs{\svarB}{\appl{\varA}{\emptystream}}}\cons
    \emptystream
  \]
  where $\vtermA$ is a typeable closed value and $\varA\not\in\svarB$,
  so that $\strA_1$ and $\strA_2$ differ only by the order of their first two bags.
  Then \[
    \typeOf{\strA_1}=\typeOf{\strA_2}
    =\mset{\dualtype{\emptystream}}\cons\mset{\dualtype{\emptystream}}\cons\emptystream
  \]
  because the variables of $\svarB$ have no occurrences in subterms, while
  \[
    \ctxOf{\strA_1}=\ctxOf{\strA_2}
    =\varA:\mset{\dualtype{\emptystream},\dualtype*{\mset{\typeOf{\vtermA}}\cons\emptystream}}
  \]
  because $\varA$ occurs in twice in each of $\strA_1$ and $\strA_2$,
  once applied to $\emptystream$,
  and then applied to $\mset{\vtermA}\cons\emptystream$.
  The relational semantics cannot \emph{``see''} that these occurrences are swapped.\footnote{
    This contrasts with the one-to-one correspondence between
    normal extensional resource terms and (isomorphism
    classes of) augmentations in game semantics,
    to be developed in \cref{section:gs}.
  }
\end{example}

Moreover, not all types are inhabited by a closed term.
This is easily observed on normal terms:
\begin{proposition}
  \label{proposition:rel:nondefinable}
  There is no normal term $\termA$ such that 
  \begin{itemize}
    \item $\termA=\btermA$ with $\emptyrctx\rjugBase\btermA:\basetype$, or
    \item $\termA=\vtermA$ with $\emptyrctx\rjugVal\vtermA:\dualtype{\emptystream}$, or
    \item $\termA=\bagA$ with $\emptyrctx\rjugBag\bagA:\mset{\dualtype{\emptystream}}$, or
    \item $\termA=\strA$ with $\emptyrctx\rjugVal\strA:\mset{\dualtype{\emptystream}}\cons\emptystream$.
  \end{itemize}
\end{proposition}
\begin{proof}
  Given the shape of the unique applicable rule for a normal base term,
  we have $\ctxOf{\btermA}\not=\emptyrctx$ when $\typeable{\btermA}$ 
  and $\btermA$ is normal.
  Each of the other three statements follows directly from the previous one.
\end{proof}
The same result holds for all terms,
as will follow from \cref{lemma:invariance},
which ensures the invariance of typing under reduction.
We first characterize typing in substitutions, in the next two lemmas.

\begin{lemma}\label{lemma:substitution:expansion}
  If $\termA'\in\rsubst{\termA}{\varA}{\bagB}$
  and $\typeable{\termA'}$ then 
  $\typeable{\termA}$, $\typeable{\bagB}$,
  $\typeOf{\bagB}=\ctxOf{\termA}(\varA)$,
  $\typeOf{\termA'}=\typeOf{\termA}$, and
  $\ctxOf{\termA'}=(\ctxOf{\termA}-\varA)\rctxcat\ctxOf{\bagB}$.
\end{lemma}
\begin{proof}
  The proof is by induction on $\termA$.

  If $\termA=\mset{\vtermA_1,\dotsc,\vtermA_k}$ then
  we must have $\termA'=\mset{\vtermA'_1,\dotsc,\vtermA'_k}$
  and $\bagB=\bagB_1\bagcat\cdots\bagcat\bagB_k$ with
  $\vtermA'_i\in\rsubst{\bagA}{\varA}{\bagB_i}$ for $1\le i\le k$.
  Then $\typeable{\vtermA'_i}$ and we apply the induction hypothesis
  to each $\vtermA_i$ for $1\le i\le k$.
  We obtain that
  $\typeable{\vtermA_i}$ and $\typeable{\bagB_i}$ for $1\le i\le k$:
  it follows that $\typeable{\termA}$ and $\typeable{\bagB}$.
  Moreover, $\typeOf{\bagB_i}=\ctxOf{\vtermA_i}(\varA)$ for $1\le i\le k$,
  so $\typeOf{\bagB}=\typeOf{\bagB_1}\bagcat\cdots\bagcat\typeOf{\bagB_k}
  =(\ctxOf{\vtermA_1}\rctxcat\cdots\rctxcat\ctxOf{\vtermA_k})(\varA)
  =\ctxOf{\termA}(\varA)$.
  We also have $\typeOf{\vtermA'_i}=\typeOf{\vtermA_i}$ for $1\le i\le k$,
  hence $\typeOf{\termA'}=\typeOf{\termA}$.
  Finally, we have 
  $\ctxOf{\vtermA'_i}=(\ctxOf{\vtermA_i}-\varA)\rctxcat\ctxOf{\bagB_i}$ for $1\le i\le k$,
  hence
  $\ctxOf{\termA'}= \ctxOf{\vtermA'_1}\rctxcat\cdots\rctxcat\ctxOf{\vtermA'_k}
  =(\ctxOf{\termA}-\varA)\rctxcat\ctxOf{\bagB}$,

  If $\termA=\emptystream$ then $\termA'=\emptystream$ and
  $\bagB=\emptybag$, which entails all the desired properties.
  The cases of $\termA=\bagA\cons\strC\not=\emptystream$ and 
  $\termA=\appl{\vtermA}{\strC}$ are similar to that of
  $\termA=\mset{\vtermA_1,\vtermA_2}$.
  The case of $\termA=\labs{\svarB}{\btermA}$ 
  (choosing $\svarB\not\ni\varA$ and $\svarB\cap\LVarOf{\bagB}=\emptyset$)
  is similar to that of $\termA=\mset{\vtermA_1}$.


  The only remaining case is that of $\termA=\appl{\varC}{\strA}$.
  If $\varC\not=\varA$, the treatment is, again, similar to that of
  $\termA=\mset{\vtermA_1}$.
  Now assume $\varC=\varA$.
  Then we can write $\termA'=\appl{\vtermB}{\strA'}$
  and $\bagB=\mset{\vtermB}\bagcat\bagB_1$
  with $\strA'\in\rsubst{\strA}{\varA}{\bagB_1}$.
  We have $\typeable{\vtermB}$ and $\typeable{\strA'}$,
  and moreover $\typeOf{\vtermB}=\dualtype{\typeOf{\strA'}}$.
  We apply the induction hypothesis to $\strA$.
  We obtain $\typeable{\strA}$ and $\typeable{\bagB_1}$:
  it follows that $\typeable{\termA}$ and $\typeable{\bagB}$.
  Moreover, $\typeOf{\strA'}=\typeOf{\strA}$ and
  $\typeOf{\bagB_1}=\ctxOf{\strA}(\varA)$.
  So $\typeOf{\bagB}=\mset{\typeOf{\vtermB}}\bagcat\typeOf{\bagB_1}
  =\mset{\dualtype{\typeOf{\strA'}}}\bagcat\ctxOf{\strA}(\varA)
  =\mset{\dualtype{\typeOf{\strA}}}\bagcat\ctxOf{\strA}(\varA)
  =\ctxOf{\termA}(\varA)
  $ and
  $\ctxOf{\termA'}=\ctxOf{\vtermB}\rctxcat\ctxOf{\strA'}
  =\ctxOf{\vtermB}\rctxcat(\ctxOf{\strA}-\varA)\rctxcat\ctxOf{\bagB_1}
  =(\ctxOf{\termA}-\varA)\rctxcat\ctxOf{\bagB}
  $
  since
  $\ctxOf{\termA}-\varA=\ctxOf{\strA}-\varA$.
\end{proof}

\begin{lemma}\label{lemma:substitution:reduction}
  If $\typeable{\termA}$, $\typeable{\bagB}$,
  and $\typeOf{\bagB}=\ctxOf{\termA}(\varA)$,
  then there exists $\termA'\in\rsubst{\termA}{\varA}{\bagB}$
  such that $(\ctxOf{\termA}-\varA)\rctxcat\ctxOf{\bagB}\rjug\termA':\typeOf{\termA}$.
\end{lemma}
\begin{proof}
  Write $\termsA'\eqdef\rsubst{\termA}{\varA}{\bagB}$.
  The proof is by induction on $\termA$.

  If $\termA=\mset{\vtermA_1,\dotsc,\vtermA_k}$ then
  $\typeable{\vtermA_i}$ for $1\le i\le k$.
  Since $\typeOf{\bagB}=\ctxOf{\termA}(\varA)
  =\ctxOf{\vtermA_1}(\varA)\bagcat\cdots\bagcat\ctxOf{\vtermA_k}(\varA)$,
  we can write $\bagB=\bagB_1\bagcat\cdots\bagcat\bagB_k$ with
  $\typeOf{\bagB_i}=\ctxOf{\vtermA_i}(\varA)$ for $1\le i\le k$.
  We apply the induction hypothesis to $\vtermA_i$ for $1\le i\le k$:
  we obtain $\vtermA'_i\in\rsubst{\vtermA_i}{\varA}{\bagB_i}$ such that
  $(\ctxOf{\vtermA_i}-\varA)\rctxcat\ctxOf{\bagB_i}\rjugVal\vtermA'_i:\typeOf{\vtermA_i}$.
  We conclude by setting 
  $\termA'\eqdef\mset{\vtermA'_1,\dotsc,\vtermA'_k}$.

  As in the previous lemma, the cases of streams, values, and base 
  terms $\appl{\vtermC}{\strA}$ or $\appl{\varC}{\strA}$
  with $\varC\not=\varA$ follow the same pattern as for bags.

  Finally, if $\termA=\appl{\varA}{\strA}$ then $\typeable{\strA}$
  and $\typeOf{\bagB}=\mset{\dualtype{\typeOf{\strA}}}\bagcat\ctxOf{\strA}(\varA)$.
  Then we can write $\bagB=\mset{\vtermB}\bagcat\bagB_1$ with
  $\typeOf{\vtermB}=\dualtype{\typeOf{\strA}}$
  and $\typeOf{\bagB_1}=\ctxOf{\strA}(\varA)$.
  We apply the induction hypothesis to $\strA$, 
  and obtain $\strA'\in\rsubst{\strA}{\varA}{\bagB_1}$ such that
  $(\ctxOf{\strA}-\varA)\rctxcat\ctxOf{\bagB_1}\rjugStr\strA':\typeOf{\strA}$.
  We conclude by setting 
  $\termA'\eqdef\appl{\vtermB}{\strA'}$.
\end{proof}

\begin{lemma}
We have $\typeable{\vtermA}$ iff $\typeable{\labs{\varA}{\vtermA}}$
and then:
\begin{itemize}
  \item $\typeOf{\labs{\varA}{\vtermA}}=
    \dualtype*{\ctxOf{\vtermA}(\varA)\cons\rstrB}$
    iff $\typeOf{\vtermA}=\dualtype{\rstrB}$;
  \item $\ctxOf{\labs{\varA}{\vtermA}}=
    \ctxOf{\vtermA}-\varA$.
\end{itemize}
\end{lemma}
\begin{proof}
  Direct application of the definitions.
\end{proof}

Now we extend the typing system to sums:
if $\termsA\in\ResourceSums$, we set 
$\rctxA\rjug\termsA:\rtermA$
when there exists $\termA\in\support{\termsA}$ such that
$\rctxA\rjug\termA:\rtermA$.
We obtain:
\begin{lemma}\label{lemma:invariance}
  Assume $\termsA\sresred\termsA'$.
  We have $\rctxA\rjug\termsA:\rtermA$
  iff $\rctxA\rjug\termsA':\rtermA$.
\end{lemma}
\begin{proof}
  We first treat the case of redexes.
  First assume $\termsA=\appl*{\labs{\varA}{\vtermA}}{\bagB\cons\strC}$
  and $\termsA'=\appl*{\rsubst{\vtermA}{\varA}{\bagB}}{\strC}$.

  If $\rctxA\rjugBase\termsA:\basetype$
  then $\typeable{\vtermA}$, $\typeable{\bagB}$ and $\typeable{\strC}$,
  and moreover:
  $\typeOf{\bagB}=\ctxOf{\vtermA}(\varA)$,
  $\typeOf{\vtermA}=\dualtype{\typeOf{\strC}}$
  and $\rctxA=(\ctxOf{\vtermA}-\varA)\rctxcat\ctxOf{\bagB}\rctxcat\ctxOf{\strC}$.
  \Cref{lemma:substitution:reduction} yields
  $\vtermA'\in\rsubst{\vtermA}{\varA}{\bagB}$ such that
  $(\ctxOf{\vtermA}-\varA)\rctxcat\ctxOf{\bagB}\rjugVal\vtermA':\typeOf{\vtermA}$,
  and we set $\termA'\eqdef\appl{\vtermA'}{\strC}$ to obtain
  $\termA'\in\termsA'$ and $\rctxA\rjugBase\termA':\basetype$.

  Conversely, if $\rctxA\rjugBase\termsA':\basetype$,
  there exists $\vtermA'\in\rsubst{\vtermA}{\varA}{\bagB}$
  such that $\rctxA\rjugBase\appl{\vtermA'}{\strC}:\basetype$.
  It follows that $\typeable{\vtermA'}$ and $\typeable{\strC}$,
  and moreover $\rctxA=\ctxOf{\vtermA'}\rctxcat\ctxOf{\strC}$
  and $\typeOf{\vtermA'}=\dualtype{\typeOf{\strC}}$.
  \Cref{lemma:substitution:expansion} entails that
  $\typeable{\vtermA}$, $\typeable{\bagB}$, and moreover
  $\typeOf{\bagB}=\ctxOf{\vtermA}(\varA)$,
  $\typeOf{\vtermA}=\typeOf{\vtermA'}$,
  and
  $\ctxOf{\vtermA'}=(\ctxOf{\vtermA}-\varA)\rctxcat\ctxOf{\bagB}$.
  It follows that
  $\typeOf{\labs{\varA}{\vtermA}}=\dualtype*{\typeOf{\bagB}\cons\typeOf{\strC}}$
  and $\rctxA=\ctxOf{\labs{\varA}{\vtermA}}\rctxcat\ctxOf{\bagB}\rctxcat\ctxOf{\strC}$,
  hence $\rctxA\rjugBase\termsA:\basetype$.

  Now assume $\termsA=\appl*{\labs{\svarA}{\btermA}}{\emptystream}$
  and $\termsA'=\erase{\btermA}{\svarA}$.

  If $\rctxA\rjugBase\termsA:\basetype$
  then $\typeable{\btermA}$ and $\ctxOf{\btermA}(\svarA)=\emptystream$.
  By \cref{lemma:typing:occurrences},
  $\svarA\cap\LVarOf{\btermA}=\emptyset$ so $\termsA'=\btermA$.
  Moreover, $\rctxA=\ctxOf{\btermA}-\svarA=\ctxOf{\btermA}$.

  Conversely, if 
  $\rctxA\rjugBase\termsA':\basetype$
  then $\termsA'\not=0$ hence $\termsA'=\btermA$
  with $\svarA\cap\LVarOf{\btermA}=\emptyset$.
  Hence, by \cref{lemma:typing:occurrences},
  $\ctxOf{\btermA}(\svarA)=\emptystream$
  so $\typeOf{\labs{\svarA}{\btermA}}=\dualtype{\emptystream}$
  and $\ctxOf{\labs{\svarA}{\btermA}}=\rctxA$.
  We obtain $\rctxA\rjugBase\termsA:\basetype$.
  
  Next, we treat the case of $\termsA=\termA\in\ResourceTerms$,
  and $\termsA\resred\termsA'$.
  We obtain the result by a straightforward induction on 
  the definition of $\resred$.

  Finally, if $\termsA=\termA+\termsB$ and $\termsA'=\termsA'+\termsB$,
  with $\termA\resred\termsA'$, we conclude directly from the previous case,
  observing that 
  $\rctxA\rjug\termA+\termsB:\rtermA$
  (resp.\ $\rctxA\rjug\termsA'+\termsB:\rtermA$)
  iff $\rctxA\rjug\termA:\rtermA$ (resp.\ $\rctxA\rjug\termsA':\rtermA$) or
  $\rctxA\rjug\termsB:\rtermA$.
\end{proof}

\begin{lemma}\label{lemma:typeable:is:normalizable}
  For any $\termA\in\ResourceTerms$:
  \begin{itemize}
    \item either $\not\typeable{\termA}$ and $\NF{\termA}=0$;
    \item or $\typeable{\termA}$, and then
      $\rctxA\rjug\NF{\termA}:\rtermA$
      iff $\rctxA=\ctxOf{\termA}$ and $\rtermA=\typeOf{\termA}$
      --- in particular $\NF{\termA}\not=0$.
  \end{itemize}
\end{lemma}
\begin{proof}
  Since $\termA\sresredRT\NF{\termA}$, 
  we can apply the previous lemma:
  \begin{itemize}
    \item if $\not\typeable{\termA}$ then $\not\typeable{\NF{\termA}}$,
      hence $\NF{\termA}=0$ by \cref{lemma:typing:unicity};
    \item if $\typeable{\termA}$ then
      $\rctxA\rjug\NF{\termA}:\rtermA$
      iff $\rctxA\rjug\termA:\rtermA$
      iff $\rctxA=\ctxOf{\termA}$ and $\rtermA=\typeOf{\termA}$,
      again by \cref{lemma:typing:unicity}.
      \qedhere
  \end{itemize}
\end{proof}

\subsection{Taylor expansion of relational semantics}

We now show that the equational theory induced by the relational interpretation
of \(λ\)-terms in \(\RelTypes\) is nothing but \(\eqExtTay\), or equivalently \(\Hs\).
We first reformulate the relational type system for $λ$-terms,
to account for the distinction between $\ValueTypes$ and $\StreamTypes$:
\begin{gather*}
  \begin{prooftree}
    \infer0{\varA:\mset{\rtypeA}\rjug \varA:\rtypeA}
  \end{prooftree}
  \qquad
  \begin{prooftree}
    \hypo{\rctxA\rjug \ltermA:\dualtype{\rstrB}}
    \infer1{\rctxA-\varA\rjug\labs{\varA}{\ltermA}:\dualtype*{\rctxA(\varA)\cons\rstrB}}
  \end{prooftree}
  \\[1em]
  \begin{prooftree}
    \hypo{\rctxA\rjug \ltermA:\dualtype*{\rbagA\cons\rstrB}}
    \hypo{\rctxB\rjugb \ltermB:\rbagA}
    \infer2{\rctxA\rctxcat\rctxB\rjug\appl{\ltermA}{\ltermB}:\dualtype{\rstrB}}
  \end{prooftree}
  \qquad
  \begin{prooftree}
    \hypo{\rctxA_1\rjug \ltermA:\rtypeA_1}
    \hypo{\cdots}
    \hypo{\rctxA_k\rjug \ltermA:\rtypeA_k}
    \infer3{\rctxA_1\rctxcat\cdots\rctxcat\rctxA_k\rjugb\ltermA:\mset{\rtypeA_1,\dotsc,\rtypeA_k}}
  \end{prooftree}
\end{gather*}

For any \(\termvA\in\ResourceVectors[\Bool]\),
we set $\rctxA\rjug\termvA:\rtermA$
iff there is $\termA\in\termvA$ with $\rctxA\rjug\termA:\rtermA$.
\Cref{lemma:invariance} entails:
\begin{lemma}\label{lemma:invariance:vresred}
  If $\termvA,\termvA'\in\ResourceVectors[\Bool]$
  and $\termvA\vresred\termvA'$ then 
  $\rctxA\rjug\termvA:\rtermA$
  iff
  $\rctxA\rjug\termvA':\rtermA$.
\end{lemma}

Then we show that the relational interpretation factors through Taylor expansion:
\begin{theorem}\label{theorem:rel:taylor}
  For any $λ$-term $\ltermA$, the following are equivalent: 
  \[
    \rctxA\rjug\ltermA:\rtypeA;
    \qquad
    \rctxA\rjug\ExtTayExp{\ltermA}:\rtypeA;
    \qquad
    \rctxA\rjug\HeadTayExp{\ltermA}:\rtypeA
    \;.
  \]
\end{theorem}

The equivalence between the last two statements follow from
\cref{lemma:invariance:vresred,theorem:exttohead}.
A first step to establish the equivalence with the first one is to consider the
expansion of variables.

\begin{lemma}
  Let $\varA\in\LVar$ and $\svarA\in\SVar$. Then:
  \begin{itemize}
    \item for each $\rtypeA\in\ValueTypes$, 
      there is a unique $\witnessOf{\varA}{\rtypeA}\in\IdExp{\varA}$
      with $\typeOf{\witnessOf{\varA}{\rtypeA}}=\rtypeA$;
    \item for each $\rbagA\in\BagTypes$, 
      there is a unique $\witnessBagOf{\varA}{\rtypeA}\in\prom*{\IdExp{\varA}}$
      with $\typeOf{\witnessBagOf{\varA}{\rbagA}}=\rbagA$;
    \item for each $\rstrA\in\StreamTypes$, 
      there is a unique $\witnessStrOf{\svarA}{\rstrA}\in\IdStrExp{\svarA}$
      with $\typeOf{\witnessStrOf{\svarA}{\rstrA}}=\rstrA$.
  \end{itemize}
  Moreover:
  \[
      \varA:\mset{\rtypeA}\rjugVal\witnessOf{\varA}{\rtypeA}:\rtypeA,\quad
      \varA:\rbagA\rjugBag\witnessBagOf{\varA}{\rbagA}:\rbagA,\quad
      \svarA:\rstrA\rjugStr\witnessStrOf{\svarA}{\rstrA}:\rstrA,\quad
  \]
  and:
  \begin{itemize}
    \item if $\vtermA\in\IdExp{\varA}$ then
      \(\typeable{\vtermA}\) and 
      $\vtermA=\witnessOf{\varA}{\typeOf{\vtermA}}$;
    \item if $\bagA\in\prom*{\IdExp{\varA}}$ then 
      \(\typeable{\bagA}\) and 
      $\bagA=\witnessBagOf{\varA}{\typeOf{\bagA}}$;
    \item if $\strA\in\IdStrExp{\svarA}$ then 
      \(\typeable{\strA}\) and 
      $\strA=\witnessStrOf{\svarA}{\typeOf{\strA}}$.
  \end{itemize}
\end{lemma}
\begin{proof}
  We define 
  $\witnessOf{\varA}{\rtypeA}$, $\witnessBagOf{\varA}{\rbagA}$ and $\witnessStrOf{\svarA}{\rstrA}$
  by mutual induction on $\rtypeA$, $\rbagA$ and $\rstrA$.
  Given $\rtypeA=\dualtype{\rstrA}$, we chose $\svarB\not\ni\varA$ and set
  $\witnessOf{\varA}{\rtypeA}\eqdef\labs{\svarB}{\appl{\varA}{\witnessStrOf{\svarB}{\rstrA}}}$.
  If $\rbagA=\mset{\rtypeA_1,\dotsc,\rtypeA_k}$, we set
  $\witnessBagOf{\varA}{\rbagA}\eqdef
  \mset{\witnessOf{\varA}{\rtypeA_1},\dotsc,\witnessOf{\varA}{\rtypeA_k}}$.
  Finally, if $\rstrA=\tuple{\rbagA_i}_{i\in\N}$, we set
  $\witnessStrOf{\svarA}{\rstrA}\eqdef\tuple{\witnessBagOf{\svarA[i]}{\rbagA_i}}_{i\in\N}$.

  The fact that these expressions are the only ones satisfying the requirements,
  together with the induced typing judgements,
  are easily established by induction on type terms.
  Finally, the last three items are obtained by mutual induction on resource terms.
\end{proof}

\begin{proof}[Proof of Theorem~\ref{theorem:rel:taylor}]
  We establish that:
  \begin{itemize}
    \item $\rctxA\rjug\ltermA:\rtypeA$ iff
      there exists $\vtermA\in\ExtTayExp{\ltermA}$
      with $\rctxA\rjugVal\vtermA:\rtypeA$, and
    \item $\rctxA\rjugb\ltermA:\rbagA$ iff
      there exists $\bagA\in\ExtArgTayExp{\ltermA}$
      with $\rctxA\rjugBag\bagA:\rbagA$;
  \end{itemize}
  by induction on $\ltermA$.
  The second statement follows directly from the first:
  if $\rbagA=\mset{\rtypeA_1,\dotsc,\rtypeA_k}$ then 
  \begin{center}\begin{tabular}{rcl}
    $\rctxA\rjugb\ltermA:\rbagA$
    &iff&
    $\rctxA=\rctxA_1\rctxcat\cdots\rctxcat\rctxA_k$ with
    $\rctxA_i\rjug\ltermA:\rtypeA_i$ for $1\le i\le k$
    \\
    &iff&
    $\rctxA=\rctxA_1\rctxcat\cdots\rctxcat\rctxA_k$ with
    $\rctxA_i\rjugVal\ExtTayExp{\ltermA}:\rtypeA_i$ for $1\le i\le k$
    \\
    &iff&
    $\rctxA\rjugBag\ExtTayExp{\ltermA}:\rbagA$.
  \end{tabular}\end{center}

  First assume $\ltermA=\varA$.
  If $\rctxA\rjug\ltermA:\rtypeA$ 
  then $\rctxA=\varA:\mset{\rtypeA}$
  and we have $\varA:\mset{\rtypeA}\rjugVal\witnessOf{\varA}{\rtypeA}:\rtypeA$
  with $\witnessOf{\varA}{\rtypeA}\in\IdExp{\varA}$.
  Conversely, if $\vtermA\in\IdExp{\varA}$ with
  $\rctxA\rjugVal\vtermA:\rtypeA$,
  we must have $\vtermA=\witnessOf{\varA}{\rtypeA}$,
  hence $\rctxA=\varA:\mset{\rtypeA}$.

  Now assume $\ltermA=\labs{\varA}{\ltermB}$.
  If $\rctxA\rjug\ltermA:\rtypeA$
  then $\rctxA=\rctxB-\varA$ and $\rtypeA=\dualtype*{\rctxB(\varA)\cons\rstrC}$,
  with $\rctxB\rjug\ltermB:\dualtype{\rstrC}$.
  By induction hypothesis, this holds iff
  there exists $\vtermB\in\ExtTayExp{\ltermB}$
  with $\rctxB\rjugVal\vtermB:\dualtype{\rstrC}$,
  which is equivalent to
  $\rctxB-\varA\rjugVal\labs{\varA}{\vtermB}:\dualtype*{\rctxB(\varA)\cons\rstrC}$,
  \emph{i.e.}\ 
  $\rctxA\rjugVal\labs{\varA}{\vtermB}:\rtypeA$.

  Finally, assume $\ltermA=\appl{\ltermB}{\ltermC}$.
  If $\rctxA\rjug\ltermA:\dualtype{\rstrA}$
  then we can write $\rctxA=\rctxB\rctxcat\rctxC$
  and there exists $\rbagC\in\BagTypes$ 
  such that $\rctxB\rjug\ltermB:\dualtype{\rbagC\cons\rstrA}$
  and $\rctxC\rjugb\ltermC:\rbagC$.
  By induction hypothesis, we obtain 
  $\vtermB\in\ExtTayExp{\ltermB}$ and
  $\bagC\in\prom{\ExtTayExp{\ltermC}}$ such that
  $\rctxB\rjugVal\vtermB:\dualtype{\rbagC\cons\rstrA}$
  and $\rctxC\rjugBag\bagC:\rbagC$.
  Let $\svarA$ be a fresh sequence variable. We have
  $\rctxC\rctxcat\svarA:\rstrA\rjugStr\bagC\cons\witnessStrOf{\svarA}{\rstrA}:\rbagC\cons\rstrA$
  and then
  $\rctxB\rctxcat\rctxC\rctxcat\svarA:\rstrA\rjugBase\appl{\vtermB}{\bagC\cons\witnessStrOf{\svarA}{\rstrA}}:\basetype$,
  and finally
  $\rctxA\rjugVal\labs{\svarA}{\appl{\vtermB}{\bagC\cons\witnessStrOf{\svarA}{\rstrA}}}:\dualtype{\rstrA}$,
  We conclude since 
  $\labs{\svarA}{\appl{\vtermB}{\bagC\cons\witnessStrOf{\svarA}{\rstrA}}}\in\ExtTayExp{\ltermA}$.

  Conversely, if $\vtermA\in \ExtTayExp{\ltermA}$ with 
  $\rctxA\rjugVal\vtermA:\dualtype{\rstrA}$,
  then we must have $\vtermA=\labs{\svarA}{\appl{\vtermB}{\bagC\cons\strA}}$
  with
  \(\svarA\) fresh,
  $\vtermB\in\ExtTayExp{\ltermB}$,
  $\bagC\in\prom{\ExtTayExp{\ltermC}}$
  and $\strA\in\IdStrExp{\svarA}$,
  so that $\rctxA=\rctxA'-\svarA$
  with $\rctxA'(\svarA)=\rstrA$
  and $\rctxA'\rjugBase\appl{\vtermB}{\bagC\cons\strA}:\basetype$.
  Since \(\svarA\) is fresh,
  $\rctxA'(\svarA)=\ctxOf{\strA}(\svarA)$,
  hence $\strA=\witnessStrOf{\svarA}{\rstrA}$
  and $\svarA:\rstrA\rjugStr\strA:\rstrA$.
  Then there must exist $\rbagC\in\BagTypes$ 
  and contexts $\rctxB$ and $\rctxC$ such that
  $\rctxA=\rctxB\rctxcat\rctxC$,
  $\rctxB\rjugVal\vtermB:\dualtype*{\rbagC\cons\rstrA}$
  and
  $\rctxC\rjugBag\bagC:\rbagC$.
  By induction hypothesis, 
  we obtain
  $\rctxB\rjug\ltermB:\dualtype*{\rbagC\cons\rstrA}$ and
  $\rctxC\rjugb\ltermC:\rbagC$,
  hence
  $\rctxA\rjug\ltermA:\dualtype{\rstrA}$.
\end{proof}

\Cref{corollary:NFT,lemma:invariance:vresred,theorem:rel:taylor,theorem:unsolvable},
together with the inductive definition of the relational semantics 
entail that:
\begin{corollary}
  Setting $\ltermA\eqRel\ltermB$ when $\sem{\ltermA}=\sem{\ltermB}$ 
  defines a sensible extensional $λ$-theory.
\end{corollary}
As we have stated above, this property is already well known
\cite{DBLP:conf/csl/BucciarelliEM07}: 
the originality of our approach is to relate this semantics
with extensional Taylor expansion.
In particular, this allows us to give a new proof 
that $\eqRel$ is $\Hs$: this was first established
by Manzonetto \cite{DBLP:conf/mfcs/Manzonetto09}, who gave 
sufficient axiomatic conditions on a cartesian closed category to 
host a reflexive object modelling $\Hs$.
Here the result comes directly from the properties of extensional
Taylor expansion:
\begin{corollary}
  The relations $\eqExtTay$ and $\eqRel$ coincide
  -- hence $\eqRel$ is \(\Hs\).
\end{corollary}
\begin{proof}
  We have just established that $\eqRel$ is a consistent sensible $λ$-theory,
  hence $\mathord{\eqRel}\subseteq\mathord{\eqOh}$
  and we obtain $\mathord{\eqRel}\subseteq\mathord{\eqExtTay}$
  by \cref{theorem:separation}.
  For the reverse inclusion, assume $\ltermA\eqExtTay\ltermB$:
  by \cref{corollary:NFT,lemma:invariance:vresred} 
  we obtain $\sem{\ExtTayExp{\ltermA}}=\sem{\ExtTayExp{\ltermB}}$,
  and \cref{theorem:rel:taylor} yields \mbox{$\sem{\ltermA}=\sem{\ltermB}$}.
\end{proof}

\section{Normal resource terms as augmentations}\label{section:gs}

In this section, we show the correspondence with game semantics that
initially motivated the introduction of extensional resource terms.  

We briefly recall the basics of \emph{pointer concurrent games} in the
style of Blondeau-Patissier and Clairambault
\cite{DBLP:conf/fscd/Blondeau-Patissier21}. In this brief presentation
we shall introduce \emph{augmentations} as representations of
normal extensional resource terms,\footnote{
  In fact, it is the other way around: extensional resource terms were
  initially designed by the authors so that normal extensional resource terms
  are representations of augmentations on the universal arena.
}
but only consider them as static objects:
we leave for further work an account of the dynamics of extensional resource
terms in game semantics, as previously developed in the typed
case~\cite{BPCVA25}.

\subsection{Arenas and their constructions}

As usual in game semantics, the game is played on an \emph{arena} --
the arena presents all the available observable computational events on
a given type, along with their causal dependencies, given by a forest order:
we call \definitive{forest} any poset $A$ such that the principal ideal
$[a]_A\eqdef \set{a'\in\ev{A} \mid a'\leq_A a}$
of any element $a$
is finite and linearly ordered by $\leq_A$.

\begin{definition}\label{def:arena}
An \definitive{arena} is $A = \tuple{\ev{A}, \leq_A, \pol_A}$ where
the pair $\tuple{\ev{A}, \leq_A}$
(the set of \definitive{events} and their \definitive{causal} order)
defines a countable forest,
and the \definitive{polarity function} $\pol_A : \ev{A} \to \{-, +\}$ is
\definitive{alternating}: whenever $a_1 \imc_A a_2$,
we have $\pol_A(a_1) \neq \pol_A(a_2)$
(where $a_1 \imc_A a_2$ means $a_1 <_A a_2$ with no event strictly
between).
A \definitive{negative arena} additionally satisfies:
$\pol_A(a) = -$ for all $a \in \min(A) \eqdef \{a \in \ev{A} \mid
\text{$a$ minimal}\}$. 
\end{definition}

We say that a negative arena is \definitive{pointed} if it has a unique minimal
event, called the \definitive{initial move}. An \definitive{isomorphism} of
arenas, written $\varphi : A \iso B$, is a bijection on events
preserving and reflecting all structure.

\paragraph{Universal arena.} Usually, the arena represents the
\emph{type}. In this paper, as we are interested in an untyped
language, we shall mainly work with one ambient arena $\U$,
that we call the \emph{universal arena}
(this is the maximal single-tree arena in the sense of
Ker \emph{et al.}~\cite{KerNO02}):

\begin{definition}
The \definitive{universal arena}, denoted by $\U$, has the following data:
\[
\begin{array}{rl}
\text{\emph{events}:} &
\text{the set $\ev{\U} = \WordsOf{\mathbb{N}}$ of lists of natural numbers,}\\
\text{\emph{causality}:} &
\text{the prefix ordering,}\\
\text{\emph{polarity}:} & 
\text{$\pol_\U(l) = -$ if $l$ has even length, $+$ otherwise.}
\end{array}
\]
\end{definition}

\begin{figure}
\[
\scalebox{\ThCSscalefactor}{\small
\xymatrix@R=5pt@C=5pt{
&&&&\emptyword^-
	\ar@{.}[dlll]
	\ar@{.}[d]
	\ar@{.}[drrr]\\
&\tuple{0}^+ 	\ar@{.}[dl]
	\ar@{.}[d]
	\ar@{.}[dr]&&&
\tuple{1}^+	\ar@{.}[dl]
	\ar@{.}[d]
	\ar@{.}[dr]&&&
\tuple{2}^+	\ar@{.}[dl]
	\ar@{.}[d]
	\ar@{.}[dr]&\dots\\
\tuple{0,0}^-&\tuple{0,1}^-&\tuple{0,2}^-\dots&
\tuple{1,0}^-&\tuple{1,1}^-&\tuple{1,2}^-\dots&
\dots&\dots&\dots\\
\vdots&\vdots&\vdots&
\vdots&\vdots&\vdots&
}
}
\]
\caption{The universal arena $\U$}
\label{fig:univ_ar}
\end{figure}

The universal arena is represented in \cref{fig:univ_ar}, read
from top to bottom and where dotted lines represent immediate causal
dependency.
It is clear that $\U$ is a negative arena, with unique minimal event
$\emptyword$.
We can think of events in this arena as moves exploring the structure of a
closed Nakajima tree.
Opponent starts computation with $\emptyword^-$, exploring the 
root of the tree, prefixed by a countable sequence of
abstractions $\labs{\varA_0}{\cdots\labs{\varA_i}{\cdots}}$.
This Opponent move enables countably many Player moves
$\tuple{i}^+$, one for each $i \in \mathbb{N}$.
Playing $\tuple{i}^+$ corresponds to having an
occurrence of $x_i$ as head variable, with countably many arguments.
Opponent can access the $j$-th argument of this occurrence by playing
$\tuple{i,j}^-$. This process continues in this way, indefinitely.
Note that the arena only describes the types of moves that can be played
along this exploration process:
a state of the process will be captured by the notion of \emph{position}
(an isomorphism class of \emph{configurations})
and a run of the process itself will be abstracted as an \emph{isogmentation}
(an isomorphism class of \emph{augmentations}~\cite{BPCVA25})
-- these concepts will be recalled below.

\paragraph{Constructions.}
We shall use some other constructions on arenas: the \definitive{atomic
arena} $o$ has just one (negative) move $\qu$. The \definitive{tensor}
$A_1 \tensor A_2$ of arenas $A_1$ and $A_2$ has as events the tagged
disjoint union $\ev{A_1 \tensor A_2} = \ev{A_1} + \ev{A_2} = \{1\}
\times \ev{A_1} \uplus \{2\} \times \ev{A_2}$, and other components
inherited. The \definitive{hom-arena} $A_1 \vdash A_2$ is $A_1^\perp
\tensor A_2$, where the \definitive{dual} $A^\perp$ of $A$ has the same
components but the polarity reversed. If $A_1$ is a negative arena and $A_2$
is pointed, the \definitive{arrow} $A_1 \tto A_2$, a pointed arena, has the
same components as $A_1 \vdash A_2$ except that the events of $A_1$ are set
to depend on the initial event of $A_2$.
In the sequel, we shall also use the obvious generalization of the tensor with
arbitrary (at most countable) arity, and write $\SeqsOf{A} = \bigtensor_{n\in \N}
A$.

\paragraph{Isomorphisms.}
The following easy lemma states a few useful isomorphisms.

\begin{lemma}\label{lem:arisos}
For all arenas $\Gamma, A, B, C$
with $\Gamma$ and $A$ negative and $B$ pointed,
we have isomorphisms:
\[
\begin{array}{rcrclcl}
  \curry_{\Gamma, A, B} &:& \Gamma \tensor A \tto B
  &\iso & \Gamma \tto (A \tto B) &:& \uncurry_{\Gamma, A, B}\\
  \pack_C &:& C \tensor \SeqsOf{C} &\iso& \SeqsOf{C} &:& \unpack_C\\
  \unfold &:& \U &\iso& \SeqsOf{\U} \tto o &:& \fold\\
\end{array}
\]
  supported by the bijections:
  \[\begin{array}{rcrclcl}
    \left.\begin{array}{r}
        \tuple{1,\tuple{1,g}}\\
        \tuple{1,\tuple{2,a}}\\
        \tuple{2,b}\\
    \end{array}\right\}
    &\in& \ev{\Gamma \tensor A \tto B} & \leftrightarrow
    & \ev{\Gamma \tto (A \tto B)} &\ni&
    \left\{\begin{array}{r}
        \tuple{1,g}\\
        \tuple{2,\tuple{1,a}}\\
        \tuple{2,\tuple{2,b}}\\
    \end{array}\right.
    \\
    \\
    \left.\begin{array}{r}
        \tuple{1,c}\\
        \tuple{2,\tuple{i,c}}\\
    \end{array}\right\}
    &\in& \ev{C\tensor\SeqsOf{C}} & \leftrightarrow
    & \ev{\SeqsOf{C}} &\ni&
    \left\{\begin{array}{r}
        \tuple{0,c}\\
        \tuple{i+1,c}\\
    \end{array}\right.
    \\
    \\
    \left.\begin{array}{r}
        \emptyword\\
        i\cons \seq{n}\\
    \end{array}\right\}
    &\in& \ev{\U} & \leftrightarrow
    & \ev{\SeqsOf{\U}\tto o} &\ni&
    \left\{\begin{array}{r}
        \tuple{2,\qu}\\
        \tuple{1,\tuple{i,\seq{n}}}\\
    \end{array}\right.
  \end{array}
  .
\]
\end{lemma}

There is a category $\ArIso$ of arenas and isomorphisms between them,
with full subcategories $\ArIso_{-}$, restricted to negative arenas,
and $\ArIso_\bullet$, restricted to pointed arenas.
The tensor and arrow constructions act functorially on isomorphisms of
arenas, yielding functors
\[
\begin{array}{rcrcl}
\tensor &:& \ArIso \times \ArIso &\to& \ArIso\\
\tto &:& \ArIso_{-} \times \ArIso_\bullet &\to& \ArIso_\bullet
\end{array}
\]
obtained again in the obvious way, since the events of both $A \tensor
B$ and $A \tto B$ are obtained as a disjoint union. Note $\tto$ is
covariant in both components.

With this observation, \cref{lem:arisos} induces isomorphisms
$\U \iso \SeqsOf{\U} \tto o \iso (\U \tensor \SeqsOf{\U}) \tto o
\iso \U \tto (\SeqsOf{\U} \tto o) \iso \U \tto \U$,
yielding the following:
\begin{corollary}
We have an isomorphism of arenas:
\[
\begin{array}{rcrclcl}
\fun &:& \U &\iso& \U \tto \U &:& \unfun
\end{array}
\]
supported by the bijection
\[
  \left.\begin{array}{r}
      \emptyword\\
      0\cons\seq{n}\\
      (i+1)\cons\seq{n}\\
  \end{array}\right\}
  \in \ev{\U}  \leftrightarrow
  \ev{\U\tto\U} \ni
  \left\{\begin{array}{r}
      \tuple{2,\emptyword}\\
      \tuple{1,\seq{n}}\\
      \tuple{2,i\cons\seq{n}}\\
  \end{array}\right.
  \qquad.
\]
\end{corollary}

We shall establish a correspondence between closed normal extensional
resource terms and \emph{isogmentations} on $\U$.
But before defining those, we must first define an adequate notion of
\emph{state} on an arena. 

\subsection{Configurations and positions}

Usually, in Hyland-Ong games, after defining arenas, one would go on to
define \emph{plays}, which are certain strings of moves of the arena
(allowing replicated moves) additionally equipped with \emph{pointers},
and satisfying a few further conditions. In pointer concurrent games,
instead, our first step is to define \emph{configurations}, which
correspond to plays with pointers without the chronological ordering.

\paragraph{Configurations.} We start with the definition:

\begin{definition}\label{def:conf}
A \definitive{configuration} of the arena $A$, written $x \in \conf(A)$, is $x =
\tuple{\ev{x}, \leq_x, \display_x}$ s.t. $\tuple{\ev{x}, \leq_x}$ is a
finite forest, and the \definitive{display map} $\display_x : \ev{x} \to
\ev{A}$ is a function such that
\[
\begin{array}{rl}
\text{\emph{minimality-respecting:}} & \text{for $a \in \ev{x}$,
$a$ is $\leq_x$-minimal iff $\display_x(a)$ is $\leq_A$-minimal,}\\
\text{\emph{causality-preserving:}} & \text{for $a_1, a_2 \in \ev{x}$, if $a_1 \imc_x a_2$ then $\display_x(a_1) \imc_A \display_x(a_2)$,}
\end{array}
\]
and $x$ is \definitive{pointed} (noted $x\in\conf_\bullet(A)$) if it
has exactly one minimal event $\init(x)$.
We write $\emptyconf$ for the empty configuration.
\end{definition}

Configurations are a first step in capturing a notion of ``thick
subtree'' \cite{DBLP:conf/tlca/Boudes09} of an arena $A$: an
exploration of the arena, keeping the same causal structure, but with
the additional capacity to duplicate branches.
\begin{sidefigure}
\[
\scalebox{\ThCSscalefactor}{
\xymatrix@R=10pt@C=5pt{
&\varepsilon^-
	\ar@{.}[dl]
	\ar@{.}[dr]
	\ar@{.}[d]\\
\tuple{3}^+&\tuple{2}^+&\tuple{2}^+
	\ar@{.}[dl]
	\ar@{.}[d]
	\ar@{.}[dr]\\
&\tuple{2,0}^-&
\tuple{2,0}^-&
\tuple{2,2}^-
	\ar@{.}[d]\\
&&&\tuple{2,2,0}^+
}
}
\]
\caption{A configuration on $\U$}
\label{fig:ex_conf_U}
\end{sidefigure}
As an example, we display in \cref{fig:ex_conf_U} a configuration
of the arena $\U$ introduced above. The figure shows the tree structure
along with the display map, conveyed via the labelling of the nodes; and
showcases the ability to replay the same move.

Though the definition of a configuration does not include a polarity
function on $x$, one may be readily deduced by setting 
$\pol_x(a) = \pol_A(\display_x(a))$ for $a \in x$.  We write $a^-$
(respectively $a^+$) for $a$ such that $\pol_x(a) = -$ (respectively
$\pol(a)=+$).

Let us give a few constructions on configurations.
We start with the \definitive{tensor product}:
for any arenas $A$ and $B$, and configurations $x\in \conf(A)$ and $y \in \conf(B)$,
we define $x \tensor y \in \conf(A \tensor B)$ by
\( \ev{x \tensor y} = \ev{x} + \ev{y}\), the rest of the structure being inherited.
This construction is easily generalized to
the tensor \(\tensor{\seq{x}}= \bigtensor_{i\in I} x_i\) of a
finite family of configurations,
with \(\ev{\bigtensor_{i\in I} x_i} = \sum_{i\in I} \ev{x_i}\).
The other constructions are variations on this theme.
For $x \in \conf(A)$ and $y \in \conf(B)$,
$x \vdash y \in \conf(A \vdash B)$ is defined like $x \tensor y$.
If $B$ and $y$ are pointed, then
$x \tto y \in \conf(A\tto B)$ is defined as $x \vdash y$ where, additionally,
all events of $x$ are set to depend on the unique initial move of $y$. 

If $x\in \conf(A)$ and $y \in \conf(\SeqsOf{A})$, we define
$x\conspack y\in\conf(\SeqsOf{A})$ from
$x\tensor y\in\conf(A\tensor\SeqsOf{A})$ by applying the isomorphism of arenas
$\pack_{A}: A \tensor \SeqsOf{A} \bij \SeqsOf{A}$:
in other words, events and causality are the same, and 
$\display_{x \conspack y}=\pack_{A}\circ\display_{x\tensor y}$.
Note that $x\conspack y=\emptyconf$ iff $x$ and $y$ are both empty.
By analogy with the range of streams, we also define the \definitive{range} of
$x\in\confOf{A^\N}$ as
$
  \rangeOf{x}\eqdef \max\set{i+1\in\N\st \exists a\in\ev{x},\ \display_x(a)\in\set{i}\times\ev{A}}
$.
By construction, $\rangeOf{\emptyconf}=0$ and 
$\rangeOf{x\conspack y}=\rangeOf{y}+1$ if at least one of 
$x$ and $y$ is non-empty.

We will also consider \definitive{sums} of configurations.
Given a finite family
$\seq{x}=\tuple{x_i}_{i\in I} \in \confOf{A}^I$,
we define $\sum\seq{x}=\sum_{i\in I}x_i \in \confOf{A}$,
as for the tensor \(\tensor{\seq{x}}\), except for
the display map:
\(\display_{\sum{\seq{x}}}(j, a) = \display_{x_j}(a) \).
If moreover $\seq{x} \in \pconfOf{A}^I$, we also write $\bagpack{\seq{x}}$,
$\bagpack{x_i\st i \in I}$ or $\bagpack{x_i}_{i \in I}$
for $\sum{\seq x} \in \confOf{A}$.
Note that $\sum\emptyword=\emptyconf$ is the empty configuration.
Moreover, if $\seq{x}=\tuple{x_i}_{i\in I}\in\pconfOf{A}^I$, then
$\bagpack{x_i}_{i \in I}=\emptyconf$ iff $\seq{x}=\emptyword$.

\paragraph{Symmetry and positions.}
Configurations do not adequately capture thick subtrees by
themselves, because they carry arbitrary ``names'' for events, which
are irrelevant.  Configurations should rather be considered up to
\emph{symmetry}: a \definitive{symmetry} $\varphi : x \sym_A y$ is an
order-isomorphism such that $\display_y \circ \varphi = \display_x$.
Symmetry classes of configurations are called \definitive{positions}: the
set of positions on $A$ is written $\pos(A)$, and they are ranged over
by $\x, \y$, \emph{etc.} A position $\x$ is \definitive{pointed}, written
$\x \in \pos_\bullet(A)$, if any of its representatives is.
We also define the \definitive{range} of $\x \in \pos(\U^\N)$ as
$\rangeOf{\x}\eqdef\rangeOf{x}$ for any representative $x$ of $\x$.
It is easy to check that our constructions on configurations
preserve symmetry so that, e.g., the sum \(\sum\seq{\x}\in\posOf{A}\)
of a finite family \(\seq{\x}\in\posOf{A}^I\) is well defined.

If $x \in \conf(A)$, we write $\ic{x} \in \pos(A)$ for the corresponding
position. Reciprocally, if $\x \in \pos(A)$, we fix $\rep{\x} \in
\conf(A)$ a representative.  

\begin{lemma}\label{lem:pos_isos}
For any arenas $A$, $A'$ and $B$ with $B$ pointed,
the earlier constructions on configurations induce bijections:%
\[
\begin{array}{rcrcl}
(- \tensor -) &:& \pos(A) \times \pos(A') &\bij& \pos(A \otimes A')\\
( - \vdash -) &:& \pos(A) \times \pos(A') &\bij& \pos(A \vdash A')\\
(- \tto -) &:& \pos(A) \times \pos_\bullet(B) &\bij& \pos_\bullet(A \tto B) \\
(- \conspack -) &:& \posOf{A} \times \posOf{\SeqsOf{A}} &\bij& \posOf{\SeqsOf{A}}\\
\bagpack{-} &:& \Mf(\pos_\bullet(B)) &\bij& \pos(B)
\end{array}
\]
\end{lemma}
\begin{proof}
For the first four bijections, this is a straightforward verification
that the corresponding constructions on configurations preserve and reflect
symmetry, and are essentially surjective.

For \(\bagpack{-}\), consider $\mu \in \Mf(\pos_\bullet(B))$. Fix
$\tuple{\x_i}_{i\in I}$ an enumeration of $\mu$, and for each $i \in I$, fix
$x_i \in \x_i$ a representative. Then, set
\[
  \bagpack{\mu} \eqdef \ic{\bagpack{x_i \mid i \in I}} \in \pos(B)
  \,.
\]

First, we show that this does
not depend on the choice of the enumeration $\tuple{\x_i}_{i\in I}$
nor of the representatives $\tuple{x_i}_{i\in I}$:
any other enumeration of $\mu$ is $\tuple{\x_{\pi(j)}}_{j \in J}$ with $\pi: J \bij I$,
and then for any choice of representatives $\tuple{y_{j}}_{j \in J}$
we have symmetries $\theta_j : x_{\pi(j)} \sym_B y_j$ for $j\in J$;
then 
\[
\begin{array}{rcrcl}
  \theta &:& \bagpack{x_i}_{i \in I} &\sym_B& \bagpack{y_j}_{j \in J}\\
         && (\pi(j), a) &\mapsto& \tuple{j, \theta_j(a)}
\end{array}
\]
is a valid symmetry.
It follows that $\bagpack{-} : \Mf(\pos_\bullet(B)) \to \pos(B)$ is well-defined.

For injectivity, consider representatives
$\tuple{x_i}_{i\in I}$, $\tuple{y_j}_{j\in J}$, and some symmetry
$\theta : \bagpack{x_i \mid i \in I} \sym_B \bagpack{y_j \mid j \in J}$. As $\theta$ is an
order-isomorphism between forests, it is determined by a bijection
between minimal events, and a subsequent order-isomorphism for each
minimal event. But since the $x_i$'s and $y_j$'s are pointed, minimal
events of $\bagpack{x_i \mid i \in I}$ are in canonical bijection with $I$ --
and likewise for $\bagpack{y_j \mid j \in J}$, so that $\theta$ is determined
by some $\pi : J \bij I$ and a family $\tuple{\theta_j : x_{\pi(j)} \sym_B
y_{j}}_{j\in J}$. But this is exactly what is needed to ensure
that $\tuple{x_i}_{i\in I}$ and $\tuple{y_j}_{j\in J}$ represent the same element
of $\Mf(\pos_\bullet(B))$.

For surjectivity, consider $\x \in \pos(B)$ represented by
some $x \in \conf(B)$. As a configuration, $x$ is a
forest. If $I$ is the set of its minimal events and, for $i \in I$,
$x_i$ is the corresponding tree with root $i$, then clearly $x \sym_B
\bagpack{x_i \mid i \in I}$.
\end{proof}

Note that, as a particular case, we have
\[ 
(- \tto o) : \pos(A)\times\pposOf{o} \bij \pos_\bullet(A\tto o)
\]
and $\pposOf{o}$ is a singleton so that $\pos(A)\times\pposOf{o}\bij\pos(A)$.
Up to the isomorphism of arenas $\fold:\U^\N\tto o \iso \U$,
and up to a minor abuse of notation, we obtain
\[ 
(- \tto o) : \pos(\U^\N) \bij \pos_\bullet(\U)\;.
\]

\subsection{Positions of the universal arena} 
\label{subsec:pos_u}

We use the above to show that the pointed positions of the universal arena $\U$
are in one-to-one correspondence with the elements of the particular
relational model of the pure \(λ\)-calculus~\cite{DBLP:conf/csl/BucciarelliEM07}
that we studied in \cref{section:rel}.
We recall from our presentation in \cref{section:rel:resource}
that these elements are the value types $\dualtype{\rstrA}\in\ValueTypes$ where,
inductively, $\rstrA$ ranges over the set of stream types $\StreamTypes=\SsOf{\ValueTypes}$;
and a bag type is just a bag of values types $\bagA\in\BagTypes=\MfOf{\ValueTypes}$.
Formally, we exhibit bijections
\[
    \kappa_{\mathrm v} : \ValueTypes \bij \pposOf{\U}
    \,,
    \quad
    \kappa_{\oc}: \BagTypes \bij \posOf{\U}
    \quad
    \text{and}
    \quad
    \kappa_{\mathrm s} : \StreamTypes \bij \posOf{\U^\N}
    \,.
\]

\begin{definition}
  We define the functions $\kappa_{\mathrm v}$, $\kappa_{\oc}$
  and $\kappa_{\mathrm s}$ simultaneously by induction on type terms:%
  \begin{align*}
    \kappa_{\mathrm v}(\dualtype{\rstrA})
    &\eqdef \kappa_{\mathrm s}(\rstrA)\tto o
    &
    \kappa_{\oc}(\mset{\rtypeA_1,\dotsc,\rtypeA_k})
    &\eqdef
    \bagpack{\kappa_{\mathrm v}(\rtypeA_1),\dotsc,\kappa_{\mathrm v}(\rtypeA_k)}
    \\
    \kappa_{\mathrm s}(\emptystream)
    &\eqdef
    \ic{\emptyconf}
    &
    \kappa_{\mathrm s}(\rbagA\cons\rstrB)
    &
    \eqdef\kappa_{\oc}(\rbagA)\conspack\kappa_{\mathrm s}(\rstrB)
    \quad\text{(if $\rbagA\cons\rstrB\not=\emptystream$)}.
  \end{align*}
\end{definition}
By construction, $\kappa_{\mathrm s}(\rbagA\cons\rstrB)=\emptypos$
iff $\rbagA\cons\rstrB=\emptystream$,
so that
$\kappa_{\mathrm s}(\rbagA\cons\rstrB)
=\kappa_{\oc}(\rbagA)\conspack\kappa_{\mathrm s}(\rstrB)$
also in this case.

To show that these define bijections as stated above, we will reason on the
\emph{size} of positions: if $A$ is an arena and $\x \in \pos(A)$, then the
\definitive{size} of $\x$, written $\SizeOf{\x}$, is the number of events
in any of its representatives.
In particular, $\SizeOf{\x}=0$ iff $\x$ is the empty position.
From the definitions of bijections in \cref{lem:pos_isos}, it is simple to
obtain the following identities:
\[
  \SizeOf*{\x \tto o}                
  = \SizeOf{\x}+1
  \qquad
  \SizeOf{\bagpack{\x_i \mid i \in I}}
  = \sum_{i\in I} \SizeOf{\x_i}
  \qquad
  \SizeOf*{\x\conspack\y} 
  = \SizeOf{\x}+\SizeOf{\y}
\]
-- indeed, it suffices to observe these identities on representatives.

\begin{theorem}\label{th:bij_pos_reltypes}
  The functions
  \[
      \kappa_{\mathrm v} : \ValueTypes \to \pposOf{\U}
      ,
      \quad
      \kappa_{\oc}: \BagTypes \to \posOf{\U}
      \quad
      \text{and}
      \quad
      \kappa_{\mathrm s} : \StreamTypes \to \posOf{\U^\N}
  \]
  are bijections.
\end{theorem}
\begin{proof}
  For injectivity, we reason by induction on type terms.
  If $\kappa_{\mathrm v}(\dualtype{\rstrA})
  =\kappa_{\mathrm v}(\dualtype{\rstrB})$ then 
  $\kappa_{\mathrm s}(\rstrA)\tto o =\kappa_{\mathrm s}(\rstrB)\tto o$,
  hence $\kappa_{\mathrm s}(\rstrA) =\kappa_{\mathrm s}(\rstrB)$
  by \cref{lem:pos_isos}:
  by induction hypothesis we have $\rstrA=\rstrB$
  hence $\dualtype{\rstrA}=\dualtype{\rstrB}$.

  If $\kappa_{\oc}(\rbagA)=\kappa_{\oc}(\rbagB)$
  with $\rbagA=\mset{\rtypeA_{1},\dotsc,\rtypeA_{k}}$
  and $\rbagB=\mset{\rtypeB_{1},\dotsc,\rtypeB_{l}}$,
  we obtain
  $\bagpack{\kappa_{\mathrm v}(\rtypeA_{1}),\dotsc,\kappa_{\mathrm v}(\rtypeA_{k})}
  =\bagpack{\kappa_{\mathrm v}(\rtypeB_{1}),\dotsc,\kappa_{\mathrm v}(\rtypeB_{l})}$.
  \Cref{lem:pos_isos} ensures that $k=l$ and,
  up to reordering,
  $\kappa_{\mathrm v}(\rtypeA_{i})=\kappa_{\mathrm v}(\rtypeB_{i})$
  for $1\le i\le k$:
  by induction hypothesis, we have $\rtypeA_i=\rtypeB_i$
  for $1\le i\le k$, hence $\rbagA=\rbagB$.

  We have already observed that
  $\kappa_{\mathrm s}(\rstrA)=\emptypos$
  iff $\rstrA=\emptystream$.
  If $\x=\kappa_{\mathrm s}(\rbagA\cons\rstrB)=\kappa_{\mathrm s}(\rbagA'\cons\rstrB')$ with $\rbagA\cons\rstrB\not=\emptystream$,
  then $\x\not=\emptypos$, hence $\rbagA'\cons\rstrB'\not=\emptystream$.
  We then write $\x=\kappa_{\oc}(\rbagA)\conspack\kappa_{\mathrm v}(\rstrB)
  = \kappa_{\oc}(\rbagA')\conspack\kappa_{\mathrm v}(\rstrB')$,
  and it is again sufficient to apply \cref{lem:pos_isos} and the induction hypothesis.

  For surjectivity, we prove by induction on $d\in\N$ that,
  for each $\x\in\pposOf{\U}$ (resp.\ $\x\in\posOf{\U}$,
  $\x\in\posOf{\U^\N}$) such that $\SizeOf{\x}\le d$,
  there exists $\rtypeA\in\ValueTypes$ (resp.\ 
  $\rbagA\in\BagTypes$, $\rstrA\in\StreamTypes$)
  such that $\kappa_{\mathrm v}(\rtypeA)=\x$
  (resp.\  $\kappa_{\oc}(\rbagA)=\x$,
  $\kappa_{\mathrm s}(\rstrA)=\x$).

  If $\x\in\pposOf{\U}$ then, by \cref{lem:pos_isos}, there exists
  $\y\in\posOf{\U^\N}$ such that $\x=\y\tto o$.
  Then $\SizeOf{\y}<\SizeOf{\x}$ and by induction hypothesis 
  there is $\rstrA$ such that $\y=\kappa_{\mathrm s}(\rstrA)$:
  we obtain $\kappa_{\mathrm v}(\dualtype{\rstrA})=\y\tto o=\x$.

  If $\x\in\posOf{U}$, by \cref{lem:pos_isos}, there exists
  $\mset{\x_1,\dotsc,\x_k}\in\MfOf{\pposOf{\U}}$ such that 
  $\x=\bagpack{\x_1,\dotsc,\x_k}$.
  Then $\SizeOf{\x_i}\le\SizeOf{\x}$ and 
  the previous paragraph gives $\rtypeA_i$ such that
  $\x_i=\kappa_{\mathrm v}(\rtypeA_i)$ for $1\le i\le k$: we obtain
  $\x=\kappa_{\oc}(\mset{\rtypeA_1,\dotsc,\rtypeA_k})$.

  If $\x\in\posOf{\U^\N}$, we prove by a further induction
  on $\rangeOf{\x}$ that there exists $\rstrA\in\StreamTypes$
  such that $\kappa_{\mathrm s}(\rstrA)=\x$.
  We have $\emptypos=\kappa_{\mathrm s}(\emptystream)$ by definition,
  so we can assume $\x\not=\emptypos$.
  \Cref{lem:pos_isos} yields $\y\in\posOf{\U}$ and $\z\in\posOf{\SeqsOf{\U}}$
  such that $\x=\y\conspack\z$.
  Then $\SizeOf{\y}\le\SizeOf{\x}$, $\SizeOf{\z}\le\SizeOf{\x}$,
  and, since one of $\y$ and $\z$ must be non-empty, $\rangeOf{\z}<\rangeOf{\x}$:
  the induction hypothesis yields $\rstrB$ such that
  $\z=\kappa_{\mathrm s}(\rstrB)$.
  Moreover, the previous paragraph gives $\rbagA$ such that
  $\y=\kappa_{\oc}(\rbagA)$.
  We get $\x=\kappa_{\mathrm s}(\rbagA\cons\rstrB)$.
\end{proof}

\subsection{Augmentations and isogmentations}

\subsubsection{Definitions}

\paragraph{Augmentations.} 
From the above, we have seen that
\emph{positions} are exactly points of the relational semantics (or
intersection types), and that configurations are concrete
representations for positions. 
An \emph{augmentation} is then a configuration enriched with
causal links: we will see that this extra structure amounts
exactly to the specification a normal term of the corresponding
type.

\begin{sidefigure}
    \includegraphics[page=10,width=.5\textwidth]{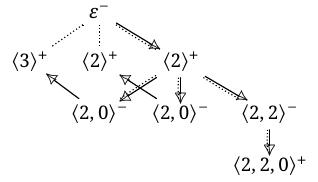}
\caption{An augmentation on $\U$}
\label{fig:ex_aug_1}
\end{sidefigure}
\begin{csidefigure}
    \includegraphics[page=11,width=.4\textwidth]{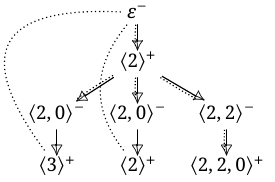}
\caption{The same augmentation following the dynamic causality}
\label{fig:ex_aug_2}
\end{csidefigure}
\begin{definition}\label{def:augmentation}
An \definitive{augmentation} on the arena $A$ is a tuple $\q = \tuple{\ev{\q},
        \leq_{\deseq{\q}}, \leq_\q, \display_\q}$, where $\deseq{\q} =
\tuple{\ev{\q}, \leq_{\deseq{\q}},\display_\q} \in \conf(A)$, and
$\tuple{\ev{\q}, \leq_\q}$ is a forest satisfying:
\[
\begin{array}{rl}
\text{\emph{rule-abiding:}} & \text{if
        $a_1 \leq_{\deseq \q} a_2$, then $a_1 \leq_\q a_2$,}\\ 
\text{\emph{courteous:}} & \text{if $a_1 \imc_\q a_2$ with
        $\pol(a_1) = +$ or $\pol(a_2) = -$, then $a_1 \imc_{\deseq \q}
a_2$,}\\
\text{\emph{deterministic:}} & \text{if $a^- \imc_\q a^+_1$ and
        $a^- \imc_\q a^+_2$, then $a_1 = a_2$,}\\
\text{\emph{$+$-covered:}}
& \text{if $a \in \ev{\q}$ is maximal in $\q$, then
$\pol(a)=+$,}\\
\text{\emph{negative:}}
& \text{if $a \in \min(\q)$, then $\pol(a) = -$.}
\end{array}
\]

We then write $\q \in \Aug(A)$, and call
$\deseq{\q} \in \conf(A)$ the \definitive{desequentialization} of $\q$.

Finally, $\q$ is \definitive{pointed} if it has a unique minimal event. In
that case, we write $\q \in \Aug_\bullet(A)$.
\end{definition}

We sometimes refer to $\leq_{\deseq{\q}}$ as the \definitive{static}
partial order, given that it corresponds to the causal constraints
imported from the ``type''. Likewise, we often refer to $\leq_\q$ as
the \definitive{dynamic} partial order, which will reflect the structure of a term.
The \definitive{size} of an augmentation $\q$, written $\SizeOf{\q}$, is
simply the cardinality of its set of events.

Moreover, for any negative arena $\Gamma$ and
$\q\in\Aug(\Gamma\vdash \U^{\N})$,
we define the \definitive{range} of $\q$ as
$\rangeOf{\q}\eqdef \max\set{i+1\in\N\st \exists a\in\ev{\q},\ 
q\in\ev{\U},\ \display_{\q}(a)=\tuple{2,\tuple{i,q}}}$.
By \emph{negative}, we have $\rangeOf{\bq}=0$ iff $\bq=\emptyisog$.

\Cref{fig:ex_aug_1} shows a pointed augmentation, obtained by enriching
the configuration in \cref{fig:ex_conf_U} with an adequate dynamic
causal order, presented here with triangle arrows $\imc$ -- following our drawing
convention, we display the (Hasse diagram of the) static causality
with dotted lines, and the (Hasse diagram of the) dynamic causality as
$\imc$. \Cref{fig:ex_aug_2} shows the same augmentation, but
rearranged so as to emphasize the tree structure of $\imc$ -- it is
this tree structure that we will show exactly corresponds to the tree
structure of an extensional resource term; here
\begin{eqnarray}
\labs{\svarA}
     {\appl{\svarA[2]}{\mset{\labs{\svarB}{\appl{\svarA[3]}{\emptystream}},\labs{\svarB}{\appl{\svarA[2]}{\emptystream}}}\cons\mset{}\cons\mset{\labs{\svarB}{\appl{\svarB[0]}{\emptystream}}}}\cons\emptystream}\,,\label{eq:ex_term}
\end{eqnarray}
and the reader may already match the tree structure of this term with
\cref{fig:ex_aug_2}. However, as for configurations, augmentations
are too rigid and carry explicit names for all events which should be
quotiented out before we establish the link with extensional resource
terms. This is done next.

\paragraph{Isogmentations.}
An \definitive{isomorphism} of augmentations $\varphi : \q \iso \p$ is a bijection
\(\ev\q\iso\ev\p\), preserving and reflecting all structure.
An \definitive{isogmentation} is an isomorphism class of augmentations,
ranged over by $\bq, \bp$, \emph{etc.}:
we write $\Isog(A)$ (resp.\ $\Isog_\bullet(A)$) for
isogmentations (resp.\ \emph{pointed} isogmentations).
If $\q \in \Aug(A)$, we write $\ic{\q} \in \Isog(A)$
for its isomorphism class; reciprocally, if $\bq \in \Isog(A)$, we
fix a representative $\rep{\bq} \in \bq$.

Clearly, the size of an augmentation is invariant under isomorphism:
the \definitive{size} of an isogmentation $\bq$, written
$\SizeOf{\bq}$, is the size of any representative.
Similarly, if $\bq\in\Aug(\Gamma\vdash \U^{\N})$,
we set the \definitive{range} of $\bq$ to be 
the range of any of its representatives.

\subsubsection{Constructions on Isogmentations}

Next, we provide a few constructions on augmentations and
isogmentations.

\paragraph{Tupling.}

Consider negative arenas $\Gamma, A_1, \dots, A_k$,
and $\q_i \in \Aug(\Gamma \vdash A_i)$ for $1\le i\le k$.
We set $\q=\tuplepack{\q_i \mid 1 \leq i \leq k} \in
\Aug(\Gamma \vdash \tensor_{1\leq i \leq k} A_i)$ with
\[
\ev{\q} = \sum_{i=1}^k \ev{\q_i}\,,
\qquad
\left\{
\begin{array}{rclcl}
  \display_{\q}(i, m) &=& \tuple{1, g}
  &\quad& \text{if $\display_{\q_i}(m) = \tuple{1, g}$,}\\
  \display_{\q}(i, m) &=& \tuple{2, \tuple{i, a}}
  && \text{if $\display_{\q_i}(m) = \tuple{2, a}$,}
\end{array}
\right.
\]
with the two orders $\leq_{\q}$ and
$\leq_{\deseq{\q}}$ inherited.

\begin{proposition}\label{prop:def_seq}
For any negative arenas $\Gamma, A_1, \dots, A_k$, this yields
a bijection
\[
\tuplepack{-, \dots, -} : \prod_{i=1}^k \Isog(\Gamma \vdash A_i)
\bij
\Isog(\Gamma \vdash \tensor_{1 \leq i \leq k} A_i)\,.
\]
\end{proposition}
\begin{proof}
First, it is direct that this construction preserves symmetry so that
it extends to isogmentations. We show that it is a bijection.

\emph{Injective.} As an iso must preserve $\imc$
and display maps, any $\varphi : \tuplepack{\q_i \mid i \in I} \iso
\tuplepack{\p_i \mid i \in I}$ decomposes uniquely into a sequence of
$\varphi_i : \q_i \iso \p_i$, as required.

\emph{Surjective.}
Consider $\q \in \Aug(\Gamma \vdash \tensor_{1\leq i \leq k} A_i)$.
Since $\q$ is a forest, any $m \in \ev{\q}$ has a unique minimal antecedent,
called $\init(m)$, sent by the display map (via \emph{negative}) 
to one of the $A_i$s --- we say that $m$ is \emph{above $A_i$}. This
allows us to assign every move of $\q$ to exactly one of the $A_i$s, so
that moves connected by (the transitive symmetric closure of) $\leq_\q$
are in the same component. This determines a partition of $\q$ into a
family of $\q_i \in \Aug(\Gamma \vdash A_i)$, one for each component --
and it is then a straightforward verification that 
$\q \iso \tuplepack{\q_i\mid 1 \leq i \leq k}$.
\end{proof}

In particular, we obtain a bijection
\begin{align*}
  \Isog(\Gamma \vdash \U)\times \Isog(\Gamma \vdash \SeqsOf{\U})
  &\bij
  \Isog(\Gamma \vdash \SeqsOf{\U})
  \\
  \tuple{\bp,\bq}
  &\mapsto
  \bp\conspack\bq
\end{align*}
by setting $ \bp\conspack\bq
  \eqdef 
  \pack_\Gamma(\tuplepack{\bp,\bq})
$
where \[
  \pack_\Gamma:
  \Isog(\Gamma \vdash \U \tensor \SeqsOf{\U})
  \bij \Isog(\Gamma \vdash \SeqsOf{\U})
\]
is induced by the isomorphism of arenas
$\pack_{\U}: \U \tensor \SeqsOf{\U}\bij \SeqsOf{\U}$.

\paragraph{Bags.}

Consider negative arenas $\Gamma$ and $A$, and $\q_1, \q_2 \in \Aug(\Gamma
\vdash A)$. We set $\q_1 * \q_2 \in \Aug(\Gamma \vdash A)$ with events
$\ev{\q_1 * \q_2} = \ev{\q_1} + \ev{\q_2}$,
and display
$\display_{\q_1 * \q_2}(i, m) = \display_{\q_i}(m)$,
and the two orders $\leq_{\q_1 * \q_2}$ and $\leq_{\deseq{\q_1 *
\q_2}}$ inherited. This generalizes to an $n$-ary operation
$\Pi_{i \in I} \q_i$ in the obvious way, which preserves
isomorphisms. The operation induced on isogmentations is associative
and admits as neutral element the empty isogmentation $\emptyisog \in
\Isog(\Gamma \vdash A)$ (with a unique representative $\emptyaug$) with no event. 
If $\bar\bq=\mset{\bq_1,\dotsc,\bq_k}\in\MfOf{\pIsogOf{\Gamma\vdash A}}$,
then we write $\bagpack{\bq_1,\dotsc,\bq_k}\eqdef\Pi\bar\bq$.

\begin{proposition}\label{prop:def_bag}
For any negative arenas $\Gamma, A$, this yields a bijection
\begin{align*}
\Mf(\Isog_\bullet(\Gamma \vdash A))
&\bij
\Isog(\Gamma \vdash A)
\\
\mset{\bq_1,\dotsc,\bq_k}
&\mapsto
\bagpack{\bq_1,\dotsc,\bq_k}\,.
\end{align*}
\end{proposition}
\begin{proof}
\emph{Injective.} Assuming the $\q_i$'s and $\p_i$'s are pointed, an iso
$\varphi : \Pi_{1\leq i \leq k} \q_i \iso \Pi_{1\leq j \leq l}
\p_j$ forces $k = l$ and induces a permutation $\pi$ on $k$ and a
family of isos $(\varphi_i : \q_i \bij \p_{\pi(i)})_{1\leq i \leq k}$,
which implies $[\ic{\q_i} \mid 1\leq i \leq k] = [\ic{\p_i} \mid 1 \leq
i \leq k]$ as bags.

\emph{Surjective.} As any $\q \in \Aug(\Gamma \vdash A)$ is finite, it
has a finite set $I$ of initial moves. As $\q$ is a forest, any
$m \in \ev{\q}$ is above exactly one initial move. For $i \in I$, write
$\q_i \in \Aug_\bullet(\Gamma \vdash A)$ the restriction of $\q$ above
$i$; then $\q \iso \Pi_{i\in I} \q_i$.
\end{proof}

\paragraph{Currying.}
For negative arenas $\Gamma$, $A$ and $B$, with $B$ pointed, we have 
\[
\Lambda_{\Gamma, A, B} : \Aug(\Gamma \tensor A \vdash B) \bij \Aug(\Gamma \vdash A \tto B)
\]
a bijection compatible with isos, which leaves the core of the
augmentation unchanged and only reassigns the display map, following
the associativity up to isomorphism of the tagged disjoint union $(X +
Y) + Z \bij X + (Y + Z)$. 
Additionally, this bijection preserves symmetry, and
hence extends to isogmentations. In linking resource terms and
isogmentation, we shall only use the instance:
\[
  \Lambda_{V,\svarA} : \Isog_\bullet((\U^{\N})^{V,\svarA} \vdash o)
  \bij
  \Isog_\bullet((\U^{\N})^V\vdash \U^\N \tto o)
\]
where $V$ is a finite set of sequence variables
and $\svarA\not\in V$  is a sequence variable.

\paragraph{Head occurrence.} We are missing one last bijection,
corresponding to normal base terms.  We start with the construction on
augmentations. Fix for now a negative arena
$\Gamma = (\U^\N)^V$
for $V$ a finite set of sequence variables. Consider also 
\(
\q \in \Aug(\Gamma \vdash \U^\N)
\)
an augmentation,
a sequence variable $\svarA\in V$, and a rank $i \in \N$. With this, we form
\[
\lift_{\svarA,i}(\q) \in \Aug_\bullet(\Gamma \vdash o)
\]
the \definitive{$(\svarA, i)$-lifting} of $\q$: it calls the value variable
$\svarA[i]$, giving it $\q$ as a stream of arguments, mimicking a base term
$\svarA[i]\,\q$. Concretely, $\lift_{\svarA,i}(\vec \q)$ starts by
playing the initial move in the component $\tuple{\svarA, i}$ of $\Gamma$,
then proceeds as $\q$. More precisely:

\begin{definition}\label{def:lifting}
Consider $\Gamma = (\U^\N)^V$ for $V$ a finite set
of sequence variables, $\q \in \Aug(\Gamma \vdash \U^\N)$ an
augmentation, $\svarA \in V$, and $i \in \N$. 
The \definitive{$(\svarA, i)$-lifting} of $\q$, written
$\lift_{\svarA,i}(\q) \in \Aug_\bullet(\Gamma \vdash o)$, has dynamic
partial order $\leq_{\q}$ prefixed with two additional moves, \emph{i.e.}\ 
$\ominus \imc \oplus \imc \q$. Its static causality is the least
partial order comprising the dependencies of the form
\[
\begin{array}{rclcl}
m &\leq_{\deseq{\lift_{\svarA,i}(\q)}}& n &\qquad& 
\text{for $m, n \in \ev{\q}$ with $m \leq_{\deseq{\q}} n$,}\\
\oplus &\leq_{\deseq{\lift_{\svarA,i}(\q)}}& m&&
\text{for all $m \in \ev{\q}$ with $\partial_{\q}(m) = \tuple{2, -}$,}
\end{array}
\]
and with display map given by the following clauses:
\[
\begin{array}{rclcl}
  \partial_{\lift_{\svarA,i}(\q)}(\ominus) &=& \tuple{2, q}\\
  \partial_{\lift_{\svarA,i}(\q)}(\oplus) &=& \tuple{1, \tuple{\svarA, \tuple{i, \emptyword}}}\\
  \partial_{\lift_{\svarA,i}(\q)}(m) &=& \tuple{1, a}
  &\qquad&
  \text{if $\partial_{\q}(m) = \tuple{1, a}$,}\\
  \partial_{\lift_{\svarA,i}(\q)}(m) &=& \tuple{1, \tuple{\svarA, \tuple{i, k::l}}}
  &&\text{if $\partial_\q(m) = \tuple{2, \tuple{k,l}}$,}
\end{array}
\]
altogether defining $\lift_{\svarA,i}(\q) \in \Aug_\bullet(\Gamma \vdash o)$ as
required.
\end{definition}

It is a routine exercise to check that this yields a well-defined augmentation
and that it preserves isomorphisms, hence extends to isogmentations, to give:

\begin{proposition}\label{prop:def_base} 
Let $\Gamma = (\U^\N)^V$ for $V$ a finite set
of sequence variables. 
We obtain a bijection
\begin{align*}
\lift : V \times \N \times \Isog(\Gamma \vdash \U^\N)
&\bij
\Isog_\bullet(\Gamma \vdash o)
\\
\tuple{\svarA,i,\bq}
&\mapsto
\lift_{\svarA,i}(\bq)\,.
\end{align*}
\end{proposition}
\begin{proof}
\emph{Injective.} Assume $\varphi : \lift_{\svarA, i} \q \iso
\lift_{\svarB, j} \p$ an iso. From the commutation with the display
map, the first Player moves must be in the same component, hence
$\svarA = \svarB$ and $i=j$. Removing the first two events, $\varphi$
restricts to $\psi : \q \iso \p$.

\emph{Surjective.}
Any $\q \in \Aug_\bullet(\Gamma \vdash o)$ has a unique initial move
(write it $\ominus$), which cannot be maximal by \emph{$+$-covered}. By
\emph{determinism}, there is a unique subsequent Player move (write it
$\oplus)$, displayed to the initial move of some occurrence of $\U$ in
$\Gamma = (\U^\N)^V$, determining some $\svarA \in V$
and $i \in \N$. We obtain $\p \in \Aug(\Gamma \vdash \U^\N)$ by
removing $\ominus$ and $\oplus$, and redisplaying on the right the
events statically depending on $\oplus$, reversing the reassignment of
\cref{def:lifting}.
We obtain $\p\in\Aug(\Gamma\vdash\U^\N)$ and
$\q\iso\lift_{\svarA,i}\p$ by a routine inspection of the
definitions.
\end{proof}

\subsection{Isogmentations as normal resource terms}

We spell out the bijection between normal resource terms and
isogmentations. If $\Gamma$ is a finite set of sequence variables,
write $\NormalValueTerms(\Gamma)$
(\emph{resp.}\ $\NormalBaseTerms(\Gamma), \NormalBagTerms(\Gamma),
\NormalStreamTerms(\Gamma)$) for the \emph{normal} value
terms $\vtermA \in \ValueTerms$ s.t. $\SVarOf{\vtermA} \subseteq
\Gamma$ (\emph{resp.}\ likewise for base, bag, and stream terms). A
finite set of sequence variables is interpreted as an arena:
\[
\intr{\Gamma} \eqdef (\U^\N)^\Gamma\,.
\]

We shall in fact establish \emph{four} mutually inductive bijections
for value, base, bag and stream normal extensional resource terms,
presenting them as isogmentations on appropriate arenas.

\begin{definition}
  We define 
  \begin{align*}
    \sintr{-}^{\Gamma}_{\mathrm b} &: \NormalBaseTerms(\Gamma)   \to \Isog_\bullet(\intr{\Gamma} \vdash o)\\
    \sintr{-}^{\Gamma}_{\mathrm v} &: \NormalValueTerms(\Gamma)  \to \Isog_\bullet(\intr{\Gamma} \vdash \U)\\
    \sintr{-}^{\Gamma}_{\oc}       &: \NormalBagTerms(\Gamma)    \to \Isog(\intr{\Gamma} \vdash \U)\\
    \sintr{-}^{\Gamma}_{\mathrm s} &: \NormalStreamTerms(\Gamma) \to \Isog(\intr{\Gamma} \vdash \U^\N)
  \end{align*}
  by mutual induction as follows:
  \begin{align*}
    \sintr{\appl{\svarA[i]}{\strB}}^{\Gamma}_{\mathrm b}
    &\eqdef \lift_{\svarA,i}\pars[\big]{\sintr{\strB}^{\Gamma}_{\mathrm s}}\\
    \displaybreak[3]
    \sintr{\labs{\svarA}{\btermA}}^{\Gamma}_{\mathrm v}
    &\eqdef \Lambda_{\intr{\Gamma},\svarA}\big(\sintr{\btermA}^{\Gamma,\svarA}_{\mathrm b}\big)\\
    \displaybreak[3]
    \sintr{\mset{\vtermA_1,\dotsc,\vtermA_k}}^{\Gamma}_{\oc}
    &\eqdef \bagpack{\sintr{\vtermA_1}^{\Gamma}_{\mathrm v},\dotsc,\sintr{\vtermA_k}^{\Gamma}_{\mathrm v}}[\big] \\
    \displaybreak[3]
    \sintr{\emptystream}^{\Gamma}_{\mathrm s}
    &\eqdef \emptyisog\\
    \displaybreak[3]
    \sintr{\bagA\cons\strB}^{\Gamma}_{\mathrm s}
    &\eqdef \sintr{\bagA}^{\Gamma}_{\mathrm b}\conspack\sintr{\strB}^{\Gamma}_{\mathrm s}
    \quad\text{if $\bagA\cons\strB\not=\emptystream$.}
  \end{align*}
\end{definition}
Note that $\sintr{\strA}^{\Gamma}_{\mathrm s}=\emptystream$
iff $\strA=\emptystream$, and that the identity
$\sintr{\bagA\cons\strB}^{\Gamma}_{\mathrm s}
=\sintr{\bagA}^{\Gamma}_{\mathrm b}\conspack
\sintr{\strB}^{\Gamma}_{\mathrm s}$
always holds.

Unlike in the typed setting \cite{BPCVA25},
the proof that these define bijections cannot be done by induction on
types. 
As we did in the proof of \cref{th:bij_pos_reltypes}, we shall reason by
induction on terms for injectivity, and by induction on the size 
and range of isogmentations for surjectivity.
Note that $\SizeOf{\bq}=0$ iff $\bq$ is empty.
By a straightforward inspection of the definitions, we moreover obtain:
\begin{align*}
  \SizeOf{\lift_{\svarA,i}(\bp)} &= \SizeOf{\bp}+2&
  \SizeOf{\Lambda_{\Gamma,\svarA}(\bp)} &= \SizeOf{\bp}
  \\
  \SizeOf{\bagpack{\bq_i \mid i \in I}} &= \sum_{i\in I} \SizeOf{\bq_i}&
  \SizeOf*{\bq\conspack\bp} &= \SizeOf{\bq}+\SizeOf{\bp}
\end{align*}
and $\rangeOf{\bp\conspack \bq}=\rangeOf{\bq}+1$ if at least one of 
$\bp$ and $\bq$ is non-empty.

\begin{theorem}\label{thm:isog_terms}
Consider $\Gamma$ any finite set of sequence variables. 
Then
$\sintr{-}^{\Gamma}_{\mathrm b}$,
$\sintr{-}^{\Gamma}_{\mathrm v}$,
$\sintr{-}^{\Gamma}_{\oc}$, and
$\sintr{-}^{\Gamma}_{\mathrm s}$
are bijections.
\end{theorem}
\begin{proof}
  For injectivity, we reason by induction on terms.

  If $\sintr{\appl{\svarA[i]}{\strB}}^{\Gamma}_{\mathrm b}
  = \sintr{\appl{\svarB[j]}{\strC}}^{\Gamma}_{\mathrm b}$
  then $\svarA=\svarB$, $i=j$ and 
  $\sintr{\strB}^{\Gamma}_{\mathrm s}
  =\sintr{\strC}^{\Gamma}_{\mathrm s}$,
  by the injectivity of $\lift$.
  We moreover obtain $\strB=\strC$ by induction hypothesis.

  If $\sintr{\labs{\svarA}{\btermA}}^{\Gamma}_{\mathrm v}
  = \sintr{\labs{\svarB}{\btermB}}^{\Gamma}_{\mathrm v}$,
  we can always assume $\svarA=\svarB$ by $α$-equivalence,
  hence $\sintr{\btermA}^{\Gamma,\svarA}_{\mathrm b}
  = \sintr{\btermB}^{\Gamma,\svarA}_{\mathrm b}$
  by the injectivity of $\Lambda_{\Gamma,\svarA}$.
  Then we obtain $\btermA=\btermB$ by induction hypothesis.

  If $\sintr{\mset{\vtermA_1,\dotsc,\vtermA_k}}^{\Gamma}_{\oc}=
  \sintr{\mset{\vtermB_1,\dotsc,\vtermB_l}}^{\Gamma}_{\oc}$,
  then by the injectivity of $\bagpack{-}$, we have
  $\mset{\sintr{\vtermA_1}^{\Gamma}_{\oc},\dotsc,\sintr{\vtermA_k}^{\Gamma}_{\oc}}
  =\mset{\sintr{\vtermB_1}^{\Gamma}_{\oc},\dotsc,\sintr{\vtermB_l}^{\Gamma}_{\oc}}$,
  hence $k=l$ and, up to reordering the $\vtermB_i$'s,
  $\sintr{\vtermA_i}^{\Gamma}_{\oc}=\sintr{\vtermB_i}^{\Gamma}_{\oc}$
  for $1\le i\le k$.
  The induction hypothesis yields $\vtermA_i=\vtermB_i$
  for $1\le i\le k$.

  We have already observed that
  $\sintr{\strA}^{\Gamma}_{\mathrm s}=\emptystream$
  iff $\strA=\emptystream$.
  If $\bq=\sintr{\bagA\cons\strB}^{\Gamma}_{\mathrm s}=
  \sintr{\bagA'\cons\strB'}^{\Gamma}_{\mathrm s}$
  with $\bagA\cons\strB\not=\emptystream$,
  then $\bq\not=\emptyisog$, hence
  $\bagA'\cons\strB'\not=\emptystream$.
  The injectivity of $\conspack$ then yields
  $\sintr{\bagA}^{\Gamma}_{\oc}=\sintr{\bagA'}^{\Gamma}_{\oc}$
  and $\sintr{\strB}^{\Gamma}_{\mathrm s}
  =\sintr{\strB'}^{\Gamma}_{\mathrm s}$.

  For surjectivity, we prove by induction on $d\in\N$ that,
  for each $\bq\in\Isog_\bullet(\intr{\Gamma} \vdash o)$
  (resp.\ $\bq\in\Isog_\bullet(\intr{\Gamma} \vdash \U)$,
  $\bq\in\Isog(\intr{\Gamma} \vdash \U)$, or
  $\bq\in\Isog(\intr{\Gamma} \vdash \U^\N)$)
  such that $\SizeOf{\bq}\le d$,
  there exists a term 
  $\btermA\in \NormalBaseTerms(\Gamma)$
  (resp.\ $\vtermA\in \NormalValueTerms(\Gamma)$,
  $\bagA\in\NormalBagTerms(\Gamma)$, or
  $\strA\in\NormalStreamTerms(\Gamma)$)
  such that
  $\sintr{\btermA}^{\Gamma}_{\mathrm b}=\bq$
  (resp.\ $\sintr{\vtermA}^{\Gamma}_{\mathrm v}=\bq$,
  $\sintr{\bagA}^{\Gamma}_{\oc}=\bq$, or
  $\sintr{\strA}^{\Gamma}_{\mathrm s}=\bq$).

  For each $\bq\in\Isog_\bullet(\intr{\Gamma} \vdash o)$,
  the surjectivity of $\lift$ yields $\svarA\in\Gamma$,
  $i\in\N$ and $\bp\in\Isog(\intr{\Gamma} \vdash \U^{\N})$
  such that $\bq=\lift_{\svarA,i}(\bp)$.
  In particular, $\SizeOf{\bp}<\SizeOf{\bq}$,
  and the induction hypothesis gives $\strA$ with 
  $\sintr{\strA}^{\Gamma}_{\mathrm s}=\bp$.
  We obtain $\bq=\sintr{\appl{\svarA[i]}{\strA}}^{\Gamma}_{\mathrm b}$.

  For each $\bq\in\Isog_\bullet(\intr{\Gamma} \vdash \U)$
  and sequence variable $\svarA\not\in\Gamma$,
  the surjectivity of $\Lambda_{\Gamma,\svarA}$
  yields $\bp\in\Isog_\bullet(\intr{\Gamma,\svarA} \vdash o)$
  such that $\bq=\Lambda_{\Gamma,\svarA}(\bp)$.
  In particular, $\SizeOf{\bp}\le\SizeOf{\bq}$,
  and the previous paragraph gives $\btermA$ with 
  $\sintr{\btermA}^{\Gamma,\svarA}_{\mathrm b}=\bp$.
  We obtain $\bq=\sintr{\labs{\svarA}{\btermA}}^{\Gamma}_{\mathrm v}$.

  For each $\bq\in\Isog(\intr{\Gamma} \vdash \U)$,
  the surjectivity of $\bagpack{-}$
  yields $\bp_1,\dotsc,\bp_k\in\Isog_\bullet(\intr{\Gamma} \vdash \U)$
  such that $\bq=\bagpack{\bp_1,\dotsc,\bp_k}$.
  In particular, $\SizeOf{\bp_i}\le\SizeOf{\bq}$,
  and the previous paragraph gives $\vtermA_i$ with 
  $\sintr{\vtermA_i}^{\Gamma}_{\mathrm v}=\bp_i$.
  We obtain $\bq=\sintr{\mset{\vtermA_1,\dotsc,\vtermA_k}}^{\Gamma}_{\oc}$.

  Finally, if $\bq\in\Isog(\intr{\Gamma} \vdash \U^{\N})$
  we prove by a further induction on $\rangeOf{\bq}$
  that there exists $\strA\in\NormalStreamTerms(\Gamma)$
  such that $\sintr{\strA}^{\Gamma}_{\mathrm s}=\bq$.
  Since $\emptyisog=\sintr{\emptystream}^{\Gamma}_{\mathrm s}$
  we can assume that $\bq\not=\emptyisog$.
  By the surjectivity of $\conspack$, we can write
  $\bq=\bq_1\conspack\bq_2$ with 
  $\bq_1\in\Isog(\intr{\Gamma} \vdash \U)$ and
  $\bq_2\in\Isog(\intr{\Gamma} \vdash \U^{\N})$.
  Then $\SizeOf{\bq_1}\le\SizeOf{\bq}$,
  $\SizeOf{\bq_2}\le\SizeOf{\bq}$, and 
  $\rangeOf{\bq_2}<\rangeOf{\bq}$ since $\bq$ is not empty.
  The previous paragraph gives
  $\bagA_1\in\NormalBagTerms(\Gamma)$
  such that $\sintr{\bagA_1}^{\Gamma}_{\oc}=\bq_1$, and
  the induction hypothesis gives
  $\strA_2\in\NormalStreamTerms(\Gamma)$ such that 
  $\sintr{\strA_2}^{\Gamma}_{\mathrm s}=\bq_2$.
  We obtain $\bq=\sintr{\bagA_1\cons\strA_2}^{\Gamma}_{\mathrm s}$.
\end{proof}

As claimed, normal extensional resource terms are a
syntax for isogmentations on the universal arena.
The reader may apply the implicit algorithm on the augmentation of
\cref{fig:ex_aug_2}, and observe that we indeed get the term
\eqref{eq:ex_term}.

\subsection{Compatibility with the bijection on positions}

Each isogmentation yields (forgetting the dynamic causal links) a
position. In \cref{subsec:pos_u}, we established a bijection between
positions of the universal arena and the relational types of
\cref{section:rel}. In \cref{section:rel:resource}, we proved that each
normal resource term admits \emph{exactly one} typing derivation.
In this final technical section, we tie everything together by proving that
these correspondences are compatible -- that is, that the following diagram
commutes:
\[
  \begin{tikzcd}[column sep=large] 
    \text{Resource terms}
    \arrow[r,"\intr{-}"]
    \arrow[d,"\text{typing}"']&
    \text{Isogmentations}
    \arrow[d,"\text{desequentialization}"]\\
    \text{Type terms}
    \arrow[r,"\intr{-}"']&
    \text{Positions}
  \end{tikzcd}
\]

To express this more formally, we need to introduce a few additional
notations.

\paragraph{From normal terms to types.}
In \cref{lemma:typing:unicity}, we introduced the notations
$\ctxOf{\termA}$ and $\typeOf{\termA}$ respectively for the unique
context and type term such that
$\ctxOf{\termA}\rjug\termA:\typeOf{\termA}$ is derivable.
Given a finite set \(\Gamma\) of sequence variables, 
we moreover write $\NormalTerms(\Gamma) = \NormalValueTerms(\Gamma) \uplus
\NormalBaseTerms(\Gamma) \uplus \NormalBagTerms(\Gamma) \uplus
\NormalStreamTerms(\Gamma)$ for the set of normal resource terms with
free sequence variables in $\Gamma$,
and we define a function
\(\ctxCtx{\Gamma}:\NormalTerms(\Gamma) \to \StreamTypes^\Gamma\),
using the notations of \cref{section:rel}:
\(\ctxCtxOf{\Gamma}{\termA}(\svarA)\eqdef\tuple{\ctxOf{\termA}\pars{\svarA[i]}}[\big]_{i\in\N}\).

\paragraph{From isogmentations to positions.} Recall that any
augmentation $\q \in \Aug(A)$ is based on a configuration $\deseq{\q}
\in \conf(A)$, its \emph{desequentialization}.
In turn, this yields a position by taking the symmetry class.
Composing with the constructions (tensor and \(\vdash\)) of \cref{lem:pos_isos}, we get functions:
\[
\begin{array}{rrcl}
  \deseq{-}_{\mathrm b}^{\Gamma}:&
\Isog_\bullet(\intr{\Gamma} \vdash o)
	&\to& \pos(\U^\N)^\Gamma\\
  \deseq{-}_{\mathrm v}^{\Gamma}:&
\Isog_\bullet(\intr{\Gamma} \vdash \U)
	&\to& \pos(\U^\N)^\Gamma \times \pos_\bullet(\U) \\
  \deseq{-}_{\oc}^{\Gamma}:&
\Isog(\intr{\Gamma} \vdash \U)
	&\to& \pos(\U^\N)^\Gamma \times \pos(\U)\\
  \deseq{-}_{\mathrm s}^{\Gamma}:&
\Isog(\intr{\Gamma} \vdash \U^\N)
	&\to& \pos(\U^\N)^\Gamma \times \pos(\U^\N)\,.
\end{array}
\]

\paragraph{Compatibility.} We are now equipped to state and prove:

\begin{theorem}
The correspondences of \cref{th:bij_pos_reltypes} and
\cref{thm:isog_terms} are compatible.
More formally, the diagrams of \cref{fig:correspondences} commute.
\end{theorem}
\begin{proof}
For any arena \(A\), \(x\in\confOf{A}\) and \(i\in\N\),
we write \(i\cdot x\in\confOf{A^\N}\) for the configuration with the same
events and order, and such that \(\display_{i\cdot x}(a)=\tuple{i,\display_x(a)}\) for each \(a\in\ev{x}\).
This operation obviously preserves symmetry, thus extends to positions.
Moreover, for any \(\svarA\in\Gamma\) and \(\x\in\posOf{A}\),
we write \((\svarA\mapsto\x)\in\posOf{A}^\Gamma\) for the 
\(\Gamma\)-indexed family sending \(\svarA\) to \(\x\)
and each \(\svarB\not=\svarA\) to \(\emptypos\).

A careful analysis of the action on configurations of the constructions
on isogmentations involved in the bijection of
\cref{thm:isog_terms} then yields the following equalities:
\[
\begin{array}{rclcl}
\deseq{\lift_{\svarA,i}(\bq)}_{\mathrm b}^{\Gamma}
&=&\gamma+ \pars{\svarA\mapsto i\cdot (\y \tto o)}
&~&
\text{%
  for $\bq \in \Isog(\sintr{\Gamma} \vdash \U^\N)$,
  with $\deseq{\bq} = \tuple{\gamma, \y}$;
}
\\[5pt]
\deseq{\Lambda_{\sintr{\Gamma}, \svarA}(\bq)}_{\mathrm v}^{\Gamma}
&=& \tuple{\gamma' \setminus \svarA, \gamma'(\svarA) \tto o}[\big]
&&
\text{%
  for $\bq \in \Isog_\bullet(\sintr{\Gamma, \svarA} \vdash o)$, 
  with $\deseq{\bq} = \gamma'$;
}
\\[5pt]
\deseq{\bagpack{\bq_1, \ldots, \bq_k}}_{\oc}^{\Gamma}
&=& \tuple{\sum \gamma_i, \bagpack{\x_i \mid 1 \leq i \leq k}}
&&
\text{%
  for $\bq_i \in \Isog_\bullet(\sintr{\Gamma} \vdash \U)$,
  with $\deseq{\bq_i} = \tuple{\gamma_i, \x_i}$;
}
\\[5pt]
\deseq{\bq \conspack \bp}_{\mathrm s}^{\Gamma}
&=& \tuple{\gamma_1 + \gamma_2, \x \conspack \y}
&& 
\text{%
  for $\bq \in \Isog(\sintr{\Gamma} \vdash \U)$
  and $\bp \in \Isog(\sintr{\Gamma} \vdash \U^\N)$,
}
\\&&&&
\text{%
  with $\deseq{\bq} = \tuple{\gamma_1, \x}$ and $\deseq{\bp} = \tuple{\gamma_2, \y}$;
}
\end{array}
\]
and $\deseq{\iota}_{\mathrm s}^{\Gamma}=\tuple{\tuple{\emptypos}_{\svarB\in\Gamma},\emptypos}$.
Here the sum of \(\Gamma\)-indexed families of positions is defined point-wise,
$\x\in\posOf{\U}$, $\y \in \posOf{\U^\N}$,
$\x_i\in\pposOf{\U}$ for $1\le i\le k$,
$\gamma, \gamma_1, \gamma_2 \in \posOf{\U^\N}^\Gamma$,
$\gamma' \in \posOf{\U^\N}^{\Gamma, \svarA}$, 
and \(\gamma'\setminus\svarA\) denotes the restriction of \(\gamma'\) to \(\Gamma\).

We then show that the diagrams of \cref{fig:correspondences} commute, by
mutual induction on type terms, observing that the above equations match the
rules of the intersection type system in \cref{section:rel:resource}.
\end{proof}

\begin{figure}
\[\begin{array}{cc}
    \begin{tikzcd}[column sep=large] 
        \NormalBaseTerms(\Gamma) 
        \arrow[d,"\ctxCtx{\Gamma}"']
        \arrow[r,"\sintr{-}^{\Gamma}_{\mathrm b}"]
        &
        \Isog_\bullet(\sintr{\Gamma} \vdash o) 
        \arrow[d,"\deseq{-}^{\Gamma}_{\mathrm b}"]
        \\
        \StreamTypes^\Gamma
        \ar[r,"\kappa_{\mathrm s}^\Gamma"']
        &
        \posOf{\U^\N}^\Gamma 
    \end{tikzcd}
    &
    \begin{tikzcd}[column sep=large] 
        \NormalValueTerms(\Gamma)
        \arrow[d,"\tuple{\ctxCtx{\Gamma},\type}"']
        \arrow[r,"\sintr{-}^{\Gamma}_{\mathrm v}"]
        &
        \Isog_\bullet(\sintr{\Gamma} \vdash \U)
        \arrow[d,"\deseq{-}^{\Gamma}_{\mathrm v}"]
        \\
        \StreamTypes^\Gamma \times \ValueTypes
        \ar[r,"\kappa_{\mathrm s}^\Gamma \times \kappa_{\mathrm v}"']
        &
        \posOf{\U^\N}^\Gamma \times \pos_\bullet(\U)
    \end{tikzcd}
    \\
    \begin{tikzcd}[column sep=large] 
        \NormalBagTerms(\Gamma)
        \arrow[d,"\tuple{\ctxCtx{\Gamma},\type}"']
        \arrow[r,"\sintr{-}^{\Gamma}_{\oc}"]
        &
        \Isog(\sintr{\Gamma} \vdash \U)
        \arrow[d,"\deseq{-}^{\Gamma}_{\oc}"]
        \\
        \StreamTypes^\Gamma \times \BagTypes
        \arrow[r,"\kappa_{\mathrm s}^\Gamma \times \kappa_{!}"']
        &
        \posOf{\U^\N}^\Gamma \times \pos(\U)
    \end{tikzcd}
    &
    \begin{tikzcd}[column sep=large] 
        \NormalStreamTerms(\Gamma)
        \arrow[d,"\tuple{\ctxCtx{\Gamma},\type}"']
        \arrow[r,"\sintr{-}^{\Gamma}_{\mathrm s}"]
        &
        \Isog(\sintr{\Gamma} \vdash \U^\N)
        \arrow[d,"\deseq{-}^{\Gamma}_{\mathrm s}"']
        \\
        \StreamTypes^\Gamma \times \StreamTypes
        \arrow[r,"\kappa_{\mathrm s}^\Gamma \times \kappa_{\mathrm s}"']
        &
        \posOf{\U^\N}^\Gamma \times \posOf{\U^\N}
    \end{tikzcd}
\end{array}\]
\caption{Compatibility of the bijections}
\label{fig:correspondences}
\end{figure}

\section{Directions for future work}\label{section:conclusion}

As stated in our introduction, although we consider extensional Taylor
expansion as a valuable contribution in itself, it is but a first step
in a ongoing program to exhibit a precise, compositional relationship between
Taylor expansion and game semantics.

We have made significant progress in the typed case by describing the game
semantics interpretation \(\gs{\vtermA}\) of a resource term \(\vtermA\) in a
compositional way, as is exposed by the structure of a resource
category, which ensures its compatibility with resource
reduction~\cite{BPCVA25}.
Moreover, on \(η\)-long resource terms in normal form, this interpretation
coincides with the simply typed version of the bijection \(\sintr{-}\):
this makes the diagram $(1)$ of \cref{fig:gs:taylor} commute --
the commutation of diagram $(1')$ follows directly, by linearity.
We are convinced that the extensional resource calculus provides the
appropriate syntax and dynamics to adapt this approach in the untyped setting.

\begin{cfigure}
  \begin{center}
  \begin{tikzcd}[
      column sep=.5cm,
      row sep=1cm,
      ]
      \vtermA
      \arrow[|->,"{\NF[\nodelim]{}}"]{rr}
      \arrow[|->,"{\gs[\nodelim]{}}"']{d}
      \arrow[phantom,"(1)"]{rrd}
      &&
      \NF{\vtermA}
      \arrow[phantom,"\simeq" sloped]{d}
      \\
      \gs{\vtermA}
      \arrow[phantom,"="]{r}
      &
      \gs{\NF{\vtermA}}
      \arrow[phantom,"="]{r}
      &
      \sintr{\NF{\vtermA}}_{\mathrm v}
  \end{tikzcd}
  \qquad
  \begin{tikzcd}[
      column sep=1cm,
      row sep=1cm,
      ]
      \ltermA
      \arrow[|->,"{\TayExp[\nodelim]{}}"]{r}
      \arrow[|->,"{\gs[\nodelim]{}}"']{d}
      \arrow[phantom,"(2)"]{rd}
      &
      \TayExp{\ltermA}
      \arrow[|->,"{\NF[\nodelim]{}}"]{r}
      \arrow[|->,"{\gs[\nodelim]{}}"]{d}
      \arrow[phantom,"(1')"]{rd}
      &
      \NFTayExp{\ltermA}
      \arrow[phantom,"\simeq" sloped]{d}
      \\
      \gs{\ltermA}
      \arrow[phantom,"="]{r}
      &
      \gs{\TayExp{\ltermA}}
      \arrow[phantom,"="]{r}
      &
      \sintr{\NFTayExp{\ltermA}}_{\mathrm v}
  \end{tikzcd}
  \end{center}
  \caption{A correspondence between extensional Taylor expansion and game semantics}
  \label{fig:gs:taylor}
\end{cfigure}

The next step, that is the subject of ongoing work, is to check that the
game semantics \(\gs{\ltermA}\) of a typed \(λ\)-term \(\ltermA\) coincides
with that of \(\TayExp{\ltermA}\).
More precisely we aim to establish how, given a resource category with enough
structure (satisfied by the category of games interpreting the resource
calculus), one can construct a cartesian closed category (which, in the case of
games, is the usual interpretation of the \(λ\)-calculus), in such a way that
Taylor expansion describes exactly the relation between both interpretations.
This would make the diagram \((2)\) of \cref{fig:gs:taylor} commute,
altogether establishing that \(\gs{\ltermA}\) and \(\NFTayExp{\ltermA}\)
are isomorphic interpretations.
Again, we believe that extensional Taylor expansion can provide the appropriate
framework to adapt this line of research in the untyped setting.

Another direction for future work is to lift the completeness condition on the
semiring of coefficients, by means of uniformity.
Indeed, ordinary Taylor expansion is amenable to a more general quantitative
setting:
\begin{itemize}
  \item it is easy to check that the definition of \(\TayExp{\ltermA}\)
    uses finite sums of non-zero coefficients only;
  \item it is possible to show, relying on a suitable notion of parallel
    resource reduction, that the compatibility of Taylor expansion
    with \(β\)-reduction does not require infinite sums
    either~\cite{DBLP:journals/lmcs/Vaux19}:
  \item and Ehrhard and Regnier have shown that, due to the uniformity
    properties of the Taylor expansion of pure \(λ\)-terms,
    no sum of coefficients is actually performed during normalization,
    in the computation of
    \(\NFTayExp{\ltermA}\)~\cite{DBLP:journals/tcs/EhrhardR08}.
\end{itemize}

The first observation is easy to reproduce in the extensional
setting, and we are confident that the third one can be adapted,
because the very same uniformity arguments apply
(in fact, we already know that the coefficients are finite,
both in the extensional Taylor expansion of a term and in its normal form,
thanks to \cref{lemma:taylor:coef,lemma:NFT:coef}).
The compatibility with \(β\)-reduction needs more care, though.
Indeed, our simulation results
(\cref{theorem:simulation:beta,theorem:simulation:eta})
rely on unbounded iterations of resource reductions 
(see \cref{section:resourcevectors:reduction}).
This was essential, in particular to deal with the \(η\)-expansion of variables:
\cref{lemma:cc} does require the iterated reduction of possibly created redexes,
and could not be reproduced with a single step of parallel reduction.

Nonetheless, we believe that one could adapt the approach of Cerda, in his
study of Taylor expansion for the infinitary \(λ\)-calculus
\cite[esp. Section 4.4]{cerda:tel-04664728}:
there, he considers a restricted version of parallel reduction,
that is sufficient to capture \(β\)-reduction,
yet preserves uniformity,
which allows him to safely consider iterated reductions
and normalization, without completeness hypothesis.

\printbibliography

\end{document}